%% file: main.tex
\documentclass[11pt]{article}
\usepackage{amsmath,amssymb,amsthm,amsfonts}
\usepackage{latexsym}
\usepackage{epsfig}
\usepackage[right=0.8in, top=1in, bottom=1in, left=0.8in]{geometry}
\usepackage{setspace}
\usepackage{bbold}
\usepackage{mathtools}
\usepackage{dirtytalk}
\usepackage{enumerate}
\usepackage{color,soul}
\usepackage{stackengine}
\usepackage{algpseudocode}
\usepackage{algorithm}
\usepackage{multicol}
\usepackage{extarrows}
\usepackage{helvet}
\usepackage{array}
\usepackage{bigstrut}

\usepackage{scalerel,stackengine}
\stackMath
\newcommand\reallywidehat[1]{%
\savestack{\tmpbox}{\stretchto{%
  \scaleto{%
    \scalerel*[\widthof{\ensuremath{#1}}]{\kern-.6pt\bigwedge\kern-.6pt}%
    {\rule[-\textheight/2]{1ex}{\textheight}}%
  }{\textheight}%
}{0.5ex}}%
\stackon[1pt]{#1}{\tmpbox}%
}

\def\compactify{\itemsep=0pt \topsep=0pt \partopsep=0pt \parsep=0pt}
\let\latexusecounter=\usecounter
\newenvironment{enumerate*}
  {\def\usecounter{\compactify\latexusecounter}
   \begin{enumerate}}
  {\end{enumerate}\let\usecounter=\latexusecounter}

\newcommand{\N}{{\mathbf{N}}}

\newcommand{\R}{{\mathbb{R}}}

\newcommand{\eps}{\epsilon}
\newcommand{\tO}{\tilde{O}}

\newcommand{\piComp}{{\overleftarrow{\pi}}}

\newcommand{\SC}{{\mathcal{S}}}

\newcommand{\RC}{{\mathcal{R}}}
\newcommand{\XC}{{\mathcal{X}}}
\newcommand{\YC}{{\mathcal{Y}}}
\newcommand{\TC}{{\mathcal{T}}}
\newcommand{\one}[1]{{\mathbb 1}\left[{#1}\right]}
\newcommand{\Sw}{S_w}
\newcommand{\wbase}{\gamma}

\newcommand{\A}{{\mathcal{A}}}

\newtheorem{theorem}{Theorem}[section]
\newtheorem*{theorem*}{Theorem}

\newtheorem{corollary}[theorem]{Corollary}
\newtheorem{lemma}[theorem]{Lemma}

\newtheorem{definition}[theorem]{Definition}
\newtheorem{claim}[theorem]{Claim}
\newtheorem{fact}[theorem]{Fact}

\makeatletter
\newtheorem*{rep@theorem}{\rep@title}
\newcommand{\newreptheorem}[2]{%
\newenvironment{rep#1}[1]{%
 \def\rep@title{#2 \ref{##1}}%
 \begin{rep@theorem}}%
 {\end{rep@theorem}}}
\makeatother
\newreptheorem{theorem}{Theorem}

\DeclareMathOperator{\E}{{\mathbb E}}

\newcommand{\norm}[1]{\left\lVert#1\right\rVert}
\newcommand{\normo}[1]{\norm{#1}_1}

\newcommand{\un}{{\ensuremath {\mathbf u}}}
\newcommand{\cc}{{\mathbf c}}
\newcommand{\basec}{{\mathfrak c}}

\usepackage[pdftex,dvipsnames]{xcolor}  %
\usepackage[utf8]{inputenc}
\usepackage{amsmath}
\usepackage{hyperref}
\usepackage{cryptocode}   
\graphicspath{ {img/} }
\usepackage[colorinlistoftodos,prependcaption,textsize=tiny]{todonotes}
\newcommandx{\unsure}[2][1=]{\todo[linecolor=red,backgroundcolor=red!25,bordercolor=red,#1]{#2}}
\newcommandx{\change}[2][1=]{\todo[linecolor=blue,backgroundcolor=blue!25,bordercolor=blue,#1]{#2}}
\newcommandx{\info}[2][1=]{\todo[linecolor=OliveGreen,backgroundcolor=OliveGreen!25,bordercolor=OliveGreen,#1]{#2}}
\newcommandx{\improvement}[2][1=]{\todo[linecolor=Plum,backgroundcolor=Plum!25,bordercolor=Plum,#1]{#2}}
\newcommandx{\thiswillnotshow}[2][1=]{\todo[disable,#1]{#2}}

\newcommand{\ns}[1]{}
\newcommand{\aanote}[1]{}

\newcommand{\remove}[1]{}

\DeclareMathOperator{\Poi}{Poi}

\DeclareMathOperator{\den}{{d}}
\DeclareMathOperator{\relden}{{rd}}

\newcommand{\totalmu}{n}

\newcommand{\normalizedEllOne}[1]{\tfrac{#1}{\|#1\|_1}}

\def\poly{\operatorname{poly}}

\renewcommand{\bf}{\normalfont \bfseries}
\renewcommand{\sc}{\normalfont \scshape}

\renewcommand{\paragraph}[1]{\vspace{0.7ex}\noindent{\bf #1}}

\newcommand{\hkappa}{{\widehat{\kappa}}}

\newcommand{\NC}{\mathcal{N}}

\newcommand{\PC}{\mathcal{P}}
\newcommand{\CC}{\mathcal{C}}
\newcommand{\CD}{\mathcal{D}}

\newcommand{\Kappa}{\mathit{K}}

\newcommand{\WC}{\mathcal{X}}
\newcommand{\HC}{\mathcal{Y}}
\newcommand{\IC}{\mathcal{I}}
\newcommand{\JC}{\mathcal{J}}
\newcommand{\VC}{\mathcal{V}}

\newcommand{\QC}[0]{\mathcal{Q}}
\newcommand{\EC}[0]{\mathcal{E}}

\newcommand{\BC}{\mathcal{B}}

\newcommand{\cad}{\mathtt{ad}}
\newcommand{\ed}{\mathtt{ed}}
\newcommand{\distG}[0]{\mathfrak{D}}
\DeclareMathOperator{\dist}{dist}
\DeclareMathOperator{\dd}{\mathtt{dd}}

\renewcommand{\d}{{\mathsf{d}}}

\newcommand{\1}[0]{\mathbb{1}}

\title{
  Edit Distance in Near-Linear Time: it's a Constant
  Factor\footnote{An extended abstract of this submission appeared in
    the {\em Proceedings of the 61st Annual IEEE Symposium on Foundations
    of Computer Science}. Research supported in part by NSF grants (CCF-1617955 and CCF-1740833),
and Simons Foundation (\#491119). Also, this research was supported in part by a grant from the Columbia-IBM center for Blockchain and Data Transparency, and by JPMorgan Chase \& Co. Any views or opinions expressed herein are solely those of the authors listed, and may differ from the views and opinions expressed by JPMorgan Chase \& Co. or its affiliates.}}
\author{Alexandr Andoni\\Columbia University\\\texttt{andoni@cs.columbia.edu} \and Negev Shekel Nosatzki\\Columbia University\\\texttt{ns3049@columbia.edu}}

\begin{document} 

\maketitle

\begin{abstract}
We present an algorithm for approximating the edit distance between
two strings of length $n$ in time $n^{1+\eps}$ up to a constant
factor, for any $\eps>0$. Our result completes a research direction
set forth in the recent breakthrough
paper~\cite{chakraborty2018approximating}, which showed the first
constant-factor approximation algorithm with a (strongly)
sub-quadratic running time.  The recent results
\cite{koucky2019constant, brakensiek2019constant} have shown
near-linear time algorithms that obtain an additive approximation,
near-linear in $n$ (equivalently, constant-factor approximation when
the edit distance value is close to $n$). In contrast, our algorithm
obtains a constant-factor approximation in near-linear time for any
input strings.

In contrast to prior algorithms, which are mostly recursing over
smaller substrings, our algorithm gradually smoothes out the local
contribution to the edit distance over progressively larger
substrings. To accomplish this, we iteratively construct a distance
oracle data structure for the metric of edit distance on all
substrings of input strings, of length $n^{i\epsilon}$ for
$i=0,1,\ldots,1/\eps$. The distance oracle approximates the edit
distance over these substrings in a certain average sense, just enough
to estimate the overall edit distance.

\end{abstract}

\thispagestyle{empty}
\newpage
\setcounter{page}{1}

\newpage
{\hypersetup{linkcolor=black}
\tableofcontents
}
\newpage

\input{intro}

\input{setup}

\input{technical}

\input{algo}

\input{imperfect-N}

\input{imp-correctness.tex}

\input{imp-complexity.tex}

\input{triangle}

\input{fastSampling.tex}

\bibliography{main,bibfile}
\bibliographystyle{alpha}

\end{document}

%% file: intro.tex
\section{Introduction}

Edit distance is a classic distance measure between sequences that
takes into account the (mis)alignment of strings. Formally, edit
distance between two strings of length $n$ over some alphabet $\Sigma$
is the number of insertions/deletions/substitutions of characters to
transform one string into the other.
Being of
key importance in several fields, such as computational biology and
signal processing, %
computational problems
involving the edit distance were studied extensively.

Computing edit distance is also a classic dynamic programming problem, with
a quadratic run-time solution. It has proven to be a poster challenge
in a central theme in TCS: %
improving the run-time from polynomial %
towards close(r) to linear. Despite significant research attempts 
over many decades,
 little progress was obtained,
with a $O(n^2/
\log^2 n)$ %
run-time algorithm \cite{MP80} remaining the fastest %
one 
known to
date. See also the surveys of \cite{Navarro01} and
\cite{Sah-Encyclopedia}. With the emergence of the fine-grained
complexity field, researchers crystallized the reason why beating
quadratic-time is hard by connecting it to the Strong Exponential Time
Hypothesis (SETH) \cite{BI15-edit} (and even more plausible
conjectures \cite{abboud2016simulating}).

Even before the above hardness results, researchers started
considering faster algorithms that approximate edit distance. A
linear-time $\sqrt{n}$-factor approximation follows immediately from
the exact algorithm of \cite{ukkonen1985algorithms, Myers86, LMS98}, which
runs in time $O(n + d^2)$, where $d$ is the edit distance between the
input strings. Subsequent research improved the approximation factor,
first to $n^{3/7}$ \cite{BJKK04}, then to $n^{1/3+o(1)}$ \cite{BES06},
and to $2^{\tilde O(\sqrt{\log n})}$ \cite{AO-edit} (based on the $\ell_1$ embedding of~\cite{OR-edit}). In the regime of
$O(n^{1+\eps})$-time algorithms, the best approximation is $(\log
n)^{O(1/\eps)}$ \cite{AKO-edit}. Predating some of this work was the
sublinear-time algorithm of \cite{BEK+03} achieving $n^\eps$
approximation when $d$ is large.

In a recent breakthrough, \cite{chakraborty2018approximating} showed
that one can obtain constant-factor approximation in
$O(n^{12/7})$ time. Subsequent developments \cite{koucky2019constant,
  brakensiek2019constant} give $O(n^{1+\eps})$-time algorithms for
computing edit distance up to an {\em additive $n^{1-g(\eps)}$ term}
and $f(1/\eps)$-factor approximation, for some non-decreasing
functions $f,g$, and any $\eps>0$.

\paragraph{Our main result} %
is a $n^{1+\eps}$ algorithm for
computing the edit distance %
up to a constant approximation.

\begin{theorem}
  \label{thm:main}
  For any $\eps>0$, $n\ge 1$, and alphabet $\Sigma$, there is a randomized algorithm that, given two
  strings $x,y\in \Sigma^n$, approximates the edit
  distance between $x$ and $y$ in $O(n^{1+\eps})$ time up to
  $f(1/\eps)$-factor approximation, where $f(1/\eps)$ depends solely
  on $\eps$.
\end{theorem}

While we do not derive the function $f(1/\eps)$ explicitly, we note
that it is doubly exponential in $1/\eps$. We present a technical
overview of our approach in Section~\ref{sec:technical}, after setting
up our notations in Section~\ref{sec:setup}. The proof of the main
theorem will follow in subsequent sections, in particular the
top-level algorithm and its main guarantees are in
Section~\ref{sec:topAlgo}.

\subsection{Related work}

A quantum algorithm for edit distance was introduced in
\cite{boroujeni2018approximating}. Some of the basic elements of the
algorithmic approach are related to
\cite{chakraborty2018approximating} (and the algorithm in this
paper). Another recent related paper is \cite{GRS19}, who obtain
$3+\eps$ approximation in $\tilde O(n^{1.6})$ time; independently, the
first author obtained a slightly worst time for the same
approximation~\cite{andoni18-edit3}. Similarly, independently,
\cite{chakraborty2018online} and \cite{andoni18-edit3} extended
the constant-factor edit distance algorithm from
\cite{chakraborty2018approximating} to solve the text searching
problem.

Sublinear time algorithms have drawn renewed
attention~\cite{goldenberg2019sublinear, kociumaka2020sublinear,
  brakensiek2020simple, bringmann2022almost, gkks22}; see also
earlier~\cite{BEK+03, BJKK04, AO-edit}. Another related line of work has
been on computing edit distance for the semi-random models of
input~\cite{AK-smoothed, kuszmaul2019efficiently}. Parallel (MPC)
algorithms were developed
in~\cite{boroujeni2018approximating,hajiaghayi2019massively}.

Progress on edit distance algorithms also inspired the first
non-trivial algorithms for approximating the longest common
subsequence (LCS)
\cite{hajiaghayi2019approximating,rubinstein2019approximation,rubinstein2020reducing,
bringmann2021linear},
     leading to a linear time, $n^{o(1)}$-approximation algorithm \cite{andoni2021estimating,
nosatzki2021approximating}. Also of note is \cite{rubinstein2020reducing} which shows that a $O(1)$-factor
approximation to edit distance implies a $2-\Omega(1)$ factor approximation
to LCS over a binary alphabet in (essentially) the same time.

\subsection{Acknowledgements}
We would like to thank FOCS and SICOMP anonymous reviewers for helpful comments and suggestions that helped improve this paper.

%% file: setup.tex
\section{Preliminaries: Setup and Notations}
\label{sec:setup}

Fix a pair of strings $(x,y) \in \Sigma^n \times \Sigma^n$ for which
we care to estimate the edit distance. 
We define $\ed_n(x,y)$ as half the number
of insertions/deletions to transform one string into the other. Note
that the standard edit
distance (allowing substitutions) can be reduced to this case (see, e.g.,~\cite{tiskin2008semi}). When length $n$ is clear from the context, we omit the subscript.

$\Sw$ is the set of integer powers of 2 up to $w$: namely, $\Sw=\{1,2,4,8\ldots
w\}\cup \{1/2,1/4\ldots 1/2n\}$.

$[n]$ denotes the set $\{1,2,3,\ldots n\}$
throughout the paper except where stated explicitly otherwise
(notably, in Section~\ref{sec:CADalgo}).

When describing intuitive parts, we sometimes use $O^*(f(n))$ to
denote $f(n)\cdot n^{O(\eps)}$ (where $\eps$ is the small constant
from the algorithm).

Finally, we use the standard notion of {\em with high probability}
(whp), meaning with probability at least $1-n^{-C}$ for large enough
constant $C>1$.

\subsection{Intervals}

An interval is a substring $x[i:j]\triangleq x_ix_{i+1}\ldots
x_{j-1}$, for $i,j\in[n]$, where $i\le j$ (i.e., starting at $i$ and
ending at $j-1$, of length $j-i$).

For $i \in [n]$, let $X_{i,w}$ ($Y_{i,w}$) denote the interval of $x$
($y$) of length $w$ starting at position $i$. Let $\WC_w, \HC_w$ the
set of all such $X_{i,w}$ and $Y_{j,w}$ strings respectively. We use
$\IC_w = \WC_w \cup \HC_w$ to denote all $x$ and $y$ axis
intervals. When clear from context, we drop subscript $w$.

By convention, if $i\not\in [1,n-w]$, we pad $X_i/Y_i$ with a default
character, say, \$. Also $Y_{\bot,w}$ is a string of unique
characters. In particular, for all distance functions $\tau_w(\cdot,
\cdot)$ on two length-$w$ strings in this paper, we define
$\tau_w(X_{i,w}, Y_{\bot,w})=w$; e.g., $\ed_w(X_{i,w},Y_{\bot,w})=w$.

Usually, by $I\in \IC_w$ we refer not only to the corresponding
substring but also to the ``meta-information'', in particular the string
it came from, start position, and length (e.g., for $I=X_{i,w}$, the
meta-information is $x,i,w$). This difference will be clear from
context or stated explicitly.

In particular, the notation
$I+j$, for an interval $I$ and integer $j$,
represents the interval $j$ positions to the right; e.g.,  if
$I=X_{i,w}$, then $I+j=X_{i+j,w}$.

\paragraph{Alignments.} An alignment between $x$ and $y$ %
is a function $\pi:[n]\to
[n]\cup \{\bot\}$, which is injective and strictly monotone on
$\pi^{-1}([n])$. The set of all such alignments is called $\Pi$. Note
that $\ed(x,y)=\min_{\pi\in \Pi} \sum_{i\in[n]} \ed_1(x_i,y_{\pi(i)})$
(recall that, by convention, $\ed_1(c,y_\bot)=1$ for all $c\in
\Sigma$).

It is convenient for us to think of $\pi$ as function from $\IC
\rightarrow \IC$, via the following extension.  For a given input
alignment $\pi: \XC \rightarrow \YC \cup \{ \bot \}$, its extension
$\widehat{\pi}: \IC \rightarrow \IC \cup \{ \bot \}$ is:

$$\widehat{\pi}[I] = \begin{cases}
\pi[I] & I \in \XC \\
\pi^{-1}[I] & I \in \YC
\end{cases},$$
where $\widehat\pi[X_{i,w}]$ means $Y_{\pi(i),w}$, and $\widehat\pi[Y_{j,w}]$
means $X_{\pi^{-1}(j),w}$, with $\pi^{-1}(j)=\bot$ if there's no $i$
with $\pi(i)=j$.
Throughout this paper, we overload notation to use $\pi$ for the
extension $\widehat\pi$ as well. We also define $\piComp(i)$ as the minimum $\pi(j)$, $j\ge i$, which is
defined ($\neq \bot$).

Finally, we also define $\pi(i)\triangleq i$ when $i< 1$ and $i>n$ for convenience.

\subsection{Interval distances}

Our algorithms will use distances/metrics over intervals in
$\IC_w$. One important instance is the {\em alignment distance},
denoted $\cad_w(\cdot,\cdot)$. At a high level, $\cad_w(\cdot,\cdot)$ is a
distance metric that approximates edit distance on length-$w$
intervals. We discuss $\cad(\cdot,\cdot)$ metric in
Section~\ref{sec:topAlgo} as well as~\ref{sec:CADalgo}.

\begin{definition}[Neighborhood]
Fix $c \geq 0$ and $I \in \IC_w$. The {\em $c$-neighborhood} of $I$ is
the set $\NC_c(I) = \{ J \in \IC_w \mid \cad(I,J) \leq c \}$, i.e. all
$x$ and $y$ intervals which are $c$-close to $I$ in terms of their
alignment distance.
\end{definition}

\begin{definition}[Ball of intervals]
A {\em ball of intervals} is a
set of consecutive intervals in either $\XC_w$ or $\YC_w$ (i.e., it's
a ball in the metric where distance between $X_i$ and $X_j$ is
$|i-j|$). The {\em smallest enclosing ball} of a set $\SC$ is the
minimal ball $\BC \supseteq \SC$.
\end{definition}

\paragraph{Average approximation for an optimal alignment.}
A common theme in our algorithm is constructing metrics on $\IC_w$
approximating $\ed_w$ in a certain ``average
sense''. In particular, this differs from the standard notion of
approximation in that the upper bound holds only on average, and for an
optimal alignment $\pi\in \Pi$. Formally, we define:

\begin{definition}[Align-approximation]
  \label{def:alignApp}
	Fix space $(\IC_w,\d)$ over $\IC_w$ (which is
        often a metric space, but need not be). We
        say $\d$ {\em $T$-align-approximates} $\ed$ if the following holds:
	\begin{enumerate}
		\item For all $I,J \in \IC_w$: $\d(I,J)\ge\ed_w(I,J)$.
		\item 
		$
\min_{\pi\in \Pi}\sum_{i\in [n]}
    \tfrac{1}{w}\d(X_{i},Y_{\pi(i)})\le T\cdot \ed(x,y).
    $

	\end{enumerate}
\end{definition}

\subsection{Operations on sets and the $*$ notation}

By convention, applying numerical functions to a set refers to the sum
over all set items; e.g., $f(S)=\sum_{i\in S} f(i)$.  When applying
set operators on other sets, we use the union; e.g., $\pi(\SC) =
\cup_{I \in \SC} \pi(I)$ and $\NC_c(\SC) = \cup_{I \in \SC}
\NC_c(I)$. Abusing notation, we use $f(S)$ even when $S$ is not fully
contained in the domain of $f$ --- in which case, we simply ignore
elements outside the domain. Any exception to the above will be
clearly specified.

We also use the notation $*$ as argument of a function, by which we
mean a vector of all possible entries. E.g., $f(*)$ is a vector of
$f(i)$ for $i$ ranging over the domain of $f$ (usually clear from the
context). Similarly, $f(*_\RC)$ means a vector of $f(i)$ for $i$
satisfying property $\RC$. Overloading notation, sometimes $f(*_\RC)$
will also mean a vector of $f(i)$ for all $i$ in the domain, with
coordinates $i\notin\RC$ being zeroed-out.

%% file: technical.tex
\section{Technical Overview}
\label{sec:technical}

\subsection{Prior work and main obstacles}
\label{sec::prior_obstacles}
As our natural starting point is the breakthrough $O(n^{12/7})$-time
algorithm of \cite{chakraborty2018approximating} (and related~\cite{boroujeni2018approximating}), we first describe
their core ideas as well as the challenges to obtaining a near-linear
time algorithm.  In particular, we highlight two of their enabling
ideas. At a basic level, their algorithm computes edit distance
$\ed_n(x,y)$ by computing $\ed_w$ between various length-$w$ {\em
  intervals} (substrings) of $x,y$ recursively, and then uses
edit-distance-like dynamic programming on intervals to put them back
together. The main algorithmic thrust is to reduce the number of
recursive $\ed_w$ computations: e.g., if the intervals are of length
$w$, and we only consider non-overlapping intervals, there are still
$n/w\times n/w$ calls to $\ed_w$, each taking at best $\Omega(w)$
time.  Hence, \cite{chakraborty2018approximating} employ two ideas to
do this efficiently: 1) use the triangle inequality to deduce distance
between pairs of intervals for which we do not directly estimate
$\ed_w$, 2) two nearby $x$-intervals (e.g., consecutive) are likely to
be matched into two nearby $y$-intervals (also consecutive) under the
optimal edit distance alignment $\pi$. (The earlier quantum
result~\cite{boroujeni2018approximating} employed the first idea
already, but relied on a quantum component instead of the second
idea.) Indeed, these ideas are enough to
reduce the number of recursive calls from $(n/w)^2$ to $\approx
(n/w)^{1.5}$.

One big challenge in the above is that, in general, one has to
consider all, overlapping intervals from $x,y$, of which there are $n$
--- since, in an optimal $\ed_n$ alignment, an $x$-interval might have
to match to a $y$-interval whose start position is far from an integer
multiple of $w$. An alternative perspective is that if one considers
only a restricted set of interval start positions, say every $s\le w$
positions in $y$, then one obtains an extra {\em additive} error of
about $s\cdot n/w$ from the ``rounding'' of start positions in
$y$. That's the reason that a bound of $(n/w)^{1.5}$ recursive calls
did not transform into $n^{1.5}$ runtime
in~\cite{chakraborty2018approximating}: to compute edit distance when
$\ed<n^{1-\Omega(1)}$, they employ a standard (exact) $\tilde
O(n+\ed^2(x,y))$ algorithm~\cite{ukkonen1985algorithms, Myers86}.

Recent improvements by \cite{koucky2019constant,
  brakensiek2019constant} showed how to reduce the number of recursive
calls to $\approx n/w$, but some fundamental obstacles remained. The
linear number of recursive calls was leveraged to obtain near-linear
time but with an additive approximation only: when $\ed(x,y)\ge
n^{1-\delta}$, the overall runtime is $n^{1+f(\delta)}$ for some
increasing function $f$.

In particular, in addition to the aforementioned challenge, a new
challenge arose: to be able to reduce to near-linear number of
recursive $\ed_w$ calls, the algorithms from \cite{koucky2019constant,
  brakensiek2019constant} might miss a large fraction of ``correct''
matches. In particular this fraction is $\approx n^{-\delta}$, which
results in an additive error of $\approx n^{1-\delta}$.  To put this
into perspective, for $w=\sqrt{n}$, if we allow an additive error
$n^{1-\delta}$, then it suffices to analyze $b=n^{0.5 + O(\delta)}$
intervals (which barely overlap) and misclassify $b \cdot n^{-\delta}$
of them. We use the following example to showcase the challenges: 

\vspace{3mm}

\noindent
\framebox{\parbox{\dimexpr\linewidth-2\fboxsep-2\fboxrule}{\itshape%
\textbf{A running example illustrating the challenges}

Consider an instance
where $\Delta = n^{-0.01}$ fraction of intervals in $\XC_w$ are
``sparse'' --- have a single cheap match in $\YC_w$ --- and the rest
of the intervals are {\em dense} (they are close to many other
intervals). Assume further that such sparse intervals are spread
around in multiple ``sparse sections'' (sequences of consecutive
sparse intervals). }}

\vspace{3mm}

Note that if we can afford large additive errors,
we can simply {\em ignore} all these sparse intervals (certifying them
at max cost $w$) and output the distance based on the dense intervals
only, with at most $n\Delta=n^{0.99}$ additive approximation. To avoid
this, one must first identify some sparse intervals (since the dense
intervals do not provide sufficient information about the sparse
sections). Even if we manage to find some of the sparse intervals
efficiently, we still need to apply knowledge of the location of such
intervals to deduce information on other intervals which might be in
completely different areas in the string. 
We will return to this
example later.

Below we describe the high-level approach to our algorithm, including
how we overcome these obstacles. We note that, while our algorithm is
based on the two key ideas from~\cite{chakraborty2018approximating},
the high-level algorithm departs from the general approach undertaken
in \cite{chakraborty2018approximating, koucky2019constant,
  brakensiek2019constant}. That said, some algorithmic steps are
similar to those developed in~\cite{koucky2019constant,
  brakensiek2019constant}. We do not rely on previous results (for any
distance regime) such as~\cite{ukkonen1985algorithms}.

\subsection{Our high-level approach}

While there are many ideas going in overcoming the above challenges,
one common theme is {\em averaging over the local proximity of
  intervals}. In particular, the algorithm proceeds by, and analyzes
over, ``average characteristics'' of various intervals of $x,y$, in a
``smooth'' way. For example decisions for a fixed interval
$I\in\IC_w$, such as whether something is close, or something is
matched, are done by considering the statistics collected on nearby
intervals (to the left/right of $I$ in the corresponding
string). While we expand on our technical ideas below, this is the
guiding principle to keep in mind.

Addressing the first challenge, we consider intervals (of fixed length
$w$) at all $n$ starting positions, i.e., the entire set $\IC_w$. Note
that recursion becomes prohibitive: we can't perform even $n$ edit
distance evaluations each taking $\Omega(w)$ time ($w$ is set to be
$\approx n^{1-\eps}$). Instead, our top-level algorithm iterates
bottom--up over all interval lengths $w=\wbase, \wbase^2,\ldots n$,
where $\wbase=n^\eps$, and for each $w$ computes a good-enough
approximation to {\em the entire metric} $(\IC_w, \ed_w)$. Recall that
$\IC_w=\WC_w\cup \HC_w$ consists of all $w$-length intervals
(substrings); i.e., $|\IC_w|=2n$. The metric, termed
$\distG_w(\cdot,\cdot)$, will be accessible via a distance oracle
(fast data structure), with $n^\eps$ query time, and will approximate
the distance between most of the pairs (an $x$-interval, $y$-interval)
that participate in an optimal alignment in an average
sense. Specifically, for $\pi$ ranging over all alignments, we will have
that $\min_\pi \tfrac{1}{w}
\distG_w(X_{i,w},Y_{\pi[i],w})=\Theta(\ed(x,y))$. Formally, we say
$\distG_w$ {\em $O(1)$-align-approximates $\ed$} (see
Def.~\ref{def:alignApp}).

In each iteration, we build $\distG_w$ using $\distG_{w'}$, where
$w'=w/\wbase$. Conceptually we do so in two phases. First, we build
another metric on $w$-length strings, $(\IC_w,\cad_w)$,  
accessible via a fast distance oracle, that uses
$\gamma^{O(1)}=n^{O(\eps)}$ time and $\distG_{w'}$ oracle
calls. Crucially, $\cad_w$ will similarly align-approximate
$\ed$. Second, equipped with a fast oracle for $\cad_w$ (itself
using $\distG_{w'}$), we build an ``efficient representation'' for the
entire metric $(\IC_w,\cad_w)$, while using only $n^{1+O(\eps)}$ calls
to $\cad_w$ oracle. Naturally, this ``efficient representation'' will
not be able to capture the entire $\cad_w$ metric (that would require
$\gg n$ query complexity), but it will capture just enough to preserve
the edit distance between $x$ and $y$---again, formally,
align-approximate $\ed$. Then we build an efficient distance
oracle for this efficient representation, which will yield the desired
metric $\distG_w$. Note that the final approximation to $\ed(x,y)$ is
computed by (essentially) querying $\distG_n(X_{1,n}, Y_{1,n})$.

In particular, the ``efficient representation'' of $\cad_w$ is a
weighted graph $G_w$ with vertex set $\IC_w$ and $n^{1+\eps}$ edges,
such that the shortest path between $I,J\in\IC_w$ approximates
$\cad_w(I,J)$, again, in an average sense for an optimal alignment. In
particular, the shortest path distance is non-contracting, and
non-expanding for interval pairs that ``matter'', i.e., which are part
of the optimal alignment $\pi$ corresponding to $\ed(x,y)$. An edge
$(I,J)$ of the graph $G_w$ will always correspond to an explicit call
to $\cad_w(I,J)$; and the main question in constructing $G_w$ is
deciding {\em which} $n^{1+\eps}$ pairs to compute $\cad_w$ for.

Once we have the graph $G_w$, we build a fast distance oracle data
structure on it to obtain the metric $\distG_w$. In particular, our fast
distance oracle is merely an embedding of the shortest path metric on $G_w$ into
$\ell_\infty^{d}$, where $d=|\IC|^{\eps}$, incurring an approximation
of $O(1/\eps)$, via \cite{Matousek1996}. We note that we cannot use
some other common distance oracle, such as, e.g.,
\cite{Thorup:2005:ADO:1044731.1044732,Chechik:2014:ADO:2591796.2591801},
because they do not guarantee that the resulting output is actually a
metric, and in particular, that it satisfies the triangle inequality,
which is crucial for us (as mentioned above). We remark that
this particular step is somewhat reminiscent of the approach from~\cite{AO-edit},
who similarly build an efficient representation for the metric
$(\IC_w,\ed_w)$ using metric embeddings. However, the similarity
ends here: first \cite{AO-edit} used Bourgain's embedding into $\ell_1$,
which incurs $\Theta(\log n)$ distortion, and second, more
importantly, the construction of $G_w$ was altogether different
(incurring a much higher approximation).

Computing the graph $G_w$ itself is the most algorithmically novel
part of our approach, and is termed {\em Interval Matching Algorithm},
as it corresponds to matching intervals that are close in $\cad_w$
distance.  This algorithmic part should be thought of as the analogue
of the algorithm deciding for which pairs of intervals to
(recursively) estimate the edit distance
in~\cite{chakraborty2018approximating}.

We sketch the Interval Matching Algorithm next in this technical
overview. We also sketch how to compute the $\cad_w$ distance in
$n^{O(\eps)}$ time, which presents its own new challenges, especially
to guarantee its metric properties.

\subsection{Interval matching algorithm}
\label{sec:toMatching}

The main task here is to efficiently compute a graph $G_w$ that
approximates $\cad_w$, in an average sense over an optimal alignment
$\pi$.
Specifically, to generate $G_w$, we iterate
  over all costs $c \in S_w$ (powers of 2), and for
  each such cost, we then generate a (sub-)graph $G_{w,c}$, with edges
  of weight $\Theta(c)$. The following is the main guarantee: for any
  fixed alignment $\pi$, for any
  pair $(I,\pi[I])\in \IC_w^2$ at a distance $\cad(I,\pi[I])\le c$, we generate a
  1- or 2-hop path for it in $G_{w,c}$---except for $O(k_c)$ such
  pairs $(I,\pi[I])$ where $k_c$ is the number of pairs with
  $\cad(I,\pi[I])> c$ (i.e., the ``error'' increases by at most a
  constant factor). The union\footnote{When there are multiple 
    edges $(I,J)$ from different $c$'s, we naturally take the minimum-weight
    edge---i.e., the smallest distance certificate.}  of such graphs
  $G_{w,c}$ yields the final graph $G_w$. Below we focus on a single
  scale graph $G_{w,c}$, which is supposed to capture nearly all pairs
  $(I,\pi[I])$ where $\cad_w(I,\pi[I])\le c$. We refer to such a pair as
  a {\em $\pi$-matchable pair} $(I,\pi[I])$.

At its core, our algorithm can be thought of as a {\em partitioning}
algorithm, where we partition $\IC_w$ into sets of intervals, such
that for nearly all $\pi$-matchable pairs $(I,\pi[I])$, both intervals
belong to the same set. We start with a single set of intervals and we
iteratively partition the set into progressively more refined
partitions (consisting of smaller parts), with the goal of keeping
$\pi$-matchable pairs $I,\pi[I]$ together. (This algorithm will use
a significant amount of notation, and, while this high-level overview
will mention only a fraction of them, the reader may refer to the
Table~\ref{tab:notation-new} in Sec.~\ref{sec::match-correctness} for some important definitions and
formulas.)

In particular, the matching algorithm proceeds in $\approx 1/\epsilon$
steps. In each step $t$, for $\lambda=n^\eps$, we generate $\lambda^t$
parts, each of size $\lesssim n/\lambda^{t}$. To construct a part, we sample a random interval $A$, termed
{\em anchor}, and estimate $\cad_w(A,I)$ for all other intervals
$I$ in its part, generating a {\em cluster} of
intervals at distance $O(c)$ from the anchor. The main desideratum is
that the two intervals from a $\pi$-matchable pair $I,\pi[I]$ are either
both close to $A$ or both far from $A$, and hence always remain together (this is
related to the triangle inequality idea
from~\cite{chakraborty2018approximating}). However, this cannot be
guaranteed, and ensuring this desideratum is a major challenge for us,
which we will address later. For now, in order to build intuition, we
first develop our ideas under the following
the assumption, that the desideratum holds:
\begin{center}
{\em Perfect Neighborhood Assumption} (PNA): \\any two intervals are at
distance either $\le c$ or $\omega(c)$; hence %
$\NC_{O(c)}(I)=\NC_{c}(I)$.
\end{center}

Every anchor will generate precisely one part, of target size. Notice that if the cluster is sufficiently large (i.e., $\gtrsim n/\lambda^{t}$), we are basically done, and in fact, is where our algorithm ``converges'' (as will be described later). Otherwise, 
we use the cluster to construct one part (set) by taking the clustered intervals together
with their {\em local extensions}: intervals around the clustered intervals
(to the left/right of the clustered ones). The parameters are
set up such that the resulting part has size $\lesssim n/\lambda^{t}$.
Note that the iterative nature of the process helps ensure the runtime:
As the  
size of parts  decreases with step $t$, we can afford to use
more anchors. In particular, at step $t$, we start with $\lambda^{t-1}$
partitions, each of size about $n/\lambda^{t-1}$, and hence, for each
of $\lambda^t$ anchors, we need to estimate $\cad$ distance to
$n/\lambda^{t-1}$ intervals (in its part), for an overall of $n\lambda$ distance
computations.

\begin{figure}[h!]
\centering
\includegraphics[width=0.8\textwidth]{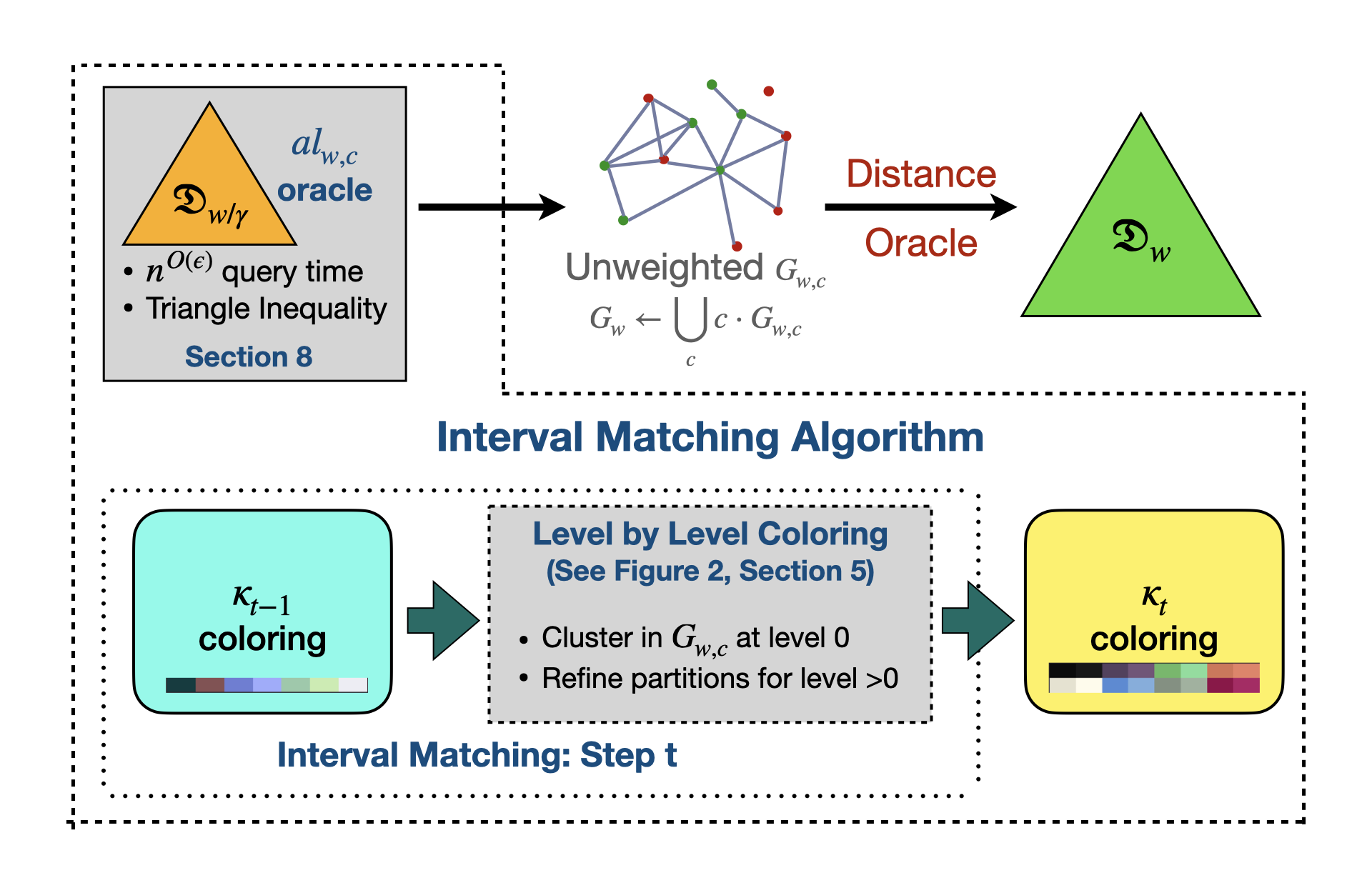}
\caption{High Level Algorithm Scheme (Section \ref{sec:topAlgo})}
\end{figure}

A direct implementation of partitions as above however runs in various
issues, yielding additive errors. In particular, it only guarantees to
correctly partition ``any fixed $\pi$-matchable pair with some
probability''; instead of the needed ``with some probability, all
except a few $\pi$-matchable pairs are partitioned correctly'' (akin
to the ``for each'' vs ``for all'' guarantees). For the latter goal,
bounding the ``except a few'' so that it's only a $O(1)$-factor
approximation, we use the notion of {\em corruption}, defined
later.

\vspace{2mm}
\paragraph{Colorings.}
To describe a partition, we use the slightly generalized concept of a
{\em coloring}: a coloring $\kappa$ is a mapping from each interval
$I\in\IC$ to a {\em distribution of colors} in a color-set $\nu$,
where a fixed color should be thought of as a part. We denote the
mapping by $\mu_\kappa: \IC\times \nu \rightarrow [0,1]$. For each
interval $I$, we require $\|\mu_\kappa(I,*)\|_1=1$, i.e., we think of
the interval $I$ as being split into fractions each assigned to a
part: fraction $\mu_\kappa(I,\chi)>0$ is assigned to color
$\chi$. While under the ``perfect neighborhood assumption'',
standard partitions are sufficient (i.e., $\mu_\kappa\in\{0,1\})$), fractional
colorings will be crucial for removing the assumption later.

Most of the colors $\chi \in \nu$ correspond to a part constructed
from a fixed anchor (i.e., its cluster of intervals together with the
local extension) with the exception of two special ``colors'':
$\nu=[\lambda^t]\cup\{\bot,\un\}$. First, the color $\bot$ that
corresponds to the {\em already-matched} intervals, i.e., intervals
for which we've already added a short path to their $\pi$-match in the
graph $G_{w,c}$ (typically ``dense'' intervals that have already
``converged''). Second, the \un-color (``uncolored'') consists of
intervals which so far have failed to be captured in a part and remain
``tbd'' (here, the progress will be that such intervals gain a certain
``sparsity'' properties, improving the chances to be colored later).

\vspace{2mm}
\paragraph{Coloring construction via potentials.} To construct a new,
more refined step-$t$ coloring (from the step-$(t-1)$ coloring), our algorithm assigns {\em
  potential scores} to clustered intervals. Using these potential
scores, we assign colors\footnote{We'll often just use the verb
  ``color'' to describe that process.} to other nearby intervals in
their proximity, as suggested above.  The main intuition is that a
$\pi$-matchable pair $(I,\pi[I])$ typically has a large set of other
$\pi$-matchable $(J,\pi[J])$ in its respective proximity (i.e., to
the left/right).

How large of a ``proximity'' a cluster can color depends on the size
of the $\cad$-neighborhood of $A$ (its cluster size). To quantify
this, we introduce the notion of {\em density} of an interval $I$ of
color $\chi$, termed $\den(I,\chi)$: the measure $\mu$ of its
$\cad$-neighborhood $\NC_{c}(I)$ that share the color $\chi$.  If an interval
$I$ (and hence its aligned $\pi[I]$) is ``dense'' (large $\NC_c$),
then we have a higher probability to cluster such a pair to an anchor;
but we can only afford a small extension for each one (i.e., each
clustered interval is used to color few other proximal intervals). In
contrast, ``sparse'' matches will be clustered with a small
probability, but can be used to generate large extensions in their
proximity.

In particular, to compute the new step-$t$ coloring, we color the intervals gradually
in levels, indexed by $l=0,1,\ldots, 1/\eps$, each level taking care
of a density scale. In each level $l$, we define potentials
$\phi$ and $\varphi$. First, for each anchor $A$ corresponding to
color $\chi$, we allocate
potential $\phi(I,\chi)\approx \tfrac{n}{d \cdot \lambda^t}$ to each clustered
interval $I$ (i.e., $I$ at distance $O(c)$ from $A$ in the same
``part''\footnote{More precisely, to the fraction of $I$ that shares a
  specific sampled color $\chi$ with $A$. } as $A$) of density
$d$. Next, we define a derivative potential $\varphi(*,\chi)$ by splitting
the allocated potential $\phi(I,\chi)$ 
 across the $\un$-colored intervals in a
proximity ball of radius $\zeta=\zeta(l)\approx n^{\eps l}$ 
around each
clustered interval $I$.   
At the end of each level, we augment %
$\mu_\kappa(I)$ by replacing
some of the $\un$-color mass with other colors, proportional to the
potential vector $\varphi(I,*)$ (to be discussed further later)\ns{was:At the end of each level, we add color to an interval $I$ by moving
some of the $\un$-color mass into proper colors, proportional to the
potential $\varphi(I,*)$ (to be discussed further later)}. Overall, the following is a high-level diagram of algorithm
computation from a $(l-1)$-level coloring $\kappa$ to an ``amended''
$l$-level coloring $\hat\kappa$ (all at the same step $t$):
$$
\mu_\kappa
\quad \xlongrightarrow{\text{clustering}} \quad\phi_\kappa\quad
\xlongrightarrow{\text{extension: splitting in the proximity ball}} \quad\varphi_\kappa\quad
\xlongrightarrow{\text{minhash (under PNA)}}\quad \mu_{\hat\kappa}
$$

A more detailed diagram and a step-by-step coloring example (both require more detailed technical setup) are presented, respectively, in Section~\ref{sec::int_match}, Figure~\ref{fig::level_coloring} and Figure~\ref{fig::level_potentials}\ns{need to fix wrong reference for Figure~\ref{fig::level_potentials} some reason...}. 

In each level $l\ge1$, our goal is to color all intervals of density
in a certain range $[n^{-\epsilon},1]\cdot d_l$, where $d_l\approx
\tfrac{n/\lambda^t}{\beta^l}$, for $\beta=n^\eps$,
as long as there are
sufficiently many intervals of that density range
overall\footnote{Notice there can be multiple matches in a ball,
  hence the quantity we care about (and bound) is the {\em relative
    density} which is the ratio between ``global'' and ''local''
  densities.}.
  At level
$l=0$, corresponding to the highest density for step $t$, we generate a 2-hop path for each pair of intervals in $A$'s cluster by ``planting a star'' in $G$; i.e., adding an
edge
between the anchor $A$ and each corresponding clustered
interval $I$. Then, we mark (fraction of) $I$ as ``already-matched'' with
the color $\bot$. In other levels, we color intervals found through
extensions of clusters using $\varphi$.

The remaining case --- ``sparse'' intervals that we did not yet color via an
anchor cluster extension as above --- will be addressed by the careful
use of $\un$-color, described next. We remark that, at the end of step
$t$, there may be left some pairs of small densities which we still could
not color, and are left \un-colored, and we will show a
bound on those as well. We now expand on the latter.

\vspace{2mm}
\paragraph{Controlling sparse sections: the $\un$-color.} 
In order to carry out the level-by-level coloring, we use the special
color \un\ (for un-colored). This color should be thought of as a
``part'' in the partition as well. At the beginning of a step, at level $l=0$,
all (fractions of) intervals which are not ``already-matched'' are assigned the $\un$ color, and (fractions of) intervals are moved from
$\un$-color to ``standard'' colors $\in[\lambda^t]$ as levels
progress. 
The $\un$-color helps with three
aspects. First, it provides a way to track {\em sparse intervals} which
cannot yet be colored and hence left pending
for future levels.  Second, if some sparse sections of intervals are
never colored in the current step,
then these intervals
will remain $\un$-colored (and hence, at the end of the current step,
form a part that is also bounded in size). Third, and more
nuanced, it allows us to ``group together'' sparse sections of
intervals that are far apart (in their starting index).

In particular, such grouping of far intervals is done using the
aforementioned ``proximity balls'', formally defined via {\em
  $\Lambda$-balls}: $\Lambda_\kappa^\zeta(I)$ is the smallest interval
ball around $I$ containing $\zeta\approx \beta^l$ \un-colored $\ell_1$-mass on both left and right of $I$. Note that
$\Lambda$-balls can contain a significantly larger set of intervals
than $\zeta$ (if the in-between intervals are mostly colored $\neq \un$). At the
same time, the ball contains at most $2\zeta$ mass of \un-colored
intervals, meaning that the potential $\phi(I)$ is distributed to a
mass $\mu \leq 2\zeta$ of intervals, ensuring that, were $I$ to be
``corrupted'' (e.g., $\pi[I]=\bot$, or the pair $(I,\pi[I])$ happens to be already
separated), we will incur only an $O(1)$ factor of total {\em corrupted
  potential} to $\varphi$'s of intervals in the $\Lambda$-ball of $I$.

To showcase the use of $\un$-color, consider again the running example
introduced in Section \ref{sec::prior_obstacles}. In the early steps $t$, our algorithm will first
color the dense intervals (i.e., via clustering/proximity balls, at lower levels $l$), leaving the sparse sections mostly unaffected, all
colored in
$\un$
(during such early steps, the dense intervals are partitioned into
progressively smaller parts while most of the sparse ones remain
$\un$-colored). Now consider a step $t$ where the part sizes so far are
$\lesssim n/\lambda^{t-1}\approx \Delta n$ and we sample $\gtrsim 1/\Delta$
anchors. At the lower levels $l$, the dense intervals will be
partitioned further (continuing the process from the previous steps)
and assigned a color $\chi\in[\lambda^t]$. However, when we reach the
high levels $l$, and $\zeta\approx \beta^l= n^{\eps l}$ is close to $\Delta n$, some
fraction of  
the sparse intervals will be clustered.
Furthermore, since at that point the dense intervals are already colored
(and have little $\un$-color),
the $\Lambda$-balls around the clustered sparse intervals will be wide
and cover most of the $\un$-colored sparse intervals. That allows
us to finally
partition the sparse intervals into smaller parts  as well.

\vspace{2mm}
\paragraph{Keeping track of errors in analysis: Corruption.}
To measure and bound errors, in particular, $\pi$-matched interval
that are separated, we use a formal notion of corruption. First we
define what it means for a (interval, color) pair to be corrupted.  
Below, $F\gg1$ is a ``distortion'' factor (used for the non-PNA case),
which one can think of being $F=2$ for now.

\begin{definition}[Corrupted pairs]
  \label{def::corruption}
	Fix alignment $\pi \in \Pi$, interval $I \in \IC$, distortion
        $F\ge 1$, and graph $G$ on $\IC$. For a color $\chi \in \nu$, %
        we say $(I,\chi)$ is a {\em ($F,\pi,G,c,\kappa$)-corrupted pair} if any of the following holds:
	\begin{enumerate*}
		\item $\pi[I] = \bot$;
		\item  $\cad(I,\pi[I]) > c$;
		\item $\chi \neq \bot$ and $\mu_\kappa(I,\chi)>F\cdot \mu_\kappa(\pi[I],\chi)$; or
		\item  $\chi = \bot$ and $(I,\pi[I])$ are at hop-distance $>2$ in
        $G$.
	\end{enumerate*}
\end{definition}

For each interval $I \in \IC$, we also define {\em corruption parameter}
$\xi_F^{\kappa,\pi,G,c}(I)\in [0,1]$ as follows:
\begin{equation}
  \label{eqn:xiDef}
  \xi_F^{\kappa,\pi,G,c}(I) = \sum_{\chi: (I,\chi) \text{ is } (F,\pi,G,c,\kappa) \text{-corrupted pair}} \mu_{\kappa}(I,\chi).
\end{equation}

In particular, an interval $I$ is fully corrupted (in a coloring
$\kappa$) if it does not have its $\pi$-matchable counterpart; and
otherwise $I$ is corrupted by the total $\ell_1$ color-mass of
$\mu_\kappa(I,*)$ where there is insufficient corresponding mass in
$\mu_\kappa(\pi[I],*)$ (intuitively, the distribution of colors is too
different).  While our statements hold for any alignment in $\Pi$, we
only care about a fixed optimal alignment $\pi$, a single graph
$G=G_{w,c}$ and a fixed cost $c$. Hence, for ease of exposition, we
say $I$
is {\em $F$-corrupted} and the corruption
is $\xi^{\kappa}_F(I)\triangleq\xi^{\kappa,\pi,G,c}_F(I)$.  Our main
goal is to bound the total corruption $\xi^\kappa_F \triangleq
\xi_F^\kappa(\IC)$, and in particular show it grows by at most a constant
factor in any level/step. Our algorithm runs for a constant number of
levels/steps, and hence finishes with %
corruption which is
proportional to the number of intervals without a $\pi$-matchable
counterpart (starting corruption), upper-bounded by $O(\tfrac{w}{c}
\cdot \ed(x,y))$.  Also, at the end of the interval matching
algorithm, all intervals are ``already matched'', i.e., all mass is on
$\mu_\kappa(\IC,\bot)$, and hence the un-corrupted intervals have a
2-hop path to their $\pi$-match.

To bound the corruption growth, we also introduce the parameter
$\rho(I)$, which measures the ``local amount of corruption'' of an
$\un$-colored interval, based on the nearby corrupted intervals. In
particular, $\rho(I)$ is defined for a $\Lambda$-ball around $I$ as
the ratio of the corruption to the $\un$-mass inside the
$\Lambda$-ball (formally defined in Section
\ref{sec:twoKeyLemmas}). One can observe that for any fixed
$\zeta$ radius of $\Lambda$, the sum of $\rho$ over $\un$-colored
intervals is proportional to the total sum of corruption.

\vspace{2mm}
\paragraph{Completing the algorithm under the perfect neighborhoods
  assumption (PNA).}
Once we compute the palettes $\varphi$ of all intervals (as a function
of the sampled anchors), we then use them to update $\mu$ for the next level.
For illustrative purposes, we now complete the algorithm under the
PNA, although our general algorithm will differ significantly from the
PNA one. Recall that under PNA, all intervals are either at
$\cad$-distance $\le c$ or $\gg c$, and hence the intervals form equivalence classes according to their $c$-neighborhood.

Under PNA, we are guaranteed that the uncorrupted $\pi$-matchable
pairs will get similar potentials ---in fact, $\phi(I)$ and
$\phi(\pi[I])$ are precisely equal (and non-zero whenever they are
clustered by an anchor). More importantly, if we consider the $\varphi$
palettes of $I,\pi[I]$, which gather the contributions from clusters
containing $I,\pi[I]$ in their proximity ball, then one can prove that
(the average) $\ell_1$ distance between the two $\varphi$ palettes is
bounded as a function of the ``local corruption'', namely $\rho(I)$
and $\rho(\pi[I])$.

Using the $\ell_1$-distance property of $\varphi$, we can assign a
single color $\chi \in \nu$ to each interval (i.e., $\mu_\kappa(I,*)$
has support one), obtaining disjoint partitions.  To generate such a
color (for each interval) we can use a random {\em weighted min-wise
  hash function} $h \sim \mathcal{H}$ for $\R^\nu$ (say, using
\cite{Char}) and use it to partition the vectors $\varphi(I,*)$ of all
$\un$-colored intervals $I$. Specifically, sample a minhash
$h:\R^\nu\to \nu$ and set the updated coloring to be
$\mu_\hkappa(I,h(\varphi(I))) \gets 1$ (the rest are 0) for all $I$
for which $\mu_\kappa(I,\un) = 1$ at the end of the previous level.

For a glimpse of the analysis, recall that our overall goal is to make
each part of the partition smaller (for runtime complexity) with only a
constant-factor corruption growth (for correctness); also, we care
only to partition areas with large mass of intervals with density in
some range $[n^{-\eps},1]\cdot d_l$ (in a fixed level $l$). To control the size of
parts, we cannot afford to assign the same color to 
 too many
intervals; %
hence we drop from $\varphi$ all colors of potential below
some fixed threshold $o(n^{-\epsilon})$, ensuring that each color appears
in the $\varphi$ palettes of at most
$\tfrac{n^{1+O(\epsilon)}}{\lambda^{t-1}}$ intervals. For bounding the
corruption, we note that minhash gives us a bound proportional to the
Jaccard distance between $\varphi(I)$ and $\varphi(\pi[I])$, while the
$\rho$ bound we have is in $\ell_1$ distance. Showing these bounds are close is where the \un-color  
plays a central role. First, consider an interval $I$ such that $\Omega(\epsilon)$
portion of its proximity ball $\Lambda$ is composed of intervals of
density $\in [n^{-\eps},1]\cdot d_l$, where $\Lambda$ is of radius
$\zeta$. Then the palette $\varphi(I)$ has $\ell_1$ mass
$\Omega(\epsilon)$ whp after thresholding since: 1) some intervals in
$\Lambda$ will be clustered whp (they are dense enough), and 2) once
clustered, the generated potential is large enough to pass the
threshold (since they are not too dense). In this case, we are done
(without using the color $\un$): the Jaccard distance is proportional
to the $\ell_1$ distance between $\varphi(I)$ and $\varphi(\pi[I])$,
and hence the probability of separating $I$ from $\pi[I]$ is bounded
by the ``local corruption'' (which overall is bounded by the total
corruption). Second, consider the case when $1 - o(\epsilon)$ portion
of intervals in the $\Lambda$ ball are outside the aforementioned
density range. Then we add mass to $\varphi(\cdot, \un)$ of intervals in the ball
$\Lambda$, filling it up to reach $\|\varphi(I)\|_1 =
\Omega(\epsilon)$ --- this will increase the probability that such
intervals are mapped (again) to $\un$ (this increases corruption by a factor $\leq 2$).
This process also guarantees $\Lambda$ balls at the next level $l$ have
$(1-o(\epsilon l))\cdot \zeta$ mass of {\em sparse intervals}, i.e., of
density $\lesssim \tfrac{n/\lambda^t}{\beta^l}$. %
One can
then prove that, after running this process for $\approx 1/\epsilon$
levels, the set of intervals corresponding to each color, including
$\un$, is of size $\lesssim n/\lambda^t$ only.

Since we could not directly extend the minhash construction to the
general non-PNA case, we do not present this construction in the
paper, but rather use it as an intuition for its ``robust'' version as
we describe next.

\subsection{Imperfect neighborhoods}

To eliminate the perfect neighborhood assumption (PNA), we must rely
on the weaker form of transitivity instead, from the triangle
inequality: $\NC_{c}(I) \subseteq 
\NC_{2c}(J)\subseteq \NC_{3c}(I)$ for any $J\in \NC_c(I)$.  Note that the usual ideas to deal
with such ``weaker transitivity'' do not seem applicable here. For instance,
if we pick the threshold of ``close'' in the cluster construction to be uniformly random $\in [c,O(c)]$,
there's still a constant probability of separating $I$ from $\pi[I]$. One
could instead apply the more nuanced metric random partitions, such as
from~\cite{MendelN06}, which would {\em partition} the metric
$(\IC_w,\cad_w)$ (thus putting us back into the perfect neighborhood
assumption), with the probability of $I,\pi[I]$ ending up in the same
part being $\ge n^{-\eps}$ --- which has been useful in other contexts by
repeating such partition $\approx n^{\eps}$ times.  However, such a
process results in a random partition retaining only $n^{1-\eps}$
$\pi$-matched pairs, which is not enough to reconstruct even those matched
pairs (intuitively, the strings are ``too corrupted'', as if the edit
distance is $(1-o(1))\cdot n$), making it inapplicable for our algorithm (here again, this challenge would be more manageable if additive approximation were allowed). Overall, dealing with
imperfect neighborhoods proved to be a substantial challenge for us,
and we develop several first-of-a-kind tools specifically to deal with
it.

Eventually, we still sample a cost $c_i$ from some ordered set
$E_c=\{c_1,c_2,\ldots\}\subset [c,O(c)]$. Since we want that the cost
$c_i$ satisfies that $\NC_{c_{i}}(\NC_{c_{i}}(I)) \subseteq
\NC_{c_{i+1}}(I)$, we set $E_c$ of costs to be exponentially-growing,
i.e., $c_i = \Theta(c) \cdot 3^i$. The formal definition is in Section
\ref{sec::int_match}.

\vspace{2mm}
\paragraph{Distortion Resilient Distance.}
Relying purely on triangle inequality forces us to assign somewhat
different potential scores to $I$ and $\pi[I]$, hence we will quantify
the ratio between the two, 
referring to it as ``distortion''.

While it may be tempting to try to keep the distortion to a constant,
it turns out one cannot do that without introducing super-constant
factor corruption growth (number of pairs with a large distortion), which is prohibitive for us. 
To control corruption, we allow the distortion (i.e., the multiplicative difference) between $\phi(I)$ and $\phi(\pi[I])$ for $\pi$-matchable pairs can be as high as $n^{\alpha}$ for some
small constant $\alpha>0$. However, such a distortion makes it impossible to
obtain a 
bound on $\ell_1$ distance between $\varphi$'s of a $\pi$-match, which is proportional to $\rho$. 
Instead, we deal with such distortion by employing
a {\em distortion resilient} (robust) version of $\ell_1$. 

\begin{definition}\label{def::dd}
	Fix $p,q \in \R^k_+$. We define the {\em $F$-distortion resilient distance},
	$$\dd_F(p,q) = \sum_{i:p_i > F\cdot q_i} p_i$$.
\end{definition}

This function allows us to define and control {\em corruption} of
(interval, color) pairs by differentiating distortion (which captures
multiplicative errors) from corruption (which captures additive
ones).
 As part of our analysis, we will show several basic properties
of the $\dd$ distance and develop {\em $\dd$-preserving
  soft-transformations}, which will replace the hard thresholds from the
minhash construction.
Intuitively, $\dd$ replaces the use of the $\ell_1$/Jaccard metric on the vectors
$\phi$ and $\varphi$, which was a key enabler for using minhash under
PNA. 
However, $\dd$ is not a metric in any reasonable sense (it's not
even symmetric), 
rendering the minhash construction obsolete (e.g., it is unreasonable to expect
any kind of LSH under $\dd$).

\vspace{2mm}
\paragraph{Assigning potential to (interval, color) pairs.}
Since maintaining equivalence classes is essential for our
construction, we analyze pairs of interval and colors in $\IC \times
\nu$ (which, combinatorially, can be thought of as ``fractions of
intervals'').
Thus, when we increase the $\phi$-potential of a clustered pair $(I,
\chi')$, we do so proportionally to its $\mu$ mass, meaning we set the
potential to $\approx\mu(I,\chi') \tfrac{n}{d \cdot \lambda^t}$. Similarly,
splitting the $\phi$ potential to $\un$-colored-pairs in $\Lambda$
balls (i.e., assigning $\varphi$) is done in a pro-rated fashion,
weighted according to respective $\mu(\cdot,\un)$ masses. \ns{maybe
  write something about $\Lambda$-balls and $\un$-pairs are defined
  w.r.t to output coloring, but $\phi$ assigned to pairs of old
  coloring?}\aanote{not sure this is the place to do it--don't we do
  that elsewhere?}

\vspace{2mm}
\paragraph{Assigning $\un$ potential: pivot sampling.} Having so far
discussed assigning of the non-$\un$ colors in $\varphi$ of (fractions
of) intervals, we now discuss how to assign $\un$-color in
$\varphi_\kappa$ and, eventually, amended coloring $\mu_{\widehat
  \kappa}$. It may be tempting to merely subtract the assigned
fraction of non-$\un$-color from the $\un$-color mass, but this would
result in additive errors (for the color $\un$), while $\dd$ only
allows for multiplicative distortion. In particular, we need $\un$
colors to agree up to a fixed distortion as well (to avoid more
corruption), and hence we compute the new $\un$-color mass directly,
via a different technique for explicitly measuring {\em sparseness}
(which is the central purpose of $\un$-color). To accomplish this
measurement, we developed a procedure called {\em pivot sampling},
which somewhat resembles the way we assign potentials for the non-$\un$
colors. First, we downsample $\IC \times \nu$ into a smaller set of
{\em pivots} $\VC$. Second, we {\em approximate the density} of each
pivot in $\VC$, for each possible cost in $E_c$, thus generating {\em
  $\theta$-potential} scores for each pivot. Third and last, such
$\theta$ potential is split among the intervals in a $\Lambda$-ball in
a similar fashion to how we split $\phi$, generating
$\varphi(\cdot,\un)$ potentials to {\em intervals in sparse
  areas}. This rather involved process, specific to dealing with
imperfect neighborhoods, requires much care to be able to control: (1)
corruption of \un-colors; (2) balance of palettes $\varphi$ (as we
describe next); (3) sparsity guarantees for $\un$-color (part) at the end of
the step $t$; and (4) computational efficiency of such sampling
mechanism.

\vspace{2mm}
\paragraph{Amending a coloring in a level, using $\varphi$.} While $\dd$ is
a convenient analytical tool for bounding corruption, it lacks the
basic properties to allow coordinated sampling between $\pi$-matchable
pairs. Instead of sampling a color from $\varphi(I,*)$ (as was done
under PNA), we add {\em all} colors in $\varphi(I,*)$ to the amended
coloring $\mu_\hkappa(I,*)$. To maintain a distribution of colors, we
first combine $\varphi$ with pre-existing non-$\un$-colors in $\mu$,
and then normalize to have $\ell_1$-mass of $\mu_\kappa(I,\un)$, i.e.,
what ``remains to be colored'', thus ensuring that the overall amended
$\mu_{\hat \kappa}(I,*)$ is a distribution. As in the PNA case, we
need to bound extra {\em corruption from normalization} by ensuring
that the palettes $\varphi$ have constant norms. While this analysis
for the PNA solution is immediate (by construction), here, instead, we
employ several combinatorial arguments that analyze mass of pairs with
certain density over certain set of costs, eventually showing that in
each level, we either add sufficient regular colors (corresponding to
anchors/clusters) or \un-colors to all intervals while maintaining guarantees
(1)--(4) above.

\vspace{2mm}
\paragraph{Controlling the growth of distortion.} Our arguments require that
throughout the matching phase, the $\dd$ distortion $F$ is bounded by
$n^{o(\epsilon)}$ (in particular, to maintain control over the aforementioned soft-transformations). Many of our algorithmic steps generate extra distortion. To control both distortion (multiplicative error) and corruption (additive error),
we parametrize maximum distortion $F=F(t,l)$ for each step/level a
priori, and bound corruption $\xi_F^{\kappa}$ at each step/level using
the pre-determined distortion parameter $F=F(t,l)$. The final
approximation factor is a function of the maximum cost in
$\tfrac{1}{c}E_c$ (which is further determined by the ``base
distortion'' $F(1,0) = n^\alpha$), together with the corruption factor
we show in each step.
At the end of the day, a distortion $F$ bounded by $n^{o(\eps)}$
allows us to carry out the above arguments (i.e., some of the above
arguments can only work under small distortion $F$).
\ns{Maybe add/ derive the exact $f(\epsilon)$ bound here?}

\subsection{The metrics $\cad_w$}

We now briefly discuss the algorithm for computing the $\cad_w(I,J)$
distance, using oracle calls to $\distG_{w/\wbase}$ metric. This
metric is used to compute distances when building the graph $G_w$, in
the Interval Matching algorithm. Note that the latter makes
$n^{1+O(\eps)}$ oracle calls to $\cad$, and hence the algorithm has to
run in time $n^{O(\eps)}=\poly(\wbase)$.

Intuitively, $\cad_w(I,J)$ is meant to capture the following distance,
which should be thought of as an extension of the edit distance over an
alphabet with the metric $(\IC_{w'},\distG_{w'})$ where $w'=w/\wbase$:\footnote{$I+i$
  means the interval starting $i$ positions to the right of the start of $I$.}
$$
\min_\pi \sum_{i\in[w]} \tfrac{1}{w'}\distG_{w'}(I+i,
J+\pi(i)),
$$ where $\pi$ ranges
over all alignments of indexes of $I$ to indexes of $J$.
One can show that essentially, if $\distG(I,J)=\ed(I,J)$, then, for $I=X_{i,w},
    J=Y_{j,w}$, the above distance is between $\Omega(\ed(X_{i,w},Y_{j,w}))$
    and $O(\ed(X_{i,2w},Y_{j,2w}))$.\footnote{While we are not aware
      of an explicit proof of this statement, it is in the spirit of
      statements that appeared in, e.g., \cite[Lemma 5]{OR-edit},
      \cite[Lemma 3.2]{AK-smoothed}, \cite[Theorem 3.3]{AKO-edit}.}

However, this distance function is hard to compute fast:
not only it is as hard as computing edit
distance on $w$-length strings, but even linear time (in
$w\gg\poly(\wbase)$) is too much for us. In particular, it does not use the
fact that $\distG_{w'}(I+i, J+\pi(i))$ captures the information
of blocks of length $w'$. Hence, it is natural to approximate the
above by considering a ``rarefication'' of the above sum as follows:
\begin{equation}
  \label{eqn:cadRarefied}
  \min_\pi \sum_{i\in[\wbase]}
\tfrac{1}{w'}\distG_{w'}(I+iw',
J+\pi(iw')).
\end{equation}
However, the latter will not satisfy the triangle inequality --- which
is crucial in the Interval Matching Algorithm --- and in fact is not
even symmetric: e.g., if the optimal $\pi(i)=i+1$, the $\cad_w(J,I)$
would be using $\distG$ on completely different arguments. This is
especially an issue since $\distG$ may substantially over-estimate
$\ed$ on some of the pairs (and hence ``shift by one'' can change the
distance a lot).

Indeed, ensuring triangle inequality is the main challenge for
defining and computing $\cad_w$ here. We manage to define an
appropriate distance $\cad$, satisfying triangle inequality %
for ``one scale
only'' metrics $\cad_{w,c}$, designed for distances in the range
$\approx [c,\wbase c]$, which turns out to be enough for the Interval
Matching algorithm.

First, we note that we have 
 two different algorithms, corresponding to two distinct distance formulations:
 (i) for large distance regime, where $c>w/\wbase$; and (ii) for small distance regime, where $c\le
w/\wbase$.
The reason there's a big 
difference
between the two cases is that when $c>w/\wbase$, the alignment $\pi$
may have a large displacement $|i-\pi(i)|\ge w/\wbase$, bigger than
the length of ``constituent'' intervals for which we have the base
metric $\distG_{w'}$. Hence, for the ``large distance'' regime, when
$c\ge w/\wbase$, we uses a slightly different (and simpler) algorithm
that runs in time $\approx \poly(w/c)$, and hence is only good when
$c$ is sufficiently large.

Finally, we sketch the harder, $\poly(\wbase)$-time algorithm, for not-so-large
$c$. The idea is to allow alignment shifts in both intervals.
More formally, let $T=\wbase^3$ and let $\A$ be the set of functions $A=(A_x,A_y)$ where $A_x,A_y: [-\wbase,\wbase]
\rightarrow \{0,1,\ldots,T-1\}$
 are non-decreasing functions with $A_x[-\wbase]=A_y[-\wbase]=0$ and
$A_x[\wbase]=A_y[\wbase]$. We define the distance
$\cad_{w,c}(I,J)$, to be, where $\theta$ is essentially $\approx c/w$ and $\mathfrak{S}_i = \{0,1,\ldots,3T - A_x[i] - A_y[i]%
-1\}$:
\begin{align*}
&\cad_{w,c}(I,J)\triangleq
\min_{A\in \A} \\
&\left(A_x[\wbase] + A_y[\wbase]\right)\theta w' + \sum_{i \in [-\wbase,\wbase)} \tfrac{1}{T}\sum_{\Delta \in \mathfrak{S}_i} \distG_{w'} \left(I + w'(i + \theta (\Delta+ A_x[i])),J + w'(i + \theta(\Delta + A_y[i]))\right).
\end{align*}

\ns{Perhaps to explain the shifts, we can split up the def by defining\\
 $\widehat{\mu}_{A,i,\theta}(I,J) = \tfrac{1}{T}\sum_{\Delta \in \mathfrak{S}_i} \distG_{w'} \left(I + w'(i + \theta (\Delta+ A_x[i]),J + w'(i + \theta(\Delta + A_y[i])\right)$ \\ as ``sum over shifts'', and then defining\\ $\cad_{w,c}(I,J)=
\min_{A\in \A} A_x[\wbase-1] + A_y[\wbase-1] + \sum_{i \in [-\wbase,\wbase)} \widehat{\mu}_{A,i,\theta}(I,J)$\\ as the standard alignment sum?}

\ns{Alex, please read the below 2 amended sentences}
Intuitively, ignoring $\Delta$-sum (i.e., think $\Delta=0$), we
obtain an alignment of $I$ to $J$ where the starting positions (of
$w'$-length intervals) are close to multiples of $w'$ in {\em both}
strings (as opposed to only one string, as in
Eqn.~\eqref{eqn:cadRarefied}). While allowing such an alignment is
enough for ensuring symmetry, it is still not enough to ensure triangle inequality. Consider intervals $I,J,K \in \IC$ where we want to guarantee that $\cad(I,K) \leq \cad(I,J) + \cad(J,K)$. Apriori, there is no way to ensure an optimal alignment between $\cad(I,J)$ and $\cad(J,K)$ will use the same shifts and hence makes it hard to offset the distance of large blocks in $\cad(I,K)$. To solve the inconsistencies between the shifts, we use the $\Delta$-sum over all possible shifts (this is yet another
instance of ``averaging it out''). The last definition can also be
computed in $\poly(\wbase)$ time by a standard dynamic programming.

Nonetheless, the following issue remains:
think of the case when
$\cad(I,J)$ and $\cad(J,K)$ use the maximally-allowed values of the alignment
(namely $\wbase^2$), in which case $\cad(I,K)$ cannot use the natural
composition of the two alignments (since it's out of bounds). To solve this issue, we upper bound each $\cad_{w,c}$ with a maximal value $\gg c$ (which solves the triangle inequality issue), and define the (non-metric) distance $\cad_w(\cdot,\cdot)$ as the summation over all costs $c \in S_w$ which can be upper-bounded by $\cad_{w,c}$ (formal definition in section \ref{sec:CADalgo}).

Finally, we remark that, at the end of the day, we cannot guarantee a per-pair
upper bound on $\cad(I,J)$, but only on average, and only when
comparing $X_{i,w}$ against $Y_{\pi(i),w}$ (although the triangle
inequality is true everywhere). This is, nonetheless, just enough for
estimating $\ed(x,y)$.

%% file: algo.tex
\section{Top Level Algorithm}
\label{sec:topAlgo}

We now describe our ``top-level'' algorithm. We assume here that the
first $(1-o(1))n$ positions of $x$ and $y$ are equal; we can remove
this assumption by padding $x,y$ with some fixed unique character \$,
increasing the size of $x,y$ by a factor of, say, $O(\log n)$.

Our algorithm consists of $\log_\wbase n$ iterations, where $\wbase=n^\eps$. %
For each
$w=\wbase^i$, $i\in [\log_\wbase n]$, we construct the metric
$(\IC_w,\distG_w)$, which {\em align approximates} the metric
$(\IC_w,\ed_w)$ (see Def.~\ref{def:alignApp}). Each
iteration consists of two components: the {\em alignment distance
algorithm}, and the {\em interval matching algorithm}, described in later
sections. Below we assume that $n$ is a power of $\wbase$, which is
without loss of generality (as we can increase $n$ appropriately by
padding the strings).

\vspace{2mm}

\paragraph{Alignment Distance algorithm.} Assuming oracle access to
$(\IC_{w/\wbase},\distG_{w/\wbase})$, the metric constructed at the
previous iteration, 
our alignment algorithm is an oracle for computing the distance
$\cad_w(\cdot,\cdot)$ on $w$-length intervals. In fact, we have
$O(\log n)$ such distance measures, $\cad_{w,c}$, one for each target
cost scale $c\in \Sw=\{1,2,4,\ldots w\}\cup \{1/2,1/4,\ldots, 1/2n\}$. Each such function
$\cad_{w,c}(\cdot,\cdot)$ evaluation is an edit-distance-like dynamic
programming of size $\poly(\wbase)$, and overall can be computed using
$\poly(\wbase)=n^{O(\eps)}$ time and oracle calls to
$\distG_{w/\wbase}$.

Note that this algorithm does not run directly, but instead is used as
an oracle inside the matching algorithm described next.

\vspace{2mm}

\paragraph{Interval Matching algorithm.} We construct
a weighted graph $G_w$ on $\IC_w$, such that the shortest path distance in $G_w$
approximates the $\cad_w$ distance on intervals. Again, this won't be
achieved for all pairs of intervals, but only for interval pairs that
``matter'', i.e., that are in an optimal alignment for $\ed(x,y)$. The
graph $G_w$ is the union of edges of the graphs $G_{w,c}$, for $c\in \Sw$, each of them
align approximating $\cad_{w,c}$ at ``scale $c$''.  Constructing the graphs
$G_{w,c}$ is the heart of the matching algorithm.

\vspace{3mm}

Once we have the graph $G_w$, we build a fast distance oracle data
structure on it, using the $\ell_\infty$ embedding of
\cite{Matousek1996}, and whose output is the desired metric
$\distG_w$. Overloading the notation, we call $\distG_w$ both the
distance oracle data structure as well as the metric it produces.  In
particular, \cite{Matousek1996} shows how one can embed any $n$-point
metric into $\ell_\infty$ of dimension
$d=O(\tfrac{1}{\eps}n^{\eps}\log n)$ while incurring $O(1/\eps)$
distortion only. Once we have such an embedding, we can compute
distance between two points by evaluating $\ell_\infty$ distance in
dimension $d$. Furthermore the embedding itself can be computed by running
$d$ single-source shortest path (SSSP) computations. Hence, we obtain
the following Theorem~\ref{thm::distanceOracle} as an immediate
corollary of \cite{Matousek1996} together with a standard $\tilde
O(m)$-time SSSP algorithm.

\begin{theorem}[corollary of \cite{Matousek1996}]\label{thm::distanceOracle}
For any constant $\eps>0$, given any weighted graph $G$ on $n$
nodes and $m$ edges, we can build a distance oracle data structure with the
following properties:
\begin{itemize}
\item
supports distance queries: given $u, v$, output $\distG_G(u,v)$ which
is a $O(1/\eps)$-factor approximation to the shortest path distance
between $u,v$ in the graph;
\item
  $\distG_G(u,v)$ is a ($n$-point) metric;
\item
  runtime per query is $\tilde O_\eps(n^{\eps})$;
\item
  data structure uses $\tilde O_\eps(n^{1+\eps})$ space, and pre-processing time is $\tilde O_\eps(m n^{\eps})$.
\end{itemize}
\end{theorem}

\paragraph{Top-level algorithm} is described in Algorithm~\ref{alg::ed}.
At the beginning, when $w=\wbase$, we use the metric $(\IC_1,
\distG_1)$, which is just the metric on all positions in $x$ and $y$,
where two positions are at distance 0 iff the positions contain the
same character, and 1 otherwise. At the end, when $w=n$, we can
extract the distance between $x$ and $y$, which is our final
approximation. The algorithm \textsc{MatchIntervals}, described in
Section~\ref{sec::int_match}, returns an unweighted graph
$G_{w,c}$. The full graph $G_w$ for the scale $w$, is obtained by
union of the graph edges of $G_{w,c}$ each scaled by $C_mc$, over all
$c\in S_w$, together with some extra edges.

\begin{algorithm}[H]
\caption{EstimateEditDistance$(x,y,\eps,%
n)$}
\label{alg::ed}
\begin{algorithmic}
\Function{EstimateEditDistance}{$x,y,\eps,%
n$}
\State Fix $C_m$ to be the constant from Theorem~\ref{thm::matching_guarantee}.
\State $\wbase \gets n^{\eps}$.
\State $\distG_1$ is a data structure that, given two positions into
$x$ and/or $y$,
outputs 0 iff the characters in those positions are equal and 1 otherwise.
\For{$w \in \{ \wbase, \wbase^2, \ldots, n\}$}
\State Let $\WC_{w}, \HC_{w}$ be sets of all $w$-length intervals on
$x$-axis and $y$-axis respectively, with $\IC_w = \WC_{w} \cup \HC_{w}$.
\For{ $c \in \Sw$ }
\State $G_{w,c} \gets \textsc{MatchIntervals}(\distG_{w/\wbase}, c)$.
\EndFor
\State $G_{w} \gets \cup_{c\in \Sw} C_m \cdot c\cdot G_{w,c}$ and add edges
$(X_{i,w},X_{i+1,w}), (Y_{i,w},Y_{i+1,w})$ with unit cost for all $i$. 
\State $\distG_{w} \gets $ data structure from Theorem~\ref{thm::distanceOracle} on
graph $G_{w}$ for approximation $10/\eps$.
\EndFor
\State \Return $\distG_n(X_{i,n},Y_{i,n})$ for a randomly
chosen $i\in[n]$.
\EndFunction
\end{algorithmic}
\end{algorithm}

\subsection{Main guarantees}

The guarantees of the algorithm follow from the following
two central theorems.

\aanote{edited statement slightly: with whp, and $pi$ moved} 
\begin{theorem}[\textsc{MatchIntervals}; see Sections~\ref{sec::int_match},~\ref{sec::match-correctness},~\ref{sec:keyCorrectness},~\ref{sec:runtime}]\label{thm::matching_guarantee}
  Fix $\eps>0$, $w, n \in \N$, and  cost $c\in \Sw$. Suppose $(\IC_w,\cad_{w,c})$ is a metric,
  for which we have query access running in time $T_{\cad}$. Then, the
  algorithm \textsc{MatchIntervals} builds an
  undirected graph $G_{w,c}$, over intervals $\IC_w$, such that:
  \begin{enumerate}
  \item \label{it:matchingDistances} For all edges $(I,J)\in
    G_{w,c}$, we have $\cad_{w,c}(I,J) \leq C_m\cdot c$, where
    $C_m=C_m(\eps)$ is a constant.
  \item \label{it:matchingError}
   For
  any alignment $\pi\in \Pi$, with high probability:
  	 $|\{ i\mid
\dist_{G_{w,c}}(X_{i,w},Y_{\pi(i),w}) >2\}| \leq C_m \cdot |\{ i\mid
          \cad_{w,c}(X_{i,w},Y_{\pi(i),w}) > c \}|$\footnote{Recall from the preliminaries that
            $\cad_{w,c}(X_i,Y_{\bot})=w$.},
        where
          $\dist_{G_{w,c}}$ is the hop-distance in ${G_{w,c}}$.
  	\item \label{it:matchingRuntime}
          The runtime of the algorithm is $O(T_{\cad} \cdot
        n^{1+O(\eps)})$ whp.
  \end{enumerate}
\end{theorem}

As described above, using the algorithm from the above theorem, we
build a graph $G_w$, which is the union of scale graphs $G_{w,c}$. Then we
take $\distG_w$ to be the fast distance oracle of the shortest path on
the graph $G_w$, using Theorem~\ref{thm:cadGuarantees}.

Next theorem says that, given access to $\distG_{w/\wbase}$, we can
compute $\cad_{w,c}(I,J)$ for any two intervals $I,J\in \IC$, which
corresponds to the ``natural extension'' of $\distG$ from
length-$w/\wbase$ to length-$w$ substrings.

\begin{theorem}[alignment distance $\cad$; see Section~\ref{sec:CADalgo}]
\label{thm:cadGuarantees}
Fix $w$ and $w'=w/\wbase$, and suppose we have a data structure for a
metric $\distG_{w'}$ that $C$-align-approximates $\ed$ for
some constant $C\ge 1$, while also $\distG_{w'}(X_{i,w'}, X_{i+1,w'})\le C$
for all $i\in[n]$ (and same for $Y$ intervals).  Then, for any $c\in
\Sw$, the algorithm from Section~\ref{sec:CADalgo} defines
a function $\cad_{w,c}(\cdot,\cdot)$ on $\IC_w\times \IC_w$ with the following properties:
\begin{enumerate}
\item
  $(\IC_w,\cad_{w,c})$ is a metric;
\item
  For all $X_i,Y_j\in \IC_w$, $\cad_{w,c}(X_i,Y_j)\ge \min\{\ed_w(X_i,Y_j),c\sqrt{\wbase}\}$;
\item Define $\cad_w(X_i,Y_j)\triangleq \sum_{c\in \Sw}
  c\cdot\one{\cad_{w,c}(X_i,Y_j)\ge c}$. Then, $\cad_{w}$ $O(C)$-align-approximates $\ed$.
  \item 
  For all $I,J\in \IC_w$, $\cad_{w,c}(I,J)$
  can be computed using $\tilde O(\wbase^{O(1)})$ time and queries to $\distG_{w'}$.

\end{enumerate}
\end{theorem}

We remark that $\cad_w$ is not guaranteed to be a metric, which is the
reason why we use $\cad_{w,c}$ in the theorem statement. Also, the
algorithm from Section~\ref{sec:CADalgo} requires no further preprocessing.

\subsection{Proof of Theorem~\ref{thm:main}}

To prove Theorem~\ref{thm:main}, we just combine the above two
theorems,~\ref{thm::matching_guarantee}
and~\ref{thm:cadGuarantees}. In particular, the inductive hypothesis
is that, for $w=\wbase^i$, where $i\in[\log_\wbase n]$, 
the distance oracle data structure $\distG_{w}$ outputs a metric
$\distG_w$ with the following properties, for some constant
$C_w=C(\eps,\log_\wbase w)$, whp:
\begin{enumerate}
\item
 $\distG_w$ $C_w$-align-approximates $\ed$;
\item
  $\distG_{w}(X_{i,w}, X_{i+1,w})\le 10/\eps$ (and same for $Y$ intervals).
\end{enumerate}

Base case: for $w=1$, this is immediate by construction of $\distG_1$.

Now assume the inductive hypothesis for $w'=w/\wbase$ and we need to
prove it for $w$. By inductive hypothesis, $\distG_{w'}$ satisfies
hypothesis of Theorem~\ref{thm:cadGuarantees}, and hence we can apply
it to obtain an oracle query to metrics $\cad_{w,c}$; each oracle
query takes $O(\wbase^{O(1)})$ time. Let $\pi$ be optimizer for
$\min_\pi \sum_i \cad_w(X_{i,w},Y_{\pi(i),w})$ (guaranteed to
align-approximate $\ed$). 

Define $\tau(\cdot,\cdot)$ to be the distance in the graph $G_w$ constructed in the algorithm. We will prove below that $\tau$ is a metric
satisfying the above properties. Hence, once we build a fast
distance oracle $\distG_w$ on the graph $G_w$ (using
Theorem~\ref{thm::distanceOracle}),
its output metric
$\distG_w$ satisfies $\tau\le \distG_w\le O(\tau/\eps)$, and hence the
inductive hypothesis.

To prove the first property of the inductive hypothesis, consider any
two intervals $I, J$, and the shortest path $v_1,\ldots v_k$ between
them, where $v_1=I$ and $v_k=J$. We have that $(v_i,v_{i+1})$ is an
edge in some graph $G_{w,c_i}$, or is an extra edge of cost $1$; call
$E$ the set of the latter $i$'s. Hence the cost $\tau(I,J)=\sum_{i\not
  \in E} C_mc_i+|E|$. For $i\not\in E$, by
Theorem~\ref{thm::matching_guarantee} and
Theorem~\ref{thm:cadGuarantees}, we have that $C_mc_i\ge
\cad_{w,c_i}(v_i,v_{i+1})\ge \ed_w(v_i,v_{i+1})$ (note that the other
part of the min cannot happen as $C_m\ll \sqrt{\gamma}$). Also, for $i\in E$, we have that
$1=\tau(v_i,v_{i+1})\ge \ed_w(v_i,v_{i+1})$ as $\ed_w(X_i,X_{i+1})\le1$
(and same for $Y$'s). Hence
$\distG(I,J)\ge \tau(I,J)=C_m\sum_{i\not \in E} c_i+|E|\ge
\ed_w(v_1,v_2)+\ldots+\ed_w(v_{k-1},v_k)\ge \ed_w(I,J)$.

Next, we note that $\tau(I,J)$ is upper bounded by
$2C_m\min_{c\in S_w} c+ \sum_{c\in \Sw} 4\cdot C_mc\cdot \one{\dist_{G_{w,c}}(I,J)>2}=C_m/n+4C_m\sum_{c\in \Sw} c\cdot \one{\dist_{G_{w,c}}(I,J)>2}$. Hence:
  $$
  \sum_{i\in [n]} \tau(X_{i,w},Y_{\pi(i),w})
\le C_m+
  4C_m\sum_{c\in \Sw}  \sum_{i\in [n]}  c\cdot
  \one{\dist_{G_{w,c}}(X_{i,w},Y_{\pi(i),w})>2}.
  $$

  For fixed $c\in S_w$, we have that, by Theorem~\ref{thm::matching_guarantee}:
  $$
  \sum_{i\in [n]} 
  \one{\dist_{G_{w,c}}(X_{i,w},Y_{\pi(i),w})>2}
  \le
C_m \cdot %
\sum_i \one{\cad_{w,c}(X_{i,w},Y_{\pi(i),w})>c}.%
  $$
  Therefore,
  $$
  \sum_{i\in [n]} \tau(X_{i,w},Y_{\pi(i),w})\le O\left(1+\sum_{c\in \Sw}
  \sum_{i\in[n]}c\cdot
  \one{\cad_{w,c}(X_{i,w},Y_{\pi(i),w})>c}\right)\le O\left(1+\sum_{i\in[n]} \cad_w(X_{i,w},Y_{\pi(i),w})\right).
  $$
  Since $\pi$ is the optimizer for the right-hand-side, using Theorem~\ref{thm:cadGuarantees} again,
  together with the inductive hypothesis for $w'=w/\wbase$, we
  conclude (where constant depends on $\eps$ and $\log_\gamma w$):
  $$
  \sum_{i\in [n]} \distG_w(X_{i,w},Y_{\pi(i),w})
  \le
  O\left(\sum_{i\in [n]} \tau(X_{i,w},Y_{\pi(i),w})\right)
  \le
  O\left(1+\sum_{i\in [n]} \cad_w(X_{i,w},Y_{\pi(i),w})\right)
  $$
  $$
  \le
  O\left(\min_\pi \sum_i \ed_w(X_{i,w},Y_{\pi(i),w})\right),
  $$
  since we can assume wlog that $\ed(x,y)\ge 1$ (checking the opposite
  is immediate), thus completing
  the proof of the inductive hypothesis.

The second property is immediate by construction of the graph $G_w$
and the fact that approximation of the distance oracle $\distG_w$ is
taken to be $10/\eps$.

Now we argue that the final output produced by the top-level algorithm
is a constant factor approximation. Consider the $\distG_w$ guarantees
for $w=n$, and fix the minimizing $\pi$, and constant $C_n=C_w$. For a
random index $i\in[n]$, with probability at least $0.9$, we have that:
1) $i\in[1,n-o(n)]$, 2) $\pi(i)\neq \bot$, and
$\distG(X_{i,n},Y_{\pi(i),n})\le O(C_n)\cdot \ed(x,y)$. Furthermore
note that $|i-\pi(i)|\le C_n\cdot \ed(x,y)$, and hence,
$\distG(X_{i,n},Y_{i,n})\le
\distG(X_{i,n},Y_{\pi(i),n})+\tfrac{10}{\eps}\cdot|i-\pi(i)|\le
O(C_n/\eps)\cdot \ed(x,y)$. Also, since $i\le n-o(n)$, we have that
$\distG(X_{i,n},Y_{i,n})\ge
\ed(X_{i,n},Y_{i,n})=\ed(x,y)$.

Concluding, the algorithm produces a $O(C_n/\eps)$ approximation to
$\ed(x,y)$, with probability $\ge 0.9$. Note that $C_n$ is a constant,
depending on $\eps$, as we have only a constant number of iterations,
each incurring a constant factor approximation.

The runtime guarantee follows from time guarantees of
Theorems~\ref{thm::distanceOracle}, \ref{thm::matching_guarantee},
and \ref{thm:cadGuarantees}. In particular, by
Theorems~\ref{thm::distanceOracle} and~\ref{thm:cadGuarantees}, the
runtime $T_\cad$ is $\tilde O(\gamma^{O(1)})\cdot
n^{O(\eps)}=n^{O(\eps)}$. Hence runtime in the matching algorithm, for
every fixed $w$, is $n^{1+O(\eps)}$. Thus the graph $G_w$ has size
$n^{1+O(\eps)}$ and the preprocessing time of
Theorem~\ref{thm::distanceOracle} is $n^{1+O(\eps)}$ as well. Overall,
we have a constant number of $w$'s to consider, and thus we obtain a
runtime of $n^{1+O(\eps)}$.

%% file: imperfect-N.tex
\section{Interval Matching Algorithm}\label{sec::int_match}

In this section we describe our main interval matching algorithm, used
to prove Theorem \ref{thm::matching_guarantee}.  The correctness and
runtime complexity analysis will follow in Sections
\ref{sec::match-correctness} and \ref{sec:runtime} respectively.

\input{matching-imp-pre.tex}

\input{matching-desc.tex}

\input{imp-algo.tex}

%% file: matching-imp-pre.tex
Our matching algorithm iterates over a constant number of steps $t$,
each iterating over a constant number of levels $l$.  The following
\textsc{MatchIntervals} algorithm is the main loop over steps.  We note that, in order to maintain high probability
statements, in each step we output $O(\log n)$ output coloring for
each input coloring.
 
 \begin{minipage}{\textwidth}
 \begin{algorithm}[H]
\caption{Matching Algorithm}
\label{alg::matching.main}	
    \hspace*{\algorithmicindent} \textbf{Input:} \hspace{2mm} Base cost $\basec$ \\
    \hspace*{\algorithmicindent} \textbf{Output:} A matching graph $G=G_{w,\basec}$.
\begin{algorithmic}[1]
\Function{MatchIntervals}{$\basec$}
	\State $\Kappa_0 \gets $ \{initialize a new coloring of a single color $\un$ assigned with mass 1 to all $\IC$\}.
	\State $G \gets$ unweighted and undirected graph with nodes $I \in \IC_w$ and no edges. 
	\For {$t=1,2,\ldots$}
		\State $\Kappa_t \gets \emptyset$.
		\For {$\{\kappa' \in \Kappa_{t-1} \mid \mu_{\kappa'}(\IC,\nu \setminus \bot) > 0 \}$}\label{alg::halt_cond}
			\For {$O(\log n)$ times}
				\State $(G,\Kappa_t) \gets (G,\Kappa_t) \uplus$\footnote{By $\uplus$, we refer to the coordinate-wise union, meaning we add the output edges to $G$ and coloring $\kappa$ to $\Kappa_t$.}%
					$\textsc{MatchStep}(\basec, \kappa',t)$.
			\EndFor
			\EndFor
			\State \textbf{break if} $\Kappa_t = \emptyset$.
		\EndFor
	\State \Return $G$.
\EndFunction
\end{algorithmic}
\end{algorithm}
 \end{minipage}
 
 \vspace{3mm}
 
It remains to describe the \textsc{MatchStep} algorithm, which
iterates over a constant number of levels. In each such level, the
algorithm updates the coloring $\kappa$ under construction, to obtain
the ``amended'' coloring $\widehat\kappa$, while also using the ``step
input coloring'' $\kappa'$, obtained at the end of the previous step.
At the end of the iteration over the levels, the algorithm
\textsc{MatchStep} produces a number of edges to add to the graph as
well as a coloring.  The pseudo-code for \textsc{MatchStep} is
presented in Alg.~\ref{alg::matching.new2}, after a detailed
description of the mechanisms and subroutines of \textsc{MatchStep}.

%% file: matching-desc.tex
\subsection{Components of \textsc{MatchStep}: setup and notations}

We first introduce basic notions used in our algorithm. To help the
reader keep track of the many definitions and notations, we summarize
them in Table~\ref{tab:notation-setup} and
Table~\ref{tab:notation-new} for quick reference (the latter table is
in Section~\ref{sec::match-correctness}).
To avoid confusion, we use $\basec$ to denote $c$ from the theorem
statement. Also, to simplify notation, we refer to the oracle
$\cad_{w,\basec}$ simply as $\cad$.  Note that $\cad$ algorithm uses
the data structure $\distG_{w/\lambda}$, which we assume all our
algorithms in this section have access to.

\vspace{2mm}
\paragraph{The set of costs $E_\basec$.}
For each fixed {\em base cost $\basec\in S_w$}, we use a fixed set of costs
$E_\basec\subset \basec\cdot [1,O(1)]$. The set $E_\basec$ is defined as $\{c_1, c_2,\ldots
 c_{1/\eta}\}$, where $c_i = \tfrac{\basec}{\eps \cdot \alpha} \cdot 3^i$, %
 for small constants $\alpha, \eta > 0$ (to be fixed later).

\vspace{2mm}
\paragraph{Colorings.}
 Recall our definition of {\em coloring} $\kappa$ %
 over color-set $\nu$ as a mapping from intervals to {\em
   distribution of colors} in $\nu$, denoted by
 $\mu_\kappa: \IC\times \nu \rightarrow [0,1]$.

Our construction analyzes pairs of intervals and colors $(I,\chi) \in
\IC \times \nu$. For a set $\RC \subseteq \IC \times \nu$ of (interval, color) pairs, $\mu_\kappa(\RC)$ is the $\ell_1$ mass of colors in $\RC$, i.e. $\mu_\kappa(\RC) \triangleq \sum_{(I,\chi) \in \RC} \mu_\kappa(I,\chi)$. We often use the shorthand {\em pairs} when referring to (interval, color) pairs.

We use colorings to partition $\IC$ into smaller (overlapping)
parts. For a color $\chi$, we denote its part by $\PC_\kappa
^\chi \triangleq
\{ I \in \IC \mid \mu_\kappa(I,\chi) > 0\}$.

Finally, we equip each $\mu_\kappa$ with a data structure that allows
efficient sampling from it (Theorem~\ref{thm::ds_sampling}). The reader should henceforth %
consider that sampling takes time proportional
to the output size, up to $\poly(\log n)$ factors.

\vspace{2mm}
\paragraph{Proximity balls $\Lambda$.}
For coloring $\kappa$ and $\zeta \in \N$, we define
$\Lambda_\kappa^\zeta(I)$ as the %
largest interval ball around $I$
containing at most $\zeta$ \un-colored $\ell_1$-mass on each of left and right
of $I$, where $\mu_\kappa(I,\un)$ is counted on both sides. %
An exception to the above is for $\zeta=0$, when $\Lambda_\kappa^0(I) = \{ I \}$.
We note that for all $I \in \IC$, $\zeta \geq 1$, we have
$\mu_\kappa(\Lambda_\kappa^\zeta(I),\un) \leq 2\zeta$.

Extending the data structure for $\mu_\kappa$ above, we also use the data
structure from Theorem~\ref{thm::ds_sum} to be able to compute the
boundaries of any $\Lambda_\kappa^\zeta(I)$, given $I,\zeta$, in
$\poly(\log n)$ time.

\vspace{2mm}
\paragraph{Interval and Pair densities.}
For $I \in \IC$, $\kappa$, color $\chi \in \nu$ and interval set $\SC
\subseteq \IC$, we define a density parameter $\den_\kappa(I,\chi,\SC)_{\hat c} \triangleq \mu_\kappa(\NC_{\hat c}(I) \cap \SC,\chi)$. We  use the shorthand $\den_\kappa(I,\chi)_{\hat c} \triangleq \den_\kappa(I,\chi,\IC)_{\hat c}$.
We also define the density
vector $\den_\kappa(I,\chi,\SC)
\in \R_+^{E_\basec}$ such that:
$$\den_\kappa(I,\chi,\SC) \triangleq [\den_\kappa(I,\chi,\SC)_{\hat c}]_{\hat c \in E_\basec}$$

We also define {\em relative density} as follows:
\begin{definition}[Relative density.]
  Fix $I \in \IC$, $\kappa$, color $\chi$ and an interval set $\SC \ni I$. The relative density $\relden_\kappa(I,\chi,\SC) \in [1,\infty)^{E_\basec}$ is the density of $I$ w.r.t. each cost $\hat{c} \in E_\basec$, divided by such density restricted to the set $\SC$, i.e.,
		$\relden_\kappa(I,\chi,\SC) \triangleq \den_\kappa(I,\chi) \oslash \den_\kappa(I,\chi,\SC)$.\footnote{$\oslash$ denotes the Hadamard coordinate-wise division.}
\end{definition}

By convention, we set the relative density of empty colors to 1 (or $\infty$ if empty in $\SC$ only). Note
that $\relden_\kappa(I,\chi,\SC) \geq 1$ and is monotonically
decreasing in $\SC$. Both are important properties that will be used
when we bound density mass in ``growing balls of intervals''.

\input{notation-match.tex}

Our main matching algorithm uses estimates of densities of (intervals,
color) pairs. In order to estimate the densities fast, we use standard
sampling, as implemented by
algorithms $\textsc{ApproxDensity}$ and
$\textsc{ApproxRelativeDensity}$
(Alg.~\ref{alg::matching.approx_densities}), whose guarantees are as
follows.

\begin{lemma}[Approximating Densities, proof in Section~\ref{sec:keyCorrectness}]\label{lm::approx_density}
	Fix interval $I \in \IC$, interval ball $S \subseteq \IC$,
        color $\chi$ in coloring $\kappa$, and cost $c\in\E_\basec$.
        Then, recalling that $T_{\cad}$ is the runtime of a
        $\cad(\cdot,\cdot)$ oracle call: 
	\begin{enumerate}
		\item For any given minimal density $\den_m > 0$, the algorithm
                  $\textsc{ApproxDensity}$ 
                  outputs $\widehat{\den} \in [\Omega(1),1] \cdot \max\{\den_\kappa(I,\chi,\SC)_c,
                  \den_m\}$
                  whp, in time $T_{\text{D}} =
                   T_{\cad} \cdot  \tO\left(\tfrac{\mu_\kappa(\SC,\chi)}{\den_m} + 1\right)$.
		\item Assume $I \in \SC$. For any given minimal
                  relative density $\relden_m \geq 1$, the algorithm
                  $\textsc{ApproxRelativeDensity}$ 
                 outputs $\widehat{\relden} \in [\Omega(1),1] \cdot \max\{\relden_\kappa(I,\chi,\SC)_c,
                  \relden_m\}$ 
                  whp, in time 
                  $$T_{\text{RD}} =
                  T_{\cad} \cdot \tO\left(\tfrac{1}{\relden_m}\cdot\tfrac{\mu_\kappa(\IC,\chi)}{\mu_\kappa(I,\chi)}
                  +\tfrac{\mu_\kappa(\SC,\chi)}{\mu_\kappa(I,\chi)}\right).$$
	\end{enumerate}

\end{lemma}

\begin{algorithm}[H]
\caption{Matching Phase: Approximating Densities}
\label{alg::matching.approx_densities}
    \hspace*{\algorithmicindent} \textbf{Input:} \hspace{2mm} interval
    $I$, color $\chi$, interval set $\SC$, cost $c$, additive parameter
    $\den_m$\ /\ $\relden_m$ and access to $\mu_\kappa$ \\
    \hspace*{\algorithmicindent} \textbf{Output:} $\Theta(\max\{\den_\kappa(I,\chi,\SC)_c,
                  \den_m\})$ and $\Theta(\max\{\relden_\kappa(I,\chi,\SC)_c,
                  \relden_m\})$ respectively
\begin{algorithmic}[1]
\Function{ApproxDensity}{$I, \chi, \SC, \den_m, c, \mu_\kappa$}
\State $\gamma \gets O\left(\tfrac{\log n}{\den_m}\right)$ %
\State $\SC' \gets$ Sample each $I' \in \SC$ independently with probability $\min\{\gamma \cdot \mu_{\kappa}(I',\chi),1\}$.%
\State $\SC^* \gets \NC_{c}(I) \cap \SC'$ by computing $\cad(I,I')$ for all $I' \in \SC'$ and keeping all $c$-matches.
\State $\widetilde{\SC^*} \gets \{I' \in S^* \mid \gamma \cdot \mu_{\kappa}(I',\chi) > 1\}$.
\State \Return $\Theta\left(\max\left\{\mu_\kappa(\widetilde{\SC^*},\chi) + \tfrac{1}{\gamma}\cdot|\SC^* \setminus \widetilde{\SC^*}|,\den_m\right\}\right)$. %
 \Comment $\Theta \gets$ rescaling to get $[\Omega(1),1]$-approx.
\EndFunction

\Function{ApproxRelativeDensity}{$I, \chi, \SC, \relden_m, c,\mu_\kappa$}
\State $\hat{d} \gets \textsc{ApproxDensity}(I, \SC, \mu_\kappa(I,\chi), c,\mu_\kappa)$.
\State $\hat{D} \gets \textsc{ApproxDensity}(I, \IC, \relden_m \cdot \mu_\kappa(I,\chi), c,\mu_\kappa)$.
\State \Return $\Theta\left(\max\left\{\tfrac{\hat{D}}{\hat{d}},\relden_m\right\}\right)$. %
 \Comment $\Theta \gets$ rescaling to get $[\Omega(1),1]$-approx.
\EndFunction

\end{algorithmic}
\end{algorithm}

\vspace{2mm}
\paragraph{Soft transformations.}
We define a couple of ``soft'' transformations used by the algorithm:
soft thresholding, and soft quantile. Their purpose is to replace
``hard'' thresholds, thus balancing complexity vs correctness.

The {\em soft thresholding} transformation helps
us with preserving sparsity of palettes. For $\delta \in (0,1]$ and $\gamma > 0$, define $T^{q}_{\delta,\gamma}:
\R^d_+ \rightarrow \R^d_+$ to be the transformation:

$$T^{q}_{\delta,\gamma}(x)_i \triangleq \begin{cases}
	x_i	&	x_i \geq \gamma \\
	0	& x_i < \delta \gamma \\
	\gamma \cdot (\tfrac{x_i}{\gamma})^q	&	\text{otherwise.}
\end{cases}$$

The basic intuition is that $T^{q}_{\delta,\gamma}$ softens the
threshold at about $\gamma$ to decay (polynomially) between $\gamma$
and $\delta\gamma$, when it becomes 0. Our algorithms use a few thresholding transformations with different
parameters.

Also, define the {\em soft quantile} transformation $Q_{\delta,s,F}: \R^d_+ \rightarrow
\R$, for parameters $\delta, s\in[0,1]$, and $F\ge 1$:
$$Q_{\delta,s,F}(x) \triangleq \max_{J \subseteq [d]} a_{|J|} \cdot \min_{j \in J} x_j; 
\hspace{5mm} 
\text{ where } a_i = \begin{cases}
	1	& i \geq s\cdot d \\
	0	& i < (s-\delta)d \\
	1/F^{sd-i}	&	\text{otherwise.}
\end{cases}$$

We use one such transformation: $Q_l(x) \triangleq
Q_{\tfrac{\epsilon}{4},1-\tfrac{(l+1)\epsilon}{2},n^{O(\alpha^2)}}(x)$. The
intuition is that $Q_l$ is a smoothing between
$\tfrac{(l+1)\epsilon}{2}$ fractional rank element from $x$ (in
sorted order), to $\tfrac{(l+1)\epsilon}{2}+\eps/4$ fractional rank.

We discuss and prove the properties of these soft transformations in
Section~\ref{sec:keyCorrectness}.

\subsection{Components of \textsc{MatchStep}: main coloring procedure}

To amend the coloring $\kappa$, we compute a set of potential scores, $\phi, \varphi$,
and $\theta$. In particular, we sample {\em anchor} intervals to
define $\phi$ potential scores. Then, such scores are divided over
$\Lambda$-balls centered at the anchors and $\cad$-close intervals, to
generate $\varphi$ palettes to all other intervals. We also sample
{\em pivots} to define $\theta$ potential scores, and use it to
augment $\varphi$.  Lastly, we use $\varphi$ palettes to amend
coloring $\kappa \rightarrow \hkappa$ for the next level. We describe
this procedure in detail next.

\paragraph{Steps and levels.}  
In each step $t \in [\log_\lambda n]$, where $\lambda= n^\epsilon$,
for a given input coloring $\kappa'$, we produce $O(\log n)$ output
colorings. The goal for step $t$ is that for each interval $I \in
\IC$, either: (1) we cluster it together with $\pi[I]$ and mark them
as ``already matched'', or (2) color $I$ similarly to $\pi[I]$ up to
some {\em bounded distortion} (while also ensuring the color parts are
decreasing with $t$). Here, our goal is to efficiently assign colors
to intervals in $\IC$ (which can be thought of as overlapping parts in
a partition), so that we can compare sampled anchors to some limited
number of other intervals that share the same color. Overall, as $t$
increases, the number of colors in the coloring grows and the size of
each color (number of intervals of that color) becomes smaller,
allowing us to increase the number of sampled anchors.

We maintain the following set of {\em coloring properties} for a
coloring $\kappa$ at a step $t$, which uses the color-set $\nu =
[\lambda^t] \cup \{ \un, \bot\}$, analyzed and proved in later sections:

 \begin{itemize}
 	\item {\em Sparsity:} Each color $\chi \in \nu \setminus
          \{\bot \}$ will be non-zero for (i.e., shared by) few intervals:
          $|\PC_\kappa^\chi| = O(n^{1+O(\eps)}/\lambda^t)$. See Lemma~\ref{lm::partition_size_imp}.
 	\item {\em Correctness:} Coloring $\kappa$ will have {\em
          bounded corruption} (as per Eqn.~\eqref{eqn:xiDef}), meaning ``accumulated error'' is
          bounded by $O(1)$ times the ``original error'',
          $O(\ed(x,y))$. See Lemma~\ref{lm::corruption_growth_imp}.
 \end{itemize}

At each step, our algorithm iterates over {\em levels} $l \in
\{0,\ldots,\log_\beta n\}$, where $\beta=n^\eps$. Intuitively, each
level takes case of different {\em density regimes} of pairs: lower levels will correspond to {\em high-density pairs} and high levels to {\em low-density pairs}. At the lowest level $l=0$, our goal will be to cluster and match high density pairs, and mark them as {\em already-matched}. 
In the subsequent levels, our goal is to color balls of intervals which have large mass of pairs of the corresponding densities.

\paragraph{Anchor Sampling.} In each level $l\ge 0$, for each output color
 $\chi \in \nu_t \setminus \{\un,\bot\}$, we sample an {\em anchor},
which is an interval, color pair $(A, \chi') \in \IC \times \nu_{t-1}
\setminus \{\bot\}$, from the distribution
$\tfrac{\mu_{\kappa'}(*,*)}{|\IC|}$ (if sampled $\chi'=\bot$, we skip
this anchor/$\chi$). Each such anchor is compared to all other
intervals in $\PC_{\kappa'}^{\chi'}$ to form a family of $O_\eps(1)$ clusters
$\CC_\chi$. In level $l=0$, we
add a graph-$G$ edge from $A$ to each clustered interval, and mark such
intervals as {\em already-matched} with $\bot$-color.  In the other
levels $l\ge1$, such clusters are extended to other intervals in the
clusters' {\em proximity}, adding $\chi$ to their palettes $\varphi$,
as described in detail below.

\paragraph{Clustering.}
For each sampled pair $(A, \chi')$, we estimate the $\cad$ distance between $A$ and all other intervals $I
 \in \PC_{\kappa'}^{\chi'}$, using the $\cad$ oracle. We also sample a cost $\hat{c} \sim E_\basec$ 
 uniformly at random.
Next consider subsets $\A_j \subseteq \PC_{\kappa'}^{\chi'}$ defined as $\A_j =
\NC_{\hat{c} + \basec\cdot j}(A) \cap \PC_{\kappa'}^{\chi'}$, for $j \in
\{0,\ldots,j_{\max}\}$ where $j_{\max}=O(1/\alpha)$ ($\alpha$ is still
tbd). While the use of such sampling process will be shown later, the important {\em clustering property} to note here is that for $j < j_{\max}$, we have $\NC_{\basec}(\A_j) \cap \PC_{\kappa'}^{\chi'} \subseteq \A_{j+1}$.
The formal clustering algorithm {\sc ClusterAnchor} is presented in Alg.~\ref{alg::imp.subroutines}.

\begin{algorithm}[H]
\caption{Matching Phase: Clustering around random anchors}
\label{alg::imp.subroutines}	
    \hspace*{\algorithmicindent} \textbf{Input:} \hspace{2mm} Anchor pair $(A,\chi')$, interval set $\RC$, base cost $\basec$ \\
    \hspace*{\algorithmicindent} \textbf{Output:} ``slowly-growing''
    clusters around $A$: $\{(\A_{\hat{c}+\basec j}, \hat{c}+\basec j,j,d_\A)\}_j$ for a random cost $\hat{c} \sim E_\basec$.\\
      \hspace*{\algorithmicindent}\hspace{18mm} where $d_\A$ is an upper-bound for $\{d_{\kappa'}(I,\chi')_{\hat c + \basec j} \mid (I,j) \in \IC \times [j_{\max}]\}$
\begin{algorithmic}[1]
\Function{ClusterAnchor}{$A,\RC, \chi', \basec$}
\State Sample $\hat{c} \sim E_\basec$ uniformly at random.
\State For all $I \in \RC$, compute $c_{A,I}\triangleq\cad(A,I)$.
\State Below, for a parameter $\tau$, we use notation $\A_\tau = \{I : c_{A,I} \leq \tau\}$.
\State \Return $\{(\A_{\hat{c}+\basec j}, \hat{c}+\basec j,j,\mu_{\kappa'}(\A_{3\hat{c}},\chi')) \mid j=0,1,\ldots,j_{\max}\}$.
\EndFunction
\end{algorithmic}
\end{algorithm}

\paragraph{Coloring: assigning $\phi$ potential to clusters.}
Next we assign potentials $\phi$ to the clustered
intervals in $\CC_\chi$.
Later, using potentials $\phi$, we will assign $\varphi$ potential to other nearby
intervals in the proximity of the clustered intervals (as described in Section~\ref{sec:toMatching}).

\ns{amended...}
From the above clustering algorithm, for a fixed output color $\chi$,
and corresponding, sampled
$A,\chi',\hat{c}$, we get a family of clusters
$\CC_\chi \triangleq \{(\A_j,c_{\A_{j}},j,d_\A)\}_j$ holding the clustering
invariant as described above. %
We then define potential $\phi_\kappa(I,\chi)$ 
in new coloring $\kappa$ for all intervals $I\in\cup_j \A_j$, using
Alg.~\ref{alg::imp_matching.process_cluster}. 

Intuitively, for sampled $(A,\chi',\hat c)$, we would like to
distribute $n/\lambda^t$ potential credits ``equally'' among
$I\in\PC_{\kappa'}^{\chi'}\cap \NC_{\hat c}(A)$, namely
$\phi_{\kappa}(I,\chi)=\tfrac{2n}{\lambda^t}\cdot
\tfrac{\mu_{\kappa'}(I,\chi')}{\den_{\kappa'}(A,\chi')_{\hat c}}$
(note that this sums up to $2n/\lambda^t$ over all $I$, and to $2n$ over
all anchors/$\chi$'s). This method however does not satisfy the
necessary $\phi$ properties, requiring couple adjustments. Before
describing the adjustments, we state these necessary properties,
termed {\em $\phi$ scoring
  invariants}, which we will guarantee:

\begin{claim}[$\phi$ Invariants, proved in Section~\ref{sec:keyCorrectness}]\label{cl::phi_uncorrupted}
	Fix step $t$ and level $l$. Fix a color $\chi$ in an output coloring $\kappa$,
          for which we 
          have sampled an anchor pair $(A,\chi'') \in \IC \times \nu_{t-1}
          \setminus \{\bot\}$ where $\chi''$ is some color
          in $\kappa'$, and a sampled cost $\hat c\in E_\basec$. The $\phi$ scores from
        Alg.~\ref{alg::imp_matching.process_cluster} satisfy the following invariants:

 \begin{enumerate}
 	\item {\em Correctness:}
          For any $I
          \in \IC$ and distortion $F\ge 1$, if $(I,\chi')$ is {\em not
            $F$-corrupted pair} (for some fixed alignment $\pi$), then $\dd_{O(F \cdot
            n^{2\alpha})}(\phi_\kappa(I,\chi),\phi_\kappa(\pi[I],\chi))
          \leq n^{-10}$.
 	\item {\em Maximal Contribution}: %
 	For all $(I,\chi') \in \IC \times \nu_{t-1}\setminus \{\bot\}$, the
          expected potential contribution of $(I,\chi')$ to $\phi_\kappa(I,\chi)$ satisfies:
 	$$\E_{A,\chi'',\hat c}\left[\phi_{\kappa}(I,\chi) \cdot
          \1[\chi'=\chi''] \right] = O\left(\tfrac{\mu_{\kappa'}(I,\chi')}{\lambda^{t}}\right).$$

 	\item {\em Minimal Contribution/Balance}: %
 	For all $(I,\chi') \in \IC \times \nu_{t-1}\setminus \{\bot\}$, the
          expected potential contribution of $(I,\chi')$ to $\phi_\kappa(I,\chi)$
          at least equals its mass
          $\mu_{\kappa'}(I,\chi')$ on almost all costs in $E_\basec$:
          in particular, for all but $O(\eta/\alpha)$ fraction of
          costs $\hat c\in E_\basec$:
    $$\E_{A,\chi''}\left[\phi_{\kappa}(I,\chi) \cdot
          \1[\chi'=\chi''] \right] \ge \tfrac{\mu_{\kappa'}(I,\chi')}{\lambda^{t}}.$$
     
	\end{enumerate}
\end{claim}

The first invariant ensures uncorrupted pairs $(A,\chi')$
add uncorrupted $\phi$ potential, in particular that
$I$ and $\pi[I]$ get similar potential
$\phi$. To guarantee the invariant, we use the weaker transitivity property, namely the clustering
property that
$\NC_{\basec}(\A_j)  \cap \PC_{\kappa'}^{\chi'} \subseteq \A_{j+1}$, and
hence that for uncorrupted pairs, $I \in \A_j \Rightarrow \pi[I] \in
\A_{j+1}$ (this is the reason we have multiple clusters $\A_j$ to
start with). In addition, we approximate densities with a threshold lower-bounded by a factor $n^{O(\alpha)}$ from the maximum density of the cluster, so that all approximated densities are within a $n^{O(\alpha)}$ bound (this treshold also helps maintaining efficiency constraints). An additional caveat is that we cannot use this argument
for $j=j_{\max}$. To fix this, our algorithm multiplies the potential of
each $\A_j$ by an exponentially decreasing coefficient $\gamma_j =
n^{-j \alpha}$, which ensures that each meaningful potential added to
$I$ generates a similar potential in $\pi[I]$.

The second invariant ensures that we do not assign too much mass to any pair
(notably a corrupted one). The na\"ive assignment of $\phi$ would
jeopardize this invariant because the corrupted intervals
may be ``$\cad$-centers'', i.e., slightly denser than their neighbors (for
every cost), and hence receive more $\phi$ in expectation. To overcome
this issue, we estimate the density $\den_{\kappa'}(I,\chi')_{c_\A}$ of each
$I \in \A$, denoted
$\widehat{d_{I,\chi'}}$, and add potential proportional to
$1/\widehat{d_{I,{\chi'}}}$ (instead of
$1/\den_{\kappa'}(A,\chi')_{\hat c}$). 

The third invariant guarantees we assign enough potential overall at each step and its importance will become clear later once we discuss the balance of colors. %

The formal algorithm {\sc AssignPhiPotential} appears in Alg.~\ref{alg::imp_matching.process_cluster}.

\begin{algorithm}[H]
\caption{Matching Phase: assign potential to Intervals in $\XC, \YC$}
\label{alg::imp_matching.process_cluster}
    \hspace*{\algorithmicindent} \textbf{Input:} \hspace{2mm} Output color $\chi$, cluster $(\A,\chi')$, density bound $d_\A$, cost $c$, and a ``decaying'' parameter $j$. \\
    \hspace*{\algorithmicindent} \textbf{Output:} Assign $\phi$ potential to all $I \in \A$.    
\begin{algorithmic}[1]

\Function{AssignPhiPotential}{$\chi,\A,d_\A,\chi',c,j, \mu_{\kappa'}$}
\State $\gamma \gets n^{-\alpha \cdot j}$.
\For{ $I \in \A$ }
					\State $\widehat{d_{I,j}} \gets %
					\textsc{ApproxDensity} (I,
                                        \chi', \IC, d_{\A} \cdot n^{-\alpha}, c, \mu_{\kappa'})%
					$.
					\State $\phi_{\kappa}(I,\chi) \gets \phi_{\kappa}(I,\chi) + \gamma \cdot \tfrac{2n}{\lambda^t} \cdot \tfrac{%
					\mu_{\kappa'}(I,\chi')}{\widehat{d_{I,j}}}$.
				        \EndFor
\EndFunction	
\end{algorithmic}
\end{algorithm}

\paragraph{Coloring: assigning $\varphi$ scores to intervals.}
Given $\phi$ potentials, we assign $\varphi$ scores to other
non-clustered intervals in the proximity of the clustered
ones. Intuitively, we would like to assign potential $\varphi$ to each
interval within the ball (in index distance) of fixed radius centered at any
clustered interval. However, sometimes we need to group together far
sections, and hence we define ball radiuses with respect to fixed {\em
  \un-color} mass (in the current coloring $\kappa$).

Specifically, let %
$Z_l = \beta^l \cdot \{1,2,4,8,\ldots,n/\beta^l\}$
be an exponentially growing set of
radiuses. Then for each $\zeta \in Z_l$, we define potential score vector
$\varphi_\kappa^\zeta(
\cdot,\cdot):\IC\times [\lambda^t]\cup \{\un\}\to [0,1]$ as
follows. For $\chi\in[\lambda^t]$, we define:

\begin{equation}\label{eq::varphi_not0}
\varphi_\kappa^\zeta(I,\chi) = T^\cc\left(\mu_\kappa(I,\un) \cdot \tfrac{\beta^{2l}}{\zeta^2} \cdot \tfrac{\phi_{\kappa}(\Lambda_\kappa^\zeta(I),\chi)}{\beta^l}\right)	
\end{equation}
where $T^\cc = T^{\Theta(\epsilon^{-1})}_{\beta^{-1},\beta^{-3}}$. We
define $\varphi(\cdot, \un)$ later (using potentials $\theta$).

We will guarantee the following {\em $\varphi$ scoring properties} for
all $I\in \IC$
except for intervals whose proximity is sufficiently
corrupted (and hence do not need guarantees):

\begin{itemize}
	\item {\em Correctness:} for any $\zeta \in Z_l$, the
          contribution of $\phi$-potential from $F$-uncorrupted pairs
          to $\varphi_\kappa^\zeta(I,*_{\neq \un})$ matches, up to
          $F^{O(1/\epsilon)}$ factor distortion, the contribution to
          $\varphi_\kappa^{\zeta'}(\pi[I],*_{\neq \un})$, for slightly
          larger $\zeta'$, up to an extra additive error $n^{-10}$.  See Lemma~\ref{lm::dd_dist_imp}.
        \item {\em Complexity:} For any color $\chi$, the number of
          $I\in \IC$ with $\varphi_\kappa^\zeta(I,\chi)\neq 0$ is $\tfrac{n
          \beta^{O(1)}}{\lambda^t}$. This will follow from the fact that
          thresholding $T^\cc$  ensures that
          $\varphi^\zeta_\kappa(I,\chi) = 0$ whenever
          $\mu_{\kappa}(I,\un) \cdot
          \phi_\kappa(\Lambda_\kappa^\zeta(I),\chi) < \zeta^2/\beta^{l+O(1)}$. Implicitly in Lemma~\ref{lm::partition_size_imp}.
\end{itemize}

\paragraph{Coloring: assigning the \un-colored $\varphi(\cdot,\un)$.}
We recall that the $\un$-color is used to 1) efficiently partition
intervals of any density together with their corresponding matches in $\pi$,
and 2) ``group'' together sparse sections that might be far apart.

We assign $\varphi(\cdot,\un)$ in a slightly different manner than
$\varphi(\cdot, *_{\neq \un})$. First, we sample a number
of random pairs $\VC$ termed {\em pivots}, directly estimate their
densities, and assign sparsity $\theta$ scores. We then use $\theta$
scores to assign $\varphi(\cdot, \un)$ scores to nearby intervals in a
similar manner to how we used $\phi$ to assign $\varphi(\cdot, *_{\neq
  \un})$.  Our \un-coloring procedure will have the following {\em
  \un-coloring guarantees}:

\begin{itemize}
   	\item {\em Correctness:} The distortion between uncorrupted $(I,\un)$ and
  $(\pi[I],\un)$ is bounded (as for the non-$\un$
  colors). See Lemma~\ref{lm::dd_dist_imp}.
	\item {\em Sparsity:} At level $l$, for any interval $I \in
          \IC$, for $\delta=\varphi^{\beta^l}_\kappa(I,\un)$, there is a set of pairs $\QC$ in the proximity of
          $\Lambda_{\kappa}^{\beta^l}(I)$ %
           with
          $\mu_{\kappa'}(\QC) = \Omega(\delta\beta^l)$ such that
          all pairs in $\QC$ are sparse on majority of possible costs
          $\in E_\basec$. The exact property will be described in the proof of Claim \ref{cl::Q_U_bound}.
	\item {\em Balance:} for every interval $I \in \IC$, we will
          have $\normo{\mu_\kappa(I,*_{\neq \un})+\varphi_\kappa(I,*)}
          = \Omega_\eps(1)$. This ensures that we can re-normalize the
          coloring at level $l$ with only $O(1)$-factor corruption
          blow-up. See Lemma~\ref{lm::min_color_mass}.
\end{itemize}

The high-level idea is as follows. Consider an interval $I$. If there is a set of pairs $(J,\chi')$ in $I$'s proximity of total mass $\delta \beta^l$, 
where for each such pair the relative density
$\relden_{\kappa'}(J,\chi',\Lambda_\kappa^{\beta^l}(J))_{c}$
is at most $\Theta^*(\tfrac{\totalmu}{\lambda^t \cdot \beta^l})$, for
sufficiently many costs $c \in E_\basec$, then we have the ``sparsity
guarantee'' we need for keeping $I$ \un-colored for the next level with $\delta$ mass. However, if $\delta$ is small enough, then we expect to find sufficiently many intervals of the ``right'' density to color $I$ with non-$\un$ colors to obtain the ``balance'' guarantee as above. 

We assign $\varphi(\cdot,\un)$ in three stages. First, we randomly sample a multi-set $\VC$ consisting of
$O^*(\tfrac{n}{\beta^{l}})$ 
pivot pairs
by including a pair $(I,\chi')\in\IC \times
\nu \setminus \{ \bot \}$ with probability $\mu_{\kappa'}(I,\chi')\cdot \beta^{-l}$,
independently, $k=\beta^{O(1)}$ times (i.e., a sample $(I, \chi')$ can
have multiplicity up to $k$). Second, for each pivot $(V,\chi')
\in \VC$, for each possible radius $\zeta \in Z_l$, we generate {\em
  potential sparsity score} $\theta^{\zeta}_l(V,\chi')$, which can be
thought of as ``the mass of colors which $V$ is relatively sparse on, for many costs''. To obtain that, we iterate over all costs $\hat{c} \in E_\basec$ and estimate an upper bound on 
$\reallywidehat{\relden_{\kappa'}(V,\chi',\Lambda_\kappa^{\zeta}(V))_{\hat{c}}}$, using the approximation algorithm \textsc{ApproxRelativeDensity}.

To maintain near-linear runtime overall, we can afford at most
$O^*(\beta^l)$ time for \textsc{ApproxRelativeDensity}
per pivot (on average),  
and hence we set the ``min threshold'' $\relden_m$ parameter to
$\Theta^*(\tfrac{\totalmu}{\lambda^t \cdot \beta^l})$. Also for the
runtime bound, we would need that the local color-mass
$\mu_{\kappa'}(\Lambda_{\kappa}^\zeta(V),\chi')$ is at most $O^*(\beta^l) \cdot \mu_{\kappa'}(V,\chi')$.
When the latter condition doesn't hold, we do not need to do any testing as the relative density will be lower than the bound we care about on average across all potential local $(V,\chi')$ pairs. %

We use the estimate of relative density to generate $\sigma_{V,\chi'} \in [0,1]$ based on
the number of costs with sparse relative density, using the soft
transformations:

$$\sigma_{V,\chi'} = 
Q_l(T^\theta\left(\Gamma_{l,\zeta}(V,\chi')\right)),$$
where $\Gamma_{l,\zeta}(V,\chi')$ is a vector of
dimension $|E_\basec|$ with
$\Gamma_{l,\zeta}(V,\chi')_{c}=
\min\left\{\reallywidehat{\relden_{\kappa'}(V,\chi',\Lambda_\kappa^{\zeta}(V))_{c}}^{-1}\cdot
\frac{\totalmu \cdot n^{4\alpha}}{ \beta^l \cdot \lambda^t},\ 1\right\}$  and $T^\theta = T^{O(1/\epsilon)}_{1/\sqrt{\beta},1}$.

Then we define the {\em sparsity potential score}
$\theta^{\zeta}_\kappa(V,\chi')$ for sampled pivots $V\in\VC$. As
mentioned above, for pivot pairs where we cannot efficiently estimate
$\sigma_{V,\chi'}$, we set the score to 1. In particular, for $m_\VC(V,\chi')$ denoting the multiplicity of $(V,\chi')$ in $\VC$:
$$
\theta_\kappa^{\zeta}(V,\chi') = %
m_\VC(V,\chi') \cdot \begin{cases}
	\sigma_{V,\chi'}	&	\mu_{\kappa'}(\Lambda_{\kappa'}^\zeta(V),\chi') \leq \lambda \cdot \beta^6 \cdot \beta^l	\\
	1	&	\text{otherwise.}
\end{cases}
$$

The algorithm for computing the $\theta$ potential, {\sc
  AssignThetaPotential}, is presented in
Alg.~\ref{alg::imp_matching.assign_theta}.

Finally, we assign $\varphi_\kappa(I,\un)$ as a function of the 
estimated sparsity potential from $\Lambda_\kappa^{\zeta}(I)$,
using the following formula, for each $\zeta\in Z_l$:

\begin{equation}\label{eq::varphi_un}
\varphi_\kappa^\zeta(I,\un) = %
T^\un\left(\frac{\beta^{2l}}{\zeta^2} \cdot \mu_\kappa(I,\un) \cdot \min\left\{\frac{1}{k} \cdot \sum_{\chi'\in \nu \setminus \{\bot \}}\theta_\kappa^{\zeta}(\Lambda_\kappa^{\zeta}(I),\chi'),1\right\}\right),
\end{equation}

where $T^\un = T^{1/\eps}_{1/\beta,\Omega_\epsilon(1)}$, and $k =
\beta^{O(1)}$ is a coefficient which guarantees concentration for all
{\em potential scores} not omitted by the transformation $T^\un$
as above.

\begin{algorithm}[H]
\caption{Matching Phase: assign $\theta$ potential to pivot pair $V, \chi'$.}
\label{alg::imp_matching.assign_theta}
    \hspace*{\algorithmicindent} \textbf{Input:} \hspace{2mm} Pivot $(V,\chi')$, multiplicity $m$, and level $l$ \\
    \hspace*{\algorithmicindent} \textbf{Output:} Assign $\theta$ potential to $(V,\chi')$  
\begin{algorithmic}[1]

\Function{AssignThetaPotential}{$V,\chi',m,l,\mu_{\kappa'}$}
		\State $\relden_m \gets  n^{4\alpha}\cdot\tfrac{\totalmu}{\beta^l \cdot \lambda^t}$.
				\For{$\zeta \in Z_l$}
					\If {$\mu_{\kappa'}(\Lambda_{\kappa}^\zeta(V),\chi') > \lambda \cdot \beta^6 \cdot \beta^l$}\label{alg::theta_treshold}
					\State $\theta_\kappa^\zeta(V,\chi') \gets m$. %
					\Else
					\For {$\hat{c} \in E_\basec$}
						\State $\reallywidehat{\relden_{\kappa'}(V,\chi',\Lambda_{\kappa}^\zeta(V))_{\hat{c}}} \gets \textsc{ApproxRelativeDensity} (V, \Lambda_{\kappa}^\zeta(V), \relden_m, \hat{c},\mu_{\kappa'})$.
						\State $\Gamma_{l,\zeta}(V,\chi')_{\hat{c}} \gets
\min\left\{\relden_m/\reallywidehat{\relden_{\kappa'}(V,\chi',\Lambda_\kappa^{\zeta}(V))_{\hat{c}}}, 1\right\}$.

					\EndFor
  				\State $\theta_\kappa^\zeta(V,\chi') \gets m \cdot Q_l(T^\theta\left(\Gamma_{l,\zeta}(V,\chi')) \right)$.\label{alg::sigma_contribution}
				\EndIf
			\EndFor
\EndFunction	
\end{algorithmic}
\end{algorithm}

\paragraph{Amending the measure $\mu_{\kappa}(I, *)$ to get
  $\mu_{\hat\kappa}(I, *)$.}  At the end of each level $l$, we update
each measure $\mu_{\kappa}(I,*)$ by moving some of the
$\mu_{\kappa}(I,\un)$ mass according to the $\varphi(I,*)$
potential. We need to ensure that we still obtain a distribution at
the end, and hence we do a certain normalization (rescaling) on
$\varphi$. Since such a rescaling can increase the corruption, we need
to ensure that the renormalization rescales the vector by a constant
factor only.  There is a caveat though, that the renomalization factor
is small only when the added $\varphi(I,*)$ potential is large, which
we can only guarantee when $\mu(I,*_{\neq \un})$ is small. As a
result, we recolor using the sum of $\mu$ and $\varphi$ according to
the following formula\footnote{Note that in the formula, the vector
  $\mu_\kappa(I, *_{\neq \un})$ is considered to be the vector with
  $\un$-coordinate zero-ed out.}:

$$\mu_{\hkappa}(I,*) = 
\mu_\kappa(I, *_{\neq \un}) + \mu_\kappa(I, \un) \cdot
\normalizedEllOne{\mu_\kappa(I, *_{\neq \un}) + \sum_{\zeta \in Z_l}
  \varphi_\kappa^\zeta(I,*)}.$$

\begin{algorithm}[H]
\caption{Matching Phase: Color Intervals}
\label{alg::matching.color_new}
    \hspace*{\algorithmicindent} \textbf{Input:} \hspace{2mm} Coloring $\kappa$ and level $l$ \\
    \hspace*{\algorithmicindent} \textbf{Output:} Amended coloring $\widehat{\kappa}$ computed from $\mu_{\kappa}$,$\phi_\kappa$ and $\theta_\kappa$. 
\begin{algorithmic}[1]
\Function{AmendColoring}{$\kappa,l$}
\State Compute the balls $\Lambda_\kappa^\zeta(I)$ for all $I \in \IC$
and $\zeta \in Z_l$ (using algorithm from Thm.~\ref{thm::ds_sum}).  
\State Compute $\varphi_\kappa^\zeta(I,\un)$ as per
Eqn. (\ref{eq::varphi_un}),(\ref{eq::varphi_base_0}) for all $I \in
\IC,\zeta \in Z_l$ (using data structure from Thm.~\ref{thm::ds_sum}).
\State Compute $\varphi_\kappa^\zeta(I,*_{\neq \un})$ as per
Eqn. (\ref{eq::varphi_not0}),(\ref{eq::varphi_base_not0}) for all $I
\in \IC,\zeta \in Z_l$ (using algorithm from Thm.~\ref{thm::ds_matrix}).
\State Compute $\mu_{\hkappa}(I) = %
\mu_\kappa(I,*_{\neq \un}) + \mu_\kappa(I, \un) \cdot
\frac{\mu_\kappa(I, *_{\neq \un}) + \sum_{\zeta \in Z_l}
  \varphi_\kappa^\zeta(I,*)}{\normo{\mu_\kappa(I, *_{\neq \un})+\sum_{\zeta \in Z_l} \varphi_\kappa^\zeta(I,*)}}$ for all $I \in \IC$.\label{alg::color_line}
\State \Return $\hkappa$.
\EndFunction	
\end{algorithmic}
\end{algorithm}

\vspace{-4mm}

\begin{figure}[H]
\caption{The flow of coloring intervals at level $l$}
\label{fig::level_coloring}
\centering
\includegraphics[width=0.98\textwidth]{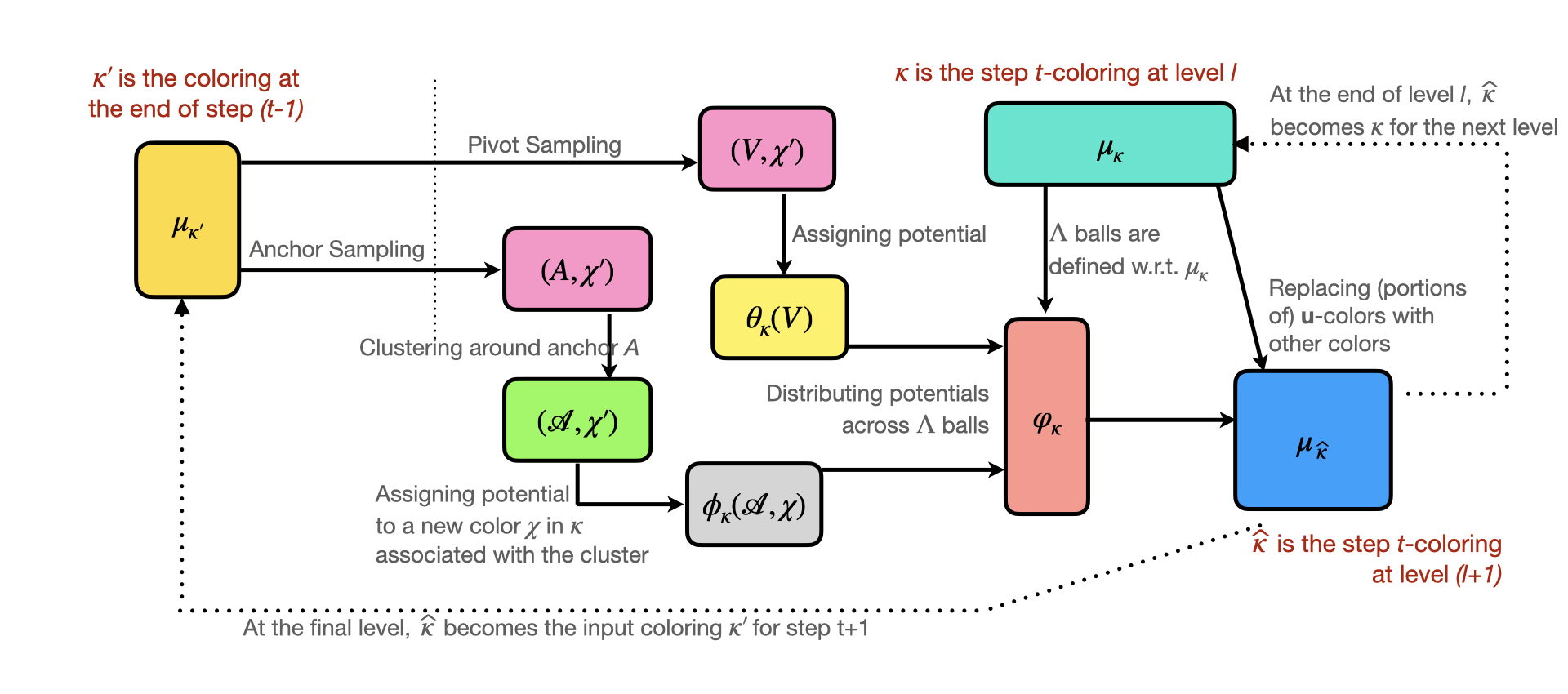}
\vspace{-8mm}
\end{figure}

\begin{figure}[H]
\caption{A step by step coloring example at level $l$. For simplification, all intervals start as $\un$-colored, we use one $\zeta=2$ radius, and omit the pivot sampling, density estimation, and thresholding processes.}
\label{fig::level_potentials}
\fbox{
\begin{tabular}{>{\centering\arraybackslash} m{8.2cm} >{\centering\arraybackslash} m{8.2cm}}
\includegraphics[width=0.49\textwidth]{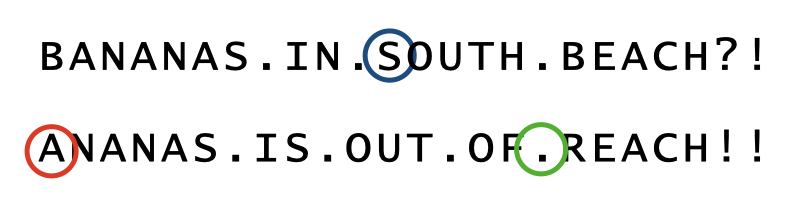} 
&
\includegraphics[width=0.49\textwidth]{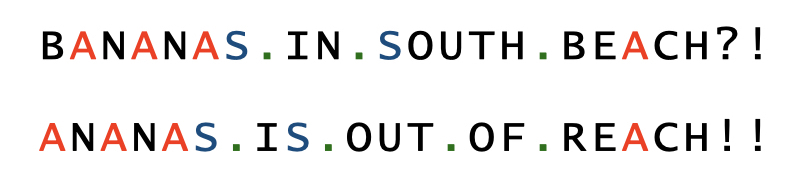} \\
Step 1: Anchor sampling ($\lambda^t = 3$).
&
Step 2: Clustering around anchors.
\\
\includegraphics[width=0.49\textwidth]{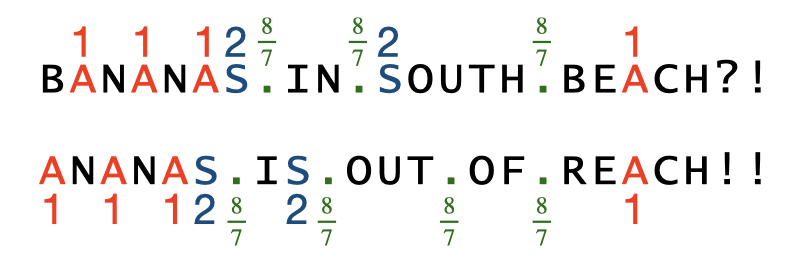} 
&
\includegraphics[width=0.49\textwidth]{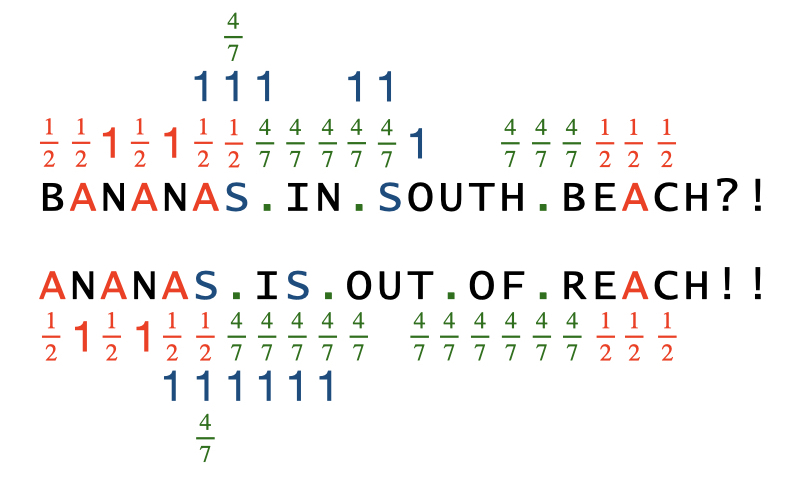} \\
Step 3: Assigning $\phi$ scores to each cluster.
&
Step 4: Assigning $\varphi$ to $\Lambda^\zeta$ balls ($\zeta=2$).
\\
\vspace{3mm}
\includegraphics[width=0.49\textwidth]{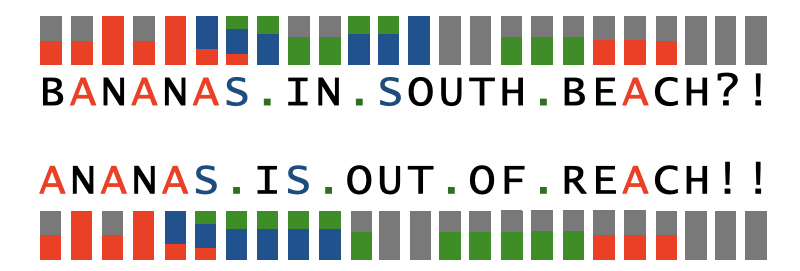} 
&
 \\
Step 5: Coloring $\mu_\hkappa$ by ``adding'' normalized $\varphi$. Gray color is for \color{gray}{\un}. 
&
\\
\end{tabular}}
\end{figure}

\paragraph{Initial level ($l=0$): marking as ``already matched'' $\bot$.} For
each step $t$, at level $l=0$, we assign $\phi$ and $\theta$ potential as
above and use them to mark intervals as ``already matched''
(with color $\bot$), instead of ``regular coloring''.
Hence for any step $t$, we only have the
colors $\un$ and $\bot$ in $\kappa$ at the end of level $l=0$, and the
rest of colors come into play later, starting with level $l=1$. We
use the following potentials, where $Z_0=\{0\}$:

\begin{align}
\varphi_\kappa^0(I,\bot) &= T^\cc\left(\mu_\kappa(I,\un) \cdot
\phi_{\kappa}(I, \nu \setminus \{\un, \bot
\})\right).\label{eq::varphi_base_not0} \\
\varphi_\kappa^0(I,\un) &= T^\un\left(\mu_\kappa(I,\un) \cdot \tfrac{\theta_\kappa^{0}(I,\nu \setminus \{ \bot \})}{k} \right).\label{eq::varphi_base_0}
\end{align}

Overall, the idea here is similar to the general case: any matching from uncorrupted $(I,\chi')$ pairs in $\kappa'$ will generate a 2-hop path in $G$ between $I$ and $\pi[I]$, and hence by projecting all of $\phi$ potential to $\bot$, we do not over-corrupt interval $I$ (in expectation), and maintain the {\em \un-coloring guarantees}. We also note that the $T^\cc$ transformation at level 0 is not required for correctness, but rather enables a more uniform analysis across levels.

The overall algorithm for computing amended coloring, {\sc AmendColoring}, is
presented in Alg.~\ref{alg::matching.color_new}.

%% file: notation-match.tex
\vspace{5mm}

\begin{table}[h!]
\vspace{-1ex}
\caption{Summary of matching algorithm setup notations.}
  \label{tab:notation-setup}
 \centering
 {\renewcommand{\arraystretch}{1.1}
   \begin{tabular}{|l|l|}
    \hline
    {\bf Notation} 	&	{\bf Description}\\
   \hline
   
      \multicolumn{2}{| >{\bf}c |}
{Intervals}
        \\
    \hline
    $\XC, \XC_w$  & space of all $w$-length intervals of $x$.    \\
        \hline
    $\YC, \YC_w$  & space of all $w$-length intervals of $x$.    \\
        \hline
    $\IC, \IC_w$  & space of all $w$-length intervals $= \XC \cup \YC$.    \\
    \hline

      \multicolumn{2}{| >{\bf}c |}
{Distances}
        \\
        \hline
    $(\IC,\cad_w),(\IC,\cad_{w,\basec})$  &
    alignment-distances. $\cad_w$ align-approximates
    $\ed$. $\cad_{w,\basec}$ is a metric.    \\
    \hline
    $(\R_+^d,\dd_F)$  & distortion resilient distance for $F>1$. $\dd_F(p,q) = \sum_{i:p_i > F\cdot q_i} p_i$   \\
        \hline

   \multicolumn{2}{| >{\bf}c |}
 	{Costs, Neighborhood}
        \\
    \hline
        $\basec$& the current base cost for which we are building the
        current graph $G_{w,\basec}$
 	\\
    \hline
        $S_w$& set of possible base costs $\{1/2n,\ldots
 	1/2,1,2,4,\ldots w\}$
 	\\
    \hline
        $E_\basec, c_i$ & $c_i= \tfrac{\basec}{\alpha \cdot \epsilon} \cdot 3^i$ and
    $E_\basec=\{c_1,\ldots c_{1/\eta}\}$, for a small constant $\eta$, dependent on
    $\epsilon$ only
    \\
        \hline
$\NC_c(I)$ & neighborhood of radius (cost) $c$: the set of all $J\in\IC$ with $\cad(I,J)\le c$
    \\
    \hline

   \multicolumn{2}{| >{\bf}c |}
{Colorings}
        \\
    \hline
    $\nu$ & $[\lambda^t]\cup \{\un, \bot\}$ color set
    \\
    \hline
    $\kappa'$ & input coloring obtained from the step $t-1$
    \\
    \hline
    $\kappa$ & current state of output coloring (at step $t$, level $l$)
           \\
 	\hline
    $\mu_\kappa(I,\chi)$ & mass of $(I,\chi)$ in
    $\kappa$; $\mu_\kappa(I,*)$ is a probability distribution
    \\
       \hline    
       $\PC_\kappa
^\chi$ & the support of $\mu_\kappa(*,\chi)$ (is a set of intervals).
    \\
       \hline
   \multicolumn{2}{| >{\bf}c |}
 	{Densities, Proximity Balls}
        \\
    \hline
    $\den_\kappa(I,\chi,\SC)_{\hat c}$ & $= 
    \mu_\kappa(\NC_{\hat c}(I) \cap \SC,\chi)$ is the $\hat c$-density of $(I,\chi)$ in interval set $\SC$   
    \\
    \hline
    $\den_\kappa(I,\chi,\SC)$ & $= 
  \den_\kappa(I,\chi,\SC)_{*}$ is a vector of densities,
    where $*$ ranges over all
    costs in $E_\basec$
    \\
    \hline
    $\relden_\kappa(I,\chi,\SC)$ & $= \den_\kappa(I,\chi, \IC) \oslash \den_\kappa(I,\chi,\SC)$
    is a vector obtained by coordinate-wise division
    \\
    \hline

    $\Lambda_\kappa^\zeta(I)$ & largest interval ball around $I$
containing $< \zeta$ \un-colored $\ell_1$-mass on each side
of $I$    
\\
        \hline
\end{tabular}}
\end{table}

%% file: imp-algo.tex
\subsection{\textsc{MatchStep} algorithm: main levels loop}
\label{sec:matchingMainLoop}

Finally, we describe the %
overall \textsc{MatchStep} algorithm, using the ingredients presented earlier. The
main algorithm is %
{\sc MatchStep} from
Alg.~\ref{alg::matching.new2}. It uses couple more functions: {\sc
  InitColoring}, in Alg.~\ref{alg::matching.init}, and {\sc Adjust-\un} in Alg.~\ref{alg::matching.adjust}.

\paragraph{Choice of Parameters.} We fix the following parameters, as a function of $n$ and $\epsilon$.
\begin{itemize}
	\item $\beta,\lambda \gets n^{\epsilon}$, ensuring convergence
          in constant  number of rounds while allowing {\em sparse} partitions.
	\item $k \gets \tO(\beta^3)$, an oversampling factor for
          concentration in pivot sampling.
	\item $\alpha \gets \epsilon^{5/\epsilon}$, sufficiently small
          constant to control the blow-up of the distortion $F$
          (noting that the starting distortion is $F_0=n^\alpha$).
	\item $\eta \gets \alpha^2\epsilon^3$, to ensure our set of costs is sufficiently large, avoiding blow-up from $Q$ transformations.
\end{itemize}

\begin{algorithm}[H]
\caption{Matching Step Algorithm}
\label{alg::matching.new2}	
    \hspace*{\algorithmicindent} \textbf{Input:} \hspace{2mm} base cost $\basec$, input coloring $\kappa'$, step $t$. \\
    \hspace*{\algorithmicindent} \textbf{Output:} a matching graph $G$ and %
    an output colorings $\kappa$.%
\begin{algorithmic}[1]
\Function{MatchStep}{%
$\basec, \kappa',t$}
				\State Initialize $\kappa \gets
                                \textsc{InitColoring}(\kappa',t)$.
				\State $G \gets$ unweighted and undirected graph with nodes $I \in \IC$ and no edges. 
                                \For{$l=0,1,\ldots$}
                                \State Break if $l>0$ and $\min\{\mu_\kappa(\XC,\un),\mu_\kappa(\YC,\un)\} < \beta^l$; 
					\State $\phi_{\kappa}, \theta_{\kappa} \gets \vec{0}$.%

					\For {$\chi \in [\lambda^t]
                                  $}
				\State Sample $(A,\chi')$ 
			from the distribution $\tfrac{\mu_{\kappa'}(*,*)}{|\IC|}$.%
				\State $\CC_\chi \gets \begin{cases}
	\textsc{ClusterAnchor}(A,\PC_{\kappa'}^{\chi'},\chi',\basec), &	\hbox{if }\chi' \neq \bot 	\\
	\emptyset,	&	\hbox{if }\chi' = \bot	
 \end{cases}$
				\For {$(\A,c_\A,j,d_\A) \in \CC_\chi$}
				\State Add edge $(A,I)$ to $G$ for each $I \in \A$.
\Comment Only important for level 0.
				\State $\textsc{AssignPhiPotential}(\chi,\A,d_\A,\chi',c_\A,j,\mu_{\kappa'})$.
				\EndFor
				\EndFor

				\State $\VC \gets$ Subsample each $(V,\chi')\in\IC \times \nu \setminus \{\bot\}$   with probability $\tfrac{\mu_{\kappa'}(V,\chi')}{\beta^l}$, independently $k$ times.
				
				\State \textsc{AssignThetaPotential}($V,\chi',m_\VC((V,\chi')),l,\mu_{\kappa'}$) for all $(V,\chi') \in \mathrm{supp}(\VC)$.
										\State $\kappa \gets \textsc{AmendColoring}(\kappa,l)$.
			\EndFor
			\State Set $\mu_\kappa(I,\chi) = 0$ whenever $\mu_\kappa(I,\chi) \in (0,n^{-10})$ and renormalize each $\mu_\kappa(I,*)$ to a distribution.\label{alg::remove_tiny}
			\State $\kappa \gets \textsc{Adjust-\un}(\kappa)$.		
	\State \Return $(G,\kappa)$.
\EndFunction
\end{algorithmic}
\end{algorithm}

The function \textsc{InitColoring} initializes an output coloring. It keeps all $\bot$ potentials from the input coloring $\kappa'$ intact, since those are {\em already matched}, and sets the rest to $\un$.

\begin{algorithm}[H]
\caption{Matching Phase - Init Coloring}
\label{alg::matching.init}
    \hspace*{\algorithmicindent} \textbf{Input:} \hspace{2mm} Input coloring $\kappa'$, step $t$ \\
    \hspace*{\algorithmicindent} \textbf{Output:} a new output coloring $\kappa$ where all $\neq \bot$ mass is set to $\un$.
\begin{algorithmic}[1]
\Function{InitColoring}{$\kappa',t$}
\State Let $\nu \gets [\lambda^t] \cup \{\un,\bot\}$.
\State Initialize new coloring $\kappa$ over $\nu$ by setting $\mu_\kappa=0$.
\State $\mu_\kappa(I,\bot) \gets \mu_{\kappa'}(I,\bot)$ for $I \in \IC$.
\State $\mu_\kappa(I,\un) \gets \mu_{\kappa'}(I, \nu \setminus \bot)$ for $I \in \IC$.
\State \Return $\kappa$.
\EndFunction	
\end{algorithmic}
\end{algorithm}

The function \textsc{Adjust-$\un$} ensures the color $\un$ will have the same {\em sparsity guarantees} as the rest of colors in the end of step $t$. Such guarantees will be discussed in Section  \ref{sec::sparsity}.

\begin{algorithm}[H]
\caption{Matching Phase - Adjust $\un$}
\label{alg::matching.adjust}
    \hspace*{\algorithmicindent} \textbf{Input:} \hspace{2mm} Coloring
    $\kappa$ at the end of step $t$ with bounded $\un$-mass (in
    $\ell_1$ sense). \\
    \hspace*{\algorithmicindent} \textbf{Output:} Adjusted coloring
    $\kappa$ with bounded $\un$-support (in $\ell_0$ sense).
\begin{algorithmic}[1]
\Function{Adjust-$\un$}{$\kappa$}
\If {$\tfrac{\mu_{\kappa}(\XC,\un)}{\mu_{\kappa}(\YC,\un)} \notin [1/\beta,\beta]$}
	\State $\mu_\kappa(I,\bot) \gets\mu_\kappa(I,\bot)+\mu_{\kappa}(I,\un)$ for $I \in \IC$.
	\State $\mu_\kappa(I,\un) \gets 0$ for $I \in \IC$.
\EndIf
\State $\left(\mu_\kappa(I,*_{\neq \un}),\mu_\kappa(I,\un)\right) \gets \tfrac{\left(\mu_{\kappa}(I,*_{\neq \un}),T^\un(\mu_\kappa(I,\un))\right)}{\normo{\left(\mu_{\kappa}(I,*_{\neq \un}),T^\un(\mu_\kappa(I,\un))\right)}}$ for $I \in \IC $.\label{alg::matching.adjust.line}
\State \Return $\kappa$.
\EndFunction	
\end{algorithmic}
\end{algorithm}

%% file: imp-correctness.tex
\section{Correctness Analysis of the Interval Matching
  Algorithm}\label{sec::match-correctness}

In this section, we prove correctness of the interval matching
algorithm, namely Theorem~\ref{thm::matching_guarantee}, items
\ref{it:matchingDistances} and \ref{it:matchingError}. Item
\ref{it:matchingRuntime} (runtime) is proven in
Section~\ref{sec:runtime} later. Note that item
\ref{it:matchingDistances} is immediate from the algorithm (as we only
add edges if $\cad$ distance is $\le O(\basec)$). Hence we focus on
item \ref{it:matchingError}. To help the reader in the ensuing proofs,
we collect important notations and definitions in
Table~\ref{tab:notation-new} for quick reference.

Our central correctness lemma shows that the ``corruption''
in each level/step grows by at most a constant factor.
Recall the notion of corruption from
Def.~\ref{def::corruption}: $(I,\chi)$ is {$F$-corrupted pair} if either:
	(1) $\pi[I] = \bot$; 
	(2) $\cad(I,\pi[I]) > \basec$; 
	(3) $\chi \neq \bot$ and $\dd_F\left(\mu_\kappa(I,\chi),\mu_\kappa(\pi[I],\chi)\right) > 0$; or
	(4) $\chi = \bot$ and $\dist_G(I,\pi[I]) > 2$.
Also, recall from  Eqn.~\eqref{eqn:xiDef} corruption per interval parameter
$\xi_F^{\kappa}(I) = \sum_{\chi: (I,\chi) \text{ is } F \text{-corrupted pair}} \mu_{\kappa}(I,\chi)$ and the total corruption is defined as $\xi^\kappa_F = \xi_F^\kappa(\IC)$.

The following central lemma bounds corruption growth per level/step:
\begin{lemma}[Corruption growth per level]\label{lm::corruption_growth_imp}
Fix $\eps,\delta \in [0,1]$, and alignment
         $\pi \in \Pi$. Fix step $t$, level $l$, input
        coloring $\kappa'$ (built at the previous step) and output
        coloring $\kappa$ (being built in the current step). 
          Then, \textsc{AmendColoring} at level $l$ amends $\kappa \rightarrow \hkappa$ such that %
          $\xi^{\hkappa}_{\widehat{F}} =
          O\left(\tfrac{\xi^{\kappa'}_{F} + \xi^\kappa_{F}}{\epsilon^{2/\epsilon} \cdot
            \delta}\right) + \tO\left(n^{-8}\right)$ 
         with probability $1-\delta$, where $\widehat{F} = F^{O(1/\epsilon^2)}$.
\end{lemma}

Recalling that $\xi^{\kappa'}_{F}$ is the corruption at the end of the
previous step, and $\xi^{\kappa}_{F}$ is the corruption at the end of
the previous level (in the current step), the above lemma bounds
 the multiplicative growth of the corruption of the amended
coloring $\hat\kappa$, modulo a very small additive term.  While the
rest of the section is devoted to proving this lemma, we first
complete the proof of Lemma~\ref{thm::matching_guarantee},
item~\ref{it:matchingError}, which requires the following fact for preserving $\dd$-distance on summations :

\begin{fact}
  \label{fct:ddTriangle}
  For any $a_1,\ldots a_m, b_1\ldots b_m \in \R_+^d$, we have that:
  $$
  \dd_{2F}(\sum_i a_i, \sum_i b_i)\le 2\sum_i \dd_F(a_i,b_i).
  $$
\end{fact}
\begin{proof}
For $j \in [d]$, let $L_j \subseteq [m]$ be the set of coordinates where $a_{i,j} > F b_{i,j}$. If $\sum_{i\in L_j} a_{i,j}\le \sum_{i \in [m] \setminus L_j} a_{i,j}$,
then $\sum_i a_{i,j} \le 2 \sum_{i \in [m] \setminus L_j} a_{i,j}\le 2F\sum_i b_{i,j}$, and hence
$\dd_{2F}(\sum_i a_{i,j}, \sum_i b_{i,j})=0$. Otherwise,
$$
\dd_{2F}(\sum_i a_{i,j}, \sum_i b_{i,j})\le \sum_i a_{i,j} < 2\sum_{i \in L_j} a_{i,j}=2\sum_i \dd_F(a_{i,j}, b_{i,j}).$$

Summing over all $j \in [d]$, we get:

$$
\dd_{2F}(\sum_i a_{i}, \sum_i b_{i}) = \sum_j \dd_{2F}(\sum_i a_{i,j}, \sum_i b_{i,j}) \leq \sum_j \sum_i 2\dd_{F}( a_{i,j},  b_{i,j}) = 2\sum_i\dd_{F}( a_{i}, b_{i})
$$

as needed.
\end{proof}

We also state the following complexity statement, bounding the size of
parts $\PC_{\kappa'}^\chi$, the set of intervals $I$ with
$\mu_\kappa(I,\chi)>0$ (in ``step input'' coloring $\kappa'$). Its
proof appears in Section~\ref{sec::sparsity}.

\begin{lemma}[Size of color parts]\label{lm::partition_size_imp}
At each step $t$, for each
        color $\chi' \in \nu \setminus \{\bot\}$,
     we have that	$|\PC_{\kappa'}^{\chi'}| =  n \cdot O_\eps(\beta^{5} \cdot
        \lambda^{1-t})$ whp. 
\end{lemma}

An immediate corollary of Lemma~\ref{lm::partition_size_imp} is that
the total number of steps is bounded by $1/\eps+O(1)$ whp.

\begin{proof}[Proof of Lemma~\ref{thm::matching_guarantee},
    item~\ref{it:matchingError} using Lemma~\ref{lm::corruption_growth_imp}]
	Fix step $t$ with input coloring $\kappa'$. Fix $F_t = n^{\alpha /\Theta(\epsilon)^{3t}}$. We first show that for each output coloring $\kappa$ %
	generated at each \textsc{MatchStep} call, we have with probability $1/2$, at the end of step $t$:
	
		$$\xi^{\kappa}_{F_{t+1}} = \epsilon^{-O(1/\epsilon^2)}
        \cdot \xi^{\kappa'}_{F_t} + \tO_\eps(n^{-8}).$$
	
	To do that, we use Lemma \ref{lm::corruption_growth_imp}, to obtain that in each level $l \leq \log_\beta n$ we have $\epsilon^{-2/\epsilon} \cdot 4\epsilon^{-1}$ factor growth in corruption with probability $1-\tfrac{\epsilon}{4}$, and by the union bound we get overall blow-up $\epsilon^{-O(1/\epsilon^2)}$ with probability $3/4$ (as we have $1/\epsilon$ levels). Observe also that removing pairs with mass $< n^{-10}$ in Line~\ref{alg::remove_tiny} introduces at most $F \cdot n^{O(\eps) - 9} < n^{-8}$ additive corruption, since there are $n^{1+ O(\eps)}$ total non-zero pairs and each pair removed can generate at most $F \cdot n^{-10}$ corrupted mass.
	Finally, notice that by Fact~\ref{fct:ddTriangle}, the excess corruption introduced by \textsc{InitColoring} in any new step is bounded by 2.

	Now, notice we
	generate $O(\log n)$ i.i.d colorings in each step
        $t$ for each input coloring $\kappa'$, hence, we must generate a ``good'' coloring $\kappa$ with
        high probability as long as we started with at least one ``good'' coloring from the previous step. Now, for a fixed $\basec$, define $k_\basec \triangleq |\{ i\mid
          \cad_{w,\basec}(X_{i,w},Y_{\pi(i),w}) > \basec \}|$ and notice we start the
        \textsc{MatchIntervals} algorithm with
        $\xi^{\kappa}_{F_1} = k_\basec$ and generate $O_\eps(1)$ blow-up
        per step.

        Finally, since
        \textsc{MatchIntervals} halts when $\mu_\kappa(\IC,\bot) = 2n$ (this is the halting condition from, 
        Line~\ref{alg::halt_cond}), and by the corollary above the total number of steps is
        $O_\eps(1)$, we have
         $\xi_{F_{t}}^\kappa
        = O_\eps(k_\basec)$ for any ``good'' coloring $\kappa$. Furthermore, for any $\pi$-matchable pair $I,
        \pi[I]$ with $\cad_{w,\basec}(I,\pi[I])\le \basec$, if the hop-path between
        $I,\pi[I]$ in $G$ is more than 2, then this pair contributes 1
        to $\xi_{F_{t}}^\kappa=O_\eps(k_\basec)$. We conclude that all but
        $O_\eps(k_\basec)$ $\pi$-matchable pairs have a 2-hop path in $G$ as needed.
\end{proof}

\subsection{Bounding corruption growth per level: two key lemmas}
\label{sec:twoKeyLemmas}

To prove our central correctness Lemma
\ref{lm::corruption_growth_imp}, we introduce the following two key
Lemmas. We refer the reader to Table~\ref{tab:notation-new} for a
quick recap of important quantities and formulas from our algorithm.

In particular, note that the amended coloring is obtained via the
following formula:
$$
\mu_{\hkappa}(I,*)=\mu_\kappa(I, *_{\neq \un}) +
    \mu_\kappa(I,\un)\cdot \normalizedEllOne{\mu_\kappa(I,*_{\neq
        \un})+\sum_{\zeta\in Z_l} \varphi_\kappa^\zeta(I,*)}
    $$

Thus, in order to bound the growth of $\xi^{\hkappa}_{\widehat{F}}$,
we need to bound the quantities: $\dd(\varphi_\kappa^{\zeta}(I,*),
\varphi_\kappa^{\zeta'}(\pi[I],*))$, for some $\zeta,\zeta'$, as well
as the $\ell_1$ normalization from above. These two goals correspond
to the two key lemmas.

To state the key lemmas, we introduce the following measures $\rho$ of
the corruption of intervals from nearby intervals.  First, define
$\BC_\pi(\SC)$ for $\SC \subseteq \IC$ as the smallest enclosing ball
around $\pi[\SC]$. Now, fix an interval $I \in \IC$, output coloring
$\kappa$, and arbitrary coloring $\kappa'$ (which can be either input
or output coloring). Fix distortion $F$, and set $\SC \subseteq \IC$
with $\mu_\kappa(\SC,\un) > 0$. We define ``local'' corruption
measures $\rho^{\kappa,\kappa'}_{F}(I,\SC)$ and
$\widetilde{\rho^{\kappa,\kappa'}_{F}}(I,\SC)$ as follows:

\begin{align*}
\rho^{\kappa,\kappa'}_{F}(I,\SC) &=
\tfrac{\mu_\kappa(I,\un)}{\mu_\kappa(\SC,\un)} \cdot \xi^{\kappa'}_F(\SC),	\\
\widetilde{\rho^{\kappa,\kappa'}_{F}}(I,\SC) &=
\tfrac{\mu_\kappa(I,\un)}{\mu_\kappa(\SC,\un)} \cdot \xi^{\kappa'}_F(\BC_\pi(\SC)).
\end{align*}

The first lemma below bounds the expected distortion-resistant
corruption of $\varphi_\kappa^\zeta(I,\cdot),\varphi_\kappa^{\zeta'}(\pi[I],\cdot)$
for a fixed interval $I$, as a function of its $\rho$ scores of the
neighborhoods of $I,\pi[I]$. The lemma considers arbitrary $\zeta$
(radius around $I$), and a convenient $\zeta'$ (radius around
$\pi[I]$) as a function of $\pi$.

\begin{lemma}\label{lm::dd_dist_imp}
	Fix step $t$, level $l\ge 0$, 
	and
	alignment $\pi \in \Pi$. Fix input coloring $\kappa'$
        (obtained at the end of step $t-1$), 
        and the current coloring $\kappa$ (obtained at the end of
        level $l-1$ or from \textsc{InitColoring}). Fix $F \in [n^\alpha,\beta^{o(\epsilon)}]$. Consider $I \in\IC$,
        where $(I,\un)$ is not $F$-corrupted pair. 
        For %
        $\zeta\in Z_l$, define $\BC_I =
        \Lambda_\kappa^{\zeta}(I)$, $\BC_I^{+\zeta} =
        \Lambda_\kappa^{2\zeta}(I)$, and let $\zeta'\in Z_l$ be such
        that $\pi[\Lambda_\kappa^{3\zeta}(I)] \subseteq
        \Lambda_\kappa^{\zeta'/2}(\pi[I])$. Let $\zeta_+ = \max\{\zeta,1\}$, and similarly $\zeta'_+= \max\{\zeta',1\}$.
        Then, for $\widehat{F} \triangleq
        F^{\Theta(1/\eps^2)}(\zeta'_+/\zeta_+)^{\Theta(1/\eps)}$:
        \begin{enumerate}
		\item  For $l \geq 1$, in expectation over the random choices of
                  $A,\chi',c$ in the algorithm, we have
                  \begin{equation}\E\left[\dd_{\widehat{F}}\left(\varphi_\kappa^{\zeta}(I,
                  *_{\neq \un}),
                  \varphi_\kappa^{\zeta'}(\pi[I],*_{\neq \un})\right)\right] = O
                  \left(\tfrac{\beta^l}{\zeta}\cdot\rho^{\kappa,\kappa'}_{F}(I,\BC_I)
                  + \tfrac{1}{n^{9}}\right).
                  \label{eqn:ddVarphiColor}
                  \end{equation}
 
		\item For $l \geq 0$, with high probability,
		  \begin{equation}
                    \dd_{\widehat{F}}\left(\varphi_\kappa^{\zeta}(I,\un),
                  \varphi_\kappa^{\zeta'}(\pi[I],
                  \un)\right) =
                  O\left(\tfrac{\beta^l}{\zeta_+}\cdot\rho^{\kappa,\kappa'}_{F}(I,\BC_I^{+\zeta})
                  + \tfrac{1}{n^{9}}\right).
                  \label{eqn:ddVarphiU}
                  \end{equation}
	\end{enumerate}

\end{lemma}

The second key lemma argues that normalizing each vector $\mu(I,*)$
does not add more than constant corruption in each level.

\begin{lemma}\label{lm::min_color_mass}
At the end of level $l\ge 0$, for each coloring $\kappa$, we have $\normo{\mu_\kappa(I,*_{\neq \un}) + \sum_{\zeta \in Z_l}\varphi_\kappa^\zeta(I,*)} = \left(\Omega(\epsilon)\right)^{l+1}$ for each $I \in \IC$ with high probability.
\end{lemma}

The proofs of Lemmas \ref{lm::dd_dist_imp} and
\ref{lm::min_color_mass} are involved and appear in 
Section~\ref{sec:keyCorrectness}.
We prove Lemma~\ref{thm::matching_guarantee} using these two key
lemmas in Section~\ref{sec:proofCorruptionGrowthLevel} after
introducing a few useful facts.

\input{imp-key-lemmas.tex}

\section{Proof of Key Correctness Lemmas \ref{lm::dd_dist_imp},
  \ref{lm::min_color_mass}}
\label{sec:keyCorrectness}

We now prove the key correctness Lemmas \ref{lm::dd_dist_imp},
\ref{lm::min_color_mass}. Before proceeding, we develop several
supporting statements, in particular:
\begin{itemize}
\item
  We argue how soft operators are resilient under the $\dd$ distance;
\item
  We provide a proof for density approximation algorithm, Lemma \ref{lm::approx_density};
\item
  We provide a proof for $\phi$ color assignment properties, Claim \ref{cl::phi_uncorrupted}.
\end{itemize}

\paragraph{Soft operators are resilient under DD.} We show that
$\dd$ distance behaves nicely under the soft operators of
thresholding $T$ and quantile $Q$ (which %
are used precisely for that
reason).

\begin{claim}\label{cl::T_trans}
For any $x,y\ge0$, we have	$\dd_{F^q}(T^{q}_{\delta,\gamma}(x),T^{q}_{\delta,\gamma}(y)) < \dd_F(x,y) + \gamma(\delta F)^q$.
\end{claim}

\begin{proof}
	Let $t_x = T^{q}_{\delta,\gamma}(x)$ and $t_y =
        T^{q}_{\delta,\gamma}(y)$ We first observe that $t_x \leq x$,
        and hence if $x>Fy$, then $\dd_{F^q}(t_x,t_y)\le t_x\le
        x=\dd_F(x,y)$. Now, suppose $x\le Fy$. If $t_y>0$, then
        $t_x\le F^qt_y$ by definition of $T$ transformation (in
        particular, in the interesting regime when $y\le \gamma$, we
        have $t_y= \gamma\cdot (y/\gamma)^q \ge
        F^{-q}\gamma\cdot (x/\gamma)^q\ge t_x/F^q$).
        Otherwise $t_x\leq \gamma \cdot (\delta F)^q$, and hence
        $\dd_{F^q}(t_x,t_y)\le t_x\le\gamma \cdot (\delta F)^q$.
	\end{proof}

\begin{claim}\label{cl::Q_trans}
	Fix $\delta < 1$ and integers $v,d$ with $v < \delta d$. Fix
        $x, y \in [0,1]^d$ with $\norm{x}_\infty \leq 1$. For any set of
        coordinates $U \subseteq [d]$ of size $d-v$ we have:
$\dd_{F' \cdot
          F^v}(Q_{\delta,s,F}(x),Q_{\delta,s,F}(y)) \leq\dd_{F'}(x_{U},y_{U}) +
        1/F^{\delta d-v}$. 
\end{claim}

\begin{proof}
  Let $q_x = Q_{\delta,s,F}(x)$, $q_y = Q_{\delta,s,F}(y)$.
        Consider the set $J$ which maximizes $q_x$. We have $|J \cap U| \geq |J|-v$ and hence either (1) $a_{|J \cap U|} \geq a_{|J|} \cdot F^{-v}$, in which case 
	\begin{align*}
	\dd_{F' \cdot F^v}(q_x,q_y) 
	&\leq \dd_{F' F^v}(a_{|J|}\cdot \min_{j \in J} x_j,a_{|J \cap U|}\cdot \min_{j \in J \cap U} y_j) \\
	&\leq \dd_{F'}(\min_{j \in J} x_j,\min_{j \in J \cap U} y_j) \\
	&\leq \dd_{F'}(\min_{j \in J \cap U} x_j,\min_{j \in J \cap U} y_j) \\
	&\leq \dd_{F'}(x_{J \cap U},y_{J \cap U}) \\
	&\leq \dd_{F'}(x_{U},y_{U}).		
	\end{align*}

	or (2) $a_{|J|} \leq 1/F^{\delta d-v}$, in which case $\dd_{F' F^k}(q_x,q_y) \leq q_x \leq F^{v-\delta d} \cdot \norm{x}_\infty \leq F^{v-\delta d}$. This concludes the proof.

\end{proof}

\begin{corollary}\label{cr::Q_trans}
  Suppose $v = \Theta(1/\alpha)$ with $\eps d=\Omega(1/\alpha^2)$. 
  Fix $x, y \in [0,1]^d$ with $\norm{x}_\infty \leq 1$. For any set of
  coordinates $U \subseteq [d]$ of size $d-v$ we have:
  $\dd_{F \cdot n^{O(\alpha)}}(Q_{l}(x),Q_{l}(y)) \leq
  \dd_F(x_{U},y_{U}) + n^{-10}$.
\end{corollary}
\begin{proof}
	Apply Claim~\ref{cl::Q_trans} with $F'=F$, $F=n^{C \cdot \alpha^2}$ for some large enough constant $C$.
\end{proof}

\vspace{2mm}
\paragraph{Analysis for Approximating Densities.}
We now analyze the performance of \textsc{ApproxDensity} and \textsc{ApproxRelativeDensity} (Alg.~\ref{alg::matching.approx_densities}), proving Lemma \ref{lm::approx_density}.%

\begin{proof}[Proof of Lemma~\ref{lm::approx_density}]%

For (1), let $d = \den_\kappa(I,\chi,\SC)_c$ and let $\widehat{d} =
\mu_\kappa(\widetilde{\SC^*},\chi) + \tfrac{1}{\gamma}\cdot |\SC^* \setminus \widetilde{\SC^*}|$ be the estimator. Also, let $d' = \den_\kappa(I,\chi,\SC\setminus \widetilde{\SC^*})_c$ and $s = |\SC^* \setminus \widetilde{\SC^*}|$. Note that
$s$ is a sum of independent random variables in
$[0,1]$ with expectation $\gamma d'$. If $d' \ge \Omega(\tfrac{\log
  n}{\gamma})$, then by Chernoff bound, $s$ would concentrate whp:
$s =  \Theta(\gamma d')$. Otherwise, if $d'<O(\tfrac{\log
  n}{\gamma})$, then whp $s/\gamma\le O( \tfrac{\log
  n}{\gamma})=\den_m$ for appropriately chosen constant. Hence the
algorithm outputs whp, up to the $\Theta(\cdot)$ rescaling:
$$\max\{\widehat{d},\den_m\} = \max\{\den_\kappa(I,\chi,\widetilde{\SC^*})_c + \Theta(1)\cdot \max\{\den_\kappa(I,\chi,\SC \setminus \widetilde{\SC^*})_c, \den_m\},\den_m\}= \Theta(1)\cdot\max\{\den_\kappa(I,\chi,\SC)_c, \den_m\}.$$
For run-time, we have from Theorem \ref{thm::ds_sampling}, that generating the set $\SC'$ costs $\tO(\gamma \mu_\kappa(\SC,\chi) + 1) = \tO\left(\tfrac{\mu_\kappa(\SC,\chi)}{\den_m} + 1\right)$, and since $|\SC'|$ is bounded by that amount as well. The total complexity hence is $T_D = T_{\cad} \cdot \tO\left(\tfrac{\mu_\kappa(\SC,\chi)}{\den_m} + 1\right)$.

For (2), we first note that since $I \in \SC$, then $d \geq
\mu_\kappa(I,\chi)$, and since we set $\den_m = \mu_\kappa(I,\chi)$ in
the first call to \textsc{ApproxDensity}, we have that $\widehat{d} =
\Theta(d)$ whp. Also, let $D = \den_\kappa(I,\chi)_c$, noting that $D/d =
\relden_\kappa(I,\chi,\SC)_c$.
Now, consider the second call to \textsc{ApproxDensity}, and let
$\widehat{D}$ be the output; by part (1), we have that
$\widehat{D}=\Theta(1)\cdot \max\{D, \relden_m\mu_\kappa(I,
\chi)\}$.

Combining the bounds on
$\widehat{D}$ and $\widehat{d}$, we have that whp, using that $d \geq \mu_\kappa(I,\chi)$:
$$
\max\left\{\tfrac{\hat{D}}{\hat{d}},\relden_m\right\} 
= \Theta(1)\max\{\tfrac{D}{d}, \tfrac{\relden_m\mu_\kappa(I,\chi)}{d}, \relden_m\} 
= \Theta(1)\cdot\max\{\tfrac{D}{d}, \relden_m\}.
$$

For runtime, we have $T_{RD} = T_D(\hat{d}) + T_D(\hat{D}) + \tO(1) = T_{\cad} \cdot \tO\left(\tfrac{\mu_\kappa(\SC,\chi)}{\mu_\kappa(I,\chi)}+\tfrac{1}{\relden_m}\cdot\tfrac{\mu_\kappa(\IC,\chi)}{\mu_\kappa(I,\chi)}\right)$ as needed.
\end{proof}

\vspace{2mm}
\paragraph{Analysis for $\phi$ invariants.} We now prove the
properties of the $\phi$ potentials, in particular Claim \ref{cl::phi_uncorrupted}.

\begin{proof}[Proof of Claim \ref{cl::phi_uncorrupted}]
Fix pair $(I,\chi')$, and 
$\hat c\in E_\basec$. Assuming that $(I,\chi')$ was clustered by some
anchor $(A,\chi'')$, for fixed $j\in \{0,\ldots j_{\max}\}$, 
let $m_{I,\chi',j} =  n^{-\alpha j} \cdot \tfrac{2\totalmu}{\lambda^t} \cdot
 \tfrac{\mu_{\kappa'}(I,\chi')}{\widehat{d_{I,j}}}$, where $\hat
 c_j=\hat c+j\basec$, and
$\widehat{d_{I,j}}$ the approximation of
 $\max\{\den_{\kappa'}(I,\chi')_{\hat c_j},n^{-\alpha}d_\A\}$ with $d_\A=\den_{\kappa'}(A,\chi'')_{3\hat c}$.
For a color $\chi \in [\lambda^t]$, the contribution of $(I,\chi')$ to
$\phi_\kappa(I,\chi)$ is $\sum_{j \in \{0,\ldots,j_{\max} \}
} m_{I,\chi',j}$, whenever $(I,\chi')$ was clustered by the
anchor at distance $\hat c_j$.

\begin{align*}
          \E_{A,\chi''}\left[m_{I,\chi',j}\cdot \one{(I,\chi') \hbox{ clustered by }
           (A,\chi'',j)}\right] 
          &= %
          \Pr_{A,\chi''}\left[A \in \NC_{\hat c_j}(I) \wedge \chi'=\chi'' \right] \cdot n^{-\alpha j} \cdot \E_{A \in \NC_{\hat c_j}(I)}\left[\tfrac{2\totalmu}{\lambda^t} \cdot \tfrac{%
	        			\mu_{\kappa'}(I,\chi')}{\widehat{d_{I,j}}}\right] \\ 
	          &= n^{-\alpha j} \cdot \sum_{A \in \NC_{\hat c_j}(I)} \tfrac{\mu_{\kappa'}(A,\chi')}{2\totalmu} \cdot \tfrac{2\totalmu}{\lambda^t} \cdot  \E_{A \in \NC_{\hat c_j}(I)}\left[\tfrac{%
	        			\mu_{\kappa'}(I,\chi')}{\widehat{d_{I,j}}}\right] \\
	          &= n^{-\alpha j} \cdot \den_{\kappa'}(I,\chi')_{\hat c_j} \cdot \tfrac{1}{\lambda^t} \cdot  \E_{A \in \NC_{\hat c_j}(I)}\left[\tfrac{%
	        			\mu_{\kappa'}(I,\chi')}{[\Omega(1),1] \cdot \max\{\den_{\kappa'}(I,\chi')_{\hat c_j},n^{-\alpha}d_\A\}}\right]
          \end{align*}	
          
                    The last equality is since $\sum_{A \in \NC_{\hat c_j}(I)}
   \mu_{\kappa'}(A,\chi') = \den_{\kappa'}(I,\chi')_{\hat c_j}$, and that
   $\widehat{d_{I,j}}$ is a constant-factor approximation to
   $\max\{\den_{\kappa'}(I,\chi')_{\hat c_j},n^{-\alpha}d_\A\}$. %
   Summing over all $j$, this immediately gives us the required upper bound
   for item (2), as we need only consider the first component of the $\max$:
    
    $$\E_{A,\chi'',\hat c}\left[\phi_{\kappa}(I,\chi) \cdot
          \1[\chi'=\chi''] \right] = O\left(\tfrac{\mu_{\kappa'}(I,\chi')}{\lambda^{t}}\right).$$
   
   For item (3), recall that $d_\A = \den_{\kappa'}(A,\chi')_{3 \hat c}$ and since $I \in \NC_{\hat c_j}(A)$, we have that $d_\A \leq\den_{\kappa'}(I,\chi')_{6\hat c}$ (by triangle inequality) . %
   Therefore, whenever $\max\{\den_{\kappa'}(I,\chi')_{\hat c_j},d_\A
   \cdot n^{-\alpha}\}> \den_{\kappa'}(I,\chi')_{\hat c_j}$ we have that
   $\den_{\kappa'}(I,\chi')_{\hat c_j}\le n^{-\alpha}\cdot
   \den_{\kappa'}(I,\chi')_{6\hat c}$, which can happen only for at most $O(1/\alpha)$ costs $\hat c \in E_\basec$. Now, consider any other ``good cost'', and consider $j=0$; we get the following bound:

$$
          \E_{A,\chi''}\left[m_{I,0}\cdot \one{(I,\chi') \hbox{ clustered by }
           (A,\chi'',j)}\right]	        			
	        \geq \den_{\kappa'}(I,\chi')_{\hat c_j} \cdot \tfrac{1}{\lambda^t} \cdot  \tfrac{%
	        			\mu_{\kappa'}(I,\chi')}{\den_{\kappa'}(I,\chi')_{\hat c_j}} \\
	        = \mu_{\kappa'}(I,\chi')/\lambda^{t}.
$$
Since $m_{I,0}$ is a lower bound on the total potential, and we are
sampling a cost from a universe of size $|E_\basec| = 1/\eta$, this
implies that for all but $O(\eta/\alpha)$ fraction of costs $\hat c\in
\E_\basec$:
               	$$\E_{A,\chi''}\left[\phi_{\kappa}(I,\chi) \cdot
          \1[\chi'=\chi''] \right] \ge \tfrac{\mu_{\kappa'}(I,\chi')}{\lambda^{t}}$$	
   \aanote{fix the extra factor 2 from above (into the statement?)}\ns{Changed $\phi$ potential, so we don't need the extra factor 2, to be fixed here only...}
as needed for (3).

It remains to prove item (1).  Assume $(I,\chi')$ is {\em not $F$-corrupted}, i.e.,
$\mu_{\kappa'}(I,\chi') \leq F \cdot \mu_{\kappa'}(\pi[I],\chi')$.
       Consider any $j <
        j_{\max}$, and recall the clustering property --- that $\NC_{\basec}(\A_j) \cap
        \PC_{\kappa'}^{\chi'} \subseteq \A_{j+1}$ --- and hence
        $\pi[I] \in \A_{j+1}$. %
	Also, we have that $\widehat{d_{I,j}}$ is at least
        $\Omega(n^{-\alpha}d_{\A}) = \Omega(n^{-\alpha}\den_{\kappa'}(A,\chi')_{3\hat c})$, and hence

        $$\tfrac{\widehat{d_{I,j}}}{\widehat{d_{\pi[I],j+1}}}
        =
        \Omega\left(\tfrac{\max\{d_{\A}n^{-\alpha},\den_{\kappa'}(I,\chi')_{\hat
            c_j}\}}{\max\{d_{\A}n^{-\alpha},\den_{\kappa'}(\pi[I],\chi')_{\hat
            c_{j+1}}\}}\right)
        =
        \Omega\left(\min\{1,\tfrac{d_{\A}n^{-\alpha}}{\den_{\kappa'}(A,\chi')_{3\hat
            c}}\}\right)
        = \Omega(n^{-\alpha}), $$
	where the second derivation is by triangle inequality. Therefore, 
	$\dd_{O(F \cdot n^{2 \alpha})}(m_{I,j},m_{\pi[I],j+1})
        = 0$ for any $j < j_{\max}$. 
        For $j=j_{\max}= O(1/\alpha)$, we have $m_{I,j_{\max}} \leq n^{-11} \cdot \totalmu \leq n^{-10}$, and by summing over all $m_{I,j}$ and $m_{\pi[I],j}$, we obtain $\dd_{O(F \cdot n^{2\alpha})}(\phi_\kappa(I,\chi),\phi_\kappa(\pi[I],\chi)) \leq n^{-10}$ as needed for (1).
        
        This concludes the proof.

\end{proof}

\input{imp-corruption.tex}

\input{imp-balance.tex}

%% file: imp-key-lemmas.tex
\subsection{Proof of Lemma~\ref{lm::corruption_growth_imp} from key
  correctness lemmas: corruption growth per level}
\label{sec:proofCorruptionGrowthLevel}

Before continuing with the correctness analysis, we establish some
auxiliary statements. 

\paragraph{Distortion Resilient Distance Properties.}
We show some properties of Distortion Resilient Distance from Def. \ref{def::dd}. Recall that $\dd_F(p,q) = \sum_{i:p_i > F\cdot q_i} |p_i|$. First, we show multiplication by scalars.

\begin{fact}\label{ft::dd_scalar}
	Fix $p,q \in \R^n_+$ and scalars $a, b \in \R_+$. We have
		$\dd_F(a\cdot p,b\cdot q) = a\cdot \dd_{F\cdot \tfrac{b}{a}}(p,q)$
\end{fact}
\begin{proof}

		$$\dd_F(a\cdot p,b\cdot q) = \sum_{i:\tfrac{a p_i}{b q_i} > F} |a p_i| = a \sum_{i:\tfrac{p_i}{q_i} > \tfrac{bF}{a}} |p_i| = a\cdot \dd_{F\cdot \tfrac{b}{a}}(p,q)$$
\end{proof}

We also show a bound on the $\ell_1$ normalization of vectors.

\begin{fact}\label{ft::dd_norm}
	Fix $p,q \in \R^n_+$. We have,
	$$\dd_{2 \cdot F^2}\left(\tfrac{p}{\normo{p}},\tfrac{q}{\normo{q}}\right) + \dd_{2 \cdot F^2}\left(\tfrac{q}{\normo{q}},\tfrac{p}{\normo{p}}\right) \leq 4\cdot \left(\tfrac{\dd_{F}(p,q)}{\normo{p}} + \tfrac{\dd_{F}(q,p)}{\normo{q}}\right)$$
\end{fact}

\begin{proof}
	By fact \ref{ft::dd_scalar}, we have that,
	$$\dd_{2 \cdot F^2}\left(\tfrac{p}{\normo{p}},\tfrac{q}{\normo{q}}\right) + \dd_{2 \cdot F^2}\left(\tfrac{q}{\normo{q}},\tfrac{p}{\normo{p}}\right) = \tfrac{1}{\normo{p}} \cdot \dd_{2 \cdot F^2 \cdot \tfrac{\normo{p}}{\normo{q}}}\left(p,q\right) + \tfrac{1}{\normo{q}} \cdot \dd_{2 \cdot F^2 \cdot \tfrac{\normo{q}}{\normo{p}}}\left(q,p\right)$$
	
	Now, consider a case where $\tfrac{\normo{p}}{\normo{q}} \in
        [1/2F,2F]$, then we have that, using the observation that
        $\dd_F(a,b)$ is decreasing in $F$:
		$$\dd_{2 \cdot F^2}\left(\tfrac{p}{\normo{p}},\tfrac{q}{\normo{q}}\right) + \dd_{2 \cdot F^2}\left(\tfrac{q}{\normo{q}},\tfrac{p}{\normo{p}}\right) \leq \tfrac{1}{\normo{p}} \cdot \dd_{F}\left(p,q\right) + \tfrac{1}{\normo{q}} \cdot \dd_{F}\left(q,p\right)$$
		
	Otherwise, assume w.l.o.g $\tfrac{\normo{p}}{\normo{q}} > 2F$,
        then on one hand we have that $\dd_{F}\left(p,q\right) \geq
        0.5\normo{p}$, and hence the RHS of the claim is at least $4 \cdot 0.5 = 2$. On the other hand, $\dd_{2F^2}\left(\tfrac{p}{\normo{p}},v\right)\leq 1$ for any $v$, and similarly for $q$, and hence the LHS is at most 2 and the claim follows.
	
\end{proof}

\vspace{2mm}
\paragraph{Alternative Counting for $\Lambda$ balls.} When bounding certain parameters, we use the following simple combinatorial claim:

\begin{claim}\label{cl::lambda_ball_func}
	Fix coloring $\kappa$, parameter $\zeta \geq 1$ and function $f: \IC \rightarrow \R_+$. Then,
	
	$$
	\sum_{I \in \IC} \mu_\kappa(I,\un) \cdot f(\Lambda_\kappa^\zeta(I)) = \sum_{I \in \IC} \mu_\kappa(\Lambda_\kappa^\zeta(I),\un) \cdot f(I) \leq 2\zeta f(\IC).
	$$
\end{claim}
\begin{proof}
	The proof follows immediate by counting the contribution of each $f(I)$ to elements in $\Lambda_\kappa^\zeta(I)$. Formally:
	
	$$
	\sum_{I \in \IC} \mu_\kappa(I,\un) \cdot f(\Lambda_\kappa^\zeta(I)) = \sum_{I \in \IC} \mu_\kappa(I,\un) \cdot \sum_{J \in \Lambda_\kappa^\zeta(I)} f(J) = \sum_{I \in \IC}  \sum_{J \in \IC} \1[\mu(\text{SEB}(I,J)\footnote{SEB$(I,J)$ = smallest enclosing ball containing both intervals},\un) \leq \zeta] \cdot \mu_\kappa(I,\un) f(J)
	$$
	
	Now, by change of summation, such quantity equals

	$$
	 \sum_{J \in \IC}  \sum_{I \in \IC} \1[\mu(\text{SEB}(I,J),\un) \leq \zeta] \cdot \mu_\kappa(I,\un) f(J) = \sum_{J \in \IC} f(J) \cdot \sum_{I \in \Lambda_\kappa^\zeta(J)} \mu_\kappa(I,\un) = \sum_{J \in \IC} \mu_\kappa(\Lambda_\kappa^\zeta(J),\un) \cdot f(J).
	$$
	
	And since $\mu_\kappa(\Lambda_\kappa^\zeta(J),\un) \leq 2\zeta$, we conclude:
	
		$$
	\sum_{I \in \IC} \mu_\kappa(I,\un) \cdot f(\Lambda_\kappa^\zeta(I)) = \sum_{I \in \IC} \mu_\kappa(\Lambda_\kappa^\zeta(I),\un) \cdot f(I) \leq 2\zeta f(\IC).
	$$
	
	as needed.
	
\end{proof}

We now proceed to the proof of Lemma
\ref{lm::corruption_growth_imp}. Recall that Lemma~\ref{lm::corruption_growth_imp} states that, in a
fixed level, the new coloring $\hkappa$ (amended from $\kappa$)
satisfies $\xi^{\hkappa}_{\widehat{F}} =
O\left(\tfrac{\xi^{\kappa'}_{F} +
  \xi^\kappa_{F}}{\epsilon^{2/\epsilon} \cdot \delta}\right) +
\tO\left(n^{-8}\right)$ with probability $1-\delta$, where
$\widehat{F} = F^{O(1/\epsilon^2)}$. For reader's convenience, we
include a summary table of notations.

\input{notation-new.tex}

\newcommand{\IpiI}[0]{\{I,\pi[I]\}}

\begin{proof}[Proof of Lemma \ref{lm::corruption_growth_imp} using Lemmas \ref{lm::dd_dist_imp}, \ref{lm::min_color_mass}]

Fix distortion $F$, level $l$, and consider an interval $I \in \IC
\setminus \pi^{-1}[\bot]$ with $\cad(I,\pi[I])\le c$.  Let $m_I =
\mu_\kappa(I,\un)$ and $m_{\pi[I]}= \mu_\kappa(\pi[I],\un)$.  
\ns{changed here such that either $I, \pi[I]$ are corrupted}If
either $(I,\un)$ or $(\pi[I],\un)$ are corrupted --- meaning $\dd_F(m_I,m_{\pi[I]}) + \dd_F(m_{\pi[I]},m_I) > 0$ and
hence $\dd_F(m_I,m_{\pi[I]}) + \dd_F(m_{\pi[I]},m_I) \geq \max\{m_I,m_{\pi[I]}\}$ --- then $\xi_F^\kappa(\IpiI)\ge \tfrac{m_I + m_{\pi[I]}}{2}$
(merely from $\un$-color). The new coloring $\hkappa$, obtained at
the end of level $l$, updates at most $m_I + m_{\pi[I]}$ total mass. Hence, 
we have:

\begin{align}
  \label{eqn:xiAtLeastMI}
  \xi_{2F}^{\hkappa}(\IpiI)-\xi_F^{\kappa}(\IpiI) &\le \sum_{I' \in \IpiI}\dd_{2F}(\mu_{\hkappa}(I',*_{\neq\bot}),\mu_{\hkappa}(\pi[I'],*_{\neq\bot})) + (\mu_{\hkappa}(I',\bot) - \mu_{\kappa}(I',\bot)) \\
  &\le
\sum_{I' \in \IpiI} \dd_{F}(\mu_{\kappa}(I',*_{\neq\bot}),\mu_{\kappa}(\pi[I'],*_{\neq\bot}))+3m_{I'} \\
&\le
O(\xi_F^{\kappa}(\IpiI)).
\end{align}

For the rest we consider intervals $I$ such that both $(I,\un)$ and $(\pi[I],\un)$ are not $F$-corrupted, i.e.,
$\tfrac{m_I}{m_{\pi[I]}} \in [1/F,F]$.

To bound the new corruption $\xi_{\hat F}^{\hkappa}(I)$, we analyze
two cases: the base case of level $l=0$, where our main goal is to
match intervals and the corruption added is for {\em already matched}
pairs (colored $\bot$) without a short path in $G$, and the general
case of $l\ge1$, and our goal is to move $\mu$-mass from $\un$-color
to $\chi\in[\lambda^t]$ color, in similar proportion for both $I$ and
$\pi[I]$. In both cases, we analyze the expectation of
$\xi^{\hkappa}_{\widehat{F}}$ as a function of
$\xi^{\kappa'}_{F},\xi^{\kappa}_{F}$.

Recall our coloring update procedure replaces the $\un$-mass
of $I$ by vector
$\mu_\kappa(I,\un) \cdot
\frac{\mu_\kappa(I,*_{\neq\un}) + \sum_{\zeta \in Z_l}
  \varphi_\kappa^\zeta(I,*)}{\normo{\mu_\kappa(I,*_{\neq\un})+\sum_{\zeta
      \in Z_l} \varphi_\kappa^\zeta(I,*)}}$.

\vspace{2mm}

\paragraph{Base case of level $l=0$.} \ns{changed}At the beginning of the step, we
have that $\mu_\kappa(I,\bot)=\mu_{\kappa'}(I,\bot)$ and $\mu_\kappa(I,\un)=\mu_{\kappa'}(I,\nu\setminus\bot)$.
Furthermore, at level 0, we amend only the
$\un$ and $\bot$ color potentials. For $\un$, we invoke Lemma \ref{lm::dd_dist_imp} (2) and obtain (whp) that 
$$\dd_{\widehat{F}}\left(\varphi_\kappa^{0}(I,\un),
                  \varphi_\kappa^{0}(\pi[I],
                  \un)\right) =
                  O\left(\rho^{\kappa,\kappa'}_{F}(I,\{I\})
                  + \tfrac{1}{n^{10}}\right) = O\left(\xi^{\kappa'}_{F}(I)
                  + \tfrac{1}{n^{10}}\right).$$
                For $\bot$, observe that $\varphi_\kappa^0(I,\bot) \leq
\phi_\kappa(I,\nu \setminus \{
\un, \bot \})$ at level 0. Now, consider an {\em
  $F$-uncorrupted pair} $(I,\chi') \in \IC \times \nu \setminus \{
\bot \}$ that has been clustered by the anchor at color $\chi$. By Claim~\ref{cl::phi_uncorrupted} (1), we have $\dd_{O(F \cdot n^{2\alpha})}(\phi_\kappa(I,\chi),\phi_\kappa(\pi[I],\chi))\le
 n^{-10}$ whp. Notice also that whenever $\phi_\kappa(I,\chi) > 0$, we add an edge $(A,I)$ to $G$ which is uniquely identified by color $\chi$. Therefore, one of the following must hold:
 \begin{enumerate}
 	\item We add both edges $(A,I)$ and $(A,\pi[I])$ to $G$; or
 	\item The total contribution of {\em $F$-uncorrupted pairs} to $\varphi_\kappa^0(I,\bot)$ is $O(n^{-9})$.
 \end{enumerate}
 
	Focusing on corrupted pairs $(I,\chi')$, note that there is
 $\xi^{\kappa'}_{F}$  corrupted
    mass in $\IC$;
    i.e., total $\mu_{\kappa'}$ of $(I,\chi')$ with
    $\dd_F(\mu_{\kappa'}(I,\chi'),\mu_{\kappa'}(\pi[I],\chi'))>0$ is
    at most $\xi^{\kappa'}_{F}$.
    Recall that we sample $\lambda^t$ anchors %
    and by Claim~\ref{cl::phi_uncorrupted} (2), each such anchor generates, in expectation, potential $O(\tfrac{\mu_{\kappa'}(I,\chi')}{\lambda^t})$.
    
    Therefore, for each
    $(I,\chi')$ pair, the added potential over all anchors is a
    r.v. with expectation $O(\mu_{\kappa'}(I,\chi'))$. Overall, we obtain the expected contribution of corrupted pairs to $\varphi_\kappa^0(I,\bot)$ is $O(\xi^{\kappa'}_{F}(I))$. Define $u_I
= \mu_\kappa(I,*_{\neq\un}) +
\varphi_\kappa^0(I,*)$ and $Z_I \triangleq \1[\dist_G(I,\pi[I])>2]$. Notice that,

\begin{align*}
	\E\left[\xi^{\hkappa}_{\widehat{F}}(I)\right] &\leq \tfrac{1}{\normo{u_I}}\left(Z_I \cdot
 \E[\mu_{\kappa}(I,\bot) + \varphi_\kappa^0(I,\bot)] %
 + \E\left[\dd_{\widehat{F}}\left(\mu_{\kappa}(I,\un) +\varphi_\kappa^{0}(I,\un),
                  \mu_{\kappa}(\pi[I],\un) + \varphi_\kappa^{0}(\pi[I],
                  \un)\right)\right]\right) \\
                  &\stackrel{Fact~\ref{fct:ddTriangle}}{\leq} \tfrac{1}{\normo{u_I}}\cdot O\left(\xi^{\kappa}_{F}(I) + \xi^{\kappa'}_{F}(I)
                  + \tfrac{1}{n^{9}}\right).
                  \end{align*}
                                    
                  as there are only $\un$ and $\bot$ colors to
                  consider at $l=0$. By
Lemma~\ref{lm::min_color_mass}, we have $\normo{u_I} =
\Omega(\epsilon)$. Taking into account Eqn.~\eqref{eqn:xiAtLeastMI} for corrupted $(I,\un)$, we therefore have in expectation
$\E[\xi^{\hkappa}_{\widehat{F}}(I) - \xi^{\kappa}_{F}(I)] =
O\left(\tfrac{1}{\eps}(\xi^{\kappa}_{F}(I) + \xi^{\kappa'}_{F}(I) +
n^{-9})\right)$ overall for each interval. Summing over all intervals %
we conclude $\E[\xi^{\hkappa}_{\widehat{F}}] = O\left(\tfrac{1}{\eps}(\xi^{\kappa}_{F} + \xi^{\kappa'}_{F} + n^{-8})\right)$. 

\vspace{2mm}

\paragraph{General case.} Consider level $l\ge1$. 
Intuitively, the new corruption
added to $I$ is driven by 1) vectors $\varphi$ that differ between $I$
and $\pi[I]$ (in $\dd$ sense), and 2) normalization. For a fixed $\zeta\in Z_l$, we shall bound the %
quantity
$\dd_{%
F'}(\varphi_\kappa^\zeta(I), \varphi_\kappa^{\zeta'}(\pi[I]))$, %
for some $\zeta'\in Z_l, F'\ge F$, by the following quantity in expectation, up
to a constant:
\begin{equation}\label{eq::rho_I_zeta}
\rho_I^\zeta \triangleq \tfrac{\beta^l}{\zeta} \cdot
\rho^{\kappa,\kappa'}_{F}(I,\Lambda_\kappa^{2\zeta}(I)) + \tfrac{1}{\log n} \cdot
\widetilde{\rho^{\kappa,\kappa}_{F}}(I,\Lambda_\kappa^{3\zeta}(I))
\ge \tfrac{m_I}{4\zeta}\left(\tfrac{\beta^l}{\zeta}\cdot
\xi_F^{\kappa'}(\Lambda_\kappa^{2\zeta}(I))+
\tfrac{1}{\log n}\xi_F^{\kappa}(\BC_\pi(\Lambda_\kappa^{3\zeta}(I)))\right).
\end{equation}

\ns{new argument}

To bound $\E[\dd_{
F'}(\varphi_\kappa^\zeta(I), \varphi_\kappa^{\zeta'}(\pi[I]))]$ by $O(\rho_I^\zeta)$, we first show such quantity is bounded by $O(m_I)$. Indeed:
	\begin{align*}
	\E[\dd_{
F'}(\varphi_\kappa^\zeta(I), \varphi_\kappa^{\zeta'}(\pi[I]))] 
&\leq \E\left[\normo{\varphi_\kappa^\zeta(I)}\right]\\
 &\stackrel{\varphi \text{ def}}{\leq}
\E\left[T^\un \left(\mu_\kappa(I,\un)\cdot
    \tfrac{\beta^{2l}}{\zeta^2}\right) +\normo{T^\cc\left(\mu_\kappa(I,\un)\cdot
    \tfrac{\beta^{2l}}{\zeta^2}\cdot
    \tfrac{\phi_\kappa(\Lambda_\kappa^\zeta(I),*)}{\beta^l}\right)}\right] \\
    &\leq
\E[m_I + m_I\cdot
    \tfrac{1}{\zeta}\cdot
    \phi_\kappa(\Lambda_\kappa^\zeta(I),\nu)] \\
    &\stackrel{Claim~\ref{cl::phi_uncorrupted}}{=}
m_I + m_I\cdot
    \tfrac{1}{\zeta}\cdot
   O(\zeta) \\   
    &=
O(m_I).    	
\end{align*}

From the above, if
$\rho_I^\zeta \geq \tfrac{m_I}{10}$, then
$\E[\dd_{
F'}(\varphi_\kappa^\zeta(I), \varphi_\kappa^{\zeta'}(\pi[I]))] = O(\rho_I^\zeta)$.
Hence we assume $\rho_I^\zeta <
\tfrac{m_I}{10}$ in the rest. 

Let $\zeta'\in Z_l$ be the smallest such that
$\pi[\Lambda_\kappa^{3\zeta}(I)] \subseteq
\Lambda_\kappa^{\zeta'/2}(\pi[I])$.
Note that $\zeta' = O(F \cdot
\zeta\cdot \log n)$ --- as
otherwise, this would mean
$\mu_\kappa(\BC_\pi(\Lambda_\kappa^{3\zeta}(I)),\un) = \omega(F \cdot \zeta \cdot \log
n)$, and since \aanote{should it be $3\zeta$ instead of $2\zeta$?}\ns{yes.}
$\mu_\kappa(\Lambda_\kappa^{3\zeta}(I),\un) \leq 6\zeta$ (by
definition of $\Lambda$-balls), this in turn implies most of
$\BC_\pi(\Lambda_\kappa^{3\zeta}(I)) \times \{\un\}$ is $F$-corrupted, and hence: %
$$\xi_F^{\kappa}(\BC_\pi(\Lambda_\kappa^{3\zeta}(I)) > 0.5 \cdot \mu_\kappa(\BC_\pi(\Lambda_\kappa^{3\zeta}(I)),\un) = \omega(\zeta\cdot F \cdot \log
n) \stackrel{Eqn.~\eqref{eq::rho_I_zeta}}{\Rightarrow}
\rho_I^\zeta > \omega(m_I),$$
which is a contradiction.

By Lemma
\ref{lm::dd_dist_imp}, we have that, for some
$F'=F^{\Theta(1/\epsilon^2)}$ satisfying $F' \ge
(F\zeta'/\zeta)^{\Theta(1/\epsilon^2)}\cdot 2\log n$:

	\begin{align*}
	  \E\left[\dd_{\tfrac{F'}{2\log n}}\left(\varphi_\kappa^{\zeta}(I),\varphi_\kappa^{\zeta'}(\pi[I])\right)\right] &=
		\E\left[\dd_{\tfrac{F'}{2\log n}}\left(\varphi_\kappa^{\zeta}(I,
                  *_{\neq \un}),\varphi_\kappa^{\zeta'}(\pi[I], *_{\neq \un})\right) + \dd_{\tfrac{F'}{2\log n}}\left(\varphi_\kappa^{\zeta}(I,\un),\varphi_\kappa^{\zeta'}(\pi[I],\un)\right)\right] \\
		&= O(\rho_I^\zeta + n^{-9}).
	\end{align*}

        We now aggregate the contribution from all $\zeta$'s, obtaining:

\begin{align*}
	\E\left[\dd_{F'}\left(\sum_{\zeta\in
            Z_l}\varphi_\kappa^{\zeta}(I),\sum_{\zeta\in Z_l}\varphi_\kappa^{\zeta}(\pi[I])\right)\right] 
	&\stackrel{\text{Fact}~\ref{ft::dd_scalar}}{=} \E\left[\dd_{\tfrac{F'}{|Z_l|}}\left(\sum_{\zeta}\varphi_\kappa^{\zeta}(I),|Z_l| \cdot \sum_{\zeta}\varphi_\kappa^{\zeta}(\pi[I])\right)\right]	 \\
	&\stackrel{\text{Fact}~\ref{fct:ddTriangle}}{\leq} 2\E\left[\sum_{\zeta}\dd_{\tfrac{F'}{2|Z_l|}}\left(\varphi_\kappa^{\zeta}(I),\sum_{\zeta}\varphi_\kappa^{\zeta}(\pi[I])\right)\right] \\
	&\leq 2\sum_{\zeta}\E\left[\dd_{\tfrac{F'}{2\log n}}\left(\varphi_\kappa^{\zeta}(I),\varphi_\kappa^{\zeta'}(\pi[I])\right)\right] \\
	&= O\left(\sum_{\zeta}(\rho_I^\zeta + n^{-9})\right).
\end{align*}

Finally, we consider the contribution from normalization. Define $u_I
= \mu_\kappa(I,*_{\neq\un}) + \sum_{\zeta \in Z_l}
\varphi_\kappa^\zeta(I)$ and $v_I = \mu_\kappa(\pi[I],*_{\neq\un}) +
\sum_{\zeta \in Z_l} \varphi_\kappa^\zeta(\pi[I])$.  Using Fact~\ref{ft::dd_norm} and Lemma~\ref{lm::min_color_mass} to bound the
normalization, we have:
	\begin{align}\nonumber
		\E\left[\dd_{8F'^2}(\normalizedEllOne{u_I},\normalizedEllOne{v_I}) + \dd_{8F'^2}(\normalizedEllOne{v_I},\normalizedEllOne{u_I})\right] &\leq O\left(\tfrac{1}{\epsilon^{l+1}}\right) \cdot \E\left[\dd_{2F'}(u_I,v_I)+\dd_{2F'}(v_I,u_I)\right] \\
		&\stackrel{Fact~\ref{fct:ddTriangle}}{=} O\left(\tfrac{1}{\epsilon^{l+1}}\right) \cdot \left(\xi_F^\kappa(I) + \xi_F^\kappa(\pi[I]) +\sum_{\zeta \in Z_l} (\rho_I^\zeta + \rho_{\pi[I]}^\zeta+ n^{-9})\right).\label{eq::corr_I}
	\end{align}

We now compute the aggregate contribution $\xi_{\hat F}^{\hat
  \kappa}(\IC)$ over all intervals $I\in\IC$, where $\hat F = 8F'^2
\cdot F = F^{\Theta(1/\epsilon^2)}$. For this, we need to also bound
the sum of all $\rho_I^\zeta$ quantities as a function of total corruption
$\xi_F^\kappa$, which we do in the next claim.

\begin{claim}\label{clm::rho_bound_imp_new}
 Fix $l>0$ and assume $\mu_\kappa(\YC,\un), \mu_\kappa(\XC,\un)  \geq 1$. %
 Then, $\sum_{I \in \IC}\sum_{\zeta\in Z_l} \rho_I^\zeta
   = O(\xi_F^{\kappa'}+\xi_F^{\kappa})$.
\end{claim}

\begin{proof}
Recall that $\rho_I^\zeta
=\tfrac{\beta^l}{\zeta}\rho^{\kappa,\kappa'}_{F}(I,\Lambda_\kappa^{2\zeta}(I)) +
\tfrac{1}{\log n}\widetilde{\rho^{\kappa,\kappa}_{F}}(I,\Lambda_\kappa^{3\zeta}(I))$.
We bound the sum of each term separately. First, let
$r_\XC=\min\{2\zeta-1,\mu_\kappa(\XC,\un)\}$ and $r_\YC=\min\{2\zeta-1,\mu_\kappa(\YC,\un)\}$. Then $\mu_\kappa(\Lambda_\kappa^{2\zeta}(I), \un) \in
        [r_\XC,2(r_\XC + 1)]$ for $I\in\XC$ (and similarly for $\YC$) and:
	
	\begin{align*}
          \sum_{I \in \IC}
        \rho^{\kappa,\kappa'}_{F}(I,\Lambda_\kappa^{2\zeta}(I))
        &\le
        \sum_{I \in \XC} \tfrac{\mu_\kappa(I,\un)}{r_\XC} \cdot
        \xi_F^{\kappa'}(\Lambda_\kappa^{2\zeta}(I))
        +
        \sum_{I \in \YC} \tfrac{\mu_\kappa(I,\un)}{r_\YC} \cdot
        \xi_F^{\kappa'}(\Lambda_\kappa^{2\zeta}(I))
       \\
       &
       \stackrel{Claim~\ref{cl::lambda_ball_func}}{\leq} \tfrac{1}{r_\XC}\sum_{I \in
          \XC}  \mu_\kappa(\Lambda_\kappa^{2\zeta}(I),\un)
        \cdot \xi_F^{\kappa'}(I)
       + \tfrac{1}{r_\YC}\sum_{I \in
          \YC}  \mu_\kappa(\Lambda_\kappa^{2\zeta}(I),\un)
        \cdot \xi_F^{\kappa'}(I)
       \\
       &\le \tfrac{1}{r_\XC}\sum_{I \in \XC}
       2(r_\XC + 1) \cdot \xi_F^{\kappa'}(I)
       +
        \tfrac{1}{r_\YC}\sum_{I \in \YC}
       2(r_\YC + 1) \cdot \xi_F^{\kappa'}(I)
       \\
       &< 3
        \xi_F^{\kappa'}.
        \end{align*}

	Summing over $\zeta\in Z_l$, we get
        $\sum_{I \in
          \IC}\sum_{\zeta\in Z_l}
        \tfrac{\beta^l}{\zeta}\rho^{\kappa,\kappa'}_{F}(I,\Lambda_\kappa^{2\zeta}(I))=
        O(\xi_F^{\kappa'})$. %
        For the second term, defining %
        $ \widetilde{\xi_F^{\kappa}}(I) \triangleq \xi_F^{\kappa}(\piComp^{-1}(I)) = \xi_F^{\kappa}(\{I' \in \IC \mid \piComp(I') = I\})$\footnote{Recall $\piComp(I) = \pi(I + j)$, where $j\ge 0$, is the minimal one such that $\pi(I + j) \neq \bot$.}:
	
		\begin{align*}
          \sum_{I \in \IC}
        \widetilde{\rho^{\kappa,\kappa}_{F}}(I,\Lambda_\kappa^{3\zeta}(I))
        &\le
        \sum_{I \in \XC} \tfrac{\mu_\kappa(I,\un)}{r_\XC} \cdot
        \xi_F^{\kappa}(\BC_{\pi}(\Lambda_\kappa^{3\zeta}(I)))
        +
        \sum_{I \in \YC} \tfrac{\mu_\kappa(I,\un)}{r_\YC} \cdot
        \xi_F^{\kappa}(\BC_{\pi}(\Lambda_\kappa^{3\zeta}(I)))
       \\
        &\le
        \sum_{I \in \XC} \tfrac{\mu_\kappa(I,\un)}{r_\XC} \cdot
       \widetilde{\xi_F^{\kappa}}(\Lambda_\kappa^{3\zeta}(I))
        +
        \sum_{I \in \YC} \tfrac{\mu_\kappa(I,\un)}{r_\YC} \cdot
         \widetilde{\xi_F^{\kappa}}(\Lambda_\kappa^{3\zeta}(I))
       \\
       &
     \stackrel{Claim~\ref{cl::lambda_ball_func}}{\leq}  
     \tfrac{1}{r_\XC}\sum_{I \in
          \XC}  \mu_\kappa(\Lambda_\kappa^{3\zeta}(I),\un)
        \cdot  \widetilde{\xi_F^{\kappa}}(I)
       + \tfrac{1}{r_\YC}\sum_{I \in
          \YC}  \mu_\kappa(\Lambda_\kappa^{3\zeta}(I),\un)
        \cdot  \widetilde{\xi_F^{\kappa}}(I)
       \\
     &\le\tfrac{1}{r_\XC}\sum_{I \in
          \YC}  \mu_\kappa(\Lambda_\kappa^{3\zeta}(\piComp(I)),\un)
        \cdot \xi_F^{\kappa}(I)
       + \tfrac{1}{r_\YC}\sum_{I \in
          \XC}  \mu_\kappa(\Lambda_\kappa^{3\zeta}(\piComp(I)),\un)
        \cdot \xi_F^{\kappa}(I)
       \\
       &\le \tfrac{1}{r_\XC}\sum_{I \in \XC}
       3(r_\XC+1) \cdot \xi_F^{\kappa}(I)
       +
        \tfrac{1}{r_\YC}\sum_{I \in \YC}
       3(r_\YC+1) \cdot \xi_F^{\kappa}(I)
       \\
       &< 4
        \xi_F^{\kappa}.
                \end{align*}

	Summing over $\zeta\in Z_l$, we get
        $\sum_{I\in\IC}\sum_{\zeta\in Z_l} \tfrac{1}{\log
          n}\widetilde{\rho^{\kappa,\kappa}_{F}}(I,\Lambda_\kappa^\zeta(I))\le
        \tfrac{|Z_l|}{\log
        n}\cdot 4\xi_F^{\kappa}=O(\xi_F^{\kappa})$.
	By summing up both terms, the claim follows.
\end{proof}

Finally, we can estimate the amount of ``new corruption'', $\sum_{I\in
  \IC}\dd_{\widehat F}(\mu_{\hkappa}(I,*),\mu_{\hkappa}(\pi[I],*))-\xi_F^\kappa$, by summing Eqn. (\ref{eq::corr_I}) over all \ns{should be uncorrupted, perhaps give them a new set for the ones we don't use, like $\IC^\xi$} intervals $I \in
\IC$, and using Claim~\ref{clm::rho_bound_imp_new} (noting we only call \textsc{AmendColoring} when $\mu_\kappa(\YC,\un), \mu_\kappa(\XC,\un) \geq \beta^l = \omega(1)$ for level $l>0$).
	
	\begin{align*}
	\E[\xi_{\widehat F}^{\hkappa}] &=
        \sum_{I \in \IC}\E[\xi_{\widehat F}^{\hkappa}(I)] \\
        &\stackrel{Fact~\ref{fct:ddTriangle}}{\le} 2\sum_{I\in \IC}\E\left[\dd_{\widehat F/2}(\mu_{\kappa}(I,*_{\neq \un}),\mu_{\kappa}(\pi[I],*_{\neq \un}))+\dd_{\widehat F/2}(m_I\normalizedEllOne{u_I},m_{\pi[I]}\normalizedEllOne{v_I})\right]
        \\
       &\stackrel{ Eqn.~\eqref{eq::corr_I}}{\le} 2\xi_F^{\kappa}+\sum_{I
          \in \IC}  O\left(\tfrac{1}{\epsilon^l}\right) \cdot
        \left(\xi_F^\kappa(I) + \xi_F^\kappa(\pi[I])+ \sum_{\zeta \in
          Z_l} (\rho_I^\zeta  + \rho_{\pi[I]}^\zeta + n^{-9})\right)
        \\
	&= 2\xi_F^{\kappa}+O\left(\tfrac{1}{\epsilon^l}\right) \cdot \sum_{I \in \IC} \left(\xi_F^\kappa(I) + \sum_{\zeta \in Z_l} (\rho_I^\zeta + n^{-9})\right) \\
        &\stackrel{Claim~\ref{clm::rho_bound_imp_new}}{\le}
        2\xi_F^{\kappa}+O\left(\tfrac{1}{\epsilon^l}\right)
        \cdot\left(\xi_F^\kappa+\xi_F^{\kappa'}+\tilde O(n^{-8})\right)
        \\
	&= O\left(\tfrac{1}{\eps^{2/\epsilon}}\right) \cdot \left(\xi_F^\kappa + \xi_F^{\kappa'}\right)+ \tO(n^{-8}).
	\end{align*}
	
	The soundness of the Lemma is obtained using the Markov inequality.
\end{proof}

%% file: notation-new.tex
\begin{table}[hp]
\vspace{-1ex}
\caption{Summary of main notations and their definitions.} %
  \label{tab:notation-new}
 \centering
{\renewcommand{\arraystretch}{1.1}
  \begin{tabular}{|l|l|}
    \hline

               \multicolumn{2}{| >{\bf}c |}
 	{Parameters}
    \\
    \hline
    \multicolumn{2}{| c |}
 	{
    \begin{tabular}{p{0.02\textwidth}|p{0.47\textwidth}|p{0.02\textwidth}|p{0.47\textwidth}}
    $\eps$ & Arbitrary small constant; complexity is
      $n^{1+O(\epsilon)}$. &   $\alpha$ & $=\eps^{5/\epsilon}$. We define $F_0 = n^\alpha$.    \\
    $\lambda$ & $=n^\epsilon$. Sample size growths by factor $\lambda$
      each step. &  $\eta$ &     $=\alpha^2\epsilon^3$. The set of
      costs $E_\basec$ has size $1/\eta$.
    \\
      $\beta$ &     $=n^\epsilon$. Base radius
    growths by $\beta$ factor each
    level. &   $k$ &     $=\tO(\beta^3)$, a pivot ``oversampling'' factor.
    \\
    \end{tabular}}
    \\
    \hline

       \multicolumn{2}{| >{\bf}c |}
 	{Transformations}
    \\
    \hline
    $T^q_{\delta,\gamma}(x)$ & 
    \begin{tabular}{p{0.3\textwidth}|p{0.52\textwidth}}
        = $\begin{cases}
	x_i	&	x_i \geq \gamma \\
	0	& x_i < \delta \gamma \\
	\gamma \cdot (\tfrac{x_i}{\gamma})^q	&	\text{Otherwise}
\end{cases}$	
&
    \begin{tabular}{p{0.05\textwidth}|p{0.45\textwidth}}
    $T^\cc$ & $=T^{\Theta(1/\epsilon)}_{1/\beta, 1/\beta^3}$
 	\\

     $T^\theta$ &     $ = T^{O(1/\epsilon)}_{1/\sqrt{\beta},1}$
    \\
         $T^\un$ &     $=T^{O(1/\eps)}_{1/\beta, \Omega_\eps(1)}$
    \\
    \end{tabular}
    \end{tabular}
    \\
    \hline
$Q_{\delta,s,F}(x)$ & $= \max_{J \subseteq [d]} a_{|J|} \cdot \min_{j
      \in J} x_j$ \hspace{5mm} where $a_i = \begin{cases}
        1       & i \geq s\cdot d \\
        0       & i < (s-\delta)d \\
        1/F^{sd-i}        &       \text{Otherwise}
\end{cases}$
    \\
    \hline
    $Q_l(x)$ & $=
    Q_{\tfrac{\epsilon}{4},1-\tfrac{(l+1)\epsilon}{2},n^{\alpha^2}}(x)$
 	\\
 	    \hline
       \multicolumn{2}{| >{\bf}c |}
 	{Potentials}
 	\\
    \hline

        $\phi_\kappa(I,\chi)$ &
    $=\mu_{\kappa'}(I,\chi')\tfrac{2n}{\lambda^t}\sum_{j=0}^{j_{\max}}
    \tfrac{n^{-\alpha j}}{\widehat{d_{I,j}}}$,
    is the basic potential from matching anchor $(A,\chi',\hat c)$
    \\
    $\widehat{d_{I,j}}$ &
    $=\Theta\left(\max\{d(I,\chi')_{\hat{c}+j\basec},n^{-\alpha} \cdot d(A,\chi')_{3\hat{c}}\}\right)$, is the approximated density for a matched interval $I$
    \\
    \hline
    $\varphi_\kappa^0(I,\bot)$ & $= T^\cc\left(\mu_\kappa(I,\un) \cdot
\phi_{\kappa}(I, \nu \setminus \{\un, \bot
\})\right)$
    \\
    \hline
    $\Lambda_\kappa^\zeta(I)$ & smallest interval ball containing
    $\zeta$ of $\ell_1$-mass of \un-color in $\mu$ on the left and right of $I$
    \\
    \hline
    $\varphi_\kappa^\zeta(I,*_{\neq \un})$ & $=T^\cc\left(\mu_\kappa(I,\un)\cdot
    \tfrac{\beta^{2l}}{\zeta^2}\cdot
    \tfrac{\phi_\kappa(\Lambda_\kappa^\zeta(I),*)}{\beta^l}\right)$
    \\
    \hline
    $\Gamma_{l,\zeta}(V,\chi')$ & a vector of costs in $E_\basec$ with
 $\Gamma_{l,\zeta}(V,\chi') =
\min\left\{\reallywidehat{\relden_{\kappa'}(V,\chi',\Lambda_\kappa^{\zeta}(V))}^{-1}\cdot
n^{4\alpha}\cdot
\frac{\totalmu}{ \beta^l \cdot \lambda^t}, 1\right\}$
    \\
    \hline
    $\sigma_{V,\chi'}$ &  $= 
    Q_l(T^\theta\left(\Gamma_{l,\zeta}(V,\chi')\right))$
    \\
    \hline
    $\VC$ & multiset of pivots; $=$ subsampling $(V, \chi') \in \IC \times
\nu \setminus \{ \bot \}$ $k$ times with prob. $\mu_{\kappa'}(V,\chi')\cdot\beta^{-l}$
\\
\hline
$m_\VC((V,\chi'))$ & multiplicity of $(V,\chi')$ in $\VC$
    \\
    \hline
$\theta_\kappa^{\zeta}(V,\chi')$ & $= 
m_\VC((V,\chi')) \cdot \begin{cases}
	\sigma_{V,\chi'}	&	\mu_{\kappa'}(\Lambda_{\kappa'}^\zeta(V),\chi') \leq \lambda \cdot \beta^6 \cdot \beta^l	\\
	1	&	\text{Otherwise}
\end{cases}$
    \\
    \hline
    $Z_l, l>0$ & $=\beta^l\cdot\{1,2,4,8,\ldots, n/\beta^l\}$, is the set of radiuses $\zeta$
    \\
    $Z_0$ %
    & $=\{0\}.$
    \\
    \hline
    $\varphi_\kappa^\zeta(I,\un)$, $\zeta>0$ & $=
    T^\un\left(\mu_\kappa(I,\un)\cdot \frac{\beta^{2l}}{\zeta^2} \cdot  \min\left\{\frac{1}{k} \cdot \sum_{\chi'\in \nu \setminus \{\bot \}}\theta_\kappa^{\zeta}(\Lambda_\kappa^{\zeta}(I),\chi'),1\right\}\right)$ 
\\
$\varphi_\kappa^\zeta(I,\un), \zeta=0$ & $= T^\un\left(\mu_\kappa(I,\un) \cdot \tfrac{\theta_\kappa^{0}(I,\nu \setminus \{ \bot \})}{k} \right)$
    \\
    \hline
    $\mu_{\hkappa}(I,*)$ & $=\mu_\kappa(I, *_{\neq \un}) +
    \mu_\kappa(I,\un)\cdot \normalizedEllOne{\mu_\kappa(I,*_{\neq
        \un})+\sum_{\zeta\in Z_l} \varphi_\kappa^\zeta(I,*)}$ (always a
    distribution over $\chi$'s)
    \\
    \hline
       \multicolumn{2}{| >{\bf}c |}
 	{Error measures}
\\
\hline
    $\xi_F^{\kappa}(I)$ & $=\begin{cases} 1 & \pi[I]=\bot
  \text{ or } \cad(I,\pi[I])>\basec\\
  \dd_F(\mu_\kappa(I,*_{\neq \bot}),\mu_\kappa(\pi[I],*_{\neq \bot})) + \mu_\kappa(I,\bot) \1[\text{dist}_G(I,\pi[I])>2] & \text{Otherwise}\end{cases}$
    \\
    \hline
$\rho^{\kappa,\kappa'}_{F}(I,\SC)$ &
  $=\tfrac{\mu_\kappa(I,\un)}{\mu_\kappa(\SC,\un)} \cdot \xi^{\kappa'}_F(\SC)$

\\ \hline
$\widetilde{\rho^{\kappa,\kappa'}_{F}}(I,\SC)$ &
$=\tfrac{\mu_\kappa(I,\un)}{\mu_\kappa(\SC,\un)} \cdot \xi^{\kappa'}_F(\BC_\pi(\SC))$,
where $\BC_\pi(\SC)$ for $\SC \subseteq \IC$ is
the smallest enclosing ball around $\pi[\SC]$
    \\
    \hline
  \end{tabular}
  }
\end{table}

%% file: imp-corruption.tex
\subsection{Controlling average $\varphi$ corruption: proof of Lemma \ref{lm::dd_dist_imp}}\label{sec::dd_dist_proof}

We prove parts 1 and 2 separately. We use the following notation: for
any $J\in\IC$, define $m_J = \mu_{\kappa}(J,\un)$, and
$\BC_{\pi[I]} = \Lambda_\kappa^{\zeta'}(\pi[I])$. %

\aanote{note that part 2 proves $+1/n^9$ instead of 10 -- need to
  update lemma}\ns{changed.}

\begin{proof}[Proof of Lemma \ref{lm::dd_dist_imp} (1)]  

Let $p_\chi=\sum_{J \in \BC_I}\phi_\kappa(J,\chi)$, $q_\chi=\sum_{J
  \in \BC_{\pi[I]}}\phi_\kappa(J,\chi), v_\chi = \sum_{J \in
  \BC_I}\varphi_\kappa^{\zeta}(J,\chi), u_\chi = \sum_{J \in
  \BC_{\pi[I]}}\varphi_\kappa^{\zeta'}(J,\chi)$. We define vectors
$p,q,v,u \in \R^{\nu \setminus \{\un, \bot\}}$ correspondingly.
Also, let $m_I = \mu_{\kappa}(I,\un)$,
and note that since $(I,\un)$ is not $F$-corrupted pair in $\kappa$, then $\dd_F(m_I,m_{\pi[I]}) = 0$. 
We have, using $T^\cc$, for some $\widehat{F} = F^{\Theta(1/\epsilon^2)}$:
\begin{align*}
	\dd_{\widehat{F}}(v,u) &\le \dd_{\widehat{F}}\left(T^\cc(m_I \cdot \tfrac{\beta^l}{\zeta^2} \cdot p),T^\cc(m_{\pi[I]} \cdot \tfrac{\beta^l}{\zeta'^2} \cdot q)\right) \\
&\stackrel{Claim~\ref{cl::T_trans}}{\leq} \dd_{\widehat{F}^{\Theta(\epsilon)}}\left(m_I \cdot \tfrac{\beta^l}{\zeta^2} \cdot p,m_{\pi[I]} \cdot \tfrac{\beta^l}{\zeta'^2} \cdot q\right) + n^{-10} \\
&\stackrel{Fact~\ref{ft::dd_scalar}}{\leq}
m_I \cdot \tfrac{\beta^l}{\zeta^2} \dd_{\widehat{F}^{\Theta(\eps)}/F \cdot \tfrac{\zeta^2}{\zeta'^2}}(p,q) + n^{-10} \\
&\leq m_I \cdot \tfrac{\beta^l}{\zeta^2} \dd_{F^4}(p,q) + n^{-10}.
\end{align*}

It remains to compute (the expected value of) $\dd_{F^4}(p,q)$, i.e.,
the difference in $\phi$ potential. 
Fix a color (coordinate) $\chi$. If an interval
 $J \in \BC_I$ 
is matched to some anchor/color $(A,\chi')$ pair,  
 and $(J,\chi')$ is not {\em
   $F$-corrupted pair}, then, by Claim \ref{cl::phi_uncorrupted}, similar potential will be
 added to $\pi[J]$ up to distortion factor $O(F \cdot n^{2\alpha}) \leq
 F^4$: specifically,
 $\dd_{F^{4}}(\phi_\kappa(J,\chi),\phi_\kappa(\pi[J],\chi)) \le \dd_{O(Fn^{2\alpha})}(\phi_\kappa(J,\chi),\phi_\kappa(\pi[J],\chi))\le
 n^{-10}$ whp.

 \aanote{remember to fix $\phi$ def in the table}\ns{was this done?}
Now focus on corrupted pairs $(J,\chi')$ where $J\in \BC_I$, $\chi'\in\nu$. Note that
 $\xi^{\kappa'}_F(\BC_I) \leq \tfrac{2\zeta}{m_I}\cdot
 \rho^{\kappa,\kappa'}_{F}(I,\BC_I)$;   
    i.e., total $\mu_{\kappa'}$ of $(J,\chi')$ with
    $\dd_F(\mu_{\kappa'}(J,\chi'),\mu_{\kappa'}(\pi[J],\chi'))>0$ is
    $O\left(\tfrac{\zeta}{m_I}\cdot\rho^{\kappa,\kappa'}_{F}(I,\BC_I)\right)$.
    Also, from Claim~\ref{cl::phi_uncorrupted}, the expected potential
    added for each corrupted pair $(J,\chi')$ for an anchor is 
   $O\left(\tfrac{\mu_{\kappa'}(J,\chi')}{\lambda^{t}}\right)$.
    Since we sample $\lambda^t$ anchor pairs, %
     for each
    $(J,\chi')$ pair, the added potential over all anchors is a
    r.v. with expectation $O(\mu_{\kappa'}(J,\chi'))$.
      Overall, we obtain, using Fact~\ref{fct:ddTriangle}:
      $$\E\left[\dd_{\widehat{F}}(v,u)\right] \leq m_I \cdot \tfrac{\beta^l}{\zeta^2} \E\left[\dd_{F^4}(p,q)\right] + n^{-10} = m_I \cdot \tfrac{\beta^l}{\zeta^2} \cdot O(n^{-9}+\tfrac{\zeta}{m_I}\rho^{\kappa,\kappa'}_{F}(I,\BC_I)) + n^{-10} = O\left( \tfrac{\beta^l}{\zeta} \cdot \rho^{\kappa,\kappa'}_{F}(I,\BC_I)  + n^{-9}\right)$$

    completing the proof of the first part.

 \end{proof}

In the 2nd part, we need to prove Eqn.~\eqref{eqn:ddVarphiU}: that
$\dd_{\widehat{F}}\left(\varphi_\kappa^{\zeta}(I,\un),
\varphi_\kappa^{\zeta'}(\pi[I], \un)\right) =
O\left(\tfrac{\beta^l}{\zeta_+}\cdot\rho + \tfrac{1}{n^{10}}\right)$,
where $\rho= \rho^{\kappa,\kappa'}_{F}(I,\BC_I^{+\zeta})$.

                  We first introduce a
                  central claim.  
                  Define
$a=
\theta_\kappa^{\zeta}(\BC_I,\nu \setminus \{ \bot \})$, and $b=\theta_\kappa^{\zeta'}(\BC_{\pi[I]},\nu \setminus \{ \bot \})$. 
We prove the following later. We note that the case $m_I=0$ trivially
proves the lemma, and hence below $m_I>0$.

\begin{claim}
  \label{clm:ddAB}
Assuming $a=\Omega(\beta^{1.5})$, then, whp,
$\dd_{F^{O(1/\epsilon)}}(a,b) = O\left(\tfrac{k}{m_I}  \cdot \tfrac{\zeta_+}{\beta^{l}}\cdot
\rho + n^{-9}\right)$.
\end{claim}

\begin{proof}[Proof of Lemma \ref{lm::dd_dist_imp} (2), using
    Claim~\ref{clm:ddAB}]

We have, by definition of $\varphi(I,\un)$, $a,b$, and using
$\BC_I=\Lambda_\kappa^\zeta(I)$ and
$\BC_{\pi[I]}=\Lambda_\kappa^{\zeta'}(\pi[I])$, for level $l>0$:
\begin{align*}
  \dd_{\widehat{F}}(\varphi_\kappa^\zeta(I,\un),\varphi_\kappa^{\zeta'}(\pi[I],\un))
&= \dd_{\widehat{F}}\left(T^\un(m_I \cdot \tfrac{\beta^{2l}}{\zeta^2} \cdot \min\{\tfrac{a}{k},1\}), T^\un(m_{\pi[I]}\cdot \tfrac{\beta^{2l}}{\zeta'^2} \cdot \min\{\tfrac{b}{k},1\})\right) \\
&\stackrel{\dd, T \text{ are monotone}}{\leq}  \dd_{\widehat{F}}\left(T^\un(m_I \cdot \tfrac{\beta^{2l}}{\zeta_+^2} \cdot \tfrac{a}{k}), T^\un(m_{\pi[I]}\cdot \tfrac{\beta^{2l}}{\zeta_{+}'^2} \cdot \tfrac{b}{k})\right) .
\end{align*}

Similarly, for $l=0$, we have using the fact $\beta^{l}=\zeta_+=\zeta'_+ = 1$ at level 0, and that $\BC_I = \{I\}$ and $\BC_{\pi[I]} = \{\pi[I]\}$:

\begin{align*}
  \dd_{\widehat{F}}(\varphi_\kappa^0(I,\un),\varphi_\kappa^{0}(\pi[I],\un))
&= \dd_{\widehat{F}}\left(T^\un(m_I \cdot \tfrac{a}{k}), T^\un(m_{\pi[I]}\cdot \tfrac{b}{k})\right) \\
&=  \dd_{\widehat{F}}\left(T^\un(m_I \cdot \tfrac{\beta^{2l}}{\zeta_+^2} \cdot \tfrac{a}{k}), T^\un(m_{\pi[I]}\cdot \tfrac{\beta^{2l}}{\zeta_{+}'^2} \cdot \tfrac{b}{k})\right) .
\end{align*}

\aanote{using fixed constant 11 -- should perhaps just have in $T^\un$
def}
Note that $T^\un=T^{O(1/\eps)}_{1/\beta,\Omega_\eps(1)}$ transformation zeros-out any $o_\epsilon(1/\beta)$
quantity, and hence if $a = o_\epsilon(k \cdot \beta^{-1.5})$, the statement is trivial. Therefore, from this point on, we assume $a = \Omega_\epsilon(k \cdot \beta^{-1.5}) = \Omega(\beta^{1.5})$. 
We continue to derive, for some $F' =
\widehat{F}^{\Theta(\epsilon)}\cdot \tfrac{\zeta_+^2}{\zeta_+'^2}\cdot 1/F$:

\begin{align*}
  \dd_{\widehat{F}}(\varphi_\kappa^\zeta(I,\un),\varphi_\kappa^{\zeta'}(\pi[I],\un))
&\leq \dd_{\widehat{F}}\left(T^\un(m_I \cdot \tfrac{\beta^{2l}}{\zeta_+^2} \cdot \tfrac{a}{k}), T^\un(m_{\pi[I]}\cdot \tfrac{\beta^{2l}}{\zeta_{+}'^2} \cdot \tfrac{b}{k})\right) \\ 
&\stackrel{Claim~\ref{cl::T_trans}}{\leq}
\dd_{{\hat F}^{11\eps}}\left(m_I \cdot \tfrac{\beta^{2l}}{\zeta_+^2}\tfrac{a}{k},m_{\pi[I]} \cdot \tfrac{\beta^{2l}}{{\zeta_{+}'}^2}\tfrac{b}{k}\right) + n^{-10}
\\
&\stackrel{Fact~\ref{ft::dd_scalar}}{\leq}
\tfrac{m_I}{k} \cdot \tfrac{\beta^{2l}}{\zeta_+^2}\cdot\dd_{{\hat F}^{11\eps}\tfrac{m_{\pi[I]}}{m_I}\tfrac{\zeta_+^2}{\zeta_{+}'^2}}\left(a,b\right) + n^{-10}
\\
&=\tfrac{m_I}{k} \cdot \tfrac{\beta^{2l}}{\zeta_+^2}\dd_{F'}(a,b) + n^{-10}
\\
&\stackrel{Claim~\ref{clm:ddAB}}{=} O\left(\tfrac{\beta^{l}}{\zeta_+} \cdot \rho + n^{-9}\right),
\end{align*}
using $F'=F^{O(1/\eps)}$ and hence $\hat
F=(F'\tfrac{\zeta_{+}'^2}{\zeta_+^2} F)^{\Theta(1/\eps)}=F^{\Theta(1/\eps^2)}(\zeta_{+}'/\zeta_+)^{\Theta(1/\eps)}$.%

\end{proof}

It remains to prove Claim~\ref{clm:ddAB}. We first establish the following
auxiliary claims. First claim argues that for a set $\SC$ of
uncorrupted intervals $I$, the density of $I$ in $\SC$ (local density)
cannot be much larger than the density of $\pi[I]$ in $\pi[\SC]$.
\begin{claim}\label{clm::den_sets}
 	Fix a color $\chi \in \nu \setminus \{\bot\}$ in a coloring $\kappa$. Fix interval sets
        $\SC, \SC_\pi \subseteq \IC$ such that: (1) $\pi[\SC]
        \subseteq \SC_\pi$; and (2) $(I,\chi)$ is {\em
          $F$-uncorrupted} for all $I \in \SC$. Then, for any $I \in
        \SC$, for all but $1/\alpha$ costs $\hat{c}\in E_\basec$ we have
        $\den_\kappa(I,\chi,\SC)_{\hat{c}} \le n^\alpha F \cdot \den_\kappa(\pi[I],\chi,\SC_\pi)_{\hat{c}}$.
 \end{claim} 
\begin{proof}
  For each $I \in \SC$, we have $\cad(I,\pi[I]) \leq c$ since  $(I,\chi)$ is {\em $F$-uncorrupted}.
Fix some $I \in \SC$. By triangle inequality, we have that
$\pi[\NC_{\hat{c}}(I) \cap \SC] \subseteq \NC_{\hat{c} + 2\basec}(\pi[I]) \cap
\SC_\pi$, and since all pairs in $\SC \times \{\chi\}$ are {\em
  $F$-uncorrupted}, then $\mu_\kappa(\NC_{\hat{c}}(I) \cap \SC) \leq F
\cdot \mu_\kappa(\pi[\NC_{\hat{c}}(I) \cap \SC])$. Overall, we obtain that that, $\den_\kappa(I, \chi,\SC)_{\hat{c}} \leq F \cdot \den_\kappa(\pi[I],\chi,\SC_\pi)_{\hat{c} + 2\basec}$. Since there can be at most $1/\alpha$ costs where $\den_\kappa(\pi[I],\chi,\SC_\pi)_{\hat{c} + 2\basec} > n^\alpha \cdot \den_\kappa(\pi[I],\chi,\SC_\pi)_{\hat{c}}$, the claim follows.
 \end{proof}
 
Second claim shows that for some interval sets $\SC$, the local density
of $I\in \SC$ in $\SC$ is at least a $1/F_d$ fraction of their density with
respect to another set of interest $\SC'$ (on most costs in
$E_\basec$), except for a mass $\approx 1/F_d$-fraction of the mass of
$\SC'$. This claim will help us both 1) to bound the mass of corrupted
pairs where their local density is significantly skewed by corrupted
pairs which can potentially generate additional corruption, as well as
2) dealing with the sharp threshold of Line \ref{alg::theta_treshold}
in \textsc{AssignThetaPotential} Algorithm.
     
\begin{claim}\label{clm::dense_skew}
Fix interval sets $\SC, \SC' \subseteq \IC$, color $\chi \in \nu \setminus \{\bot \}$ in a coloring
$\kappa$ and density factor $F_d$. Except for a mass of
$O_\eps\left(\tfrac{n^\alpha}{F_d}\cdot \mu_\kappa(\SC',\chi)\right)$ of
$(I,\chi)\in(\SC,\chi)$, for all except $\le 1/\alpha$ costs $\hat{c} \in
E_\basec$, we have $\den_\kappa(I,\chi,\SC)_{\hat{c}}\ge \tfrac{1}{F_d}\cdot \den_\kappa(I,\chi,\SC')_{\hat{c}}$.
\end{claim}

\begin{proof}
Let $\mu_I = \mu_\kappa(I,\chi)$ and $\mu_{\SC} = \mu_\kappa(\SC_,\chi)=\sum_{I \in \SC} \mu_I$. 	Define
a distribution $p_{\SC}$ over $\SC$ by setting $p_{\SC}(I) = \tfrac{\mu_I}{\mu_{\SC}}$.
For each $\hat{c}\in E_\basec$, we consider the following expectation
$\E_{I \sim p_{\SC}}
\left[\tfrac{\den_\kappa(I,\chi,\SC')_{\hat{c}}}{\den_{\kappa}(I,\chi,\SC)_{2\hat{c}}}\right]$. By triangle inequality, we have that $\den_{\kappa}(I,\chi,\SC)_{2\hat{c}}
        \geq \max_{J \in \NC_{\hat c}(I) \cap \SC'}
        \den_{\kappa}(J,\chi,\SC)_{\hat{c}}$. Hence, we can estimate the expectation as follows:
        
	$$\sum_{I \in \SC} \tfrac{\mu_I}{\mu_{\SC}} \tfrac{\den_\kappa(I,\chi,\SC')_{\hat{c}}}{\den_{\kappa}(I,\chi,\SC)_{2\hat{c}}} = \tfrac{1}{\mu_{\SC}}\sum_{I \in \SC, J \in \SC' : \cad(I,J) \leq \hat{c}} \tfrac{\mu_I\cdot \mu_J}{\den_{\kappa}(I,\chi,\SC)_{2\hat{c}}}\leq \tfrac{1}{\mu_{\SC}}\sum_{I \in \SC, J \in \SC' : \cad(I,J) \leq \hat{c}} \tfrac{\mu_I\cdot \mu_J}{\den_\kappa(J,\chi,\SC)_{\hat{c}}} = \tfrac{1}{\mu_{\SC}}\sum_{J \in \SC'} \mu_J = \tfrac{\mu_{\SC'}}{\mu_{\SC}}.$$

	By Markov Inequality, we have that $\Pr_{I \sim p_{\SC}}
        \left[\tfrac{\den_\kappa(I,\chi,\SC')_{\hat{c}}}{\den_{\kappa}(I,\chi,\SC)_{2\hat{c}}} >
          \tfrac{F_d}{n^\alpha} \right] \leq \tfrac{\mu_{\SC'}}{\mu_{\SC}}
        \cdot \tfrac{n^\alpha}{ F_d}$. By union bound over $\hat{c}\in E_\basec$, there is a mass of $O_\epsilon(\tfrac{\mu_{\SC'}\cdot n^\alpha}{F_d})$ where the inequality occurs for at least one cost. On the other hand, for intervals where the inequality does not occur for any cost, there can be at most $1/\alpha$ costs where $\tfrac{\den_{\kappa}(I,\SC)_{2\hat{c}}}{\den_{\kappa}(I,\SC)_{\hat{c}}} > n^\alpha$ (noting $n$ is an upper bound for $\mu_{\SC}$), and hence for all other costs we have $\den_\kappa(I,\chi,\SC')_{\hat{c}} \leq F_d \cdot \den_\kappa(I,\chi,\SC)_{\hat{c}}$ as needed.

\end{proof}

A corollary is that the sharp threshold of Line~\ref{alg::theta_treshold}
in \textsc{AssignThetaPotential} Algorithm cannot affect the $\theta$ scores too much. For a pair $(J,\chi')$, we use the notation
$\sigma_{J,\chi'}=Q_l \left(T^\theta\left(\Gamma_{l,\zeta}(J,\chi')\right)\right)$ as the quantity in Line~\ref{alg::sigma_contribution} of
\textsc{AssignThetaPotential}. Note that
$\theta_\kappa^\zeta(J,\chi')= m_\VC((J,\chi')) \cdot \sigma_{J,\chi'}$ unless the threshold
of Line \ref{alg::theta_treshold} passes.

\begin{corollary}\label{cor::theta_threshold}
		Fix a level $l > 0$ and $\zeta \in Z_l$.
	Consider an  
	interval ball
	$\QC$ %
with
        $\mu_{\kappa}(\QC,\un) = O(\zeta)$. Let $\QC^{+\zeta} =
        \Lambda_\kappa^\zeta(\QC)$. Then, %
        $\E[\theta_\kappa^{\zeta}(\QC,\nu \setminus \{ \bot \})] < \left(1 + o(1)\right) \cdot \tfrac{k}{\beta^l} \cdot \sum_{(I',\chi') \in \QC^{+\zeta} \times \nu \setminus \{ \bot \}} \mu_{\kappa'}(I',\chi') \cdot \sigma_{I',\chi'}$.
        \end{corollary}

        \begin{proof}

        	Let $\theta_\QC \triangleq \theta_\kappa^{\zeta}(\QC,\nu \setminus \{ \bot \})$.
        	Define the set of colors of large mass in
                $\QC^{+\zeta}$, where  Line \ref{alg::theta_treshold}
                of \textsc{AssignThetaPotential} may pass:
$$\TC = \{ \chi' \in
     \nu \setminus \{ \bot \} \mid \mu_{\kappa'}(\QC^{+\zeta},\chi') > \lambda \cdot \beta^{l+6} \}.$$ 
     Also, define $\bar{\TC} = \nu \setminus \TC \setminus \{ \bot \}$. 
     Notice that %
     pairs $(J, \chi') \in \QC \times \bar{\TC}$ will not pass the threshold from Line \ref{alg::theta_treshold} of \textsc{AssignThetaPotential} and their contribution to $\E[\theta_\QC]$ is %
     $\E[m_\VC((J,\chi'))] \cdot 
\sigma_{J,\chi'} = \tfrac{k}{\beta^l} \mu_{\kappa'}(J,\chi') \cdot
\sigma_{J,\chi'}$.

Now focus on pairs $(J, \chi') \in \QC \times \TC$. We split the ball $\QC^{+\zeta}$ into $h = O(1)$ consecutive balls each of
$\un$-mass $\leq \zeta$ each: call them $\QC^{q}$ for $q\in \{1,2,\ldots,h\}$.
By Lemma
   \ref{lm::partition_size_imp} we have that $\mu_{\kappa'}(\IC,\chi')
   \leq | \PC_{\kappa'}^{\chi'} | = \tO_\eps (\tfrac{n \cdot \lambda}{\lambda^t} \cdot \beta^{5})$. 
    For fixed $q \in [h]$, we invoke Claim \ref{clm::dense_skew}, using
   $\SC=\QC^q$, $\SC' = \IC$, and $F_d = \tfrac{n}{\lambda^t
     \beta^l}$, to obtain that there exists at most $\tO_\eps(n^{\alpha}
   \cdot \lambda \cdot \beta^{l+5}) \stackrel{\TC \text{ def.}}{=} o(\mu_{\kappa'}(\QC^{+\zeta},\chi'))$
   mass of $J\in \QC^q$ where $\relden_{\kappa'}(J,\chi',\QC^q)_{\hat{c}} >
   \tfrac{n}{\lambda^t\beta^l}$ on more than $1/\alpha$
   costs. Note that when $\relden_{\kappa'}(J,\chi',\QC^q)_{\hat{c}} \le
   \tfrac{n}{\lambda^t\beta^l}$ for all $q$, then also $\relden_{\kappa'}(J,\chi',\Lambda^\zeta_\kappa(J))_{\hat{c}} \le
   \tfrac{n}{\lambda^t\beta^l}$.

   Hence for all $J\in \QC$, except for
   $O(1) \cdot o(\mu_{\kappa'}(\QC^{+\zeta},\chi'))$, we have
   $\relden_{\kappa'}(J,\chi',\Lambda^\zeta_\kappa(J))_{\hat{c}} \le
   \tfrac{n}{\lambda^t\beta^l}$ on all but $1/\alpha$ costs $\hat{c}$. The
   latter implies that $\sigma_{J,\chi'} = 1$ whp. This in turn gives us:
    \begin{equation}\label{eq::theta_Q}
    	\sum_{(J,\chi') \in \QC \times \TC} \mu_{\kappa'}(J,\chi') \cdot
\sigma_{J,\chi'} \ge \mu_{\kappa'}(\QC,\TC)- o(\mu_{\kappa'}(\QC^{+\zeta},\TC)).
    \end{equation}

    Now, we combine both cases. Note that, by construction, we have:

$$
\E[\theta_\QC] \le
\sum_{(J,\chi') \in \QC \times \bar{\TC}} %
\E[m_\VC((J,\chi'))] \cdot
\sigma_{J,\chi'} + \sum_{(J,\chi') \in \QC \times \TC} %
\E[m_\VC((J,\chi'))]
$$
$$= \tfrac{k}{\beta^l}\cdot\left(\sum_{(J,\chi') \in \QC \times \bar{\TC}} \mu_{\kappa'}(J,\chi') \cdot
\sigma_{J,\chi'}+\mu_{\kappa'}(\QC,\TC)\right).$$

Using Eqn.~\ref{eq::theta_Q}, we conclude:

     $$\E[\theta_\QC] \le
    \tfrac{k}{\beta^l}\cdot\left(\sum_{(J,\chi') \in
      \QC\times \nu \setminus \{\bot\}}
    \mu_{\kappa'}(J,\chi')\sigma_{J,\chi'}
 +
    o\left(\mu_{\kappa'}(\QC^{+\zeta},\TC)\right)\right)    $$
    $$\le \tfrac{k}{\beta^l}\left(1+o(1)\right)\cdot\left(\sum_{(J,\chi') \in
      \QC^{+\zeta}\times \nu \setminus \{\bot\}}
    \mu_{\kappa'}(J,\chi')\sigma_{J,\chi'}\right),$$
        	
        	as needed.
        	
        \end{proof}

Finally, we are ready to prove Claim~\ref{clm:ddAB}.

\begin{proof}[Proof of Claim~\ref{clm:ddAB}]
\ns{changed the proof. Need to check.}Recall that $a=\sum_{(J,\chi') \in \VC\cap\BC_I \times \nu}
\theta_\kappa^{\zeta}(J,\chi')$. Since $a\ge
\Omega_\eps(\beta^{1.5})$, and is a sum of independent r.v. bounded by 1, we have that $a=\Theta(\E[a])$ whp, hence
we analyze $\E[a]=\sum_{J\in \BC_I,\chi'} \Pr[(J,\chi')\in \VC]\cdot \theta_\kappa^\zeta(J,\chi')$. The plan is to
show that $\E[b] \cdot F^{O(1/\eps)} \ge  \E[a]$ unless error $\rho$ is
sufficiently large. Note that this requires comparing
$\theta_\kappa^\zeta(J,\chi')$ vs $\theta_\kappa^{\zeta'}(\pi[J],\chi')$.

For reasons that will be clear later, we consider $J\in \BC_I^{+\zeta}$. 

Consider an {\em $F$-uncorrupted} pair $(J,\chi') \in \BC_I^{+\zeta} \times \nu \setminus \{ \bot \}$;
i.e., $\dd_F(\mu_{\kappa'}(J,\chi'),\mu_{\kappa'}(\pi[J],\chi')) = 0$.
Let $p = \Gamma_{l,\zeta}(J,\chi')$ and $q =
\Gamma_{l,\zeta'}(\pi[J],\chi')$. %
By the definition of $\Gamma$, we have for a cost $\hat c$:

$$\tfrac{p_{\hat{c}}}{q_{\hat{c}}} =
O\left(\tfrac{\relden_{\kappa'}(\pi[J],\chi',\Lambda_\kappa^{\zeta'}(\pi[J]))_{\hat{c}}}{\relden_{\kappa'}(J,\chi',\Lambda_\kappa^\zeta(J))_{\hat{c}}}
+ 1\right) =
O\left(\tfrac{\den_{\kappa'}(J,\chi',\Lambda_\kappa^\zeta(J))_{\hat{c}}}{\den_{\kappa'}(\pi[J],\chi',\Lambda_\kappa^{\zeta'}(\pi[J]))_{\hat{c}}}\cdot
\tfrac{\den_{\kappa'}(\pi[J],\chi')_{\hat{c}}}{\den_{\kappa'}(J,\chi')_{\hat{c}}} +
1\right),$$

since the two $\relden$ quantities are estimated
up to a constant-factor, and capped at the same value $\relden_m=n^{4\alpha}\tfrac{n}{\lambda^t\beta^l}$.
For the latter fraction, we note that since $\den_{\kappa'}(J,\chi',\SC)_{\hat{c}}$ is monotonic in $\hat{c}$ for any fixed $J,\SC$, and since $(J,\chi')$ is uncorrupted, 
then by triangle inequality, there are at most $1/\alpha$ costs in $E_\basec$ where $\tfrac{\den_{\kappa'}(\pi[J],\chi')_{\hat{c}}}{\den_{\kappa'}(J,\chi')_{\hat{c}}} \geq n^{\alpha}$.

The ratio $\tfrac{\den_{\kappa'}(J,\chi',\Lambda_\kappa^\zeta(J))_{\hat{c}}}{\den_{\kappa'}(\pi[J],\chi',\Lambda_\kappa^{\zeta'}(\pi[J]))_{\hat{c}}}$ is more tricky and we analyze it next.
Define the set $\EC \subseteq \BC_I^{+\zeta} \times \nu$ to be the set of
interval, color pairs $(J,\chi')$ which are either corrupted, or where
$\tfrac{\den_{\kappa'}(J,\chi',\Lambda_\kappa^\zeta(J))_{\hat{c}}}{\den_{\kappa'}(\pi[J],\chi',\Lambda_\kappa^{\zeta'}(\pi[J]))_{\hat{c}}}
> n^{3\alpha}\cdot F$ for more than $2/\alpha$ distinct costs $\hat{c} \in E_\basec$.%

We analyze two cases depending on whether pairs in $\EC$ contribute
significantly to $a$ or not.

\vspace{2mm}

\paragraph{Pairs in $\EC$ contribute insignificantly to $a$.}
For an uncorrupted pair $(J,\chi') \in \{\BC_I^{+\zeta} \times \nu\}
\setminus \EC$, we can combine both inequalities above to obtain:
there exists at most $3/\alpha$ costs in $E_\basec$ where $\tfrac{p_{\hat{c}}}{q_{\hat{c}}}
> n^{4\alpha}\cdot F$. Let $C \subseteq E_\basec$ be the set of all other
costs.

We now use Corollary \ref{cr::Q_trans} with $d=|E_\basec| > O(1/\eps\cdot
1/\alpha^2)$ and $v=3/\alpha$, for some $F'=F^{O(1/\eps)}$: %
\begin{align}
  \dd_{F'}(\sigma_{J,\chi'},\sigma_{\pi[J],\chi'})
        &= \dd_{F'}(Q_l(T^\theta([p_{\hat{c}}]_{\hat{c}\in
    E_\basec})),Q_l(T^\theta([q_{\hat{c}}]_{\hat{c}\in E_\basec})))
  \nonumber
  \\
	&\stackrel{Corollary~\ref{cr::Q_trans}}{\leq} \dd_{F' \cdot n^{-O(\alpha)}}([T^\theta(p_{\hat{c}})]_{\hat{c}\in C},[T^\theta(q_{\hat{c}})]_{\hat{c}\in
          C}) + n^{-10}
  \nonumber
  \\
  &\stackrel{Claim~\ref{cl::T_trans}}{\leq} \dd_{(F'n^{-O(\alpha)})^{O(\eps)}}([p_{\hat{c}}]_{\hat{c}\in C},[q_{\hat{c}}]_{\hat{c}\in C}) + 2\cdot n^{-10}
  \nonumber
  \\
  &\leq  \dd_{F^2 \cdot n^{4\alpha}}([p_{\hat{c}}]_{\hat{c}\in C},[q_{\hat{c}}]_{\hat{c} \in C}) +
  O(n^{-10})
  \nonumber
  \\
  &= O(n^{-10}),
  \label{eqn:ddSigmas}
\end{align}
where last inequality used the lemma assumption that $F\ge
n^{\alpha}$, and the last equality used the definition of $C$ above.

Now, for level $l=0$, we have $\E[a] = k \cdot\left(\sum_{(J,\chi') \in
      \BC_I \times \nu \setminus \{\bot \}}
    \mu_{\kappa'}(J,\chi')\sigma_{J,\chi'}\right)$.
Otherwise, we invoke Corollary~\ref{cor::theta_threshold} for $\QC=\BC_I$ to obtain for all levels:

    \begin{align*}
    \E[a] 
    &\le \tfrac{k}{\beta^l}\left(1+o(1)\right)\cdot\left(\sum_{(J,\chi') \in
      \BC_I^{+\zeta}\times \nu \setminus \{\bot \}}
    \mu_{\kappa'}(J,\chi')\sigma_{J,\chi'}\right) \\
    &\le \tfrac{k}{\beta^l}\left(1+o(1)\right)\cdot\left(\mu_{\kappa'}(\EC)+\sum_{(J,\chi') \in
      \BC_I^{+\zeta}\times (\nu \setminus \{\bot\}) \setminus \EC}
    \mu_{\kappa'}(J,\chi')\sigma_{J,\chi'}\right).
    \end{align*}

    Now, suppose the second term is the dominant one, i.e.,
    
    $$\E[a]\le 3\tfrac{k}{\beta^l}\cdot \sum_{(J,\chi') \in
      \BC_I^{+\zeta}\times (\nu \setminus \{\bot\}) \setminus \EC}
    \mu_{\kappa'}(J,\chi')\sigma_{J,\chi'}.$$ 
    Then, for any $(J,\chi')$
    in the sum, we have
    $\tfrac{\mu_{\kappa'}(J,\chi')}{\mu_{\kappa'}(\pi[J],\chi')}
    \leq F$ (($J,\chi'$) is uncorrupted) and
    $\tfrac{\sigma_{J,\chi'}}{\sigma_{\pi[J],\chi'}} \leq
    {F}^{O(1/\epsilon)}$ unless $\sigma_{J,\chi'}<O(n^{-10})$ (from
    Eqn.~\eqref{eqn:ddSigmas}). Hence, using the assumption that $F\le \beta^{O(\eps)}$
    $$\E[b] \geq
    F^{-O(1/\epsilon)}(\E[a]-O(n^{-9})) \geq \beta^{-1} \cdot
    \Omega_\epsilon(\beta^{1.5}) = \Omega_\epsilon(\beta^{0.5}).$$
    Therefore, we
    have that $b = \Theta(\E[b])$ whp as well, and overall we obtained
    $\dd_{F^{O(1/\epsilon)}}(a,b) = O(n^{-9})$ whp.
    
\vspace{2mm}

\paragraph{Pairs in $\EC$ contribute significantly.} Now assume that
$\E[a]\le 3\tfrac{k}{\beta^l}\cdot\mu_{\kappa'}(\EC)$.

We now bound the color mass $\mu_{\kappa'}(\EC)$.
    For a color $\chi'$, we define $\CC_{\chi'}$ to be the set of 
     corrupted pairs in $\BC_I^{+\zeta} \times \{\chi'\}$, and $\CC := \cup_{\chi'\in\nu} \CC_{\chi'}$. Also let $\JC_{\chi'}$ be the interval set representing the first coordinates of $\CC_{\chi'}$.
     Like in the proof of part (1) of the Lemma
     above, we have $\mu_{\kappa'}(\CC) \le
     \tfrac{2\zeta_+}{m_I}\cdot\rho$ where $\rho=
     \rho^{\kappa,\kappa'}_{F}(I,\BC_I^{+\zeta})$. For level $l=0$, we have $\BC_I = \{ I \}$ and hence $\EC \subseteq \CC$ (when $\BC_I = \{ I \}$, a pair $(I,\chi')$ can only be in $\EC$ if it is corrupted by definition of $\EC$). For $l \geq 1$,
     let $\BC_I^{+2\zeta} = \Lambda_\kappa^{3\zeta}(I)$ and $\BC_{\pi[I]}^{-\zeta'/2} = \Lambda_\kappa^{\zeta'/2}(\pi[I])$.
     We invoke Claim \ref{clm::den_sets} with $\SC=\BC_I^{+2\zeta} \setminus
     \JC_{\chi'}, \SC'=
     \BC_{\pi[I]}^{-\zeta'/2}\supseteq
     \pi[\BC_I^{+2\zeta}]$.
     The claim implies that, for
     $J\in \BC_I^{+\zeta}$ and hence $\BC_{\pi[I]}^{-\zeta'/2}\subseteq\Lambda_\kappa^{\zeta'}(\pi[J])$, we have for all but $1/\alpha$
     costs $\hat c$:
     $$\tfrac{\den_{\kappa'}(J,\chi',\Lambda_\kappa^\zeta(J) \setminus
       \JC_{\chi'})_{\hat{c}}}{\den_{\kappa'}(\pi[J],\chi',\Lambda_\kappa^{\zeta'}(\pi[J]))_{\hat{c}}}
 \le \tfrac{\den_{\kappa'}(J,\chi',\BC_I^{+2\zeta} \setminus
       \JC_{\chi'})_{\hat{c}}}{\den_{\kappa'}(\pi[J],\chi',\BC_{\pi[I]}^{-\zeta'/2})_{\hat{c}}}
 \le n^\alpha\cdot F.$$
     
      Let $\EC_{\chi'}$ be the set of {\em
       uncorrupted} pairs $(J,\chi')\in\BC_I^{+\zeta} \times \{\chi'\}$ with
     $\tfrac{\den_{\kappa'}(J,\chi',\Lambda_\kappa^\zeta(J))_{\hat{c}}}{\den_{\kappa'}(J,\chi',\Lambda_\kappa^\zeta(J)
       \setminus \JC_{\chi'})_{\hat{c}}} > n^{2\alpha}$ on more than
     $1/\alpha$ costs in $E_\basec$. 
    Then, for each uncorrupted
     pair $(J,\chi')\in \BC_I^{+\zeta}\times\{\chi'\}\setminus\EC_{\chi'}\setminus
     \CC_{\chi'}$, except for $\le 2/\alpha$ costs, we have
     $\tfrac{\den_{\kappa'}(J,\chi',\Lambda_\kappa^\zeta(J))_{\hat{c}}}{\den_{\kappa'}(\pi[J],\chi',\Lambda_\kappa^{\zeta'}(\pi[J])])_{\hat{c}}}
     \le n^{3\alpha}\cdot F$.

    Therefore $\EC \subseteq \cup_{\chi'} \EC_{\chi'} \cup \CC$. It
    remains to bound $\mu_{\kappa'}(\cup_{\chi'} \EC_{\chi'}) =
    \sum_{\chi'} \mu_{\kappa'}(\EC_{\chi'})$. Partition
    $\BC_I^{+\zeta}$ into 8 balls each of $\un$-mass $\leq\zeta$, called
    $\BC_I^{q}$ for $q\in\{1,2,\ldots,8\}$. Then, for each $q$ and color
    $\chi' \in \nu$, we invoke Claim \ref{clm::dense_skew} with
    $\SC=\BC_I^q\setminus \JC_{\chi'}$, $\SC'=\JC_{\chi'}$, and
    $F_d=n^{2\alpha}/2$, to obtain that there's only a mass of
    $O(\tfrac{1}{n^{\alpha}} \mu_{\kappa'}(\CC_{\chi'}))$ of $(J,\chi')\in
    \BC_I^q \times \{\chi'\} \setminus \CC_{\chi'}$ having more than $1/\alpha$ costs
    satisfying
    $\tfrac{\den_{\kappa'}(J,\chi',\JC_{\chi'})_{\hat{c}}}{\den_{\kappa'}(J,\chi',\BC_I^q\setminus
      \JC_{\chi'})_{\hat{c}}}>n^{2\alpha}/2$. Now, by definition, we have for each $(J,\chi')\in
    \EC_{\chi'}$ that
    $\tfrac{\den_{\kappa'}(J,\chi',\Lambda_\kappa^\zeta(J))_{\hat{c}}}{\den_{\kappa'}(J,\chi',\Lambda_\kappa^\zeta(J)
      \setminus \JC_{\chi'})_{\hat{c}}} > n^{2\alpha}$ on more than $1/\alpha$ distinct costs, implying that     $\tfrac{d_{\kappa'}(J,\chi',\JC_{\chi'})_{\hat{c}}}{\den_{\kappa'}(J,\chi',\Lambda_\kappa^\zeta(J)
      \setminus \JC_{\chi'})_{\hat{c}}} > n^{2\alpha}-1$, as well as
    that $\tfrac{\den_{\kappa'}(J,\chi',\JC_{\chi'})_{\hat{c}}}{\den_{\kappa'}(J,\chi',\BC_I^q
      \setminus \JC_{\chi'})_{\hat{c}}} > n^{2\alpha}/2$, for $q$
    s.t. $J\in \BC_I^q$ and hence $\BC_I^q\subseteq \Lambda_\kappa^\zeta(J)$. However, as deduced above, such pairs must have mass bounded by $O(\tfrac{1}{n^{\alpha}} \mu_{\kappa'}(\CC_{\chi'}))$, implying that
   $\mu_{\kappa'}(\EC_{\chi'})=o(\mu_{\kappa'}(\CC_{\chi'}))$. Summing
   over all colors in $\nu$, we obtain that $\mu_{\kappa'}(\EC) = (1+o(1))\cdot \mu_{\kappa'}(\CC) = O(\tfrac{\zeta_+}{m_I}\rho)$.

\vspace{2mm}

   Overall, we conclude:
    $$\dd_{F'}(a,b) \leq a\le O(\E[a]) \le  O\left(\tfrac{k}{\beta^l}\cdot
      \tfrac{\zeta_+}{m_I}\rho\right).$$
    
    This completes the proof of Claim~\ref{clm:ddAB}.
\end{proof}

%% file: imp-balance.tex
\subsection{Controlling balance of colors: proof of Lemma \ref{lm::min_color_mass}}\label{sec::min_color_proof}

	 \ns{check intuition...} In this section we analyze the quantity $M_I = \mu_\kappa(I,\nu
\setminus \{ \un \})+ \sum_{\zeta \in
  Z_l}\varphi_\kappa^\zeta(I,\nu)$, and in particular, prove that for each level $l$, $M_I = \Omega_{\eps}(1)$. This will imply
        that our normalization of $\mu_{\hat\kappa}$ introduces only a
        constant-factor further error.

Most importantly, analyzing $M_I$ requires understanding
of the quantity $\varphi_\kappa^\zeta(I,\nu)$. At a high level, we
show that one of the following must hold, for $\zeta = \beta^l$ (or
$\zeta=0$ for $l=0$):
\begin{enumerate}
	\item $\varphi_\kappa^\zeta(I,\un) = \Omega_\epsilon(1)$, or
	\item $\varphi_\kappa^\zeta(I,\nu\setminus \{\un\}) = \Omega_\epsilon(1)$, or,
	\item $\mu_\kappa(I,\nu
\setminus \{ \un \}) = \Omega(1)$.
\end{enumerate}

To do so, we proceed in a few
steps. Recall that $\varphi_\kappa^\zeta(I,\nu\setminus \{\un\})$ depends on
$\phi(\Lambda^\zeta_\kappa(I),\nu)$ potential. First, we show how to
       bound $\phi(\SC,\chi)$ for a set $\SC$ as a function of
       $\relden_{\kappa'}(I,*,\SC)$, for an interval $A$ in the
       neighborhood of the sampled anchor $A$ (Claim~\ref{clm::rel_dense_bound_imp}). Second, we use
       this bound to lower bound $\norm{\varphi(I,*_{\neq \un})}_1$ (which applies
       $T^\cc$ to a re-scaling of $\phi(\Lambda^\zeta(I),*)$) as a
       function of $\relden_{\kappa'}(J,\chi',\SC)$ for $J\in \SC$
       where $\SC$ is $\Lambda^\zeta(I)$ ---
       in particular, if enough such $J$ have $\relden$ about
       $\tfrac{n}{\beta^l\lambda^t}$ on a non-trivial fraction of
       costs $c\in E_\basec$, we get such a lower bound (Claim~\ref{clm::dense_set_imp}), as long as $\mu_\kappa(I,\un)$ is not too small (otherwise, item (3) above holds).

       We then focus on $\varphi(I,\un)$ which relies on estimates of
       $\theta(\Lambda_\kappa^{\zeta}(J),\nu \setminus \{ \bot \})$.
       Our third claim bounds the $\theta$ quantity for $J\in \SC$ as
       a function of the mass of $I'\in \SC,\chi'$ such that
       $\relden_{\kappa'}(I',\chi',\SC)_c$ is at most
       $O^*\left(\tfrac{n}{\beta^l\lambda^t}\right)$ on some fraction $\approx s_l$ of
       costs $c\in E_\basec$ (Claim~\ref{cl::Q_U_bound}). Our final fourth claim
       (Claim~\ref{cl::density_range_set}), uses the previous claim (applied to small $\zeta$ balls of the previous level $l-1$) to show that, overall,
       small $\theta(\Lambda^{\zeta}(I))$'s (i.e., insufficient for (1) above to hold) imply that many pairs in $I$'s
       proximity have the density $\relden$ in the right range for
       enough costs (allowing us to obtain (2) using %
       Claim~\ref{clm::dense_set_imp}). Finally, we use these
       four claims to prove Lemma~\ref{lm::min_color_mass}.
       
               Throughout this section, each time we refer to sets of costs,
        we refer to subsets of $E_\basec$, where $\basec$ is the base cost of
        \textsc{MatchIntervals}. Also, when we refer to contribution
        of a sampled cost $c_i \in E_\basec$, we refer to the total
        contribution of all $c_{i,j} = c_i + \basec \cdot j$ where
        $j\in\{0,\ldots j_{\max}\}$ where $j_{\max}=O(1/\alpha)$.

       In the first claim controlling $\phi$, we show that, for a
       fixed new color $\chi$ and chosen anchor $(A,\chi')$, for any
       ball $\SC \subseteq \IC$, the total potential added to $\SC$ is roughly
       concentrated.
       
		\begin{claim}[Controlling $\phi(\SC,\chi)$]\label{clm::rel_dense_bound_imp}
		Fix an interval set $\SC \subseteq \IC$ and output
                color $\chi$. For any
                interval, color pair $(I,\chi')$,
      for  all but
                $O(1/\alpha)$ costs $c_i \in E_\basec$, and 
                for all $A \in \NC_{2c_i}(I)$, we have, where
                $\relden_{c_i}=\relden_{\kappa'}(I,\chi',\SC)_{c_i}$:
                \begin{enumerate}
                	\item If $(A,\chi',c_i)$ was sampled for color $\chi$, then $\phi_\kappa(\SC,\chi) =
                   O(n^{2\alpha}) \cdot
                    \tfrac{\totalmu}{\relden_{c_i} \cdot
                      \lambda^t}$; and
                      \item If $(A,\chi',c_{i+1})$ was sampled for color $\chi$, then $\phi_\kappa(\SC,\chi) =
                   \Omega(n^{-2\alpha}) \cdot
                    \tfrac{\totalmu}{\relden_{c_i} \cdot
                      \lambda^t}$.
                \end{enumerate}
                      
                     \end{claim}

\begin{proof}
	\ns{changed, to check.}
	Let $C_{I,\chi'}^{\text{Good}} \subseteq E_\basec$ be all $c_i$ costs where:
	\begin{itemize}
		\item $\den_{\kappa'}(I,\chi')_{c_i} \cdot n^\alpha \geq \den_{\kappa'}(I,\chi')_{15c_i}$; and
		\item $\den_{\kappa'}(I,\chi',\SC)_{c_i} \cdot n^\alpha \geq \den_{\kappa'}(I,\chi',\SC)_{15c_i}$.
		\item $i$ is not maximal (and hence $c_{i+1}$ may be sampled as well).
	\end{itemize}
	We note that by construction, $15c_i < c_{i+3}$ and since
        $\mu_{\kappa'}(\SC,\chi') \in \{0\} \cup [n^{-10},n]$ for any
        set $\SC$, then there exist at most $66/\alpha + 1$ costs in
        $E_\basec \setminus C_{I,\chi'}^{\text{Good}}$.

        We now show the claim holds for all costs in $C_{I,\chi'}^{\text{Good}}$.
	Fix a sampled anchor, color pair $(A,\chi')$ in coloring
        $\kappa'$ and a cost ${c_i} \in C_{I,\chi'}^{\text{Good}}$. 
        For $J \in \IC$, denote
        $d_{J,i,j} = \den_{\kappa'}(J,\chi')_{c_{i,j}}$. By $\phi$
        definition, if $(A,\chi',c_i)$ was sampled, we have
        $\phi_\kappa(\SC,\chi) = \sum_j n^{-\alpha 
          j} \cdot m_{i,j}$, where $m_{i,j} = \sum_{J \in \NC_{c_{i,j}}(A)
          \cap \SC} \tfrac{\mu_{\kappa'}(J,\chi')}{\widehat{d_{J,i,j}}} \cdot
        \tfrac{2\totalmu}{\lambda^t}$, and
        $\widehat{d_{J,i,j}}$ is constant factor approximation to
        $\max\{d_{J,i,j},\den_m\}$ for
        $\den_m=\den_{\kappa'}(A,\chi')_{3c_i} \cdot n^{-\alpha}$
        (whp, by
        Lemma~\ref{lm::approx_density}). Hence, using the triangle inequality, we have,
         
        $$\widehat{\den_{J,i,j}} \in [\Theta(n^{-\alpha}),\Theta(1)]
        \cdot \den_{\kappa'}(A,\chi')_{3c_i} \subseteq
              [\Theta(n^{-\alpha}) \cdot
                \den_{\kappa'}(I,\chi')_{c_i},\Theta(1) \cdot
                \den_{\kappa'}(I,\chi')_{5c_i}] \subseteq
              [\Theta(n^{-\alpha}),\Theta(n^{\alpha})] \cdot
              \den_{\kappa'}(I,\chi')_{c_i};$$
              
               and similarly,
              
              $$\widehat{\den_{J,i+1,j}} \in   [\Theta(n^{-\alpha}) \cdot
                \den_{\kappa'}(I,\chi')_{3c_i},\Theta(1) \cdot
                \den_{\kappa'}(I,\chi')_{15c_i}] \subseteq 
                  [\Theta(n^{-\alpha}),\Theta(n^{\alpha})] \cdot
                  \den_{\kappa'}(I,\chi')_{c_i}.$$ 
                  
                  \aanote{why is that?}\ns{is it clearer now?} Therefore we have
                  that for each $j$:
                  
        	$$m_{i,j}
	 = \sum_{J \in \NC_{c_{i,j}}(A) \cap \SC} \tfrac{\mu_{\kappa'}(J,\chi')}{\widehat{d_{J,i,j}}} \cdot \tfrac{2\totalmu}{\lambda^t} 
	 \in \tfrac{[\Theta(n^{-\alpha}),\Theta(n^{\alpha})]}{\den_{\kappa'}(I,\chi')_{c_i}} \cdot \tfrac{\totalmu}{\lambda^t}\cdot \sum_{J \in \NC_{c_{i,j}}(A) \cap \SC} \mu_{\kappa'}(J,\chi') 
	 =\tfrac{[\Theta(n^{-\alpha}),\Theta(n^{\alpha})]}{\den_{\kappa'}(I,\chi')_{c_i}} \cdot \tfrac{\totalmu}{\lambda^t}\cdot \den_{\kappa'}(A,\chi',\SC)_{c_{i,j}} 
         $$

	and by a similar argument, we have that if $c_i  \in C_{I,\chi'}^{\text{Good}}$, then $m_{i+1,j}
	 	 \in \tfrac{[\Theta(n^{-\alpha}),\Theta(n^{\alpha})]}{\den_{\kappa'}(I,\chi')_{c_i}} \cdot \tfrac{\totalmu}{\lambda^t}\cdot \den_{\kappa'}(A,\chi',\SC)_{c_{i+1,j}} 
         $ as well.
	
	For (1), we have by triangle inequality that $\den_{\kappa'}(A,\chi',\SC)_{c_{i,j}} \leq \den_{\kappa'}(I,\chi',\SC)_{4c_i} \leq n^\alpha \cdot \den_{\kappa'}(I,\chi',\SC)_{c_i}$ and hence $m_{i,j} = \tfrac{O(n^{2\alpha})}{\relden_{\kappa'}(I,\chi')_{c_i}} \cdot \tfrac{\totalmu}{\lambda^t}$.
	
	For (2), we have by triangle inequality that $\den_{\kappa'}(A,\chi',\SC)_{c_{i+1,j}} \geq \den_{\kappa'}(I,\chi',\SC)_{c_i}$ and hence $m_{i+1,j} = \tfrac{\Omega(n^{-\alpha})}{\relden_{\kappa'}(I,\chi')_{c_i}} \cdot \tfrac{\totalmu}{\lambda^t}$.

	 Summing over all $j$, we obtain the required bounds.

\end{proof}

The next claim helps control $\varphi_\kappa^{\zeta}(I,*_{\neq \un})=T^\cc\left(\gamma\cdot
\tfrac{\phi_\kappa(\Lambda^{\zeta}_\kappa(I),*)}{\beta^l}\right)$
for some parameters $\gamma,\zeta$.
In particular, it shows that if we have a set of pairs $\CD$ with each $(I,\chi') \in \CD$
of the ``right'' density
within some ball $\QC$, then we add sufficient $\phi$ mass to $\QC$,
which will survive the $T^\cc$ thresholding whp.

\begin{claim}[Controlling $\varphi(I,*_{\neq \un})$]\label{clm::dense_set_imp}
	Fix level $l\ge 0$ and parameters $\delta, \gamma, s
        \in[\Omega_\eps(1),1]$
          as well as a set
        $\QC\subseteq \IC$. Consider the (interval, color) set $\CD
        \subseteq \QC \times \nu \setminus \{ \bot \}$ of pairs
        $(I,\chi')$ for which there exists a set of $\tfrac{s}{\eta}$ distinct
        costs $C^*_{I,\chi'} \subseteq E_\basec$, where for each $\hat{c} \in C^*_{I,\chi'}$, we have
        $\relden_{\kappa'}(I,\chi',\QC)_{\hat{c}} \in n^{3\alpha}\cdot\tfrac{\totalmu}{\beta^l \lambda^t} \cdot [1,\beta^2 +
          \1[l=0]\cdot n]$.
        If $\mu_{\kappa'}(\CD) \geq \delta \cdot \beta^l$, then, with high probability, we have at step $l$:
    \begin{enumerate}
    	\item	If $l > 0$, then
    		$\normo{T^\cc\left(\gamma \cdot \tfrac{\phi_\kappa(\QC,*)}{\beta^{l}}\right)}  \geq
        0.9 \delta \gamma \cdot (s - O(\tfrac{\eta}{\alpha})).$
      \item	If $l = 0$, then,
    		$T^\cc\left(\gamma \cdot \phi_\kappa(\QC,\nu)\right)  \geq
        0.9 \delta \gamma \cdot (s - O(\tfrac{\eta}{\alpha})).$
    \end{enumerate}
\end{claim}

To prove Claim~\ref{clm::dense_set_imp}, we cannot apply Claim~\ref{clm::rel_dense_bound_imp} directly, as the lower bound we get from there is for a different samples, and hence we show first the following useful statistical fact.

\newcommand{\bA}[0]{{\mathbf{A}}}
\newcommand{\bB}[0]{{\mathbf{B}}}
\newcommand{\bX}[0]{{\mathbf{X}}}

\begin{fact}\label{fct::stat_min}
	Let $\bX$ be a non-negative random variable over some probability space and let %
	$\bA$ and $\bB$ be events such that $\Pr[\bA]=\Pr[\bB]$. Fix
        $\gamma \in \R$ such that $\bX \mid \bB \geq \gamma$ (i.e.,
        whenever $\bB$ happens, $\bX\ge \gamma$). %
	Then, 
	$\E[\bX \cdot \1[\bX \geq \gamma]] \geq \E[\bX \cdot \1[\bA]]$.
\end{fact}
\begin{proof}
Since $\Pr[\bA] = \Pr[\bB]$, then also $\Pr[\bA \setminus \bB] = \Pr[\bB \setminus \bA]$. Hence,
	\begin{align*}
	\E[\bX \cdot \1[\bA]] &= \E[\bX \cdot \1[\bA \cap \bB]] + \E[\bX \cdot \1[\bA \setminus \bB]] \\
	&=  \E[\bX \cdot \1[\bA \cap \bB]] + \E[\bX \cdot \1[\bA \setminus \bB] \cdot \1[\bX \geq \gamma]] + \E[\bX \cdot \1[\bA \setminus \bB] \cdot \1[\bX < \gamma]]	\\
	&<  \E[\bX \cdot \1[\bA \cap \bB]] + \E[\bX \cdot \1[\bA \setminus \bB] \cdot \1[\bX \geq \gamma]] + \gamma \cdot \Pr[\bA \setminus \bB]	\\
	&=  \E[\bX \cdot \1[\bA \cap \bB]] + \E[\bX \cdot \1[\bA \setminus \bB] \cdot \1[\bX \geq \gamma]] + \gamma \cdot \Pr[\bB \setminus \bA] \\
	&\leq  \E[\bX \cdot \1[\bA \cap \bB]] + \E[\bX \cdot \1[\bA \setminus \bB] \cdot \1[\bX \geq \gamma]] + \E[\bX \cdot \1[\bB \setminus \bA] ]			\\
	&\leq \E[\bX \cdot \1[\bX \geq \gamma]]
	\end{align*}

\end{proof}

\begin{proof}[Proof of Claim~\ref{clm::dense_set_imp}]
	Note that every output color $\chi$ in
        output coloring $\kappa$ is defined via the iid sampled triplet
        $(A,\chi'',\hat{c}) \sim \IC \times \nu %
        \times E_\basec$. 
		
		Fix $\Delta = \log^2 n$. We define the random variables $\psi_{\chi} \triangleq \min\{\phi_\kappa(\QC,\chi)/\beta^l,1/\Delta\}$ and  %
    $$
    M_\chi \triangleq \begin{cases}
    	T^\cc\left(\gamma \cdot \psi_{\chi} \right)  &	l > 0 \\
    	\gamma \cdot \psi_{\chi}  &	l=0.
 \end{cases}
$$

Note that for $l>0$,
        $\normo{T^\cc\left(\gamma\cdot\tfrac{\phi(\QC,*)}{\beta^{l}}\right)}
        \geq \sum_\chi M_\chi$, and for $l=0$, we have $T^\cc\left(x\right)=x$ for any $x = \Omega_\epsilon(1)$
        and hence it suffices to show whp $\sum_\chi M_\chi \geq 0.9 \delta \gamma \cdot (s - O(\tfrac{\eta}{\alpha}))$.
        
        Now since we sample $\lambda^t$ iid triplets $(A,\chi'',\hat c)$, the quantity $\sum_\chi M_\chi$
		 is a sum of $\lambda^t$ iid r.v. 
             bounded by
                 $\tfrac{1}{\log^2 n}$. Therefore, by
                 Chernoff bound, we obtain w.h.p.: $\sum_\chi M_\chi \geq 0.9\E[\sum_\chi M_\chi] - O_\epsilon\left(\tfrac{1}{\log n}\right) = 0.9\lambda^t \cdot \E[M_\chi] - O_\epsilon\left(\tfrac{1}{\log n}\right)$. Therefore, it suffices to show $\E[M_\chi] \geq \delta \gamma \cdot (s - O(\tfrac{\eta}{\alpha})) \cdot \lambda^{-t}$.

	To do that, we introduce another random variable, $m_\chi$, and thereafter show the following inequalities:
	\begin{itemize}
		\item $\E[m_\chi] \geq \delta \cdot \beta^l \cdot (s - O(\tfrac{\eta}{\alpha})) \cdot  \lambda^{-t}$; and,
		\item $\E[M_\chi] \geq \tfrac{\gamma}{\beta^l} \E[m_\chi]$.

	\end{itemize}
	
	\paragraph{The random variable $m_\chi$.} For each $(I,\chi') \in \CD$, let $C_{I,\chi'}^{\text{Bad}}$ be union of ``bad costs'' from Claim~\ref{cl::phi_uncorrupted} (3) and Claim~\ref{clm::rel_dense_bound_imp} of size $|C_{I,\chi'}^{\text{Bad}}| = O(1/\alpha)$.
	For a fixed output color $\chi$, define the random
        variable $z_{I,\chi'}$ as the contribution of
        $(I,\chi')$ to $\phi_\kappa(I,\chi)$ just from costs in
        $C^*_{I,\chi'} \setminus
        C_{I,\chi'}^{\text{Bad}}$, and let $m_\chi := \sum_{(I,\chi') \in
          \CD}z_{I,\chi'}$. 
    \aanote{should we then say ``contribution to $\phi_\kappa(\CD,\chi)$''
      (or change the summation)?}\ns{fixed as discussed....}
    \vspace{2mm}
    
    \paragraph{The inequality $\E[m_\chi] \geq \delta \cdot \beta^l\cdot (s - O(\tfrac{\eta}{\alpha})) \cdot \lambda^{-t}$.}
      Here we use Claim~\ref{cl::phi_uncorrupted} (3), which guarantees, for each $(I,\chi') \in \CD$:
	$\E_{(A,\chi'',\hat{c}) \sim \IC \times \nu \times
          E_\basec}[\phi_\kappa(I,\chi) \cdot \1[\chi' = \chi''] \mid
          \hat{c} \in E_\basec \setminus C_{I,\chi'}^{\text{Bad}}] \geq
        \mu_{\kappa'}(I,\chi') \cdot \lambda^{-t}$,
        and hence, 
	
$$	
		\E_{(A,\chi'',\hat{c})%
		}\left[m_\chi \right] 
 = \E_{(A,\chi'',\hat{c})%
 }\left[\sum_{(I,\chi') \in \CD}z_{I,\chi'}\right] 
	= \sum_{(I,\chi') \in \CD} \E_{(A,\chi'',\hat{c})%
	}[z_{I,\chi'}] \geq \sum_{(I,\chi') \in \CD} \Pr_{\hat c\in E_\basec}[\hat c\in C^*_{I,\chi'} \setminus C_{I,\chi'}^{\text{Bad}}] \cdot \mu_{\kappa'}(I,\chi') \cdot  \lambda^{-t}
	$$
        since $m_\chi$ only counts contribution for $\hat c\in
        C^*_{I,\chi'} \setminus C_{I,\chi'}^{\text{Bad}}$. Using the fact that $\Pr_{\hat c\in E_\basec}[\hat c\in C^*_{I,\chi'} \setminus C_{I,\chi'}^{\text{Bad}}]\ge (s-O(\eta/\alpha))$ and that $\mu_{\kappa'}(\CD) = \delta \cdot \beta^l$, we obtain,
	
	$$
		\E_{(A,\chi'',\hat{c}) \sim \IC \times \nu \times E_\basec}\left[m_\chi \right] 
	\ge \mu_{\kappa'}(\CD)\cdot (s - O(\tfrac{\eta}{\alpha}))\cdot \lambda^{-t}  
	= \delta \cdot \beta^l \cdot (s - O(\tfrac{\eta}{\alpha})) \cdot \lambda^{-t}.
	$$
	as needed.
	
	\vspace{2mm}
	
   \paragraph{The inequality $\E[M_\chi] \geq \tfrac{\gamma}{\beta^l} \E[m_\chi]$.}
                 Define the event
    $$
    \bA_\CD \triangleq \left\{\text{for sampled }(A,\chi',c_i)\text{
      there exists some }  (I,\chi') \in \CD \text{ s.t. } \cad(I,A) \leq 2c_i \wedge c_i \in C^*_{I,\chi'} \setminus
        C_{I,\chi'}^{\text{Bad}}\right\}.
    $$
 
	Note that on one hand, conditioned on {\em not} $\bA_\CD$, we
        have $m_\chi= 0$
        and hence  $m_\chi = m_\chi \cdot \1[\bA_\CD]$. On the other hand, we have $\relden_{\kappa'}(I,\chi',\QC)_{\hat{c}} \geq n^{3\alpha}\cdot\tfrac{\totalmu}{\beta^l \lambda^t}$ %
        for any
        $(I,\chi') \in \CD$ and $\hat{c} \in C^*_{I,\chi'}$, and
        therefore, whenever
        $\bA_\CD$ happens, by Claim~\ref{clm::rel_dense_bound_imp},
  $\phi_\kappa(\QC,\chi) %
        \leq O(\beta^l/n^{\alpha}) < \beta^l/\Delta$.
        This implies that $\phi_\kappa(\QC,\chi)\cdot \1[\bA_\CD] = \psi_\chi \cdot \1[\bA_\CD] \cdot \beta^l$, and overall:
        
        $$
        \psi_\chi \cdot \1[\bA_\CD] \cdot \beta^l = \phi_\kappa(\QC,\chi)\cdot \1[\bA_\CD] \geq m_\chi \cdot \1[\bA_\CD] = m_\chi
        $$
	
    For $l=0$, %
    this immediately gives 
    $M_\chi = \gamma \psi_\chi \geq \gamma \psi_\chi \cdot \1[\bA_\CD] \geq \tfrac{\gamma}{\beta^l} m_\chi$
     which implies $\E[M_\chi] \geq \tfrac{\gamma}{\beta^l} \E[m_\chi]$ as needed.
    
    For $l>0$, we need to use the upper bound guarantee on the density to ``deal with $T^\cc$ thresholding''. %
    For this we first define the event
    
        $$
    \bB_\CD \triangleq \left\{\text{for sampled
    }(A,\chi',c_{i+1})\text{ there exists some } (I,\chi') \in
    \CD \text{ s.t. } \cad(I,A) \leq 2c_i \wedge c_i \in C^*_{I,\chi'} \setminus
        C_{I,\chi'}^{\text{Bad}}\right\}
    $$   
    
    and notice that $\Pr[\bB_\CD] = \Pr[\bA_\CD]$. Now, since we are guaranteed for any
        $(I,\chi') \in \CD$ and $\hat{c} \in C^*_{I,\chi'}$
that $\relden_{\kappa'}(I,\chi',\QC)_{\hat{c}} \leq n^{3\alpha}\cdot
\beta^2 \cdot \tfrac{\totalmu}{\beta^l \lambda^t}$, then by
Claim~\ref{clm::rel_dense_bound_imp} (2), we have that, whenever 
$\bB_\CD$ holds, $\psi_\chi \geq \tfrac{\Omega(1)}{\beta^2
  n^{5\alpha}} > \beta^{-3}/\gamma$. We now invoke Fact~\ref{fct::stat_min} to obtain $\E[\psi_\chi \cdot \1[\psi_\chi \geq \beta^{-3}/\gamma]] \geq \E[\psi_\chi \cdot \1[\bA_\CD]]$. Since $T^\cc(x) = x$ whenever $x \geq \beta^{-3}$, we also have that $M_\chi \geq \gamma \psi_\chi \cdot \1[\gamma \psi_\chi \geq \beta^{-3}]$. Overall, we obtain:
$$
\E[M_\chi] \geq \E[ \gamma \psi_\chi \cdot \1[\gamma \psi_\chi \geq \beta^{-3}]] \geq \gamma \cdot \E[\psi_\chi \cdot \1[\bA_\CD]] \geq \tfrac{\gamma}{\beta^l} m_\chi
$$

for $l>0$ as well. 

By combining both inequalities, we conclude:
$$\E[M_\chi] \geq \tfrac{\gamma}{\beta^l} \E[m_\chi] \geq \gamma \delta \cdot (s - O(\tfrac{\eta}{\alpha})) \cdot \lambda^{-t}$$

which implies the claim with high probability as needed.

\end{proof}

Next, in order to analyze $\varphi(\cdot,\un)$, we show an auxiliary claim that bounds $\theta_\kappa^{\zeta}(\Lambda_\kappa^{\zeta}(J),\nu
\setminus \{ \bot \})$. This claim will be used for complexity analysis as well.
Consider some interval ball $\QC$ at level
$l$ containing $\zeta$ $\un$-mass in the current coloring $\kappa$. Intuitively, this implies there is a mass of sparse pairs
proportional to $\zeta$ from previous levels. The claim below bounds
$\theta_\kappa^{\zeta}(\Lambda_\kappa^{\zeta}(J),\nu \setminus \{ \bot
\})$, for all $J \in \QC$ generated at level $l$, based on the mass of
sparse pairs around $\QC$ for the current level.

\begin{claim}
  \label{cl::Q_U_bound}
	Fix a level $l\ge0$ and $\delta \geq \beta^{-1.5}$. Fix $\zeta \in Z_l$, and let $s_l = 1-\tfrac{(l+1) \cdot \epsilon}{2}$.
	Consider an %
	interval ball
	$\QC$ %
which, for level $l=0$, contains precisely one interval, and if $l>0$ satisfies
        $\mu_{\kappa}(\QC,\un) \leq \zeta$. Let %
        $\QC^{+2\zeta} =
        \Lambda_\kappa^{2\zeta}(\QC)$. Let $\CD \subseteq \QC \times
        \nu \setminus \{ \bot \}$ be the set of all $(I,\chi') \in
        \QC \times
        \nu \setminus \{ \bot \}$ such that
        $\relden_{\kappa'}(I,\chi',\QC)_{c} \leq
        n^{3\alpha}\tfrac{\totalmu}{\beta^l \cdot \lambda^t}$ on
        at least $\tfrac{s_l}{\eta}$ costs $c\in E_\basec$, and let $\CD^+\subseteq \QC^{+2\zeta} \times \nu \setminus \{ \bot \}$ be the set of $(I,\chi') \in \QC^{+2\zeta} \times \nu \setminus \{ \bot \}$ where
	 $\relden_{\kappa'}(I,\chi',\QC^{+2\zeta})_{c} \leq
        n^{3\alpha}\cdot \beta\cdot\tfrac{\totalmu }{\beta^l \cdot
          \lambda^t}$ on at least $\tfrac{s_l -
          \epsilon/4}{\eta}$ distinct costs. Finally, let $\zeta_+ = \max\{\zeta,1\}$.
        \aanote{for $l=0$, we have $\zeta=0$ so not well defined}\ns{ok now?}

	 \begin{enumerate}
	 	\item If $\mu_{\kappa'}(\CD^{+}) \leq \delta \zeta_+$,
                  then
                  $\theta_\kappa^{\zeta}(\Lambda_\kappa^{\zeta}(J),\nu
                  \setminus \{ \bot \}) < 1.1\delta k \cdot
                  \tfrac{\zeta_+}{\beta^l}$ for each $J \in \QC$ whp.
	 	\item If $\mu_{\kappa'}(\CD) \geq \delta \zeta_+$, then
                  $\theta_\kappa^{\zeta}(\Lambda_\kappa^{\zeta}(J),\nu
                  \setminus \{ \bot \}) \geq 0.9\delta k\cdot
                  \tfrac{\zeta_+}{\beta^l}$ for each $J \in \QC$ whp.
	 \end{enumerate}
	
	        \end{claim}

	\begin{proof}
	Let $\QC^{+\zeta} =
        \Lambda_\kappa^\zeta(\QC)$.
Let $m_\QC = \theta_\kappa^{\zeta}(\QC,\nu \setminus \{ \bot \})$ and
$m_{\QC^{+\zeta}} = \theta_\kappa^{\zeta}(\QC^{+\zeta},\nu
\setminus \{ \bot \})$. Note that $m_\QC \leq
\theta_\kappa^{\zeta}(\Lambda_\kappa^{\zeta}(J),\nu \setminus \{ \bot
\}) \leq m_{\QC^{+\zeta}}$ for each $J \in \QC$, hence it suffices to
bound $m_\QC, m_{\QC^{+\zeta}}$. To compute $m_\QC,m_{\QC^{+\zeta}}$, we first compute their expectations
$\E[m_\QC], \E[m_{\QC^{+\zeta}}]$, and then use Chernoff bound to
conclude that these variables are 
roughly bounded by their expectation whp.

	We first upper-bound $\E[m_{\QC^{+\zeta}}]$.
	For level $l=0$, we have $\E[m_{\QC^{+\zeta}}] = \E[m_{\QC}] = \tfrac{k}{\beta^l} \cdot\left(\sum_{(J,\chi') \in
      \QC \times \nu \setminus \{\bot \}}
    \mu_{\kappa'}(J,\chi')\sigma_{J,\chi'}\right)$.
For levels $l>0$, we use Corollary~\ref{cor::theta_threshold} to obtain 
	$$\E[m_{\QC^{+\zeta}}] < \left(1 + o(1)\right) \cdot \tfrac{k}{\beta^l} \cdot \sum_{(I,\chi') \in \QC^{+2\zeta} \times \nu \setminus \{ \bot \}} \mu_{\kappa'}(I,\chi') \cdot \sigma_{I,\chi'}. $$

	Consider a pair $(I,\chi') \in \QC^{+2\zeta} \times \nu \setminus \{ \bot \}$. If for a cost $c$ we have
         $\relden_{\kappa'}(I,\chi',\QC^{+2\zeta})_{c} > n^{3\alpha} \cdot
         \beta\cdot \tfrac{\totalmu}{\beta^l \cdot \lambda^t}$, then $\Gamma_{l,\zeta}(I,\chi')_c = O(n^\alpha/\beta) < 1/\sqrt{\beta}$, and hence 
         $T^\theta(\Gamma_{l,\zeta}(I,\chi')) = 0$. If the above occurs on
         $\ge\tfrac{1-(s_l-\epsilon/4)}{\eta}$ distinct costs $c$ ---
         i.e., for $I\not \in \CD^+$ --- then
         $%
         \sigma_{I,\chi'}=Q_l \left( T^{\theta}(\Gamma_{l,\zeta}(I,\chi'))\right)=0$
         (based on $Q_l$ transformation). Hence, $\E[m_{\QC^{+\zeta}}]
         \le \mu_{\kappa'}(\CD^+)\cdot \tfrac{k}{\beta^l}\cdot \left(1 + o(1)\right)\leq \delta
         \zeta_+ \cdot \tfrac{k}{\beta^l} \cdot \left(1 + o(1)\right)$ for part 1.
	 
	 Similarly, to lower-bound $\E[m_\QC]$, we observe that for
         any $(I,\chi') \in \QC \times \nu \setminus \{ \bot \}$ with
         density $\relden_{\kappa'}(I,\chi',\QC)_{c} \leq n^{3\alpha}\cdot \tfrac{\totalmu
           }{\beta^l \cdot \lambda^t}$ for a cost
         $c$, we have $\Gamma_{l,\zeta}(I,\chi')_c \ge \Omega(n^\alpha) > 1$ whp. Hence
         $T^\theta(\Gamma_{l,\zeta}(I,\chi')_c) = 1$, and for
         $(I,\chi')$ where the above occurs on
         $\ge\tfrac{s_l}{\eta}$ distinct costs, we have $\sigma_{I,\chi'}=Q_l(T^\theta(\Gamma_{l,\zeta}(I,\chi')))=1$. We conclude $\E[m_\QC]
         \geq \tfrac{k}{\beta^l}\cdot \mu_{\kappa'}(\CD)\ge \delta \zeta_+
         \cdot \tfrac{k}{\beta^l}$  for part 2.
         
Finally, note that $m_\QC,m_{\QC^{+\zeta}}$ are each sum of
independent r.v. in $[0,1]$, and since $\delta \geq \beta^{-1.5} \gg \tfrac{\log n}{k}$, then one can apply Chernoff bound, and obtain that $m_{\QC^{+\zeta}} \leq 1.1 \E[m_{\QC^{+\zeta}}]$ and $m_\QC \geq 0.9 \E[m_\QC]$ with high probability. 
\end{proof}

Finally, the next claim helps control $\varphi(\cdot, \un)$. In particular, it shows that if
$\theta_\kappa^{\beta^l}(\Lambda_\kappa^{\beta^l}(I),\nu \setminus \{
\bot \})$ is too small to ensure that we $\un$-color such interval
with sufficient mass, then there must be a sufficiently large mass of
pairs in $I$'s proximity of the ``right density''. The latter property
will instead ensure that $I$ is colored with non-$\un$ color, as we
will show later.

	\begin{claim}[Small $\theta$ implies many
                            pairs in $I$'s proximity have right
                            density]\label{cl::density_range_set} Fix
          input coloring $\kappa'$, output coloring $\kappa$ at a
          level $l>0$, and let $\kappa^{-1}$ be the coloring $\kappa$
          at the previous level $l-1$. Define $b =
          \min_{I' \in \IC} \mu_{\kappa^{-1}}(I',\nu \setminus \{\un\})
          + \sum_{\zeta \in Z_{l-1}}
          \varphi_{\kappa^{-1}}^\zeta(I',\nu)$ (ie, the minimal
          un-normalized mass of new coloring at level $l-1$); and assume
          $b = \Omega_\eps(1)$.
          Consider any $I \in \IC$ with
          $\theta_\kappa^{\beta^l}(\Lambda_\kappa^{\beta^l}(I),\nu
          \setminus \{\bot\}) < \tfrac{b}{50}\cdot k$. Then, with high probability, there
          exists a set of interval--color pairs $\CD \subseteq
          \Lambda_\kappa^{\beta^l}(I) \times \nu \setminus \{ \bot \}$
          with $\mu_{\kappa'}(\CD) \geq \tfrac{b}{25} \cdot \beta^l$
          such that every $(J,\chi') \in \CD$ pair has relative
          density
          $\relden_{\kappa'}(J,\chi',\Lambda_\kappa^{\beta^l}(J))_c \in
          [1,\beta^2] \cdot n^{3\alpha}\cdot\tfrac{\totalmu}{\beta^l
            \cdot \lambda^t}$ on $\ge\tfrac{\epsilon/4}{\eta}$ distinct
          costs.
	\end{claim}

	\begin{proof}

     Let $\BC_I = \Lambda_\kappa^{\beta^l}(I)$ and hence $\mu_{\kappa}(\BC_I,\un) \geq \beta^l - 1$. 
	  For each $l$, we denote $Z_l^* = Z_l \cap [0,\beta^l \cdot
            \beta^{0.9}]$.

          Since $\mu_{\kappa}(J,\un)$
                is monotonically decreasing with level $l$
                (\textsc{AmendColoring} can only decrease it), there
                exists a set of intervals $\QC \subseteq \BC_I$ with
                $\mu_{\kappa}(\QC,\un) \geq \beta^l (1 - o(1))$, such
                that $\Lambda_{\kappa^{-1}}^{2\zeta}(\QC) \subseteq
                \BC_I$ for any $\zeta\le \beta^{l-0.1}$ (and hence $\zeta \in Z_{l-1}^*$). Also observe
                that for $\zeta \in Z_{l-1} \setminus Z_{l-1}^*$, we
                always have $\varphi_{\kappa^{-1}}^\zeta(J,\un) \leq
                T^\un(\tfrac{\beta^{2(l-1)}}{\zeta^2}) \leq
                T^\un(\beta^{-1.8}) = 0$.

		Recall $\mu$ recoloring procedure (\textsc{AmendColoring} line~\ref{alg::color_line}). Letting $m_{J} = \mu_{\kappa^{-1}}(J,\un)$ and
          $U^\zeta(J) = \tfrac{1}{k} \cdot
          \theta_{\kappa^{-1}}^{\zeta}(\Lambda_{\kappa^{-1}}^{\zeta}(J),\nu
          \setminus \{ \bot \})$, we have:
		
		$$\beta^l (1 - o(1)) \leq \mu_{\kappa}(\QC,\un) \leq \tfrac{1}{b}\sum_{J \in \QC} m_{J} \cdot \sum_{\zeta \in Z_{l-1}} \varphi_{\kappa^{-1}}^\zeta(J,\un) \leq \tfrac{1}{b} \sum_{J \in \QC} m_{J} \cdot \sum_{\zeta \in Z_{l-1}^*}  T^\un(m_{J} \cdot \tfrac{\beta^{2(l-1)}}{\zeta^2} \cdot U^\zeta(J)).$$
		\aanote{should the last be an inequality as it is a
                  min with 1?}\ns{yes. fixed}
                
		Now since $m_{J}, \tfrac{\beta^{l-1}}{\zeta} \leq 1$,
                the RHS can be further upper bounded by $\tfrac{1}{b}
                \sum_{J \in \QC}m_{J} \cdot \sum_{\zeta \in Z_{l-1}^*}
                \tfrac{\beta^{l-1}}{\zeta} \cdot
                T^\un(\tfrac{\beta^{l-1}}{\zeta}\cdot
                U^\zeta(J))$. Now for each radius $\zeta \in Z_{l-1}$ from level
                $l-1$, we decompose $\QC$ into disjoint consecutive interval balls $\QC^\zeta_i$ as follows: for level $l=1$, each $\QC^0_i$ represent the next single interval of $\QC$ starting from the left. For $l>1$, each ball $\QC^\zeta_i$ is a maximal ball 
               containing $\un$-color of mass $\leq \zeta$ at the start of
                level $l-1$ (i.e. $\mu_{\kappa^{-1}}(\QC^\zeta_i,\un)
               \leq \zeta$ for all $i$), starting at the first
               interval on the right of $\QC^\zeta_{i-1}$ \aanote{not
                 on the left? in general, on
                 the left and right is not clear...}\ns{feel free to change! I thought $x_0$ is the first character in the string and then we go the the right no? but as you prefer...}.  Also, let $W$
               be the set of indices $i$ of the obtained balls, and
               $W'$ the set of indices $i$ excluding the two smallest
               and two largest ones. 

               By switching summation we get:
		
		$$b \cdot \beta^l (1 - o(1)) \leq \sum_{\zeta \in
                  Z_{l-1}^*} \tfrac{\beta^{l-1}}{\zeta} \sum_{J \in
                  \QC} m_{J} \cdot T^\un(\tfrac{\beta^{l-1}}{\zeta}
                \cdot U^\zeta(J)) = \sum_{\zeta \in Z_{l-1}^*}
                \tfrac{\beta^{l-1}}{\zeta} \sum_{i \in W} \sum_{J \in
                  \QC^\zeta_i}m_{J} \cdot
                T^\un(\tfrac{\beta^{l-1}}{\zeta} \cdot U^\zeta(J)).$$
\aanote{the last one should be an equality?}\ns{yes. changed}

		Notice that since $U^\zeta(J) \leq 1$, then for a single $i$, we have $\sum_{J \in
                  \QC^\zeta_i}m_{J} \cdot
                T^\un(\tfrac{\beta^{l-1}}{\zeta} \cdot U^\zeta(J))
                \leq \beta^{l-1} = o(\tfrac{b}{8} \cdot \beta^l)$,
                which implies that we can replace $W$ with $W'$ in
                above: i.e., $b \cdot \beta^l (1 - o(1))\leq \sum_{\zeta \in Z_{l-1}^*}
                \tfrac{\beta^{l-1}}{\zeta} \sum_{i \in W'} \sum_{J \in
                  \QC^\zeta_i}m_{J} \cdot
                T^\un(\tfrac{\beta^{l-1}}{\zeta} \cdot U^\zeta(J))$ as
                well. Now, let $q^\zeta_i$ be the $\mu_{\kappa'}$-mass
                of  interval--color pairs $(J,\chi')\in \QC^\zeta_i \times \nu$ that are relatively sparse: where
                $\relden_{\kappa'}(J,\chi',\BC_I)_c \leq
                n^{3\alpha}\beta\cdot\tfrac{\totalmu }{\beta^{l-1}
                  \cdot \lambda^t}$ on $\ge \tfrac{s_{l-1} -
          \epsilon/4}{\eta}$ distinct costs $c$. Also let
                $q^*$ be the total $\mu_{\kappa'}$-mass of relatively sparse
                pairs in $\BC_I$ (using the same definition as
                above). By the definition of sets
                $\{\QC_i^\zeta\}_i$, we have $q^* \geq
                \sum_i q^\zeta_i$. We now invoke Claim
                \ref{cl::Q_U_bound} (1) with $\QC=\QC_i^{\zeta}$ to
                obtain that, for any $J\in \QC_i^{\zeta}$, we have $U^\zeta(J)<1.1\cdot
                \tfrac{\max\{\sum_{j \in [-2,2]} q^\zeta_{i+j},\beta^{-1.5}\zeta\}}{\beta^{l-1}}$ (noting $\Lambda_\kappa^{2\zeta}(\QC_i^{\zeta}) \subseteq \cup_{j \in [-2,2]} \QC_{i+j}^\zeta$). Hence
                $T^\un(\tfrac{\beta^{l-1}}{\zeta} \cdot U^\zeta(J)) \le
                \tfrac{1.1}{\zeta}\sum_{j \in [-2,2]} q^\zeta_{i+j}$ (noting that $T^\un(x) = 0$
                for $x < O(\beta^{-1.5})$). Also, since we defined
                $\QC^\zeta_i$ as containing a-priori $\leq \zeta$ uncolored
                mass in level $l-1$, then the quantity $\sum_{J \in
                  \QC^\zeta_i}m_{J} \cdot
                T^\un(\tfrac{\beta^{l-1}}{\zeta} \cdot U^\zeta(J))$
                can be upper bounded by $1.1\sum_{j \in [-2,2]} q^\zeta_{i+j}$.             %
           Overall, we obtain:
		
		$$b \cdot \beta^l (1 - o(1)) \leq 1.1\sum_{\zeta \in Z_{l-1}^*} \tfrac{\beta^{l-1}}{\zeta} \sum_{i \in W'} \sum_{j \in [-2,2]} q^\zeta_{i+j} \leq 6 \sum_{\zeta \in Z_{l-1}^*} \tfrac{\beta^{l-1}}{\zeta} \sum_{i \in W} q^\zeta_i \leq 6 q^* \cdot \sum_{\zeta \in Z_{l-1}^*} \tfrac{\beta^{l-1}}{\zeta}  \leq 12q^*.$$

	Hence we get that $q^* > \tfrac{b}{15} \cdot \beta^l$, meaning
        that there exists a set $\CD^* \subseteq \BC_I \times \nu \setminus \{ \bot \}$ with $\mu_{\kappa'}(\CD^*) \geq \tfrac{b}{15} \cdot \beta^l$, such that for all $(J,\chi') \in \CD^*$, we have $\relden_{\kappa'}(J,\chi',\BC_I)_c \leq
                n^{3\alpha}\beta^2 \cdot\tfrac{\totalmu }{\beta^{l}
                  \cdot \lambda^t}$ on $\ge \tfrac{s_{l-1} -
                \epsilon/4}{\eta}$ distinct costs $c$.
	On the other hand, let $U^{\beta^l}_\kappa(I) := \tfrac{1}{k}\cdot
        \theta_\kappa^{\beta^l}(\Lambda_\kappa^{\beta^l}(I),\nu
        \setminus \{\bot\})$. From the claim assumption, we have $U_\kappa^{\beta^l}(I)
        \leq \tfrac{b}{50}$, and by applying (the contrapositive of) Claim~\ref{cl::Q_U_bound}
        (2) for level $l$, we have that w.h.p, the
        $\mu_{\kappa'}$-mass in $\BC_I$ with relative sparse intervals with respect to
        level $l$ (i.e. $(J,\chi') \in B_I \times \nu \setminus \{
        \bot \}$ pairs, where $\relden_{\kappa'}(J,\chi',\BC_I)_c \le
        n^{3\alpha}\cdot\tfrac{\totalmu}{\beta^l \cdot \lambda^t}$
        on $\ge\tfrac{s_l}{\eta}$ distinct costs) is at most
        $\tfrac{b}{40} \cdot \beta^l$.
        
	The two statements can be combined to conclude that $\BC_I
        \times \nu \setminus \{ \bot \}$ contains a set of $\mu_{\kappa'}$-mass of
        at least $\left(\tfrac{b}{15} - \tfrac{b}{40}\right)\beta^l >
        \tfrac{b}{25} \cdot \beta^l$ of pairs $(J,\chi')$ with relative density
        $\relden_{\kappa'}(J,\chi',\BC_I)_c \in [1,\beta^2] \cdot
        n^{3\alpha}\cdot \tfrac{\totalmu}{\beta^l \cdot \lambda^t}$
        on $\tfrac{s_{l-1} -
          \epsilon/4}{\eta} - \tfrac{s_l}{\eta} = \tfrac{\epsilon/4}{\eta}$
        distinct costs as needed. 
		\end{proof}
		
	Finally, we prove Lemma~\ref{lm::min_color_mass}.	

	\begin{proof}[Proof of Lemma \ref{lm::min_color_mass}]
          We prove by induction on the level $l$.
          Fix coloring $\kappa$ and level $l$, and let
          $\kappa^{-1}$ be coloring $\kappa$ from the
          previous level, $l-1$. Consider each $I \in \IC$. Let $M_I = \mu_\kappa(I,\nu \setminus \{ \un \})+ \sum_{\zeta \in Z_l}\varphi_\kappa^\zeta(I,\nu)$
			and let $m_I = \mu_\kappa(I,\un)$. 
			Define
			$$b_l = \begin{cases}
			1	& l=0 \\
			\min_{I \in \IC} \mu_{\kappa^{-1}}(I,\nu \setminus \{ \un \}) + \sum_{\zeta \in Z_{l-1}}\varphi_{\kappa^{-1}}^\zeta(I,\nu)
			& \text{Otherwise}
			\end{cases}$$
                        Our induction hypothesis is
                        $b_l=\Omega(\eps)^l$. Since $b_{l+1} =
                        \min_{I \in \IC} M_I$, it suffices to prove
                        that $M_I = b_l \cdot \Omega(\epsilon)$ (whp)
                        to complete the induction. First, if $m_I \leq
                        0.5$, then $M_I \geq 0.5$ (noting
                        $\mu_\kappa(I,\nu) = 1$) and the claim is
                        immediate so we assume $m_I > 0.5$.

\vspace{2mm}

\paragraph{Base case: $l=0$.} If there exists a set of colors
$\nu^* \subseteq \nu$, s.t. $\mu_{\kappa'}(I, \nu^*) \geq
\tfrac{1}{4}$ where, for each $\chi' \in \nu^*$, we have
\aanote{shouldn't this be relative distance, in particular normalized
  by $\mu_\kappa$. and also shouldn't this be $\kappa'$ instead too?}\ns{I think you are right. Changed.}
$\relden_{\kappa'}(I,\chi')_c > n^{3\alpha}\tfrac{\totalmu}{\beta^l \cdot
  \lambda^t}$ on $\tfrac{\epsilon/4}{\eta}$ distinct costs, then we can
use Claim~\ref{clm::dense_set_imp} and obtain that, whp, 
\aanote{the second should be an equality?}\ns{yes. changed.}
$$M_I \geq \varphi^{0}_\kappa(I,\bot) = T^\cc(m_I \cdot
\phi_\kappa(I,\nu\setminus\{\un,\bot\})) \geq T^\cc(\tfrac{1}{2} \cdot \phi_\kappa(I,\nu\setminus\{\un,\bot\}))
\geq \tfrac{\epsilon}{32}.$$

 Otherwise, we use Claim \ref{cl::Q_U_bound} (2) to obtain %
$\theta_\kappa^{0}(I,\nu \setminus \{ \bot \}) \geq \tfrac{k}{5}$
and hence $T^\un\left(m_I \cdot \tfrac{1}{k} \cdot \theta_\kappa^{0}(I,\nu \setminus \{ \bot \})\right) \geq \tfrac{1}{10}$ (whp) %
which in turn implies that $M_I \geq \varphi^{0}_\kappa(I,\un) \geq \tfrac{1}{10}$.

\vspace{2mm}

\paragraph{General case: $l>0$.} Note that if %
$\theta_\kappa^{\beta^l}(\Lambda_\kappa^{\beta^l}(I),\nu \setminus \{ \bot \}) \geq \tfrac{b_l}{50} \cdot k$, then $T^\un\left(m_I \cdot \tfrac{1}{k} \cdot \theta_\kappa^{\beta^l}(\Lambda_\kappa^{\beta^l}(I),\nu \setminus \{ \bot \})\right) \geq \tfrac{1}{100} \cdot b_{l}$, and hence $M_I \geq \varphi_\kappa^{\beta^l}(I,\un) \geq \min\{\tfrac{1}{100} b_{l},1\}$.	
\aanote{formally, we have to check whether $1$ is the min, in which
  case the conclusion still hold, right?}\ns{check now please} 
Otherwise, we invoke Claim \ref{cl::density_range_set} to obtain there exists a set of pairs $\CD \subseteq \Lambda_\kappa^{\beta^{l}}(I) \times \nu \setminus \{ \bot \}$ of mass $\mu_{\kappa'}(\CD) \geq \tfrac{b_{l}}{25} \beta^l$ where $\relden_{\kappa'}(J,\chi',\Lambda_\kappa^{\beta^{l}}(I))_c \in [1,\beta^2] \cdot n^{3\alpha}\tfrac{\totalmu}{\beta^{l} \cdot \lambda^t}$ on $\tfrac{\epsilon/4}{\eta}$ distinct costs for each $(J,\chi') \in \CD$. We now use Claim~\ref{clm::dense_set_imp} and obtain that whp,
						
					$$M_I \geq \varphi^{\beta^l}_\kappa(I,\nu \setminus \{ \un,\bot \}) \geq \normo{T^\cc\left(m_I \cdot \tfrac{\phi_\kappa(\Lambda_\kappa^{\beta^{l}}(I),*)}{\beta^{l}} \right)} \ge \normo{T^\cc\left(\tfrac{1}{2} \cdot \tfrac{\phi_\kappa(\Lambda_\kappa^{\beta^{l}}(I),*)}{\beta^{l}} \right)} > \tfrac{\epsilon}{300} \cdot b_{l}$$ 
					
					as needed. We conclude that
                                        $M_I$ can be lower bounded by
                                        the recursive definition,
                                        
                                        $$b_{l+1} \geq \min\{b_l \cdot
                                        \Omega(\epsilon),1\} =
                                        \left(\Omega(\epsilon)\right)^{l+1},$$
                                         which
                                        completes the proof of the lemma.

		\end{proof}

%% file: imp-complexity.tex
\section{Runtime Complexity of the Interval Matching Algorithm}
\label{sec:runtime}

In this section we analyze the runtime complexity, establishing
Theorem~\ref{thm::matching_guarantee}, item \ref{it:matchingRuntime}.

First we prove the lemma bounding the size of color parts in each step
coloring.

\input{imp-sparsity}

\subsection{Algorithm runtime}

We now proceed to the main runtime analysis of the interval matching
algorithm.

We define $S_{\kappa'} = \norm{\mu_{\kappa'}(*,*)}_0$, i.e., the count
of non-zero-mass pairs $(I,\chi') \in \IC \times \nu$. We note that
from Lemma \ref{lm::partition_size_imp}, we have $S_{\kappa'} = n
\cdot \tO_\eps(\beta^{O(1)})$. The same Lemma also implies that we converge in $1/\epsilon + O(1)$ steps, since for $t>1/\epsilon + O(1)$, all $\nu \setminus \{\bot\}$ colors must be non-zero in 0 intervals, in which case we stop iterating. %

Before proceeding to the proof, we state a couple theorems for fast
evaluation of potentials; they are all proved in
Section~\ref{sec:fastSampling} (and can be obtained from known data
structures).  In particular, for sampling anchors and pivots, we use
the following data-structure.

\begin{theorem}\label{thm::ds_sampling}
Given a set of numbers $a_1,\ldots a_n\ge 0$ in a sparse
representation (i.e., as a set $S$ of $i$'s with $a_i>0$)
there exists a data structure $D$ supporting the following queries:
\begin{itemize}
\item
  Sample index $i\in[n]$ from the distribution $\{a_i/\sum_j a_j\}_i$ in time $O(\log n)$;
\item
  Given $k\ge 1$, produce a set $S\subseteq [n]$ that includes each
  $i$ with probability $\min\{ka_i,1\}$ independently.  The runtime is
  $O((1+k\cdot\sum_{i=1}^n a_i)\log n)$ in expectation.
\item
  Given an interval $[s,t]\subset [n]$, and $k\ge 1$, produce a set $S\subseteq [n]$ that includes each
  $i\in[s,t]$ with probability $\min\{ka_i,1\}$ independently.  The runtime is
  $O((1+k\cdot\sum_{i=s}^t a_i)\log^2 n)$ in expectation.
\end{itemize}

Furthermore, the preprocessing time is $O(|S|\cdot \log^{O(1)}n)$.
\end{theorem}

For fast calculation of
equations (\ref{eq::varphi_not0}) and (\ref{eq::varphi_un}), we need
data-structures that generate for each $I \in \IC$ the sums of only
the relevant $\phi, \theta$ colors over the $\Lambda$ balls (i.e., which
pass the threshold). We use the following data-structure and algorithm:

\begin{theorem}\label{thm::ds_sum}
Given a set of numbers $a_1,\ldots a_n\ge 0$ in a sparse
representation (i.e., as a set $S$ of $i$'s with $a_i>0$), there
exists a data structure $D$ supporting the following queries:
\begin{itemize}
\item
Given some $i<j$ compute $\sum_{k=i}^j a_k$. The query complexity is $O(\log n)$.
\item
Given some $i$ and $\zeta>0$, compute $j\in S, j\le i$ such that
$\sum_{k=j-1}^i a_k<\zeta\le \sum_{k=j}^i a_k$ (or output that none
exists). The query complexity is $O(\log n)$.
\end{itemize}

The preprocessing time is $O(|S|\log n)$.
\end{theorem}

\begin{theorem}\label{thm::ds_matrix}
  Suppose we are given an $m\times n$ matrix $A$, where $m\le n$,
  given in the sparse form, i.e., as a set $S$ of non-zero
  entries. Also, suppose we are given $n$ intervals
  $[s_i,t_i]\subseteq [n]$ such that both $\{s_i\}_i$ and $\{t_i\}_i$
  are non-decreasing. Given $\gamma>0$, we can find all pairs
  $(i,\chi)\in [n]\times[m]$ such that
  $\sum_{k=s_i}^{t_i}A_{\chi,k}\ge \gamma$. The runtime is $\tilde
  O\left(|S|+n+\tfrac{1}{\gamma} \cdot \sum_{\chi,i} \sum_{k=s_i}^{t_i}A_{\chi,k}\right)$.
\end{theorem}

Theorem \ref{thm::ds_sum} helps us calculating $\Lambda_\kappa^\zeta$
for each interval, and generating $\varphi_\kappa^\zeta(I,\un)$ while
Theorem \ref{thm::ds_matrix} is used for
$\varphi_\kappa^\zeta(I,*_{\neq \un})$. We prove these theorems in
section \ref{sec:fastSampling}.

We proceed to proving Theorem~\ref{thm::matching_guarantee}, item
\ref{it:matchingRuntime} next.

\begin{lemma}
  Clustering and assigning $\phi_\kappa$ takes $T_{\cad} \cdot \tilde{O}(n \cdot \lambda \cdot \beta^{O(1)})$ time in each step $t$ and level $l$.
\end{lemma}
\begin{proof}
	Fix a step $t$ and level $l$. We claim the following for the
        clustering and potential $\phi$ computation (in Alg.~\ref{alg::matching.new2}):
	\begin{enumerate}
		\item We call \textsc{ClusterAnchor} (Alg.~\ref{alg::imp.subroutines}) at most
                  $\lambda^t$ times. This is immediate from the fact each anchor can be clustered at most once.
		\item \textsc{ClusterAnchor} takes
                  $T_{\cad} \cdot n \cdot \tO(\lambda^{1-t} \cdot
                  \beta^{O(1)})$ time. Indeed, notice that we call
                  each $\textsc{ClusterAnchor}$ with $|\RC| =
                  |\PC_{\kappa'}^{\chi'}| = n \cdot
                  \tO_\eps(\lambda^{1-t} \cdot \beta^{5})$ (by Lemma
                  \ref{lm::partition_size_imp}), and we perform a
                  single $\cad$ calculation per interval $I \in \RC$.
		\item We call \textsc{AssignPhiPotential} (Alg.~\ref{alg::imp_matching.process_cluster}) at most
                  $\tfrac{1}{\alpha} \cdot \lambda^t$ times. Indeed, observe each call to \textsc{ClusterAnchor} outputs $\leq 1/\alpha$ clusters, and we perform \textsc{AssignPhiPotential} once for each output cluster.

		\item \textsc{AssignPhiPotential} takes $n \cdot
                  \tO(\lambda^{1-t} \cdot \beta^{O(1)})$ time in
                  expectation. We argue this below.
	\end{enumerate}
	
To analyze the time complexity of \textsc{AssignPhiPotential}, we need
to bound time spent estimating densities. For this task, notice that
for each $I \in \A$, we estimate $\den_{\kappa'}(I,\chi')_{\hat{c}}$
using $\den_m = d_\A=n^{-2\alpha}\mu_{\kappa'}(\A,\chi')$. Hence the
total complexity $T_{\text{APP}}(\A) = T_{\cad} \cdot \tO(|\A|\cdot \tfrac{\mu_{\kappa'}(\IC,\chi')}{d_\A} + |\A|)$. The term $T_{\cad} \cdot \tO(|\A|)$ term is dominated by the complexity of \textsc{ClusterAnchor} from above, hence we only need to bound:

$$ T_{\cad} \cdot \tO\left(|\A|\cdot \tfrac{\mu_{\kappa'}(\IC,\chi')}{d_\A}\right) = T_{\cad} \cdot \tfrac{|\A|}{d_\A} \cdot n \cdot \tO\left(\lambda^{1-t} \cdot \beta^{O(1)}\right).$$

 It remains to bound $\E\left[\tfrac{|\A|}{d_\A}\right]$.
	Fix $\chi \in [\lambda^t]$ for which we sampled an anchor
        triplet $(A,\chi'',c_i)$. For $(I,\chi',c_i) \in \IC \times
        \nu \setminus \{ \bot \} \times E_{c}$, define
        $Z_{I,\chi',\hat{c}}$ as the (random) event that $(I,\chi')$ is clustered in some $\A \in \CC_\chi$ for cost $\hat{c} \in E_\basec$:
	
	$$Z_{I,\chi',\hat{c}} = %
	\left[A \in \NC_{c_{i,j_{\max}}}(I) \cap \PC_{\kappa'}^{\chi'} \wedge \chi'' = \chi' \wedge \hat{c} = c_i\right]$$
	
	Also define $p_{I,\chi',\hat{c}} = \Pr[Z_{I,\chi',\hat{c}}]$ and %
		$\tau_{I,\chi',\hat{c}} = \max_{\A \in \CC_\chi} \tfrac{1}{d_\A} \cdot\one{Z_{I,\chi',\hat{c}}}$.
		Notice that $p_{I,\chi',\hat{c}} \leq \tfrac{d_{I,j_{\max}}}{2n}$,                where
                $d_{I,j_{\max}}=\den_{\kappa'}(I,\chi',\IC)_{c_{i,j_{\max}}}$.
		Also note that %
                we have $d_\A %
                \geq
                n^{-\alpha}\cdot d_{I,j_{\max}}$ (as
                $\NC_{c_{i,j_{\max}}}(I)\subseteq \NC_{3c}%
                (A)$)\ns{check this is indeed the case and no need for $4c$},
                and hence $\tau_{I,\chi',\hat{c}} \leq \tfrac{n^{\alpha}}{d_{I,j_{\max}}}$.
		We can finally bound the expectation:
		
		 $$\E\left[\tfrac{|\A|}{d_\A}\right] \leq \sum_{\substack{(I, \chi') \in \IC \times \nu \setminus \{ \bot \} \\ \mu_{\kappa'}(I,\chi')>0}} \tfrac{p_{I,\chi',\hat{c}}}{d_\A} \leq \sum_{\substack{(I, \chi') \in \IC \times \nu \setminus \{ \bot \} \\ \mu_{\kappa'}(I,\chi')>0}} p_{I,\chi',\hat{c}} \cdot \tau_{I,\chi',\hat{c}} \leq \tfrac{n^{\alpha}}{2n} \sum_{\substack{(I, \chi') \in \IC \times \nu \setminus \{ \bot \} \\ \mu_{\kappa'}(I,\chi')>0}} \tfrac{d_{I,j_{\max}}}{d_{I,j_{\max}}} = \tfrac{n^{\alpha}}{2n} S_{\kappa'} = \beta^{O(1)}.$$
		
We conclude the expected time spent on \textsc{AssignPhiPotential} for each cluster is $T_{\cad} \cdot \beta^{O(1)} \cdot n \cdot \tO(\lambda^{1-t} \cdot \beta^{O(1)}) = T_{\cad} \cdot n \cdot \tO(\lambda^{1-t} \cdot \beta^{O(1)})$ as needed.

To summarize the time complexity of clustering, we have for a fixed step $t$ and level $l$:
\begin{itemize}
	\item Total time spent on all calls to \textsc{ClusterAnchor} is $\lambda^t \cdot T_{\cad} \cdot n \cdot \tO(\lambda^{1-t} \cdot \beta^{O(1)}) = T_{\cad} \cdot n \cdot \tO(\lambda \beta^{O(1)}) = T_{\cad} \cdot n^{1+O(\eps)}$.
	\item Expected total time spent on all calls to
          \textsc{AssignPhiPotential} is $\tfrac{1}{\alpha} \cdot
          \lambda^t \cdot T_{\cad} \cdot n \cdot \tO(\lambda^{1-t}
          \cdot \beta^{O(1)}) = T_{\cad} \cdot n \cdot \tO(\lambda
          \beta^{O(1)}) = n^{1+O(\eps)}$, and since each call is
          independent and takes at most $T_{\cad} \cdot \tO(n)$ time, we have the bound with high probability as well.
\end{itemize}

\end{proof}
Second, we show that computing $\theta$'s for the pivots is also time efficient.

\begin{lemma}
  The algorithm \textsc{AssignThetaPotential} (Alg.~\ref{alg::imp_matching.assign_theta})
        takes $T_{\cad} \cdot n \cdot \tilde{O}(\lambda \beta^{O(1)})$
        time in each level $l$ (whp).
\end{lemma}

\begin{proof}
	Fix step $t$ and level $l$. We notice that the main runtime
        term to estimate is the time spend on approximating the
        relative densities. For this, we note:
	\begin{enumerate}
	 \item  We only estimate densities for pairs $(V,\chi')$ when $\mu_{\kappa'}(\Lambda_\kappa^\zeta(V),\chi') \leq \lambda \cdot \beta^{l + O(1)}$ (from Line \ref{alg::theta_treshold} threshold in \textsc{AssignThetaPotential}).
		\item We sample each $(I',\chi') \in \IC \times \nu$, $k$ times, each with probability $p_{I',\chi'} := \beta^{-l} \cdot \mu_{\kappa'}(I',\chi')$, and spend
	 $\tau_{I',\chi'} :=
                  \tO\left(\tfrac{1}{\relden_m}\cdot\tfrac{\mu_{\kappa'}(\IC,\chi')}{\mu_{\kappa'}(I',\chi')}
                  +\tfrac{\mu_{\kappa'}(\Lambda_\kappa^\zeta(V),\chi')}{\mu_{\kappa'}(I',\chi')}\right)$
                  time approximating such pair (by Lemma
                  \ref{lm::approx_density}), where $\relden_m = n^{3\alpha}\cdot\tfrac{n}{\lambda^t \beta^l}$.
	\end{enumerate}

	 Combining the above, we get $\tau_{I',\chi'} = T_{\cad} \cdot \tO\left(\tfrac{\lambda^t \beta^l}{n} \cdot \tfrac{n \cdot \lambda \cdot \beta^{O(1)}}{\lambda^t} \cdot \tfrac{1}{\mu_{\kappa'}(I',\chi') }+ \tfrac{\lambda \beta^{l + O(1)}}{{\mu_{\kappa'}(I',\chi')}} \right) = T_{\cad} \cdot \tO\left(\tfrac{\lambda \beta^{l + O(1)}}{\mu_{\kappa'}(I',\chi') } \right)$.

	 Hence, the expected run-time of all pair approximation is
	 
	 \begin{align*}
\E_\VC\left[\text{Density approximation time}\right] &= \E_\VC\left[\sum_{(I',\chi') \in \VC} \tau_{I',\chi'}\right] \\
&\leq \sum_{(I',\chi') \in \IC \times \nu: p_{I',\chi'}>0} k \cdot p_{I',\chi'} \cdot \tau_{I',\chi'} \\
&= k \cdot  \sum_{(I',\chi') \in \IC \times \nu: p_{I',\chi'}>0} \beta^{-l} \cdot\mu_{\kappa'}(I',\chi') \cdot \tO\left(\tfrac{\lambda \beta^{l + O(1)}}{\mu_{\kappa'}(I',\chi') } \right)	\\
&= k \cdot  \sum_{(I',\chi') \in \IC \times \nu: p_{I',\chi'}>0} \tO\left(\lambda \cdot \beta^{O(1)}\right) \\
&= k \cdot S_{\kappa'} \cdot \tO\left(\beta^{O(1)}\right) = \tO\left(n \cdot \lambda \cdot \beta^{O(1)}\right).
 \end{align*}

Since the expectation is over sum of independent r.v., bounded by $O(n)$, we also have the bound whp.

\end{proof}

Last, we show that assigning colors using $\phi, \theta$ scores is efficient:

\begin{lemma}
	Assigning colors in \textsc{AmendColoring} (Alg.~\ref{alg::matching.color_new}) takes $\tO_\eps(n \cdot \beta^{O(1)})$ time in each level $l$.
\end{lemma}
\begin{proof}
	Let $S_\un = \norm{\mu_\kappa(*,\un)}_0$, $S_\theta = \norm{\theta_\kappa(*,*)}_0$, $S_\phi = \norm{\phi_\kappa(*,*)}_0$. We note that $S_\un \leq 2n$ and $S_\theta, S_\phi \leq S_{\kappa'} = n \cdot \tO_\eps(\beta^{O(1)})$.
	
	First, from Theorem \ref{thm::ds_sum}, the pre-processing time
        of the data structures is $\tO\left(S_\un +
        \tO(S_\theta)\right) = \tO_\eps(n \cdot \beta^{O(1)})$, and
        since we query each data structure $O(n)$ times (once for each $I \in \IC$), the total query time is $\tO(1)\cdot O(n) = \tO(n)$.
	
	Second, focusing on calculating $\varphi_\kappa^\zeta(*,*_{\neq \un})$, define $L_\phi = \sum_{I \in \IC} \sum_{\zeta \in Z_l}
        \tfrac{\beta^{2l}}{\max\{\zeta^2,1\}} \tfrac{ \mu_\kappa(I,\un) \cdot
          \phi_\kappa(\Lambda_\kappa^\zeta(I),\nu)}{\beta^l}$, and recall from Eqn. (\ref{eq::sum_all_phi}), we have $L_\phi = \lambda^t \cdot O_\eps(n^{2\alpha} \cdot \tfrac{n}{\lambda^t}) =  O_\eps(n \cdot \beta)$. 
        We invoke Theorem \ref{thm::ds_matrix} using $\gamma = \beta^{-4}$, to obtain the time spent on calculating $\varphi_\kappa^\zeta(*,*_{\neq \un})$ is $\tO\left(S_\phi + n + \tfrac{L_\phi}{\beta^{-4}}\right) = \tO_\eps(n \cdot \beta^{O(1)})$.
        
        Last, each update rule of $\mu_\kappa(I)$ using $\varphi$s takes $\tO(1)$ time, which sums up to $\tO(n)$ over all intervals.
\end{proof}

Using the above Lemmas, we prove our main complexity guarantee.
\begin{proof}[Proof of Theorem~\ref{thm::matching_guarantee}, item \ref{it:matchingRuntime}]
We conclude that the runtime of \textsc{MatchIntervals}
(Alg.~\ref{alg::matching.new2}), per one step and level, is $\tilde
O(n^{1+O(\eps)}\cdot T_\cad)$. There's is a constant number of
levels, so the same bound holds for each step. 
Since the algorithm converges in constant steps (as was shown above),
and the branching factor is $O(\log n)$,
the conclusion follows.
\end{proof}

%% file: imp-sparsity.tex
\subsection{Sparsity of colors in colorings: proof of Lemma~\ref{lm::partition_size_imp}}\label{sec::sparsity}

Recall that Lemma~\ref{lm::partition_size_imp} shows that each color
$\chi \in \nu \setminus \{ \bot\}$ is sparse in each coloring
$\kappa'$. Specifically, recalling that $\PC_{\kappa'}^{\chi'}$
is the set of intervals $I$ with $\mu_{\kappa'}(I,\chi')>0$, we need to
prove that $|\PC_{\kappa'}^{\chi'}|\le O_\eps(\beta^5
\lambda)\cdot \tfrac{n}{\lambda^t}$.

\begin{proof}[Proof of Lemma \ref{lm::partition_size_imp}]
	Consider a color $\chi \in \nu \setminus \{ \un, \bot \}$ in coloring $\kappa$. %
	For each level $l$ and for each $j$, we add $O\left(\tfrac{\totalmu}{\lambda^t} \cdot \tfrac{%
					\mu_{\kappa'}(I,\chi')}{\widehat{d_{I,j}}}\right) = O\left(n^{\alpha} \cdot \tfrac{\totalmu}{\lambda^t} \cdot \tfrac{%
					\mu_{\kappa'}(I,\chi')}{d_{A,j}}\right)$ $\phi$-potential to each $(I,\chi') \in \NC_{c_{i,j}}^{\kappa'}(A)$, hence for level $l$ we have $\phi(\IC,\chi) = O_\eps\left(n^{\alpha} \cdot \tfrac{\totalmu}{\lambda^t} \right)$. 
We use the bound on
$\phi_\kappa(\IC,\chi)$ to bound the
total number of intervals
which ``survive the $T^\cc$
transformation'' (i.e., have
$\varphi$ potential still non-zero). First we derive, for $\l\ge1$:
\aanote{just to make sure, we don't need to consider $l=0$ here right?}\ns{correct.}

\begin{equation}\label{eq::sum_all_phi}
	\sum_{I \in \IC} \sum_{\zeta \in Z_l}
        \tfrac{\beta^{2l}}{\zeta^2} \tfrac{ \mu_\kappa(I,\un) \cdot
          \phi_\kappa(\Lambda_\kappa^\zeta(I),\chi)}{\beta^l} \leq
        \sum_{\zeta \in Z_l} \tfrac{\beta^{l}}{\zeta} \sum_{I \in \IC}
        \tfrac{ \mu_\kappa(I,\un)}{\zeta} 
        \phi_\kappa(\Lambda_\kappa^\zeta(I),\chi) 
        \stackrel{Claim~\ref{cl::lambda_ball_func}}{\leq}  \sum_{\zeta \in
          Z_l} \tfrac{\beta^{l}}{\zeta} %
        2\phi_\kappa(\IC,\chi) = O_\eps(n^{\alpha} \cdot
        \tfrac{n}{\lambda^t}).		
	\end{equation}

	Now we have that $\sum_{\zeta \in Z_l} \varphi_\kappa^\zeta
        (\IC, \chi) = \sum_{I \in \IC} \sum_{\zeta \in Z_l}
        T^\cc(\tfrac{\beta^{2l}}{\zeta^2}\tfrac{ \mu_\kappa(I,\un) \cdot
          \phi_\kappa(\Lambda_\kappa^\zeta(I),\chi)}{\beta^l})$. Since $T^\cc$ zeros out
        any input smaller than $1/\beta^{4}$, we have that
        at most $O_\eps(n\beta^5 / \lambda^t)$ intervals $I$ have
        $\sum_{\zeta \in Z_l} \varphi_\kappa^\zeta (\IC, \chi)>0$.
        Finally, one can observe $I \in \PC_{\kappa'}^{\chi'}$ only if $\varphi_{\kappa'}(I,\chi') > 0$ at some level of the previous step $t-1$, hence the bound follows.
	
	It remains to show for the color $\chi=\un$. 
		For this task, let 
              $l'$ be such that $\beta^{l'} \in [1,\beta) \cdot n^{3\alpha}\cdot \beta \cdot \tfrac{\totalmu}{\lambda^t}$.
                We consider two cases.

First, suppose the algorithm reaches level $l'$ and $\mu_\kappa(\XC,\un), \mu_\kappa(\YC,\un) \geq
\beta^{l'}$. We now assume (by induction) Lemma~\ref{lm::partition_size_imp} holds for $\kappa'$ and invoke Claim
\ref{cl::Q_U_bound}. Then we have that each $I \in \IC$ $T^\un(\tfrac{1}{k} \cdot
\theta_\kappa^{\zeta}(\Lambda_\kappa^{\zeta}(I),\nu \setminus \{ \bot
\}))$ is at most proportional to the color mass where
$\relden_\kappa(I,\chi',\QC^{+2\zeta})_c \leq n^{3\alpha} \cdot \tfrac{\totalmu
  \beta}{\beta^{l'} \lambda^t}$ for some set $\QC$,
on some of the costs. However, since, by the definition of relative
density it must be at least 1, then $\relden_\kappa(I,\chi',\QC^{+2\zeta})_c \geq 1$
and hence we obtain $T^\un(\tfrac{1}{k} \cdot \tfrac{\beta^{2l}}{\zeta^2}\cdot
\theta_\kappa^{\zeta}(\Lambda_\kappa^{\zeta}(I),\nu \setminus \{ \bot
\})) \leq T^\un(1.1\beta^{-1.5} \cdot \tfrac{1}{k} \cdot \tfrac{\beta^{2l}}{\zeta^2} \cdot
1.1 k \tfrac{\zeta}{\beta^l}) \leq T^\un(1.1 \beta^{-1.5}) = 0$ for
all $I \in \IC$ (note that Claim~\ref{cl::Q_U_bound} requires
$\delta\ge \beta^{-1.5}$). This means that all $\varphi_\kappa^\zeta (\IC, \un)=0$ and $\un$ will not appear on any
interval.

Now, consider the opposite: the algorithm breaks at level $l^* \leq
l'$, meaning that $\mu_\kappa(\XC,\un)$ or
                $\mu_\kappa(\YC,\un)$ is at most
                $\beta^{l^*} \leq \beta^{l'}$. Assume w.l.o.g $\mu_\kappa(\XC,\un)
                < \beta^{l'}$; this implies
                $\mu_\kappa(\XC,\un) < \tfrac{\totalmu
                  \beta^2 \cdot n^{3\alpha}}{\lambda^t}$\aanote{shouldn't there be a
                  $n^{3\alpha}$ too?}\ns{right!}.  		
         Next, notice that if
                $\tfrac{\mu_\kappa(\YC,\un)}{\mu_\kappa(\XC,\un)} >
                \beta$, then \textsc{Adjust-$\un$} zeros out all of $\un$
                mass. Otherwise, we have
                that $\mu_\kappa(\IC,\un) = O(\tfrac{\totalmu
                  \beta^4}{\lambda^t})$ before calling to
                \textsc{Adjust-$\un$}. Notice that
                \textsc{Adjust-$\un$} omits $\un$ on any palette
                 where $\un$ is of mass at most
                 $\Omega_\epsilon(\beta^{-1})$ (by the $T^\un$
                 transformation), hence we conclude color $\un$ is
                 part of $O_\eps(\tfrac{\totalmu \beta^5}{\lambda^t})$
                 palettes as well at the end of step $t$ for each
                 $\kappa \in \Kappa$. Hence in the next step $\hat{t} = t+1$, we get $|\PC_{\kappa'}^{\chi'}| = \tO_\eps(\tfrac{\totalmu \beta^5}{\lambda^{\hat{t}-1}})$ for all  $\chi' \in \nu \setminus \{ \bot \}$ as needed.

\end{proof}

%% file: triangle.tex
\section{Alignment Distance Algorithm for $\cad$: Proof of Theorem~\ref{thm:cadGuarantees}}
\label{sec:CADalgo}

\newcommand{\constg}[0]{\tau}

In this section we show the algorithm for computing the $\cad_w$
distance, in particular proving Theorem~\ref{thm:cadGuarantees}. To
briefly recall the theorem, we want an algorithm for computing a
metric $\cad_w(I,J)$ on $w$-length strings $I,J$ given an oracle to a
metric $\distG_{w/\wbase}$ on $w/\wbase$-length strings, running in
time $\poly(\wbase)$. Note that we can assume that $w\ge \wbase^4$, as
otherwise we can afford to set $\cad(I,J)=\ed(I,J)$ and compute it
directly. We manage to ensure it only for ``single scale''
metrics $\cad_{w,c}$, designed for distances in the range $\approx
[c,\wbase c]$.

First, we reduce $\cad_{w,c}$, where $c\in \Sw$, to a set of fewer
$\cad$ functions, indexed $\cad_{w,t}$, where
$t\in\{\wbase,\wbase^{2},\ldots, w/\wbase^2\}$. In particular,
for given $c$, we set $t$ by rounding down $c/\constg$ to an integer
power of $\wbase$, for some (large) constant
$\constg$ (to be fixed later). When $c/\tau<\wbase$, we set $t=\wbase$.

We note that the algorithm for the ``largest distance regime'', when
$t=w/\wbase^2$, will be different from the rest, when $t\le
w/\wbase^3$. We describe each of the two algorithms separately
starting with the ``large'' distance regime, which is easier.

Henceforth, for simplicity, we set $w'=w/\wbase$, and often refer to
$\distG_{w'}$ as $\distG$. Also, for strings $I,J$ of length $w$, we
let $I_{w'},J_{w'}$ be the strings starting at the same position but
of length $w'$ only (e.g., for $I=X_{i,w}$, $I_{w'}=X_{i,w'}$).
Finally, we cap all the output distances at $w$; note that cannot
break any of the guarantees (e.g., triangle inequality, or the fact
that each $\cad_{w,c}\ge \ed_w$ as $\ed_w\le w$).

Below, to simplify lots of notation, we will use notation $k\in[m]$ to
mean $k\in\{0,1,2,m-1\}$ (which is different from how it was used in
the past sections). Also, the notation $[a,b)$ means all integers
  $\{a,...b-1\}$, whereas $[a,b]=\{a,...b\}$.

\subsection{Large distance algorithm: $\cad_{w,w/\wbase}$ metric}

For simplicity of notation, we set $\Gamma=\cad_{w,t}$ when
$t=w'=w/\wbase\ge \wbase$. 

The algorithm proceeds as follows. For
$m=\wbase^4$, the vertices are $V_{p,q}$ where $p \in [-m,m]$ and
$q\in [-m,m]$. For each node $V_{p,q}$, we add the following edges to
$H$:
\begin{itemize}
	\item $V_{p,q} \to V_{p + 1,q + 1}$ of cost $\tfrac{w/m}{w'}\distG_{w'}(I_{w'} + \tfrac{w}{m}p,J_{w'} + \tfrac{w}{m}q)$ (diagonal edges).
	\item $V_{p,q} \rightarrow V_{p,q+1}$ and $V_{p,q} \rightarrow
          V_{p+1,q}$ of cost $w/m$ (gap edges).
\end{itemize}

Now we run the shortest path from vertex $V_{-m,-m}$ to $V_{m,m}$. Define
$\Gamma$ to be 4 times that shortest path value. It is immediate
to note that this can be computed in $m^{O(1)}$ time.

\subsubsection{Lower and upper bounds on $\Gamma$}

First, we prove the lower bound. 

\begin{lemma}
 \label{lem:lbGamma}
$\Gamma(X_{i,w},Y_{j,w})\ge \ed_w(X_{i,w},Y_{j,w})$ for all $i,j$.
\end{lemma}

\begin{proof}
  Note that any path in $H$ corresponds to an alignment $A:[-m,m)\to
  [-m,m)\cup\{\bot\}$ as follows. For each $p\in[-m,m)$, consider the edge
  which increases $p$: if it's the diagonal edge $V_{p,q}\to
  V_{p+1,q+1}$ then $A(p)=q$; otherwise set $A(p)=\bot$.
  The number of all
  gap edges is then $2\cdot |\{p:A(p)=\bot\}|$. Hence, recalling the
  convention that $\ed_{w'}(X_i,Y_{f(\bot)})=w'$ for any function
  $f:\N\to \N$:
  $$
  \Gamma(X_{i,w},Y_{j,w})\ge
  \tfrac{w/m}{w'}\sum_{k\in [-m,m)}\ed_{w'}(X_{i+kw/m,w'},Y_{j+w/m\cdot
    A(k),w'})
    +2\tfrac{w}{m}\cdot|\{p:A(p)=\bot\}|.
  $$
  By averaging, there must exist some
  $\delta\in\left[\tfrac{w'}{w/m}\right]$ such that:
  $$
    \Gamma(X_{i,w},Y_{j,w})
    \ge 
  \sum_{k\in [-\wbase,\wbase)}\ed_{w'}(X_{i+w/m\cdot \delta+kw',w'},Y_{j+w/m\cdot
    A(\delta+km/\wbase),w'})
  +2\tfrac{w}{m}\cdot|\{p:A(p)=\bot\}|.
  $$
  We now lower bound this by the $\ed(X_{i+w/m\cdot \delta,w},Y_{j+w/m\cdot A(\delta+km/\wbase),w'})$. In particular, note that we can build an LCS
  between these $w$-length strings from the LCS of the $w'$-length
  strings $X_{i+w/m\cdot \delta+kw',w'}$ and $Y_{j+w/m\cdot
    A(\delta+km/\wbase),w'}$. Now, the latter strings $Y_{j+w/m\cdot
    A(\delta+km/\wbase),w'}$ may overlap, and hence some
  characters double-counted. To account for this
  overlap, let $k_1<k_2<\ldots<k_z$ be the indeces $k$ where $A(\delta+km/\wbase)\neq \bot$. Then, we can lower-bound the sum of $\ed$ of
  $w'$-strings by $\ed$ between the $2w$-length strings, minus the
  overlap as follows, where $i'\triangleq i+w/m\cdot \delta-w$ and $j'=j+w/m\cdot
    A(\delta-m)$ (if it exists, and $j'=j+i'-i$ otherwise):
  \begin{align*}
    \Gamma(X_{i,w},Y_{j,w})
  &\ge
  \ed_{2w}(X_{i',2w},Y_{j',2w})
  +2\tfrac{w}{m}\cdot|\{p:A(p)=\bot\}|
  \\
  &\quad
   -\sum_{l=1..z-1} \tfrac{w}{m}\cdot \left|A(\delta+k_{l+1}m/\wbase)-A(\delta+k_lm/\wbase)-\tfrac{w'}{w/m}\right|
  \end{align*}

Now we note that the subtracted overlap is upper bounded by the
(scaled) number of $p$ s.t. $A(p)=\bot$ (function $A$ must be
injective on the rest of $p$'s). Furthermore,
$\ed_w(X_{i,w},Y_{j,w})\le
\ed_{2w}(X_{i',2w},Y_{j',2w})+\tfrac{w}{m}\cdot|\{p:A(p)=\bot\}|$ and
hence: \aanote{double-check}
  \begin{align*}
\Gamma(X_{i,w},Y_{j,w})
&\ge
  \ed_{2w}(X_{i+w/m\cdot \delta-w,2w},Y_{j+w/m\cdot
    A(i+w/m\cdot \delta-w),2w})
  +\tfrac{w}{m}\cdot|\{p:A(p)=\bot\}|,
  \\
&\ge
  \ed_w(X_{i,w},Y_{j,w}).
  \end{align*}

\end{proof}

Now we prove the upper bound. Below, constant $C>1$ is from
Theorem~\ref{thm:cadGuarantees} hypothesis:

\begin{lemma}
  \label{lem:gammaUpperBound}
  For any fixed $i$ and any alignment $\pi$ with $\pi(i)\neq \bot$, we have
  $$\Gamma(X_{i,w},Y_{{\pi(i),w}})\le
  O\left(\tfrac{w}{\wbase^2}+\sum_{l\in[w]}\tfrac{1}{w'}\distG(X_{i+l,w'},Y_{\pi(i+l),w'})
  +
  C\cdot\sum_{l\in[w]} |\piComp(i+l+1)-\piComp(i+l)-1|\right).
  $$
\end{lemma}

In the above, in the last term, recall from the preliminaries that
$\piComp(i')$ is the minimum $\pi(j)$, $j\ge i'$, which is defined
($\neq \bot$). This last term should be thought of the error stemming
from $\pi$ skipping large chunks of $y$ (which, overall, can be
charged to the edit distance between $x$ and $y$).

\begin{proof}
  Defining
  $$u_\pi\triangleq\sum_{l\in[w]}\tfrac{1}{w'}\distG(X_{i+l,w'},Y_{\pi(i+l),w'})+
  C\sum_{k\in[-m,m)} |\piComp(i+(k+1)w/m)-\piComp(i+k\cdot w/m)-w/m|,$$
  we note that it is enough to prove that $\Gamma(X_{i,w},Y_{{\pi(i),w}})\le
  O(w/\wbase^2+u_\pi)$ as, by simple triangle inequality, $|\piComp(i+(k+1)w/m)-\piComp(i+k\cdot
  w/m)-w/m|\le \sum_{l\in[w/m]}|\piComp(i+k\cdot w/m+l+1)-\piComp(i+k\cdot
  w/m+l)-1|$.

  Now, let $j=\pi(i)$.  For $k\in[-m,m)$, define
  $A'(k)=\left\lfloor\tfrac{\piComp(i+kw/m)-j}{w/m}\right\rfloor$. If
  $A'(i)=A'(i-1)$ or $A'(i)\ge m$ then set $A(i)=\bot$ and
  $A(i)=A'(i)$ otherwise  (i.e., only the first copy of a sequence of
  equal numbers remains). 

  We claim that
  $$
  |\{k:A(k)=\bot\}|\le O(1+\sum_{k\in [-m,m)}
  \tfrac{|\piComp(i+(k+1)w/m)-w/m-\piComp(i+kw/m)|}{w/m})\le O(1+\tfrac{u_\pi}{w/m}),$$
  since any set of real numbers $\{r_1,\ldots r_m\}$ (for us
  $r_k=\tfrac{\piComp(i+kw/m)}{w/m}$), when rounded down, will have a
  number of duplicates bounded by $1+\sum_{k=1}^{m}
  |r_{k+1}-r_k-1|$. Furthermore, the number of distinct $A'(i)\ge 
  m$ is upper bounded by (using triangle inequality):
  $$
  1+\tfrac{m}{w}|\piComp(i+w-m)-\piComp(i)-w|
  \le
  1+\sum_{k\in [-m,m)}
  \tfrac{|\piComp(i+(k+1)w/m)-w/m-\piComp(i+kw/m)|}{w/m}.
  $$

  Now we have that for $\Gamma=\Gamma(X_{i,w},Y_{{\pi(i),w}})$:
  \begin{align*}
    \Gamma
    &\le
    4\sum_{k\in
      [-m,m), A(k)\neq \bot}\tfrac{w/m}{w'}\distG(X_{i+kw/m,w'},Y_{j+w/m\cdot
      A(k),w'}) + \tfrac{w}{m}\cdot 2|\{k:A(k)=\bot\}|
    \\
    &\le
    4\sum_{k\in
      [-m,m)}\tfrac{w/m}{w'}(\distG(X_{i+kw/m,w'},Y_{\piComp(i+kw/m),w'})+w/m) + \tfrac{w}{m}\cdot O(1+u_\pi)
    \\
    &=
    O(u_\pi)+
    4\sum_{k\in
      [-m,m)}\tfrac{w/m}{w'}\distG(X_{i+kw/m,w'},Y_{\piComp(i+kw/m),w'}),
  \end{align*}
  using the fact that we set $m\ge \wbase$.

  Now note that, for any $l\in[w/m]$, we have that:
  $$
  \distG(X_{i+kw/m,w'},Y_{\piComp(i+kw/m),w'})
  \le
  \distG(X_{i+l+kw/m,w'},Y_{\piComp(i+l+kw/m),w'})+2Cl+C|\piComp(i+(k+1)w/m)-\piComp(i+kw/m)|.
  $$
  Hence:
  \begin{align*}
    \Gamma
    &\le
    O(u_\pi)+
    4\sum_{k\in
      [-m,m)}\tfrac{1}{w/m}\sum_{l\in[w/m]}\tfrac{w/m}{w'}\left(\distG(X_{i+kw/m,w'},Y_{\piComp(i+kw/m),w'})\right)
    \\
    &\le
    O(u_\pi)+
    4\sum_{k\in
      [-m,m)}\tfrac{1}{w'}\sum_{l\in[w/m]}
    \\
   &\qquad\qquad
    \left(\distG(X_{i+l+kw/m,w'},Y_{\piComp(i+l+kw/m),w'})+2Cl+ C|\piComp(i+(k+1)w/m)-\piComp(i+kw/m)|\right)
    \\
    \vspace{2mm}
    &\le
    O(u_\pi)+
    4\sum_{l\in
      [w]}\tfrac{1}{w'}\distG(X_{i+l,w'},Y_{\piComp(i+l),w'})+\tfrac{8C}{w'}m\tfrac{w^2}{m^2}
    \\
   &\qquad\qquad
   +\tfrac{4C}{w'}\sum_{k\in[-m,m)}\tfrac{w}{m}\left(\tfrac{w}{m}+|\piComp(i+(k+1)w/m)-\piComp(i+kw/m)-\tfrac{w}{m}|\right)
    \\
        \vspace{2mm}
    &\le
    O(u_\pi)+
    4\sum_{l\in
      [w]}\tfrac{1}{w'}\distG(X_{i+l,w'},Y_{\piComp(i+l),w'})+8C\tfrac{w\wbase}{
    m}
    +4C\tfrac{w\wbase}{m}+4\tfrac{w/m}{w'}\cdot
    u_\pi
    \\
    &\le
    O(u_\pi+w/\wbase^2),
  \end{align*}
  where we use $O(C w\wbase/m)\le w/\wbase^2$ by our choice of
  $m=\wbase^4$.
\end{proof}

\subsection{Not large distance regime: $\cad_{w,t}$ for $t\le w/\wbase^3$}

To compute the $\cad_{w,t}$ distance for smaller $t$, we use a
slightly different alignment representation called {\em block
  alignment}, which maps ``grid blocks'' into $x$-axis and $y$-axis
``shifts'' (loosely speaking). We think of alignment of square
blocks of size $w'$ on $w$ by $w$ grid, defined by
coordinates $p_i,q_i$ for every square block $i \in
[\wbase]$ in the alignment. The first coordinate represents
accumulated horizontal shifts, corresponding to inter-block
insertions, and the second one represents accumulated vertical shifts,
corresponding to inter-block deletions. Such representation is
somewhat easier to handle in our case, and will be formally defined
below.

Define $\theta = t/w$.  Let $T=\wbase^3$.  To ease exposition, we use
$|\cdot|_1$ notation for the $\ell_1$ norm.
  
  \renewcommand{\normo}[1]{{|#1|_1}}

\begin{definition}[$\cad$]\label{def::cad_new}
Let $\A$ be the set of functions $A=(A_x,A_y)$ where $A_x,A_y: [-\wbase,\wbase]
\rightarrow [T]$ are non-decreasing functions with $A_x[-\wbase]=A_y[-\wbase]=0$ and
$A_x[\wbase]=A_y[\wbase]$. For $I,J\in \IC_w$, define the distance
$\cad_{w,t}(I,J)=\min\{\cad_{w,t}^*(I,J), w'\theta T\}$, where
$\cad_{w,t}^*(I,J)$ is equal to:
$$
\min_{A\in \A} \left(A_x[\wbase] + A_y[\wbase]\right)\theta w' + \sum_{k \in [-\wbase,\wbase)} \tfrac{1}{T}\sum_{\Delta \in [3T - \normo{A[k]}]} \distG \left(I_{w'} + w'(k + \theta (\Delta+ A_x[k])),J_{w'} + w'(k + \theta(\Delta + A_y[k]))\right).
$$

\end{definition}

The former term can be thought of as ``shift costs'' and the latter as ``diagonal costs''.

\subsubsection{Fast computation}

\paragraph{Algorithm.}
The vertices are $V_{k,p,q}$ where $k \in [-\wbase,\wbase]$ \ns{why
  here $[-\wbase,\wbase)$ and $[-\wbase,\wbase]$ in def?}\aanote{ok now?} and $p,q \in
[T]$, as well as final node $V^*$.
For each node $V_{k,p,q}$ we add the following edges to $H$:
\begin{itemize}
	\item $V_{k,p,q} \to V_{k+1,p,q}$ of cost $\tfrac{1}{T}\sum_{\Delta \in [3T - p - q]} \distG (I_{w'} + w'(k + \theta (\Delta+ p)),J_{w'} + w'(k + \theta(\Delta + q)))$ (diagonal edges).
	\item $V_{k,p,q} \rightarrow V_{k,p+1,q}$ and $V_{k,p,q} \rightarrow V_{k,p,q+1}$ of cost
          $\theta w'$ (gap edges).
    \item when $k=\wbase$, $V_{\wbase,p,q} \rightarrow V^*$ of cost $\theta w'\cdot |p-q|$ (target edges).
\end{itemize}

We run shortest path from $V_{-\wbase,0,0}$ to $V^*$, and output its
value, capped at $w'\theta T$ (from above).

\paragraph{Complexity.} It's immediate to see that the runtime
complexity is $\wbase^{O(1)}$ since $H$ has $O(\wbase T^2)$
edges. Note that this is also the upper bound on the number of calls to
$\distG$ distance queries.

\subsubsection{Lower bound on $\cad_{w,t}$}

We establish the lower bounds for $\cad_{w,t}$ in the following lemma.

\begin{lemma}\label{lm::cad_lb}
For all $t \in \{\gamma,\ldots w/\wbase^3\}$, intervals $I,J\in \IC_w$, we have that $\cad_{w,t}(I,J)\ge
\min\{\ed_w(I,J),w'\theta T\}$.
\end{lemma}

For both the above, as well as for the upper bound shown later, we
define distances:
\begin{align}
d_t(I,J) &\triangleq \min_A \sum_{k \in [-\wbase,\wbase)}
  \tfrac{1}{T}\sum_{\Delta \in [T]} \distG (I_{w'} + w'(k + \theta (\Delta+
  A_x[k])), J_{w'} + w'(k + \theta(\Delta + A_y[k]))) + \normo{A[\wbase]}\theta w'.	\\
D_t(I, J) &\triangleq \min_A \sum_{k \in [-\wbase,\wbase)} \tfrac{1}{T} \sum_{\Delta \in [3T]} \distG (I_{w'} + w'(k + \theta (\Delta+ A_x[k])),J_{w'} + w'(k + \theta(\Delta + A_y[k]))) + \normo{A[\wbase]}\theta w'.
\label{eqn:defDt}
\end{align}

Note that $d_{t}(I,J) \leq \cad_{w,t}^*(I,J) \leq D_{t}(I,J)$, hence
it is enough to prove Lemma \ref{lm::cad_lb} for $d_t$.  We will use
$D_t$ later for the upper bound, in Lemma \ref{lm::cad_ub}. Note that
while each of $d_t, D_t$ lower/upper bounds $\cad_{w,t}$, those may
not satisfy triangle inequality over $\IC$, hence we use
Def.~\ref{def::cad_new} for $\cad$.

\begin{proof}[Proof of Lemma~\ref{lm::cad_lb}]
Without loss of generality, assume that $I=X_i=X_{i,w}$ (i.e., it is
an interval from $x$, starting at position $i$) and $J=Y_j=Y_{j,w}$ (the cases when
$I,J$ are both from $x$ or both from $y$ are treated in exactly the
same manner).
By the above, it is enough to prove that $d_t(X_{i},Y_{j})\ge
\ed(X_{i},Y_{j})$. Note that we can rewrite:
$$
d_t(X_{i},Y_{j}) = \min_A \tfrac{1}{T}\sum_{\Delta \in [T]} \sum_{k \in [-\wbase,\wbase)}  \distG (X_{i} + w'(k + \theta (\Delta+ A_x[k])),Y_{j} + w'(k + \theta(\Delta + A_y[k]))) + \normo{A[\wbase]}\theta w',
$$
and hence, for the minimizing $A$, there's some $\Delta\in [T]$ such that
\begin{align*}
  d_t(X_{i},Y_{j})
  &\ge
  \sum_{k \in [-\wbase,\wbase)}  \distG (X_{i} + w'(k + \theta (\Delta+ A_x[k])),Y_{j} + w'(k + \theta(\Delta + A_y[k]))) + \normo{A[\wbase]}\theta w'.
  \\
  &=
  \sum_{k \in [-\wbase,\wbase)}  (\distG (X_{i} +w'\theta\Delta+ w'(k + \theta\cdot
  A_x[k]),Y_{j}+w'\theta\Delta + w'(k + \theta\cdot A_y[k]))
  \\
  &\qquad\qquad\quad+
  (A_x[k+1]-A_x[k]+A_y[k+1]-A_y[k])\cdot\theta w').
\end{align*}

Using that $\distG\ge \ed$, we get, using the notation
$i'\triangleq i+w'\theta\Delta\le i+w$ and $j'\triangleq
j+w'\theta\Delta\le j+w$:
\begin{align*}
d_t(I,J)
&\ge
\sum_{k \in [-\wbase,\wbase)}  \ed(X_{i'+ w'(k + \theta\cdot
    A_x[k])},Y_{j' + w'(k + \theta\cdot A_y[k])})+
  w'\theta\cdot (A_x[k+1]-A_x[k]+A_y[k+1]-A_y[k])
  \\
  &\ge \ed(X[i'-w:i'+w+w'\theta A_x[\wbase]],Y[j'-w:j'+w+w'\theta A_y[\wbase]])\ge \ed(X[i:i+w], Y[j:j+w]),
\end{align*}
where we also used the fact that $A_x[\wbase]=A_y[\wbase]$.

In conclusion: $\cad_{w,c}(I,J)=
\min\{\cad_{w,c}^*(I,J),w'\theta T\}\ge
\min\{d_t(I,J),w'\theta T\}\ge \min\{\ed(I,J),w'\theta T\}$.
\end{proof}

\subsection{$\cad_w$ align-approximates $\ed$}

Recall that $\cad_w(X_i,Y_j)=\sum_{c\in \Sw}
c\cdot\one{\cad_{w,c}(X_i,Y_j)\ge c}$. Hence the lower bound, that
$\cad_w\ge \ed_w$ point-wise follows immediately from Lemmas~\ref{lem:lbGamma}
and  \ref{lm::cad_lb} (for the latter, note that $w'\theta T\ge
\tfrac{c}{\tau \wbase^2}T\ge c\sqrt{\wbase}$). In particular for $c^*\in S_w$ such that $1\le c^*\le
\ed_w(X_i,Y_j)<2c^*$, we have $\cad_w(X_i,Y_j)=\sum_{c\in S_w:c\le c^*}
c\ge 2c^*-1\ge \ed_w(X_i,Y_j)$ (the case of $\ed_w(X_i,Y_j)=0$ is immediate).

We now prove the upper bound on
$\cad_w$, where $C$ is the constant from the
Theorem~\ref{thm:cadGuarantees} hypothesis.
  Fix the alignment $\pi\in \Pi$ that
  minimizes the cost $k_\pi=\sum_{i\in [n]}
  \tfrac{1}{w'}\distG_{w'}(X_{i,w'},Y_{\pi(i),w'})$. Note that the
  hypothesis of the theorem implies that $k_\pi\le C\cdot
\ed(x,y)$.
We prove the following lemma.

  \begin{lemma}\label{lm::cad_ub}
$ \sum_{i\in [n]}\tfrac{1}{w}\cad_w(X_{i,w},Y_{\pi(i),w})\le O(C)\cdot 
    \sum_{i\in[n]}\tfrac{1}{w'}\distG_{w'}(X_{i,w'},Y_{\pi(i),w'}).
$
\end{lemma}

Note that this implies the desired upper bound of $O(C^2)\cdot
\ed(x,y)$ as required for
$\cad_w$ to $O(C^2)$-align-approximate $\ed$.
  
\begin{proof}
We first introduce a bit of useful notation, $\ed_\pi(i,w)$, which,
intuitively, is the edit distance of the alignment $\pi$ from
$[i,i+w)$ to $[\piComp(i),\piComp(i+w))$:

\begin{align*}
\ed_\pi^i(i,w) &\triangleq 
\left|\{i'\in[i,i+w)\mid \pi(i')=\bot\}\right|. \\
\ed_\pi^j(i,w) &\triangleq 
\left|\{j\in[\piComp(i),\piComp(i+w)) \mid \pi^{-1}(j) = \bot \}\right|. \\
\ed_\pi(i,w) &\triangleq \ed_\pi^i(i,w)+\ed_\pi^j(i,w).
\end{align*}

We rewrite $\cad_w$ in terms of $\cad_{w,t}^*$ and $\Gamma$ as
follows, where $\constg$ is a constant from the start of the section,
still to be determined.  Recall that $\one{\cad_{w,c}(I,J)\ge
  c}=\one{\cad_{w,t}(I,J)\ge c}$, where $t$ is the rounding down of
$c/\constg$ to a power of $\wbase$ \ns{we slightly changed that def
  right?}\aanote{right, but only relevant for small regime. here's
  more of a informal statement}. Considering the ``large'' and
``not large'' cases for $t$ separately, we can further rewrite the
``large'' case as:
$$
\sum_{\substack{c\in \Sw\\c/\constg\ge w/\wbase^2}}
c\cdot \one{\cad_{w,c}(I,J)\ge c}\le 2\cdot\Gamma(I,J)\cdot
\one{\Gamma(I,J)\ge w/\wbase^2\cdot \constg}.
$$
Hence, using Lemma~\ref{lem:gammaUpperBound}, for $\constg$
at least twice the implicit constant from the lemma (in front of $w/\wbase^2$), we obtain that:
$$
\sum_{\substack{c\in \Sw\\c/ \constg\ge w/\wbase}}
c\cdot \one{\cad_{w,c}(X_{i,w},Y_{\pi(i),w})\ge c}\le
O\left(\sum_{k\in[w]}\tfrac{1}{w'}\distG(X_{i+k,w'},Y_{\pi(i+k),w'})+  C\sum_{k\in[w]} |\piComp(i+k+1)-\piComp(i+k)-1|\right).
$$
In particular, by summing over all $i$, and noting that each
$\distG(X_{i,w'},Y_{\pi(i),w'})$ appears at most $w=w'\wbase$ times
(and same with the absolute difference terms):
\begin{align}
\label{eqn:contributionDlarge}
  \sum_i\sum_{\substack{c\in \Sw\\c/\constg\ge w/\wbase}}
  c\cdot \one{\cad_{w,c}(X_{i,w},Y_{\pi(i),w})\ge c}
  &\le
w\cdot
O\left(\sum_{i\in[n]}\tfrac{1}{w'}\distG(X_{i,w'},Y_{\pi(i),w'})+C|\piComp(i+k+1)-\piComp(i+k)-1|\right)
\\
&\le
w\cdot
O\left((1+C)\sum_{i\in[n]}\tfrac{1}{w'}\distG(X_{i,w'},Y_{\pi(i),w'})\right),
\end{align}
as $\distG(X_{i,w'},Y_{\pi(i),w'})=w'$ whenever $\pi(i)=\bot$ \ns{maybe change the $O()$ notation to constant 2, to make clear use of geometric series?}.

We now upper bound the ``not large'' regime, which is significantly
more involved. For fixed $c$ with $c/\constg< w/\wbase$, recall that
$t$ is the
round-down of $c/\constg$; except if $c/\constg<\wbase$, when we set $t=\wbase$.
Hence we have
$\one{\cad_{w,c}(I,J)\ge c}=\one{\cad_{w,t}(I,J)\ge c}\le
\one{\cad_{w,t}^*(I,J)\ge c}\le \one{D_t(I,J)\ge t\constg}+\one{D_\gamma(I,J)\ge c}$, where $D_t$ is
as defined in Eqn.~\eqref{eqn:defDt}.
We can bound:
\begin{equation}
\label{eqn:cadWub}
\sum_{\substack{c\in \Sw\\c/\constg< w/\wbase^2}}
c\cdot \one{\cad_{w,c}(I,J)\ge c}
\le
O\left(\max_{t:D_t(I,J)\ge t\constg}D_t(I,J)+D_{\wbase}(I,J)\right)
\le
O\left(\sum_{t}\left[D_t(I,J)-\constg t/2\cdot \one{t>\wbase}\right]^+\right),
\end{equation}
where $[x]^+=\max\{0,x\}$, and recalling that there are $O(\log_\wbase
n)=O(1)$ different possible $t$'s considered. %

We upper bound $D_t(X_{i,w},Y_{j,w})$ for any fixed $t$ and indeces
$i\in[n]$ and $j=\piComp(i)$, by exhibiting a convenient choice for function
$A$. 
We define $A_x[k]$'s, for fixed $i$ and $t$. For $k=-\wbase$,
$A_x[k]=0$ by definition. For $k\in(-\wbase,\wbase)$, define real $\delta_k\ge
0$ to be the smallest such that $\ed_\pi^i(i+w'(k-1+\theta
A_x[k-1]),w'+w'\theta \delta_k)= w'\theta \delta_k$ (infinity if it
doesn't exist). Set
$A_x[k]=\min\{A_x[k-1]+\lceil \delta_k\rceil,T\}$. Then, if
$A_x[k]=T$, then set $A_y[k]=T$, as well as the subsequent $A_x,A_y$.
Otherwise, let $A^*_y[k]$ to be the unique integer such that
\begin{equation}
  \label{eqn:AxAyRelation}
\piComp(i+w'(k+\theta A_x[k]))
\in \piComp(i)+w'(k+\theta A_y^*[k])+[0,w'\theta).
\end{equation}

and set $A_y[k] = \min\{A^*_y[k],T\}$.
Note that for $t=\wbase$, we have that $w'\theta=1$ and hence $\piComp(i+w'(k+\theta A_x[k]))
=\piComp(i)+w'(k+\theta A_y[k])$.

Finally, set $A_x[\wbase]=A_y[\wbase]=\max\{A_x[\wbase-1],A_x[\wbase-1]\}$.

\begin{claim}
Functions $A_x,A_y$ are non-decreasing and
  $\normo{A[\wbase]}\le
O(\tfrac{\ed_\pi(i-w,3w)}{w'\theta}+\wbase)$. Futhermore, when
$t=\wbase$, we have that $\normo{A[\wbase]}\le
O(\tfrac{\ed_\pi(i-w,3w)}{w'\theta})$.
\end{claim}
\begin{proof}
While $A_x$ is non-decreasing by construction, we need to prove that
$A_y$ is non-decreasing. We have that, as long as $A_k[x]<T$, for $i_k=i+w'(k+\theta A_x[k])$
and $i_{k-1}=i+w'(k-1+\theta A_x[k-1])$, and noting that
$i_k-i_{k-1}=w'+w'\theta \lceil \delta_k\rceil$:
\begin{align*}
\piComp(i_k)
-
\piComp(i_{k-1})
&\ge
i_k-i_{k-1}
-\ed_\pi^i(i_{k-1},i_k-i_{k-1})
\\
&\ge
w'+w'\theta \lceil \delta_k\rceil-w'\theta \lceil \delta_k\rceil
\\
&=
w'.
\end{align*}

Hence $A_y[k]$ must also be non-decreasing.

To prove the bound on $\normo{A}$, we first note that the bound is
immediate if $\ed_\pi(i-w,3w)\ge \wbase^2 t/10$: then $\normo{A}\le
2T=2\wbase^3=2\tfrac{\wbase^2 t}{w' \theta}$.
Suppose for the rest that $\ed_\pi(i-w,3w)\le
\tfrac{\wbase^2}{10}t$ (and hence $\ed_\pi^i(i-w,3w)$ as well). We now
show by induction on $k$ that $A_x[k]\le T/5$ and $\delta_k\le T$ ---
satisfied for $k=-\wbase$ by definition. As for inductive step:
\begin{align}
  \ed_\pi^i(i-w,w+w'(k+\theta A_x[k]))
  \nonumber
&=
\sum_{l=-\wbase+1}^{k} \ed_\pi^i(i+w'(l-1+\theta A_x[l-1]),w'+w'\theta
\lceil \delta_l\rceil)
  \nonumber
\\
&\ge
\sum_{l=-\wbase+1}^{k} w'\theta
\delta_l
  \nonumber
\\
&\ge
\sum_{l=-\wbase+1}^{k} w'\theta
\lceil\delta_l\rceil-w'\theta
  \nonumber
\\
&=w'\theta(A_x[k]-\wbase-k),
\label{eqn:Axkk}
\end{align}

and thus $A_x[k]\le \tfrac{\ed_\pi^i(i-w,w+w'(k+\theta
  A_x[k]))}{w'\theta}+\wbase+k\le \wbase^3/10+2\wbase\le T/5$. This also
implies $w'\theta A_x[k]\le t/\wbase\cdot T/5\le t\wbase^2/5<w'$. Now we
also show that $\delta$'s are
$<T$. Indeed, define $f(\delta)\triangleq \ed_\pi^i(i+w'(k-1+\theta
A_x[k-1]),w'+w'\theta \delta)$, and note that $f(0),f(T/10)\le
\ed_\pi^i(i-w,3w)\le t\wbase^2/10$. At the same time the function
$g(\delta)=w'\theta \delta$ grows from 0 to $g(T/10)=w'\theta T/10\ge
t\wbase^2/10$. Hence the non-decreasing functions $f,g$ must intersect
somewhere at $\delta\in[0,T/10]$ (with $w'\theta \delta$ an integer).

From the above we have that $A_x[\wbase-1]\le
\tfrac{\ed_\pi^i(i-w,3w)}{w'\theta}+\wbase$. As for $A_y[\wbase-1]$,
we have that:
$$
A_y[\wbase-1]\le \tfrac{\piComp(i+w+w'\theta A_x[\wbase-1])
-(j+w-w')}{w'\theta}\le \tfrac{\piComp(i+w)+w'\theta A_x[\wbase-1]
  -\piComp(j+w)+\ed_\pi^j(i+w-w',w)}{w'\theta}\le O(\tfrac{\ed_\pi(i-w,3w)}{w'\theta}+\wbase).
$$
The conclusion follows since
$|A[\wbase]|_1=2\max\{A_x[\wbase-1],A_y[\wbase-1]\}$.

In the case when $t=\wbase$, we note that $w'\theta=1$, and hence it
is enough to consider $\delta_k$ to be integers only. Thus in
Eqn.~\eqref{eqn:Axkk} we don't need the negative term (that comes only
from rounding $\delta_k$), and hence $A_x[\wbase-1], A_y[\wbase-1]\le
O(\tfrac{\ed_\pi^i(i-w,3w)}{w'\theta})$, from which we obtain the conclusion.
\end{proof}

Since $A$ depends on $i$ and $t$, we denote it as $A^{i,t}$ below. For
each $i\in[n]$, we also define $t_i\in \{\wbase,\ldots w/\wbase^2\}$ as the
round-down of $\ed_\pi(i-w,3w)$ (i.e., $t_i\le \ed_\pi(i-w,3w)<
\wbase t_i$), unless $\ed_\pi(i-w,3w)<\wbase$, in which case $t_i=\wbase$.
Also let $\theta_i$ be the $\theta$ from $D_{t_i}$,
i.e., $\theta_i=t_i/w$.

By the claim above, we have that for any $i,t$:
\begin{equation}
  \label{eqn:AitBnd}
  w'\tfrac{t}{w}\cdot |A^{i,t}[\wbase]|_1\le O(\ed_\pi(i-w,3w)+t\cdot \one{t>\wbase}).
\end{equation}
We also observe that
\begin{equation}
  \label{eqn:sumED}
  \tfrac{1}{w}\sum_i
  \ed_\pi(i-w,3w)
  =\tfrac{1}{w}\sum_{j\in[3w]}\sum_{i=0}^{n/w-3}\ed_\pi(iw-w+j,3w)
  \le \tfrac{3w}{w}\cdot O\left(|\{i\in[n]\mid \pi(i)=\bot\}|\right)\le O(k_\pi).
\end{equation}

Hence, for $\constg$ sufficiently large constant, and
$t$ ranging over $\{\wbase,\ldots,
w/\wbase^2\}$, we can bound the total
``not-large'' contribution, $\sum_{i\in [n]}\tfrac{1}{w}\sum_{\substack{c\in \Sw\\c/\constg< w/\wbase}}
c\cdot \one{\cad_{w,c}(X_{i},Y_{\pi(i)})>c}$ as follows:
\begin{align*}
&\stackrel{\eqref{eqn:cadWub}}{\le} 
2\cdot|\{i:\pi(i)=\bot\}|+
\tfrac{1}{w}  \sum_{i\in[n]} 
\sum_t\left[D_t(X_{i},Y_{\piComp(i)})-\constg t/2\cdot \one{t>\wbase}\right]^+
\\
&\stackrel{\eqref{eqn:defDt}}{\le} 
2k_\pi+
\sum_{\substack{i\in[n]\\t}} \tfrac{1}{w}\left[\tfrac{1}{T}\sum_{\substack{k\in[-\wbase,\wbase)\\\Delta\in[3T]}} \distG (X_{i + w'(k + \theta_t
(\Delta+ A_x^{i,t}[k]))},Y_{\piComp(i) + w'(k + \theta_t(\Delta +
  A_y^{i,t}[k]))})+w'\theta_t|A^{i,t}[\wbase]|_1-\tfrac{\constg
      t\cdot \one{t>\wbase}}{2}\right]^+  
  \\
&\overset{\eqref{eqn:AitBnd}}{\le}
O(k_\pi)+
\tfrac{1}{w}\sum_{\substack{i\in[n]\\t}}\left[\tfrac{1}{T}\sum_{\substack{k\in[-\wbase,\wbase)\\\Delta\in[3T]}}
  \distG (X_{i + w'(k + \theta_t
(\Delta+ A_x^{i,t}[k]))},Y_{\piComp(i) + w'(k + \theta_t(\Delta +
  A_y^{i,t}[k]))})+O(\ed_\pi(i-w,3w))-\tfrac{\constg t\cdot \one{t>\wbase}}{3}\right]^+
\\
&\overset{\eqref{eqn:sumED}}{\le}
O(k_\pi)+
O(k_\pi)+\tfrac{1}{w}\sum_{\substack{i\in[n]\\t}}\left[\tfrac{1}{T}\sum_{\substack{k\in[-\wbase,\wbase)\\\Delta\in[3T]}}
  \distG (X_{i+w'\theta_t \Delta +
  w'(k + \theta_t A_x^{i,t}[k])},Y_{\piComp(i) + w'\theta_t \Delta +w'(k + \theta_t
    A_y^{i,t}[k])})-\tfrac{\constg t\cdot \one{t>\wbase}}{3}\right]^+
\\
&\overset{\eqref{eqn:AxAyRelation}}{\le}
O(k_\pi)+\tfrac{1}{w}\sum_{\substack{i\in[n]\\t}}\left[\tfrac{1}{T}\sum_{\substack{k\in[-\wbase,\wbase)\\\Delta\in[3T]}}
   \distG (X_{w'\theta_t \Delta +i+
  w'(k + \theta_t A_x^{i,t}[k])},Y_{w'\theta_t \Delta +\piComp(i+
  w'(k + \theta_t A_x^{i,t}[k]))})+\left(6\wbase\cdot Cw'\theta_t-\tfrac{\constg t}{3}\right)\cdot\one{t>\wbase}\right]^+
\\
&\le
O(k_\pi)+\tfrac{1}{wT}
\sum_{\stackrel{i\in[n]}{t}}\sum_{\substack{k\in[-\wbase,\wbase)\\\Delta\in[3T]}}  
  \distG (X_{w'\theta_t \Delta +i+
  w'(k + \theta_t A_x^{i,t}[k])},Y_{\piComp(w'\theta_t \Delta +i+
  w'(k + \theta_t A_x^{i,t}[k]))})
\\
&\quad+
\tfrac{O(C)}{wT}\sum_{\stackrel{i\in[n]}{t}}
\sum_{\substack{k\in[-\wbase,\wbase)\\\Delta\in[3T]}}
  \left|\piComp(w'\theta_t \Delta +i+
  w'(k + \theta_t A_x^{i,t}[k]))-(w'\theta_t \Delta+\piComp(i+
  w'(k + \theta_t A_x^{i,t}[k])))\right|,
\end{align*}

where the last two inequalities are due to triangle inequality for
$\distG$ and $\distG(Y_j,Y_{j+1})\le C$, as well as for $\tau$
satisfying $6\wbase\cdot Cw'\theta_t\le\tfrac{\constg t}{3}$, i.e.,
$\tau\ge 18C$. To estimate the last term,
we note that, again since $w'\theta T\le w'$:
\aanote{not quite true since $\piComp$ can be a lot to the right}
\begin{align*}
  S
  &\triangleq
  \tfrac{1}{wT}
    \sum_{\substack{i\in[n]\\t}}
\sum_{\substack{k\in[-\wbase,\wbase)\\\Delta\in[3T]}}
  \left|\piComp(w'\theta_t \Delta +i+
  w'(k + \theta_t A_x^{i,t}[k]))-(w'\theta_t \Delta+\piComp(i+
  w'(k + \theta_t A_x^{i,t}[k])))\right|
  \\
  &\le
  \tfrac{1}{wT}
      \sum_{\substack{i\in[n]\\t}}\sum_{\substack{k\in[-\wbase,\wbase)\\\Delta\in[3T]}}
\ed_\pi^j(i+w'(k+\theta_t A_x^{i,t}[k]),w'\theta_t\Delta).
  \\
  &\le
  \tfrac{1}{wT}\sum_{\Delta\in[3T]}
  \sum_{i\in[n],t}
2\cdot\ed_\pi(i-w,3w).
\\
&\le
  O(\log_\wbase n\cdot\ed_\pi(1,n))
  \\
  &\le O(k_\pi).
\end{align*}

Finally, overall we have, using that $A_x^{i,t}[k]\in [0,T]$ and hence
we can absorb it into $\Delta$-summation, now over $\Delta\in[4T]$
(the main reason we are using $\Delta$ to start with):
\begin{align*}
\sum_{i\in [n]}\tfrac{1}{w}\sum_{\substack{c\in \Sw\\c/\constg< w/\wbase}}
c\cdot \one{\cad_{w,c}(X_{i},Y_{\pi(i)})>c}
&\le
O(Ck_\pi)+\tfrac{1}{w'}\sum_{t}\sum_{\substack{k\in[-\wbase,\wbase)\\\Delta\in[4T]}} \tfrac{1}{\wbase T}\sum_{i\in[n]}
    \distG (X_{i+
  w'k + \tfrac{t}{\wbase} \Delta},Y_{\piComp(i+
  w'k + \tfrac{t}{\wbase} \Delta)}),
\\
&\le
O(Ck_\pi)+\tfrac{1}{w'}\sum_{t} 8\sum_i \distG(X_i,Y_{\pi(i)})
\\
&\le
O(\log_\wbase n\cdot Ck_\pi)
\\
&=
O(Ck_\pi).
\end{align*}

We complete the proof of Lemma~\ref{lm::cad_ub}, combining the bound
for ``large'' distance regime (Eqn.~\eqref{eqn:contributionDlarge},
scaled by $1/w$), and the above ``not large'' distance regime, we
obtain that:
$$\sum_{i\in [n]}\tfrac{1}{w}\cad_w(X_{i,w},Y_{\pi(i),w})\le
O(Ck_\pi)\le O(C)\cdot 
    \sum_{i\in[n]}\tfrac{1}{w'}\distG_{w'}(X_{i,w'},Y_{\pi(i),w'}),
    $$
    which completes the proof of Lemma~\ref{lm::cad_ub}.
\end{proof}

\subsection{$\cad_{w,c}$ are metrics}

Finally, we prove that $\cad_{w,c}$ is a metric for each $c$. Given
how we reduce them to $\Gamma,\cad_{w,t}$, it is enough to prove
metricity for the latter two.

\paragraph{$\Gamma$ metric.}
Identity follows by definition. Symmetry follows from the fact that
the graph $H$ is symmetric.

It remains to prove the triangle inequality. Consider three intervals
$I,J,K$. We want to prove that:
$$
\Gamma(I,K)\le \Gamma(I,J)+\Gamma(J,K).
$$

Note that, for the pair $(I,J)$, the optimal path in graph $H$
corresponds to an alignment $A_{IJ}:[-m,m)\to [-m,m)\cup\{\bot\}$. In
particular, for each $p\in[-m,m)$, consider the edge which increases $p$:
if it's the diagonal edge $V_{p,q}\to V_{p+1,q+1}$ then $A_{IJ}(p)=q$;
otherwise $A_{IJ}(p)=\bot$. The number of used gap edges is then $2\cdot
|\{p:A_{IJ}(p)=\bot\}|$. Hence:
  $$
  \Gamma(I,J)= 4
  \tfrac{w/m}{w'}\sum_{\substack{k\in [-m,m)\\A_{IJ}(i+\tfrac{w}{m}k)\neq \bot}}\distG(I_{w'}+\tfrac{w}{m}k,J_{w'}+\tfrac{w}{m}\cdot A_{IJ}(k))+8\cdot |\{p:A_{IJ}(p)=\bot\}|+2w/\wbase^2.
  $$

  Similarly we can extract $A_{JK}$. We now define an alignment $A_{IK}$
  from $I$ to $K$ as $A_{IK}(p)=A_{JK}(A_{IJ}(p))$ (and $\bot$ if either
  functions has value $\bot$). Then we have that:
  \begin{align*}
  \Gamma(I,K)&\le 4
  \tfrac{w/m}{w'}\sum_{\substack{k\in [-m,m) \\ A_{IK}(i+\tfrac{w}{m}k)\neq \bot}}\distG(I_{w'}+\tfrac{w}{m}k,K_{w'}+\tfrac{w}{m}\cdot A_{IK}(k))+8\cdot |\{p:A_{IK}(p)=\bot\}|+2w/\wbase^2
  \\
  &\le4\tfrac{w/m}{w'}\sum_{\substack{k\in [-m,m)\\ A_{IK}(i+\tfrac{w}{m}k)\neq \bot}}\distG(I_{w'}+\tfrac{w}{m}k,J_{w'}+\tfrac{w}{m}\cdot A_{IJ}(k))+8\cdot |\{p:A_{IK}(p)=\bot\}|
  \\
  &+
  4\tfrac{w/m}{w'}\sum_{\substack{k\in [-m,m)\\ A_{IK}(i+\tfrac{w}{m}k)\neq \bot}}\distG(J_{w'}+\tfrac{w}{m}\cdot A_{IJ}(k),K_{w'}+\tfrac{w}{m}\cdot A_{IK}(k))+2w/\wbase^2
  \\
  &\le4\tfrac{w/m}{w'}\sum_{\substack{k\in [-m,m)\\ A_{IJ}(k)\neq \bot}}\distG(I_{w'}+\tfrac{w}{m}k,J_{w'}+\tfrac{w}{m}\cdot A_{IJ}(k))+8\cdot |\{p:A_{IJ}(p)=\bot\}|+2w/\wbase^2
  \\
  &+
  4\tfrac{w/m}{w'}\sum_{\substack{k\in [-m,m) \\ A_{IK}(i+\tfrac{w}{m}k)\neq
      \bot}}\distG(J_{w'}+\tfrac{w}{m}\cdot A_{IJ}(k),K_{w'}+\tfrac{w}{m}\cdot
  A_{JK}(A_{IJ}(k)))+8\cdot |\{p:A_{JK}(p)=\bot\}|
  \\
  &\le \Gamma(I,J)+\Gamma(J,K).
  \end{align*}
  
\paragraph{$\cad_{w,t}$ metric.}
Identity and Symmetry are trivial (for identity, we use
$A=(\vec{0},\vec{0})$ and for symmetry we switch the coordinates of
$A$).

To show triangle inequality, fix $I, J, K \in \IC$. We prove that
$\cad_{w,c}(I,K) \leq \cad_{w,c}(I,J) + \cad_{w,c}(J,K)$. Fix $A, A'$
which minimize $\cad_{w,c}^*(I,J),\cad_{w,c}^*(J,K)$. Define $A'' = A +
A'$ (i.e. $A''_x[l] = A_x[l] + A'_x[l]$ and $A''_y[l] = A_y[l] +
A'_t[l]$). Note that if we obtain $A''_x[l]\ge T$ or $A''_y[l]\ge T$, then we
already have that $\cad_{w,c}(I,J) + \cad_{w,c}(J,K)\ge w' \theta T\ge \cad_{w,c}
(I,K)$, and hence we can assume that $A''_X[l],A''_y[l]<T$ below. Last, we use the shorthand $\normo{A} = \normo{A[\wbase]}$, meaning the total number of shifts of $A$.

\begin{align*}
	\cad_{w,c}^*(I,K) \leq& \sum_{l \in [-\wbase,\wbase)} \tfrac{1}{T} \sum_{\Delta \in [3T - \normo{A''[l]}]} \distG (I_{w'} + w'(l + \theta (\Delta+ A''_x[l])),K_{w'} + w'(l + \theta(\Delta + A''_y[l]))) + \normo{A''}\theta w' \\
		\leq& \sum_{l} \tfrac{1}{T} \sum_{\Delta \in [3T -
                    \normo{A''[l]}]} \Big( \distG (I_{w'} + w'(l + \theta (\Delta+ A''_x[l])),J_{w'} + w'(l + \theta(\Delta + A'_x[l]+A_y[l]))) \\
		&\qquad+ \distG (J_{w'} + w'(l + \theta(\Delta + A'_x[l]+A_y[l])),K_{w'} + w'(l + \theta (\Delta+ A''_y[l])))\Big) + \normo{A''}\theta w'\\
		=& \sum_{l} 
		\tfrac{1}{T} \sum_{\Delta \in [3T - \normo{A''[l]}]+A'_x[l]} \distG (I_{w'} + w'(l + \theta (\Delta + A''_x[l] - A'_x[l])),J_{w'} + w'(l + \theta(\Delta + A_y[l]))) \\
		&+ \tfrac{1}{T}\sum_{\Delta \in [3T - \normo{A''[l]}]+A_y[l]}\distG (J_{w'} + w'(l + \theta (\Delta + A'_x[l])),K_{w'} + w'(l + \theta(\Delta + A''_y[l]-A_y[l]))) + \normo{A''}\theta w'  \\
		\leq& \sum_{l} \tfrac{1}{T}\sum_{\Delta \in [3T - \normo{A[l]}]} \distG (I_{w'} + w'(l + \theta (\Delta + A_x[l])),J_{w'} + w'(l + \theta(\Delta + A_y[l]))) \\
		\qquad &+ \tfrac{1}{T}\sum_{\Delta \in [3T - \normo{A'[l]}]}\distG (J_{w'} + w'(l + \theta (\Delta + A'_x[l])),K_{w'} + w'(l + \theta(\Delta + A'_y[l]))) +\normo{A}\theta w'+\normo{A'}\theta w' \\
		\leq& \cad_{w,c}^*(I,J) + \cad_{w,c}^*(J,K)
\end{align*}

where the 2nd step is triangle inequality of $(\IC,\distG)$, 3rd step is
change of variables, 4th step is triangle inequality of
$(\N^2,\ell_1)$, and last step is by definition. The conclusion follows
from the fact that all $\cad^*$'s are thresholded at the same threshold.

%% file: fastSampling.tex
\section{Fast Data Structures}
\label{sec:fastSampling}

In this section, all reals have $O(\log^2 n)$ of precision (in fixed
point representation). 

\begin{theorem}[Theorem \ref{thm::ds_sum}, restated]
Given a set of numbers $a_1,\ldots a_n\ge 0$ in a sparse
representation (i.e., as a set $S$ of $i$'s with $a_i>0$), there exists a data structure $D$ supporting the following queries:
\begin{itemize}
\item
Given some $i<j$ compute $\sum_{k=i}^j a_k$. The query complexity is $O(\log n)$.
\item
Given some $i$ and $\zeta>0$, compute $j\in S, j\le i$ such that
$\sum_{k=j-1}^i a_k<\zeta\le \sum_{k=j}^i a_k$ (or output that none
exists). The query complexity is $O(\log n)$.
\end{itemize}

The preprocessing time is $O(|S|\log n)$.
\end{theorem}
\begin{proof}
  First of all, at preprocessing time, we can compute all the prefix
  sums $\sigma_i=\sum_{k=1}^i a_k$ for $i\in S$. Then we build a predecessor
  data structure~\cite{CLRS} on $S$ as well as on the reals $\sigma_i$ for $i\in
  S$. Then the first type of queries is answered by two predecessor queries
  (and taking their difference). The second type of queries is
  answered by a single query to the (second) predecessor data
  structure. The preprocessing and query runtime follow immediately
  from standard predecessor data structures.

\end{proof}

\begin{theorem}[Theorem \ref{thm::ds_sampling}, restated]
Given a set of numbers $a_1,\ldots a_n\ge 0$ in a sparse
representation (i.e., as a set $S$ of $i$'s with $a_i>0$)
there exists a data structure $D$ supporting the following queries:
\begin{itemize}
\item
  Sample index $i\in[n]$ from the distribution $\{a_i/\sum_j a_j\}_i$ in time $O(\log n)$;
\item
  Given $k\ge 1$, produce a set $S\subseteq [n]$ that includes each
  $i$ with probability $\min\{ka_i,1\}$ independently.  The runtime is
  $O((1+k\cdot\sum_{i=1}^n a_i)\log n)$ in expectation.
\item
  Given an interval $[s,t]\subset [n]$, and $k\ge 1$, produce a set $S\subseteq [n]$ that includes each
  $i\in[s,t]$ with probability $\min\{ka_i,1\}$ independently.  The runtime is
  $O((1+k\cdot\sum_{i=s}^t a_i)\log^2 n)$ in expectation.
\end{itemize}

Furthermore, the preprocessing time is $O(|S|\cdot \log^{O(1)}n)$.
\end{theorem}

We note that similar data structures were also developed in
\cite{bringmann2014sampling}; we include our proof below for
completeness of the algorithm.

\begin{proof}
Let $m=|S|$. First, note that we can assume wlog that $m = n$. When $m<n$, we pre-process vector $b \in \R^m$ consisting of the non-zero entries in $a$ in an increasing order. Then, for the third query type, define the semi-monotone function $\Phi: [n+1] \rightarrow [m+1]-1$, which maps each $i \in [n+1]$ to the number of non-zero entries whose index is strictly smaller than $i$ in $a$. For that task, we build during preprocessing a BST to calculate $\Phi(i)$ (in additive time $\tO(m)$), and when queried on interval $[s,t]$, we use interval $[\Phi(s)+1,\Phi(t+1)]$ instead (with $O(\log n)$ additive query overhead).

At preprocessing, precompute all partial sums $p_i=\sum_{j=1}^i a_j$ for
all $i\in[n]$. Let $p_0=0$. Also, assuming $n$ is a power of 2
(otherwise we can pad to nearest power), we build a balanced binary
search tree with each leaf corresponding to an index $i\in[n]$. Each
node $v$ of the tree stores the sum of leafs in the subtree,
$\sigma_v$. Also, sort all $a_i$ in increasing order.

\paragraph{Query of the first type.} Pick a random uniform $r\in[0,1]$
(say, to $O(\log^2)$ bits of precision) and
perform a binary search on $r\cdot \sum_j a_j$ in the set $p_0,p_1,\ldots p_n$. If $p_i\le
r\sum_j a_j\le p_{i+1}$, then output $i+1$. Note that the probability that output
is $i$ is precisely $a_i/\sum_j a_j$.

\paragraph{Query of the second type.} The algorithm works as
follows. First, using the sorted list, find all $a_i$'s such that
$a_i\ge 0.5/k$ (i.e., the ones that are included with set $S$ with
probability $\ge 1/2$). Call this set $L$. Sample each $i\in L$ into
set $S$ accordingly (taking $O(|L|)$ time).

Start at the root $r$, generate a integer $q_r$ from the
Poisson distribution with expectation $2k\sigma_r$. Then, we proceed
recursively as follows: for a node $v$ with integer $q_v$, for children with
sums $\sigma, \sigma'$, pick random integers $q,q'$ such that: 1)
$q+q'=q_v$, and 2) assuming that $q_v$ is from distribution
$\Poi(2k\sigma_v)$, then $q,q'$ are independently from
$\Poi(2k\sigma),\Poi(2k\sigma')$ respectively. Recurse into each child
with $>0$ integer $q$. At a leaf $i\in[n]\setminus L$, if $q_i>0$ (and hence
$q_i\ge 1$), then include $i$ into set $S$ with probability
$\tfrac{a_ik}{1-e^{-2a_i/k}}$. (Indeces in $L$ are ignored here.)

We now briefly argue that the above produces the desired set. First of
all note that the procedure on the tree generates $q_i$ from $\Poi(a_i
k)$ distribution independently (using Poission's distribution
property). We note that
$\tfrac{1-e^{-2x}}{x}=2-2(2x)/2!+2(2x)^2/3!-...\ge 1$ for $x\le 1/2$,
hence the probability $\tfrac{a_ik}{1-e^{-2a_i/k}}$ is indeed less
than 1 for $i\in [n]\setminus L$. For such $i$, we indeed include it
into $S$ with probability
$$
\Pr[i\in S]=\Pr[q_i\ge1]\cdot \tfrac{a_ik}{1-e^{-2a_i/k}}=
a_ik<1.
$$

We now bound the query time. First, we note that $L$ has size at most
$2k\cdot \sum_{i=1}^n a_i$ (since each included $i$ satisfies $a_i\ge1/2k$). Second, we note that $\E[q_r]=O(k)$,
and the tree procedure has runtime at most $O(\log n)$ factor of the
number of leafs with $q_i>0$, of which there are exactly $q_r$. Hence
runtime is $O(k\log n)$ in expectation.

\paragraph{Query of the third type.} We augment the data structure
$D$ from above as follows. Recall the dyadic intervals on $[n]$ are intervals
of the form $[j2^i+1,(j+1)2^i]$, where $i\in\{0,\ldots \log n\}$ and
$j\in\{0,\ldots n/2^i-1\}$. For each dyadic interval $I=[u,v]$, we build a
data structure $D_I$ on the numbers $a_u,\ldots a_v$. For a query
interval $[s,t]$, we decompose $[s,t]$ into $O(\log n)$ dyadic
intervals, and perform query of the second type in each of these.
\end{proof}

\begin{theorem}[Theorem \ref{thm::ds_matrix}, restated]
  Suppose we are given an $m\times n$ matrix $A$, where $m\le n$,
  given in the sparse form, i.e., as a set $S$ of non-zero
  entries. Also, suppose we are given $n$ intervals
  $[s_i,t_i]\subseteq [n]$ such that both $\{s_i\}_i$ and $\{t_i\}_i$
  are non-decreasing. Given $\gamma>0$, we can find all pairs
  $(i,\chi)\in [n]\times[m]$ such that
  $\sum_{k=s_i}^{t_i}A_{\chi,k}\ge \gamma$. The runtime is $\tilde
  O\left(|S|+n+\tfrac{\sum_{\chi,i} \sum_{k=s_i}^{t_i}A_{\chi,k}}{\gamma}\right)$.
\end{theorem}

\begin{proof}
Suppose $n$ is a power of 2; otherwise, just pad with zero's until the
closest power of 2.  Recall the dyadic intervals on $[n]$ are
intervals of the form $[j2^i+1,(j+1)2^i]$, where $i\in\{0,\ldots \log
n\}$ and $j\in\{0,\ldots n/2^i-1\}$. We call $i$ to be the level of
the dyadic interval.

First, for each dyadic interval $I$, we store two sorted lists $S_I,
T_I$. For each given interval $[s_i,t_i]$, we decompose it into at
most $2\log n$ dyadic intervals $I_1,I_2,\ldots$. For each such dyadic
interval $I_k$, we add $s_i$ to $S_{I_k}$ and $t_i$ to $T_{I_k}$. Note
that this takes $O(n\log n)$ time.

We also store a dictionary data structure (e.g., a BST) storing, for
each pair $(\chi, I)$, where $\chi\in[m]$ and $I$ a dyadic interval, the
quantity $\tau_{\chi, I}=\sum_{k\in I} A_{\chi,k}$. In particular,
we store only the {\em non-zero} $\tau_{\chi, I}$. We can compute
this efficiently as follows: 1) initialize an empty dictionary; 2)
enumerate through all $(\chi, k)\in S$ and add $A_{\chi,k}$ to all
$\tau_{\chi,I}$ where $k\in I$; 3) if some $\tau_{\chi,I}$ is not
yet in the dictionary, create an entry for it first. Note that the
runtime is $O(|S|\log n)$.

The rest of the algorithm proceeds as follows. Using the dictionary
structure, find all pairs $(\chi, I)$ of $\chi\in[m]$ and dyadic interval 
$I=[u,v]$ such that $\tau_{\chi,I}\ge \tfrac{\gamma}{2\log n}$. For each
one of them, using
$S_I$, let $s_i\in S_I$ be the max $s_i\in S_I$. Similarly, let $t_j\in
T_I$ be the min. For each interval
index $l$ such that $j\le l\le i$ and $s_l\in S_I$ (equivalently
$t_l\in T_I$), we compute $\sum_{k=s_l}^{t_l} A_{\chi,k}$ (using the dyadic
interval sums) in $O(\log n)$ time. We output $(l,\chi)$ if the resulting sum
is $\ge \gamma$. At the end we remove the duplicate $(l,\chi)$ if multiple
copies have been output.

We now argue correctness and runtime. Consider an interval $[s_l,t_l]$
with mass at least $\gamma$ for some $\chi\in[m]$. Then it can be
decomposed into $\le 2\log n$ dyadic intervals, at least one of which,
say $I=[u,v]$, has to have mass $\gamma_{\chi,I}\ge
\tfrac{\gamma}{2\log n}$. Hence, when we consider the pair $(\chi, I)$
where the dyadic interval $I=[u,v]$, we will have that $s_l\le s_i\le
u$ and $v\le t_j\le t_i$ (by the definition of $i,j$). Hence we will
output the pair $(l,\chi)$.

Let's argue runtime. First of all, the number of pairs $(\chi,I)$ with
non-zero $\tau$ is at most $|S|\cdot\log n$ as each $A_{\chi,i}$ can
contribute to that many pairs $(\chi,I)$. Hence enumerating the
dictionary takes $O(|S|\cdot \log n)$ time. For each pair $(\chi,I)$
with $\tau_{\chi,I}>0$, and each interval index $l\in[n]$ that we end up checking, we
have that $\sum_{k=s_l}^{t_l} A_{\chi,k}\ge \gamma/2\log
n$. Furthermore any such $(l,\chi)$ can be considered only $O(\log
n)$ times (as this is how many times $l$ appears in the data
structures $S_I,T_I$ over all dyadic intervals $I$). Hence the total
number of pairs $(l,\chi)$ for which we estimate the precise mass is
upper bounded by:
$$
O(\log n)\cdot\sum_\chi \sum_l\tfrac{\sum_{k=s_l}^{t_l} A_{\chi,k}}{\gamma/2\log
n}.
$$
Hence total time spend in this phase is $O(\log^3
n)\tfrac{\sum_{\chi,i} \sum_{k=s_i}^{t_i}A_{\chi,k}}{\gamma}$. This
completes the analysis.

\end{proof}

%% file: main.bbl
\newcommand{\etalchar}[1]{$^{#1}$}
\def\cprime{$'$} \def\cprime{$'$}
\begin{thebibliography}{AHWW16}

\bibitem[AHWW16]{abboud2016simulating}
Amir Abboud, Thomas~Dueholm Hansen, Virginia~Vassilevska Williams, and Ryan
  Williams.
\newblock Simulating branching programs with edit distance and friends: or: a
  polylog shaved is a lower bound made.
\newblock In {\em Proceedings of the forty-eighth annual ACM symposium on
  Theory of Computing}, pages 375--388. ACM, 2016.

\bibitem[AK12]{AK-smoothed}
Alexandr Andoni and Robert Krauthgamer.
\newblock The smoothed complexity of edit distance.
\newblock {\em {ACM} Transactions on Algorithms}, 8(4):44, 2012.
\newblock Previously in ICALP'08.

\bibitem[AKO10]{AKO-edit}
Alexandr Andoni, Robert Krauthgamer, and Krzysztof Onak.
\newblock Polylogarithmic approximation for edit distance and the asymmetric
  query complexity.
\newblock In {\em Proceedings of the Symposium on Foundations of Computer
  Science (FOCS)}, 2010.
\newblock Full version at \url{http://arxiv.org/abs/1005.4033}.

\bibitem[And18]{andoni18-edit3}
Alexandr Andoni.
\newblock Simpler constant-factor approximation to edit distance problems,
  2018.
\newblock Manuscript, available at
  \url{http://www.cs.columbia.edu/~andoni/papers/edit/}.

\bibitem[ANSS22]{andoni2021estimating}
Alexandr Andoni, Negev~Shekel Nosatzki, Sandip Sinha, and Clifford Stein.
\newblock Estimating the longest increasing subsequence in nearly optimal time.
\newblock In {\em Proceedings of the Symposium on Foundations of Computer
  Science (FOCS)}, 2022.
\newblock Also as arXiv preprint arXiv:2112.05106.

\bibitem[AO12]{AO-edit}
Alexandr Andoni and Krzysztof Onak.
\newblock Approximating edit distance in near-linear time.
\newblock {\em {SIAM} J. Comput. (SICOMP)}, 41(6):1635--1648, 2012.
\newblock Previously in STOC'09.

\bibitem[BCAD21]{bringmann2021linear}
Karl Bringmann, Vincent Cohen-Addad, and Debarati Das.
\newblock A linear-time $n^{0.4}$-approximation for longest common subsequence.
\newblock In {\em Proceedings of International Colloquium on Automata,
  Languages and Programming (ICALP)}, 2021.

\bibitem[BCFN22]{bringmann2022almost}
Karl Bringmann, Alejandro Cassis, Nick Fischer, and Vasileios Nakos.
\newblock Almost-optimal sublinear-time edit distance in the low distance
  regime.
\newblock In {\em Proceedings of the Symposium on Theory of Computing (STOC)},
  2022.

\bibitem[BCR20]{brakensiek2020simple}
Joshua Brakensiek, Moses Charikar, and Aviad Rubinstein.
\newblock A simple sublinear algorithm for gap edit distance.
\newblock {\em arXiv preprint arXiv:2007.14368}, 2020.

\bibitem[BEG{\etalchar{+}}18]{boroujeni2018approximating}
Mahdi Boroujeni, Soheil Ehsani, Mohammad Ghodsi, MohammadTaghi HajiAghayi, and
  Saeed Seddighin.
\newblock Approximating edit distance in truly subquadratic time: Quantum and
  mapreduce.
\newblock In {\em Proceedings of the ACM-SIAM Symposium on Discrete Algorithms
  (SODA)}, pages 1170--1189. SIAM, 2018.

\bibitem[BEK{\etalchar{+}}03]{BEK+03}
Tu\u{g}kan Batu, Funda Erg\"{u}n, Joe Kilian, Avner Magen, Sofya Raskhodnikova,
  Ronitt Rubinfeld, and Rahul Sami.
\newblock A sublinear algorithm for weakly approximating edit distance.
\newblock In {\em Proceedings of the Symposium on Theory of Computing (STOC)},
  pages 316--324, 2003.

\bibitem[BES06]{BES06}
Tu\u{g}kan Batu, Funda Erg\"{u}n, and Cenk Sahinalp.
\newblock Oblivious string embeddings and edit distance approximations.
\newblock In {\em Proceedings of the ACM-SIAM Symposium on Discrete Algorithms
  (SODA)}, pages 792--801, 2006.

\bibitem[BI15]{BI15-edit}
Arturs Backurs and Piotr Indyk.
\newblock Edit distance cannot be computed in strongly subquadratic time
  (unless {SETH} is false).
\newblock In {\em Proceedings of the Symposium on Theory of Computing (STOC)},
  2015.

\bibitem[BJKK04]{BJKK04}
Ziv {Bar-Yossef}, T.~S. Jayram, Robert Krauthgamer, and Ravi Kumar.
\newblock Approximating edit distance efficiently.
\newblock In {\em Proceedings of the Symposium on Foundations of Computer
  Science (FOCS)}, pages 550--559, 2004.

\bibitem[BR20]{brakensiek2019constant}
Joshua Brakensiek and Aviad Rubinstein.
\newblock Constant-factor approximation of near-linear edit distance in
  near-linear time.
\newblock In {\em Proceedings of the Symposium on Theory of Computing (STOC)},
  2020.
\newblock arXiv preprint arXiv:1904.05390.

\bibitem[Bri14]{bringmann2014sampling}
Karl Bringmann.
\newblock {\em Sampling from discrete distributions and computing Fr{\'e}chet
  distances}.
\newblock PhD thesis, Saarland University, 2014.

\bibitem[CDG{\etalchar{+}}18]{chakraborty2018approximating}
Diptarka Chakraborty, Debarati Das, Elazar Goldenberg, Michal Koucky, and
  Michael Saks.
\newblock Approximating edit distance within constant factor in truly
  sub-quadratic time.
\newblock In {\em 2018 IEEE 59th Annual Symposium on Foundations of Computer
  Science (FOCS)}, pages 979--990. IEEE, 2018.

\bibitem[CDK19]{chakraborty2018online}
Diptarka Chakraborty, Debarati Das, and Michal Koucky.
\newblock Approximate online pattern matching in sub-linear time.
\newblock In {\em FSTTCS}, 2019.

\bibitem[Cha02]{Char}
Moses Charikar.
\newblock Similarity estimation techniques from rounding.
\newblock In {\em Proceedings of the Symposium on Theory of Computing (STOC)},
  pages 380--388, 2002.

\bibitem[Che14]{Chechik:2014:ADO:2591796.2591801}
Shiri Chechik.
\newblock Approximate distance oracles with constant query time.
\newblock In {\em Proceedings of the Forty-sixth Annual ACM Symposium on Theory
  of Computing}, STOC '14, pages 654--663, New York, NY, USA, 2014. ACM.

\bibitem[CLRS01]{CLRS}
Thomas~H. Cormen, Charles~E. Leiserson, Ronald~L. Rivest, and Clifford Stein.
\newblock {\em Introduction to Algorithms}.
\newblock MIT Press, 2nd edition, 2001.

\bibitem[GKKS22]{gkks22}
Elazar Goldenberg, Tomasz Kociumaka, Robert Krauthgamer, and Barna Saha.
\newblock Gap edit distance via non-adaptive queries: Simple and optimal.
\newblock In {\em Proceedings of the Symposium on Foundations of Computer
  Science (FOCS)}, 2022.

\bibitem[GKS19]{goldenberg2019sublinear}
Elazar Goldenberg, Robert Krauthgamer, and Barna Saha.
\newblock Sublinear algorithms for gap edit distance.
\newblock In {\em Proceedings of the Symposium on Foundations of Computer
  Science (FOCS)}, pages 1101--1120. IEEE, 2019.

\bibitem[GRS20]{GRS19}
Elazar Goldenberg, Aviad Rubinstein, and Barna Saha.
\newblock Does preprocessing help in fast sequence comparisons?
\newblock In {\em Proceedings of the Symposium on Theory of Computing (STOC)},
  2020.

\bibitem[HSS19]{hajiaghayi2019massively}
MohammadTaghi Hajiaghayi, Saeed Seddighin, and Xiaorui Sun.
\newblock Massively parallel approximation algorithms for edit distance and
  longest common subsequence.
\newblock In {\em Proceedings of the ACM-SIAM Symposium on Discrete Algorithms
  (SODA)}, pages 1654--1672. SIAM, 2019.

\bibitem[HSSS19]{hajiaghayi2019approximating}
MohammadTaghi Hajiaghayi, Masoud Seddighin, Saeed Seddighin, and Xiaorui Sun.
\newblock Approximating lcs in linear time: Beating the barrier.
\newblock In {\em Proceedings of the ACM-SIAM Symposium on Discrete Algorithms
  (SODA)}, pages 1181--1200. SIAM, 2019.

\bibitem[KS20a]{kociumaka2020sublinear}
Tomasz Kociumaka and Barna Saha.
\newblock Sublinear-time algorithms for computing \& embedding gap edit
  distance.
\newblock In {\em Proceedings of the Symposium on Foundations of Computer
  Science (FOCS)}, pages 1168--1179, 2020.

\bibitem[KS20b]{koucky2019constant}
Michal Kouck{\`y} and Michael~E Saks.
\newblock Constant factor approximations to edit distance on far input pairs in
  nearly linear time.
\newblock In {\em Proceedings of the Symposium on Theory of Computing (STOC)},
  2020.
\newblock arXiv preprint arXiv:1904.05459.

\bibitem[Kus19]{kuszmaul2019efficiently}
William Kuszmaul.
\newblock Efficiently approximating edit distance between pseudorandom strings.
\newblock In {\em Proceedings of the ACM-SIAM Symposium on Discrete Algorithms
  (SODA)}, pages 1165--1180. SIAM, 2019.

\bibitem[LMS98]{LMS98}
Gad~M. Landau, Eugene~W. Myers, and Jeanette~P. Schmidt.
\newblock Incremental string comparison.
\newblock {\em SIAM J. Comput.}, 27(2):557--582, 1998.

\bibitem[Mat96]{Matousek1996}
Ji{\v{r}}{\'i} Matou{\v{s}}ek.
\newblock On the distortion required for embedding finite metric spaces into
  normed spaces.
\newblock {\em Israel Journal of Mathematics}, 93(1):333--344, Dec 1996.

\bibitem[MN07]{MendelN06}
M.~Mendel and A.~Naor.
\newblock Ramsey partitions and proximity data structures.
\newblock {\em Journal of the European Mathematical Society}, 9(2):253--275,
  2007.
\newblock Extended abstract appeared in FOCS 2006.

\bibitem[MP80]{MP80}
William~J. Masek and Mike Paterson.
\newblock A faster algorithm computing string edit distances.
\newblock {\em J. Comput. Syst. Sci.}, 20(1):18--31, 1980.

\bibitem[Mye86]{Myers86}
Eugene~W. Myers.
\newblock An {$O(ND)$} difference algorithm and its variations.
\newblock {\em Algorithmica}, 1(2):251--266, 1986.

\bibitem[Nav01]{Navarro01}
Gonzalo Navarro.
\newblock A guided tour to approximate string matching.
\newblock {\em ACM Comput. Surv.}, 33(1):31--88, 2001.

\bibitem[Nos21]{nosatzki2021approximating}
Negev~Shekel Nosatzki.
\newblock Approximating the longest common subsequence problem within a
  sub-polynomial factor in linear time.
\newblock {\em arXiv preprint arXiv:2112.08454}, 2021.

\bibitem[OR07]{OR-edit}
Rafail Ostrovsky and Yuval Rabani.
\newblock Low distortion embedding for edit distance.
\newblock {\em J. ACM}, 54(5), 2007.
\newblock Preliminary version appeared in STOC'05.

\bibitem[RS20]{rubinstein2020reducing}
Aviad Rubinstein and Zhao Song.
\newblock Reducing approximate longest common subsequence to approximate edit
  distance.
\newblock In {\em Proceedings of the ACM-SIAM Symposium on Discrete Algorithms
  (SODA)}, pages 1591--1600. SIAM, 2020.

\bibitem[RSSS19]{rubinstein2019approximation}
Aviad Rubinstein, Saeed Seddighin, Zhao Song, and Xiaorui Sun.
\newblock Approximation algorithms for lcs and lis with truly improved running
  times.
\newblock In {\em Proceedings of the Symposium on Foundations of Computer
  Science (FOCS)}, pages 1121--1145. IEEE, 2019.

\bibitem[Sah08]{Sah-Encyclopedia}
S{\"u}leyman~Cenk Sahinalp.
\newblock Edit distance under block operations.
\newblock In Ming-Yang Kao, editor, {\em Encyclopedia of Algorithms}. Springer,
  2008.

\bibitem[Tis08]{tiskin2008semi}
Alexander Tiskin.
\newblock Semi-local string comparison: Algorithmic techniques and
  applications.
\newblock {\em Mathematics in Computer Science}, 1(4):571--603, 2008.

\bibitem[TZ05]{Thorup:2005:ADO:1044731.1044732}
Mikkel Thorup and Uri Zwick.
\newblock Approximate distance oracles.
\newblock {\em J. ACM}, 52(1):1--24, January 2005.

\bibitem[Ukk85]{ukkonen1985algorithms}
Esko Ukkonen.
\newblock Algorithms for approximate string matching.
\newblock {\em Information and control}, 64(1-3):100--118, 1985.

\end{thebibliography}
